\documentclass[3p,10pt]{elsarticle}

\usepackage[T1]{fontenc}
\usepackage{lmodern}
\usepackage{microtype}
\usepackage{geometry}
\makeatletter
\let\ps@pprintTitle\ps@plain
\makeatother
\usepackage{amsmath,amssymb,amsthm,mathtools,mathrsfs}
\newtheorem{theorem}{Theorem}[section]
\newtheorem*{theorem*}{Theorem}
\newtheorem{lemma}{Lemma}[section]
\newtheorem*{lemma*}{Lemma}
\newtheorem{corollary}{Corollary}[section]
\theoremstyle{definition}
\newtheorem{definition}{Definition}[section]

\theoremstyle{remark}
\newtheorem{remark}{Remark}[section]

\usepackage{graphicx}
\usepackage{tikz}
\usetikzlibrary{arrows,backgrounds,calc,fit,decorations.pathreplacing,
  decorations.markings,shapes.geometric,positioning}
\tikzset{
  every fit/.append style=text badly centered,
  internal/.style={draw,fill,shape=circle},
  external/.style={shape=circle},
  square/.style={draw,fill,rectangle},
  triangle/.style={draw,fill,regular polygon,regular polygon sides=3,inner sep=3pt},
  pentagon/.style={draw,fill,regular polygon,regular polygon sides=5,
    inner sep=2pt,minimum size=14pt}
}
\def\borderColor{black!60}

\usepackage{chngcntr}
\counterwithin{figure}{section}
\usepackage{longtable,booktabs,array,multirow}
\usepackage[font=small,labelfont=bf,justification=centering]{caption}
\newcommand{\BI}{BI_{120}}
\newcommand{\yes}{\textsf{Yes}}
\newcommand{\no}{\textsf{No}}

\usepackage{enumitem}
\setlist{itemsep=2pt,partopsep=0pt,parsep=0pt,topsep=4pt}
\setlist[itemize]{leftmargin=*}
\setlist[enumerate]{leftmargin=*}
\newlist{caselist}{description}{4}
\setlist[caselist]{style=unboxed,leftmargin=0pt,labelindent=0pt,
  font=\normalfont\bfseries,labelsep=.5em,topsep=6pt,itemsep=6pt,parsep=4pt}

\usepackage{xcolor}
\newcommand{\caseitem}[1]{\item[Case #1.]}
\newcommand{\subcaseitem}[1]{\item[Subcase #1.]}
\newcommand{\caseheading}[1]{\par\addvspace{6pt}\noindent
  \textbf{Case #1.}\enspace\ignorespaces}
\newcommand{\proofpart}[1]{\par\addvspace{8pt}\noindent
  \textbf{Part #1.}\enspace\ignorespaces}
\newcommand{\caseref}[1]{Case~#1}
\newcommand{\subcaseref}[1]{Subcase~#1}
\newcommand{\partref}[1]{Part~#1}

\biboptions{numbers,square,sort&compress}
\usepackage{xurl}
\usepackage{hyperref}
\hypersetup{colorlinks=true,allcolors=black,citecolor=blue,pdfborder={0 0 0}}
\AtBeginDocument{\hypersetup{allcolors=black,citecolor=blue,pdfborder={0 0 0}}\urlstyle{same}}
\title{\texorpdfstring{The Computational Complexity of Holant Problems on 4-regular Graphs from the Stable Subgroup Sequence of $SL(2,\mathbb{C})$}{The Computational Complexity of Holant Problems on 4-regular Graphs from the Stable Subgroup Sequence of SL(2,C)}}
\author{Yuan Huang}
\author{Zhiguo Fu}
\date{April 2026}

\begin{document}

\begin{frontmatter}
\begin{abstract}
The Holant framework provides a general setting for studying counting problems and includes graph homomorphisms (\#GH) and counting constraint satisfaction problems (\#CSP) as special cases. Over the past twenty years, a series of computational complexity dichotomies have been established for Holant problems, but the classification for complex-valued signatures is still open.  The main obstacle is the case in which all signatures have even arity.
In this paper, we establish a dichotomy for Holant problems with a complex-valued 4-ary signature, which is a key base case for the full classification of Holant problems.
We present a new strategy by introducing Schur's theorem, the classification of finite subgroups of $\mathrm{SL}(2,\mathbb{C})$ and stable subgroup sequences into the proof. These new techniques are of independent interest.
\end{abstract}

\end{frontmatter}

\section{Introduction}
The goal of classifying the computational complexity of counting problems is to identify which problems are computable in polynomial time and which are \#P-hard within a given framework. Substantial progress has been made in this field. For \#GH and \#CSP, two classical frameworks of counting
problems, the full computational complexity classifications have been established \cite{BulatovCSP,Bulatov2005,Bulatov2008,CSP,CaiChenLu2010,CaiChenLu2016,Dyer2000,Dyer2013}.
Inspired by Valiant's holographic algorithms \cite{Val08,Valiant2002}, the more general Holant framework \cite{CaiLuXia2008,CaiLuXia2009,CaiLuXia2013}
was proposed. Holant problems can express a broad class of counting problems, such as
counting matchings (\#Matchings) and perfect matchings (\#PM). In particular, Holant problems contain \#GH
and \#CSP as special cases.

\#CSP and \#GH have been well studied and the computational complexity dichotomies
have been established. However, the understanding of Holant problems is still limited. Restricted to the Boolean domain, several complexity dichotomies have been established. When all signatures are restricted
to be symmetric, a full dichotomy has been proved \cite{CaiHengWilliams2016}. For asymmetric signatures, however, dichotomies are available only for special families of Holant problems, such as Holant$^*$
\cite{CaiLuXia2011}, Holant$^+$
\cite{backensholant+} and Holant$^c$
\cite{Holantc,CaiLuXia2018},
in which some auxiliary signatures are freely available.
For Holant problems without any auxiliary signatures, a full dichotomy for real-valued Holant problems was achieved \cite{RealHolant}. Currently, the result closest to the full  dichotomy for Holant problems is the dichotomy for the case in which there exists a nontrivial odd-arity signature, established by using the complexity class FP$^{\mathrm{NP}}$ \cite{oddarity}.

When all signatures have even arity, one basic case is computing the partition function of the six-vertex model \cite{Sixvertex}. Moreover, as the generalization of the six-vertex model to higher arities, counting weighted Eulerian orientations (\#EO) has been classified using the complexity class FP$^{\mathrm{NP}}$ \cite{EO}.
In addition, a dichotomy for computing the partition function of the eight-vertex model
has been established \cite{Eightvertex}, which
is the generalization of the six-vertex model to more general 4-ary signatures.

In this paper, we give the complexity dichotomy for Holant problems on the Boolean domain with a 4-ary signature without any restrictions. The key step in the proof of a counting-complexity dichotomy is to introduce some unary signatures to reduce the arities of higher-arity signatures.
If all signatures have even arity, then all the signatures we can construct have even arity. This means the signatures of lowest positive arity we can construct are binary signatures.
In this case, we need to obtain unary signatures by interpolating with binary signatures. However, interpolation fails if the ratio of the eigenvalues of the binary signature is a root of unity. This is a major obstacle in previous proofs.  In this paper, we give a
new strategy to overcome this difficulty.
More precisely, the matrix form of a binary signature is a $2\times 2$ matrix. Up to a scalar, if the matrix form has full rank, then the binary signature is in $\mathrm{SL}(2,\mathbb{C})$.
If we have a binary signature of infinite order, then we can obtain a unary signature by interpolating directly.
If all the binary signatures are of finite order,
then any subgroup generated by finitely many binary signatures is finite by Schur's theorem.
Moreover, by the classification of the finite subgroups of $\mathrm{SL}(2,\mathbb{C})$, the subgroup generated by finitely many binary signatures is conjugate to one of $C_n$, $BD_{4n}$, $BT_{24}$, $BO_{48}$ and $BI_{120}$,  which are the finite subgroups of $\mathrm{SL}(2,\mathbb{C})$.
Our strategy is to construct a subgroup sequence.
If the sequence is unstable, then we can obtain a binary signature whose order is large enough for interpolation. Otherwise, we can derive a system of linear equations that yields a signature of eight-vertex form to complete the proof by Lemma \ref{haveeightvertex}. The normalizer and the transpose closure of the five conjugate classes play a key role in our proof.



The paper is organized as follows. In Section 2, we give the background and notation used in this paper. In Section 3, we summarize the cases that can be reduced to \#CSP$^2$. \#CSP$^2$ is a variant of \#CSP in which each variable appears an even number of times. In
Section 4, we prove the dichotomy of Holant problems with a 4-ary signature and a degenerate binary signature. In Section 5, we prove the dichotomy of Holant problems with a 4-ary signature of eight-vertex form. In Section 6, we introduce the subgroup sequence and prove the main theorem of this paper.
\section{Preliminaries and notation}

In this paper, we use $\frak{i}$ to denote the imaginary unit, i.e., $\frak{i}^2=-1$.

For each positive integer $n$, we write $[n]=\{1,\ldots,n\}$. When applied to a signature family, $\langle\mathcal{F}\rangle$ denotes its closure under tensor products and permutations of inputs; for matrices, angle brackets denote the generated subgroup.

A signature $f$ of arity $k$ is
a map $\{0,1\}^{k} \rightarrow \mathbb{C}$.
Fix a set $\mathcal{F}$ of signatures.
A signature grid $\Omega =(G,\pi)$ is a tuple, where
$G=(V,E)$ is a graph, $\pi$ labels each $v\in V$ with
a signature $f_{v}\in \mathcal{F}$ of arity deg$(v)$,
and the incident edges $E(v)$ at $v$ with input variables of $f_{v}$.
We consider all $0\mbox{-}1$ edge assignments $\sigma$, each of which gives an evaluation
$\prod_{v\in V}f_{v}(\sigma |_{E(v)})$, where $\sigma |_{E(v)}$ denotes the
restriction of $\sigma$ to $E(v)$. The counting problem on the instance $\Omega$
is to compute
\[
{\rm Holant}_{\Omega}={\rm Holant}(\Omega;\mathcal{F})=\sum_{\sigma:E \rightarrow\{0,1\}} \prod_{v\in V}f_{v}(\sigma|_{E(v)}).
\]
The Holant problem parameterized by the set $\mathcal{F}$ is denoted by Holant$(\mathcal{F})$.
We use Holant$(\mathcal{F}|\mathcal{G})$ to denote the Holant problem
over signature grids
with a bipartite graph $H=(U,V,E)$ where each vertex in $U$ or $V$ is assigned
a signature
in $\mathcal{F}$ or $\mathcal{G}$ respectively.

For a binary signature $h$, we define its matrix form as a $2\times 2$ matrix $M_{x_i,x_j}(h)$, where the rows are indexed by $x_i$, and the columns are indexed by $x_j$. For example, \[
M_{x_1,x_2}(h)=\begin{pmatrix}
 h_{00} & h_{01}\\
 h_{10} & h_{11}
\end{pmatrix}, M_{x_2,x_1}(h)=\begin{pmatrix}
 h_{00} & h_{10}\\
 h_{01} & h_{11}
\end{pmatrix}.
\]  When the input order is clear, we abbreviate this matrix as $M(h)$.
Multiplying a fixed signature by a nonzero scalar does not change the complexity of its Holant problem: the resulting instance value differs by an explicitly known product of scalars. An invertible binary matrix $M$ has determinant-one representatives $\pm M/\sqrt{\det M}$. Whenever group order is discussed, we use such a representative, rather than the unnormalized matrix. Finiteness of its order is independent of the choice of sign and is equivalent to finite projective order. A nonzero rank-one matrix is a nonzero degenerate binary signature; the zero signature must be treated separately.


For a 4-ary signature $f$, we define its matrix form as a $4\times 4$ matrix $M_{x_i,x_j,x_k,x_l}(f)$, where the rows are indexed by $x_i,x_j$, and the columns are indexed by $x_k,x_l$. For example,\[
M_{x_1,x_2,x_3,x_4}(f)=\begin{pmatrix}
  f_{0000}&  f_{0001}&  f_{0010}& f_{0011}\\
  {f_{0100}} &  {f_{0101}}& {f_{0110}}& {f_{0111}}\\
 {f_{1000}} &  {f_{1001}}& {f_{1010}}& {f_{1011}}\\
  {f_{1100}} & {f_{1101}}&  {f_{1110}}& {f_{1111}}
\end{pmatrix}, M_{x_3,x_4,x_1,x_2}(f)=\begin{pmatrix}
  f_{0000}&  f_{0100}&  f_{1000}& f_{1100}\\
  {f_{0001}} &  {f_{0101}}& {f_{1001}}& {f_{1101}}\\
 {f_{0010}} &  {f_{0110}}& {f_{1010}}& {f_{1110}}\\
  {f_{0011}} & {f_{0111}}&  {f_{1011}}& {f_{1111}}
\end{pmatrix}.
\]
When the input order is clear from the context, we use $M(f)$ to denote the matrix form of $f$. We say $f$ has eight-vertex form if $f_{0001}=f_{0010}=f_{0100}=f_{1000}=f_{1110}=f_{1101}=f_{1011}=f_{0111}=0$.

\begin{definition}
A signature is called degenerate if it is
a tensor product of unary signatures. Any signature that cannot be decomposed into
a tensor product of unary signatures is called nondegenerate.
\end{definition}


A signature is symmetric if its value depends only on the
Hamming weight of its input.
A  symmetric signature $f$ on $k$ Boolean
variables can be expressed as $[f_{0},f_{1},...,f_{k}]$, where $f_{w}$ is
the value of $f$ on inputs of Hamming weight $w$. For example,
$(=_{k})$ is the \textsc{Equality} signature $[1,0,...,0,1]$ (with $k-1$ zeros) of
arity $k$; $(\neq_{2})=[0, 1, 0]$ is the \textsc{Disequality} signature. The support of a signature $f$ is the set of inputs on which $f$ is nonzero.
Let $\mathcal{EQ}=\{=_{1},=_{2},...,=_{n},...\}$ denote the set of all
\textsc{Equality} signatures.
As the following lemma shows, \#CSP can be expressed as Holant problems.
\begin{lemma}
(\cite{Mainbook}) \#{\rm CSP}$(\mathcal{F}) \equiv_{T}{\rm Holant}(\mathcal{EQ}|\mathcal{F})$.
\end{lemma}
Moreover, let $\mathcal{EQ}_{2}=\{=_{2},=_{4},...,=_{2n},...\}$ denote
the set of all even \textsc{Equality} signatures. We define a special case of \#CSP
\begin{center}
    \#{\rm CSP}$^2(\mathcal{F}) \equiv_{T}{\rm Holant}(\mathcal{EQ}_2|\mathcal{F})$,
\end{center}
 i.e., every variable appears an even number of times.


\subsection{Results from group theory}
\begin{theorem}\label{classification}(\cite{finitegroups})
Let $G$ be a finite subgroup of $\mathrm{SL}(2,\mathbb{C})$. Then $G$ is one of the following groups (up to conjugacy):
\begin{enumerate}
    \item a cyclic group $C_n$;
    \item a binary dihedral group, of the form $BD_{4n}$, with $n\in\mathbb{N}$;
    \item a binary group corresponding to one of the Platonic solids, that is $BT_{24}$, $BO_{48}$ or $BI_{120}$.
\end{enumerate}
\end{theorem}
\begin{lemma}(\cite{finitegroups})
    \begin{itemize}
        \item $C_n$ is generated by $\begin{pmatrix}
\varepsilon  &0 \\
0  & \varepsilon^{-1}
\end{pmatrix}$, where $\varepsilon$ is a primitive $n$th root of unity, $C_n=\left \{\begin{pmatrix}
\varepsilon  &0 \\
0  & \varepsilon^{-1}
\end{pmatrix}^k \Bigg |k\in[n]\right \}$;
\item $BD_{4n}$ is generated by  $A=\begin{pmatrix}
\varepsilon  &0 \\
0  & \varepsilon^{-1}
\end{pmatrix}$ and $B=\begin{pmatrix}
0  &1 \\
-1  & 0
\end{pmatrix}$, where $\varepsilon$ is a primitive $2n$th root of unity, $BD_{4n}=\left \{A^k,A^kB \Bigg |k\in[2n]\right \}$;

\item $BT_{24}$ is generated by $\begin{pmatrix}
\frak{i}  &0 \\
0  & -\frak{i} 
\end{pmatrix}$, $\begin{pmatrix}
0  &1 \\
-1  & 0
\end{pmatrix}$, and $\frac{1}{2}\begin{pmatrix}
1+\frak{i}   &-1+\frak{i}  \\
1+\frak{i}   & 1-\frak{i} 
\end{pmatrix}$.
\item $BO_{48}$ is generated by $\frac{1}{\sqrt{2}}\begin{pmatrix}
1+\frak{i}  &0 \\
0  & 1-\frak{i} 
\end{pmatrix}$, $\begin{pmatrix}
0  &1 \\
-1  & 0
\end{pmatrix}$, and $\frac{1}{2}\begin{pmatrix}
1+\frak{i}   &-1+\frak{i}  \\
1+\frak{i}   & 1-\frak{i} 
\end{pmatrix}$.
\item $BI_{120}$ is generated by $\begin{pmatrix}
\omega^3  &0 \\
0  & \omega^2 
\end{pmatrix}$, $\begin{pmatrix}
0  &1 \\
-1  & 0
\end{pmatrix}$, and $\frac{1}{\sqrt{5}}\begin{pmatrix}
\omega^4-\omega   &\omega^2-\omega^3 \\
\omega^2-\omega^3   & \omega-\omega^4 
\end{pmatrix}$, where $\omega$ is a primitive fifth root of unity.
    \end{itemize}

\end{lemma}
\begin{definition}
    Let $G$ be a group. An isomorphism from $G$ onto itself is called an automorphism of $G$. The set of all automorphisms of $G$ is denoted by $Aut(G)$.
\end{definition}
Note that if $G$ is finite, then $Aut(G)$ is also finite.
\begin{definition}
    The normalizer of a subgroup $H$ in the group $G$ is $\mathcal{N}_{G}(H)=\{g|g\in G, H=gHg^{-1}\}$.
\end{definition}

\begin{definition}
    The centralizer of a subgroup $H$ in the group $G$ is $\mathcal{C}_{G}(H)=\{g|g\in G, h=ghg^{-1}\mathrm{\ for\ all\ }h\in H\}$.
\end{definition}
\begin{corollary}\label{NClemma}(\cite{Algebradummit})
    For any subgroup $H$ of a group $G$, the quotient group
    $\mathcal{N}_{G}(H)/\mathcal{C}_{G}(H)$ is isomorphic to a subgroup of $Aut(H)$.
\end{corollary}

\begin{lemma}\label{normalizer}
The normalizers of $C_n$, $BD_{4n}$, $BT_{24}$, $BO_{48}$ and $BI_{120}$ in $\mathrm{SL}(2,\mathbb{C})$ are
\begin{itemize}
    \item $\mathcal{N}_{\mathrm{SL}(2,\mathbb{C})}(C_1)=\mathcal{N}_{\mathrm{SL}(2,\mathbb{C})}(C_2)=\mathrm{SL}(2,\mathbb{C})$;
    \item $\mathcal{N}_{\mathrm{SL}(2,\mathbb{C})}(C_n)= \left \{\begin{pmatrix}
  a&0 \\
  0&a^{-1}
\end{pmatrix},\begin{pmatrix}
  0&b \\
  -b^{-1}&0
\end{pmatrix}\Bigg|a,b\in \mathbb{C}\setminus\{0\} \right \}$, where $n\ge3$;
\item $\mathcal{N}_{\mathrm{SL}(2,\mathbb{C})}(BD_{4n})=BD_{8n}$, where $n\ge3$;
\item $\mathcal{N}_{\mathrm{SL}(2,\mathbb{C})}(BD_{8})=\mathcal{N}_{\mathrm{SL}(2,\mathbb{C})}(BT_{24})=\mathcal{N}_{\mathrm{SL}(2,\mathbb{C})}(BO_{48})=BO_{48}$;
\item $\mathcal{N}_{\mathrm{SL}(2,\mathbb{C})}(BI_{120})=BI_{120}$.
\end{itemize}
\end{lemma}
\begin{proof}
    The proof of $\mathcal{N}_{\mathrm{SL}(2,\mathbb{C})}(C_1)=\mathcal{N}_{\mathrm{SL}(2,\mathbb{C})}(C_2)=\mathrm{SL}(2,\mathbb{C})$ is trivial.

    Note that if $G$ is a finite subgroup of $\mathrm{SL}(2,\mathbb{C})$ generated by a generating set $\mathcal{G}$, and if $S\in \mathrm{SL}(2,\mathbb{C})$ such that $S\mathcal{G}S^{-1}\subseteq G$, then $SGS^{-1}=G$.

    For $C_n$, where $n\ge3$, its generator is $\begin{pmatrix}
 \varepsilon &0 \\
 0 &\varepsilon^{-1}
\end{pmatrix}$, where $\varepsilon\ne 1$. Note that for any
$\begin{pmatrix}
 a & b\\
 c &d
\end{pmatrix}\in \mathcal{N}_{\mathrm{SL}(2,\mathbb{C})}(C_n)$, we have
$\begin{pmatrix}
 a & b\\
 c &d
\end{pmatrix}\begin{pmatrix}
 \varepsilon & 0\\
 0 &\varepsilon^{-1}
\end{pmatrix}\begin{pmatrix}
 d & -b\\
 -c &a
\end{pmatrix}=\begin{pmatrix}
 ad\varepsilon-bc\varepsilon^{-1} & ab(\varepsilon^{-1}-\varepsilon)\\
 cd(\varepsilon-\varepsilon^{-1}) &ad\varepsilon^{-1}-bc\varepsilon
\end{pmatrix}\in C_n$, and $\varepsilon^{-1}-\varepsilon\ne0$ since $\varepsilon\ne1$. Clearly $\begin{pmatrix}
 ad\varepsilon-bc\varepsilon^{-1} & ab(\varepsilon^{-1}-\varepsilon)\\
 cd(\varepsilon-\varepsilon^{-1}) &ad\varepsilon^{-1}-bc\varepsilon
\end{pmatrix}\in C_n$  if and only if $ab=0$ and $cd=0$. Then we have $\mathcal{N}_{\mathrm{SL}(2,\mathbb{C})}(C_n)=
 \left \{ \begin{pmatrix}
  a&0 \\
  0&a^{-1}
\end{pmatrix},\begin{pmatrix}
  0&b \\
  -b^{-1}&0
\end{pmatrix}\Bigg|a,b\in \mathbb{C}\setminus\{0\} \right \}$, where $n\ge3$

For $BD_{4n}$, where $n\ge3$, its generators are $A=\begin{pmatrix}
 \varepsilon &0 \\
 0 &\varepsilon^{-1}
\end{pmatrix}$ and $B=\begin{pmatrix}
 0 &1 \\
 -1 &0
\end{pmatrix}$, where $\varepsilon$ is a primitive $2n$th root of unity. Moreover, $BD_{4n}=\left \{A^k, A^kB  | k\in [2n] \right\}$. We have $\varepsilon\ne \pm\frak{i}$ since $n\ge3$.  Note that the traces of two similar matrices are necessarily equal and the trace of $A^kB$ is zero for any $k$.

Thus for any $g\in\mathcal{N}_{\mathrm{SL}(2,\mathbb{C})}(BD_{4n})$, $gAg^{-1}\in \left \{A^k| k\in[2n] \right\}\subseteq C_{2n}$. This implies that $\mathcal{N}_{\mathrm{SL}(2,\mathbb{C})}(BD_{4n})\subseteq \mathcal{N}_{\mathrm{SL}(2,\mathbb{C})}(C_{2n})$. Further, if
$g=\begin{pmatrix}
 a & 0\\
 0 &a^{-1}
\end{pmatrix}$, then
$gA^kBg^{-1}=\begin{pmatrix}
 a & 0\\
 0 &a^{-1}
\end{pmatrix}\begin{pmatrix}
 0 & \varepsilon^k\\
 -\varepsilon^{-k} &0
\end{pmatrix}\begin{pmatrix}
 a^{-1} & 0\\
 0 &a
\end{pmatrix}=\begin{pmatrix}
 0 & a^2\varepsilon^k\\
 -a^{-2}\varepsilon^{-k} &0
\end{pmatrix}\in BD_{4n}$ for $k\in[2n]$. This implies that $a^{4n}=1$, i.e. $a$ is a $4n$th root of unity. Similarly, if $g=\begin{pmatrix}
 0 & b\\
 -b^{-1} &0
\end{pmatrix}$, then $gA^kBg^{-1}=\begin{pmatrix}
 0 & b^2\varepsilon^k\\
 -b^{-2}\varepsilon^{-k} &0
\end{pmatrix}\in BD_{4n}$ for $k\in[2n]$. This implies that $b^{4n}=1$ . Then we have $\mathcal{N}_{\mathrm{SL}(2,\mathbb{C})}(BD_{4n})= BD_{8n}$, where $n\ge3$.

Note that by Corollary~\ref{NClemma}, if $H$ is a finite subgroup of $G$ and $\mathcal{C}_{G}(H)$ is finite, then $\mathcal{N}_{G}(H)$ is also finite.
For each $G\in\{BD_8,BT_{24},BO_{48},BI_{120}\}$, direct calculation gives $\mathcal C_{\mathrm{SL}(2,\mathbb C)}(G)=\{I,-I\}$, so these centralizers are finite. Thus $\mathcal{N}_{\mathrm{SL}(2,\mathbb{C})}(BD_{8}),\allowbreak\mathcal{N}_{\mathrm{SL}(2,\mathbb{C})}(BT_{24}),\allowbreak\mathcal{N}_{\mathrm{SL}(2,\mathbb{C})}(BO_{48}),\allowbreak\mathcal{N}_{\mathrm{SL}(2,\mathbb{C})}(BI_{120})$ are also finite. Further, it is easy to verify that all the generators of $BO_{48}$ are in $\mathcal{N}_{\mathrm{SL}(2,\mathbb{C})}(BD_{8}),\allowbreak\mathcal{N}_{\mathrm{SL}(2,\mathbb{C})}(BT_{24}), \mathcal{N}_{\mathrm{SL}(2,\mathbb{C})}(BO_{48})$, and $BI_{120}\subseteq \mathcal{N}_{\mathrm{SL}(2,\mathbb{C})}(BI_{120})$. Note that there does not exist a finite subgroup $H$ of $\mathrm{SL}(2,\mathbb{C})$ such that $BO_{48}\subset H$ or $BI_{120}\subset H$. Thus we have $\mathcal{N}_{\mathrm{SL}(2,\mathbb{C})}(BD_{8})=\mathcal{N}_{\mathrm{SL}(2,\mathbb{C})}(BT_{24})=\mathcal{N}_{\mathrm{SL}(2,\mathbb{C})}(BO_{48})=BO_{48}$ and $\mathcal{N}_{\mathrm{SL}(2,\mathbb{C})}(BI_{120})=BI_{120}$.
\end{proof}
\begin{definition}
For a subgroup $G$ of  $\mathrm{SL}(2,\mathbb{C})$, let $G^T=\{g^T|g\in G\}$.
     $G$ is closed under transposition if and only if $G=G^T$.
\end{definition}

\begin{lemma}\label{norma}
For two subgroups $G, H$ of $\mathrm{SL}(2,\mathbb{C})$,
if $H=PGP^{-1}$, where $P\in \mathrm{SL}(2,\mathbb{C})$,  and both $G$  and $H$
    are closed under transposition,
    then $P^{T}P\in \mathcal{N}_{\mathrm{SL}(2,\mathbb{C})}(G)$.
\end{lemma}
\begin{proof}
    Note that $G=G^T$, $H=H^T$ and $H^T=(P^{-1})^TG^TP^T=H$.
    We have $H^T=(P^{-1})^TG^TP^T=(P^{-1})^TGP^T$ by $G=G^T$. Further, we have $(P^{-1})^TGP^T=PGP^{-1}$ by $H=H^T$.
    Thus $G=(P^TP)^{-1}G(P^TP)$, i.e., $P^{T}P\in \mathcal{N}_{\mathrm{SL}(2,\mathbb{C})}(G)$.
\end{proof}

\begin{theorem}\label{Schur}(Schur's theorem \cite{Algebra})
A group is said to be periodic if each of its elements has finite order.
Every finitely generated periodic subgroup $G$ of $\mathrm{SL}(2,\mathbb{C})$ is finite.
\end{theorem}

\subsection{Gadgets, holographic transformations, and polynomial interpolation}

\subsubsection{Gadget construction}

One basic tool throughout the paper is gadget construction. We say a signature $f$ is realizable or constructible
from a signature set $\mathcal{F}$ if there is a gadget with some dangling edges
such that each vertex is assigned a signature from $\mathcal{F}$, and the
resulting graph, when viewed as a black-box signature with inputs on the
dangling edges, is exactly $f$. If $f$ is realizable from a set $\mathcal{F}$,
then we can freely add $f$ into $\mathcal{F}$ while preserving the
complexity. This notion is defined by an $\mathcal{F}$-gate.
An $\mathcal{F}$-gate is similar to a signature grid $(G,\pi)$ for Holant$(\mathcal{F})$
except that $G=(V,E,D)$ is a graph with internal edges $E$ and dangling edges $D$.
The dangling edges $D$ define input variables for the $\mathcal{F}$-gate. We denote
the regular edges in $E$ by $1,2,...,m$ and the dangling edges in $D$ by $m+1,...,m+n$.
Then the $\mathcal{F}$-gate defines a signature $f$
\[
f(y_{1},...,y_{n})=\sum_{\sigma:E\rightarrow \{0,1\}}\prod_{v\in V}f_{v}(\hat{\sigma}|_{E(v)}),
\]
where $(y_{1},...,y_{n})\in \{0,1\}^{n}$ is an assignment on the
dangling edges, $\hat{\sigma}$ is the extension of $\sigma$ on $E$
by the assignment $(y_{1},...,y_{n})$, and $f_{v}$ is the signature assigned
at each vertex $v\in V$ (see Figure~\ref{fig:Fgate}). This signature $f$ is called the signature of the
$\mathcal{F}$-gate.

{\tiny
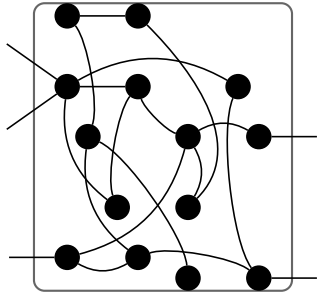
\begin{figure}[tbp]
 \centering
 \begin{tikzpicture}[scale=0.95,transform shape,node distance=28,semithick]
  \node[external]  (0)                     {};
  \node[internal]  (1) [below right of=0]  {};
  \node[external]  (2) [below left  of=1]  {};
  \node[internal]  (3) [above       of=1]  {};
  \node[internal]  (4) [right       of=3]  {};
  \node[internal]  (5) [below       of=4]  {};
  \node[internal]  (6) [below right of=5]  {};
  \node[internal]  (7) [right       of=6]  {};
  \node[internal]  (8) [below       of=6]  {};
  \node[internal]  (9) [below       of=8]  {};
  \node[internal] (10) [right       of=9]  {};
  \node[internal] (11) [above right of=6]  {};
  \node[internal] (12) [below left  of=8]  {};
  \node[internal] (13) [left        of=8]  {};
  \node[internal] (14) [below left  of=13] {};
  \node[external] (15) [left        of=14] {};
  \node[internal] (16) [below left  of=5]  {};
  \path let
         \p1 = (15),
         \p2 = (0)
        in
         node[external] (17) at (\x1, \y2) {};
  \path let
         \p1 = (15),
         \p2 = (2)
        in
         node[external] (18) at (\x1, \y2) {};
  \node[external] (19) [right of=7]  {};
  \node[external] (20) [right of=10] {};
  \path (1) edge                             (5)
            edge[bend left]                 (11)
            edge[bend right]                (13)
            edge node[near start] (e1) {}   (17)
            edge node[near start] (e2) {}   (18)
        (3) edge                             (4)
        (4) edge[out=-45,in=45]              (8)
        (5) edge[bend right, looseness=0.5] (13)
            edge[bend right, looseness=0.5]  (6)
        (6) edge[bend left]                  (8)
            edge[bend left]                  (7)
            edge[bend left]                 (14)
        (7) edge node[near start] (e3) {}   (19)
       (10) edge[bend right, looseness=0.5] (12)
            edge[bend left,  looseness=0.5] (11)
            edge node[near start] (e4) {}   (20)
       (12) edge[bend left]                 (16)
       (14) edge node[near start] (e5) {}   (15)
            edge[bend right]                (12)
       (16) edge[bend left,  looseness=0.5]  (9)
            edge[bend right, looseness=0.5]  (3);
  \begin{pgfonlayer}{background}
   \node[draw=\borderColor,thick,rounded corners,fit = (3) (4) (9) (e1) (e2) (e3) (e4) (e5),inner sep=0pt,transform shape=false] {};
  \end{pgfonlayer}
 \end{tikzpicture}
\caption{ An $\mathcal{F}$-gate with 5 dangling edges.}
 \label{fig:Fgate}
\end{figure}
}

\begin{lemma}
    (\cite{Mainbook}) If $f$ is realizable from a set $\mathcal{F}$, then $\mathrm{Holant}(f, \mathcal{F}) \equiv_T\mathrm{Holant}(\mathcal{F})$.
\end{lemma}

The following gadgets will be used repeatedly in this paper.

\paragraph{Chains}

For two binary signatures $g,h$, by connecting input $x_2$ of $g$ to input $x_1$ of $h$, we construct a binary signature with the signature matrix $M(g)M(h)$.

For two 4-ary signatures $f_1,f_2$, with the signature matrix $M_{x_i,x_j,x_k,x_l}(f_1)$ and $M_{x_s,x_t,x_v,x_u}(f_2)$, where ($i,j,k,l$) and $(s,t,v,u)$ are permutations of $(1, 2, 3, 4)$, by connecting $x_k$ with $x_s$, $x_l$ with $x_t$, we construct a 4-ary signature with signature matrix $M_{x_i,x_j,x_k,x_l}(f_1)M_{x_s,x_t,x_v,x_u}(f_2)$ by matrix product with row index $x_i,x_j$ and column index $x_v,x_u$.

\paragraph{Loops and pinning}

Let $f$ be a signature of arity $k\ge3$, $h$ be a binary signature, and $u$ be a unary signature.
We use $f_{h}^{ij}$ to denote the signature constructed by connecting variables $x_i$,$x_j$ of $f$ with, respectively, variables $x_1$, $x_2$ of $h$ using $=_2$. $S(f_h)$ denotes the set of all signatures of arity $k-2$ we can construct by adding a loop to $f$ using $h$, i.e. \[S(f_h)=\{f_{h}^{ij}|1\le i,j\le k, i\ne j\}.\]
We use $f_u^{i}$ to denote the signature constructed by connecting input $x_i$ of $f$ with $u$ using $=_2$. Moreover, \[S(f_u)=\{f_u^i|1\le i \le k\}.\]
Moreover, for a binary signature set $\mathcal{B}$, \[S(f_\mathcal{B}^{ij})=\{f_h^{ij}|h\in \mathcal{B} \},\]
\[S(f_\mathcal{B})=\bigcup_{h\in\mathcal{B}}S(f_h).\]

\begin{lemma}\label{solveeightvertexform}
    Let $f$ be a 4-ary signature. If $S(f_{\ne_2})\subseteq\{(0,b,c,0)|b,c\in\mathbb{C}\}$, then $f$ has eight-vertex form.
\end{lemma}
\begin{proof}
     We have
\[
\left\{\begin{matrix}
 f_{0100}+f_{1000}=f_{0111}+f_{1011}=0\\
 f_{0001}+f_{0010}=f_{1101}+f_{1110}=0\\
 f_{0010}+f_{1000}=f_{0111}+f_{1101}=0\\
 f_{0001}+f_{0100}=f_{1011}+f_{1110}=0\\
 f_{0001}+f_{1000}=f_{0111}+f_{1110}=0\\
 f_{0010}+f_{0100}=f_{1011}+f_{1101}=0
\end{matrix}\right.\ ,
\]which implies that $f$ has eight-vertex form.
\end{proof}
\paragraph{Binary modifications}

A binary modification to input $x_i$ of $f$ using the binary signature $h$ means connecting input $x_i$ of $f$ to input $x_2$ of $h$.

Thus, applying matrices $L$ and $R$ to the first and second inputs of a binary signature with matrix $M$ gives $LMR^T$. In particular, conjugation by $L$ uses $L$ and $(L^{-1})^T$ on the two inputs.

Permuting inputs and applying invertible matrices separately to them preserve membership in $\langle\mathcal T\rangle$: apply each matrix within its tensor factor, and apply the inverses for the converse.

\subsubsection{Holographic transformations}

For a general graph, we can always transform it into a bipartite graph
while preserving the Holant value. For each edge in the graph, we replace
it by a path of length two. Each new vertex is assigned the binary
\textsc{Equality} signature $(=_{2})$. Hence, we have Holant$(\mathcal{F})
\equiv_{T}$Holant$(=_{2}|\mathcal{F})$.

For a matrix $T\in {\rm GL}(2,\mathbb{C})$ and a signature set $\mathcal{F}$,
{we write $Tf=T^{\otimes n}f$ as the transformed signature. Then we}
define $T \mathcal{F} = \{Tf \mid f\in\mathcal{F}\}$,
and similarly for $\mathcal{F} T$.
Whenever we write $T^{\otimes n} f$ or $T \mathcal{F}$,
we view the signatures as column vectors;
similarly for $f T^{\otimes n} $ or $\mathcal{F} T$ as row vectors.
Then we introduce holographic transformations for bipartite Holant problems Holant($\mathcal{F}|\mathcal{G}$).
The holographic transformation defined by $T$ is the following operation:
Given a signature grid $\Omega = (H, \pi)$ of Holant$({\mathcal{F}}|{\mathcal{G}})$,
for the same bipartite graph $H$,
we obtain a new grid $\Omega' = (H, \pi')$ of Holant$({{ \mathcal{F}T}|T^{-1}\mathcal{G} })$ by replacing each signature in
$\mathcal{F}$ or $\mathcal{G}$ with the corresponding signature in $\mathcal{F} T$  or $T^{-1} \mathcal{G}$ .
The following theorem shows that
an invertible holographic transformation does not change the complexity of the Holant problem in the bipartite setting.
\begin{theorem}\label{Holtrans}
(\cite{Val08}) For any $T \in  {\rm GL}(2,\mathbb{C})$,
\[
{\rm Holant}_{\Omega}(\mathcal{F}|\mathcal{G}) = {\rm Holant}_{\Omega'}(\mathcal{F}T|T^{-1}\mathcal{G}).
\]
\end{theorem}

For a binary signature $h$ and $h'=P^{\otimes2}h$ where $P\in\mathrm{GL}(2,\mathbb{C})$, we have $M(h')=PM(h)P^T$.

We often use
$Z=\frac{1}{\sqrt{2}}\begin{pmatrix}
1&1\\ \frak{i} & -\frak{i}
\end{pmatrix}$, $H=\frac{1}{\sqrt{2}}\begin{pmatrix}
1&1\\ 1 & -1
\end{pmatrix}$
and
$X=\begin{pmatrix}
0&\frak{i}\\ \frak{i} & 0
\end{pmatrix}$ for holographic transformations in the following.

\subsubsection{Polynomial interpolation}

Polynomial interpolation is a powerful technique to prove \#P-hardness for counting problems. Now we
give some important reductions by interpolation.
\begin{lemma}\label{poly=_2}(\cite{Geneq2})
    Let $h$ be a nondegenerate binary signature. Then for any signature set $\mathcal{F}$, we have
       \[ \mathrm{Holant}(=_2|\mathcal{F},h,=_2)\le_T\mathrm{Holant}(=_2|\mathcal{F},h).\]
\end{lemma}

\begin{lemma}\label{polyatdia}
    Let $h$ be a binary signature with the signature matrix $M(h)=P\begin{pmatrix}
 1 & 0\\
  0&t
\end{pmatrix}P^{-1}$ where $t$ is not a root of unity and $t\ne0$. Then we have
       \[ \mathrm{Holant}(=_2|\mathcal{F},h,h')\le_T\mathrm{Holant}(=_2|\mathcal{F},h),\]
    where $h'$ has the matrix form $M(h')=P\begin{pmatrix}
 1 & 0\\
  0&\lambda
\end{pmatrix}P^{-1}$ for any $\lambda\in\mathbb{C}$.
\end{lemma}
\begin{proof}
We construct a sequence of gadgets $h_{s}$ by a chain of $s$ copies of
$h$.
$h_{s}$ has the signature matrix
$M(h_s)=M(h)^s=P\begin{pmatrix}
 1 & 0\\
  0&t^s
\end{pmatrix}P^{-1}$.

Suppose $h'$ appears $m$ times in an instance
$\Omega$ of $\mathrm{Holant}(=_2|\mathcal{F},h,h')$. We replace each
appearance of $h'$ by a copy of the gadget $h_s$ to obtain an instance
$\Omega_{s}$ of $\mathrm{Holant}(=_2|\mathcal{F},h,h_s)$.
 Since $h_s$ has the signature matrix
$M(h_s)$, we can view our construction of $\Omega_{s}$ as replacing $h_s$ by three signatures with matrices $P$, $\begin{pmatrix}
 1 & 0\\
  0&t^s
\end{pmatrix}$
and $P^{-1}$, respectively.

We divide $\Omega_{s}$ into two parts. One part is represented by $\begin{pmatrix}
 1 & 0\\
  0&t^s
\end{pmatrix}^{\otimes m}$. Here we rewrite it
as a column vector $(1,0,0,t^s)^{\otimes m}$.
The other part is the rest of $\Omega_{s}$(including the signatures corresponding to matrices $P$ and $P^{-1}$) and its signature is
represented by $A$, the signature of the remaining open network, expressed as a row vector.
Then, the Holant value of $\Omega_{s}$
is the dot product $\left \langle A,(1,0,0,t^s)^{\otimes m}\right \rangle$,
which is a summation over 2$m$ bits.
We can stratify all Boolean
assignments of these $2m$ bits having a nonzero evaluation of a term in
Holant$_{\Omega_{s}}$ into the following categories:
\begin{itemize}
    \item There are $i$ copies of $\begin{pmatrix}
 1 & 0\\
  0&t^s
\end{pmatrix}$ receiving inputs 00;

    \item There are $j$  copies of $\begin{pmatrix}
 1 & 0\\
  0&t^s
\end{pmatrix}$ receiving inputs 11,
\end{itemize}
where $i+j=m$.
For any assignment in the category with parameter $(i,j)$, the evaluation
of $\begin{pmatrix}
 1 & 0\\
  0&t^s
\end{pmatrix}^{\otimes m}$ is $t^{sj}$. Let $a_{ij}$ be the summation of
values of the part $A$ over all assignments in the category $(i,j)$.
Note that $a_{ij}$ is independent of the value of $s$ since we view the
gadget $\begin{pmatrix}
 1 & 0\\
  0&t^s
\end{pmatrix}$ as a block. We can denote $a_{ij}$ by $a_{j}$. Then, we rewrite
the dot product summation and obtain
\[
{\rm Holant}_{\Omega_{s}}=\left \langle A,(1,0,0,t^s)^{\otimes m}\right \rangle =\sum_{0\leq j \leq m}a_{j}t^{sj}.
\]
Under this stratification, the Holant value of an instance of Holant$(=_2|\mathcal{F},h,h')$ can be expressed as
\[
{\rm Holant}_{\Omega}=\left \langle A,(1,0,0,\lambda)^{\otimes m}\right \rangle =\sum_{0\leq j \leq m}a_{j}\lambda^j.
\]
Since $t$ is not a root of unity and $t\neq 0$, the system of linear equations has a nonsingular Vandermonde matrix.

By querying the oracle for the values of Holant$_{\Omega_{s}}$,
we can solve for the coefficients $a_{j}$ in polynomial
time and obtain the value of $p(\lambda)=\sum_{0\leq j \leq m}a_{j}\lambda^{j}$
for any $\lambda\in\mathbb{C}$. Therefore, we have
\[ \mathrm{Holant}(=_2|\mathcal{F},h,h')\le_T\mathrm{Holant}(=_2|\mathcal{F},h).\]
\end{proof}
\begin{lemma}\label{polyatnotdia}
    Let $h$ be a binary signature with the signature matrix $M(h)=P\begin{pmatrix}
 t & 1\\
  0&t
\end{pmatrix}P^{-1}$ where $t\ne0$. Then we have
       \[ \mathrm{Holant}(=_2|\mathcal{F},h,h')\le_T\mathrm{Holant}(=_2|\mathcal{F},h),\]
    where $h'$ has the matrix form $M(h')=P\begin{pmatrix}
 \lambda & 1\\
  0&\lambda
\end{pmatrix}P^{-1}$ for any $\lambda\in\mathbb{C}$.
\end{lemma}
\begin{proof}
    We construct a sequence of gadgets $h_{s}$ by a chain of $s$ copies of $h$.
$h_{s}$ has the signature matrix
$M(h_s)=M(h)^s=c_sP\begin{pmatrix}
 \frac{t}{s} & 1\\
  0& \frac{t}{s}
\end{pmatrix}P^{-1}$ where $c_s=s t^{s-1}\ne0$. As in Lemma~\ref{polyatdia}, let $m$ be the number of occurrences of $h'$ in $\Omega$, replace each occurrence by $h_s$ to obtain $\Omega_s$, and let $A$ denote the remaining open network after factoring out $P$ and $P^{-1}$. The stratification of all Boolean
assignments of these $2m$ bits having a nonzero evaluation of a term in
Holant$_{\Omega_{s}}$ uses the following categories:
\begin{itemize}
    \item There are $i$ copies of $\begin{pmatrix}
 \frac{t}{s} & 1\\
  0& \frac{t}{s}
\end{pmatrix}$ receiving inputs 01;

    \item There are $j$  copies of $\begin{pmatrix}
 \frac{t}{s} & 1\\
  0& \frac{t}{s}
\end{pmatrix}$ receiving inputs 00 and 11,
\end{itemize}
where $i+j=m$.

Then, we rewrite
the dot product summation and obtain
\[
c_s^{-m}{\rm Holant}_{\Omega_{s}}=\left \langle A,(\frac{t}{s},1,0,\frac{t}{s})^{\otimes m}\right \rangle =\sum_{0\leq j \leq m}a_{j}(\frac{t}{s})^{j}.
\]
Under this stratification, the Holant value of an instance of Holant$(=_2|\mathcal{F},h,h')$ can be expressed as
\[
{\rm Holant}_{\Omega}=\left \langle A,(\lambda,1,0,\lambda)^{\otimes m}\right \rangle =\sum_{0\leq j \leq m}a_{j}\lambda^j.
\]
The values $t/s$ for $s=1,\ldots,m+1$ are distinct because $t\ne0$, so the resulting Vandermonde matrix is nonsingular.

By querying the oracle for the values of Holant$_{\Omega_{s}}$ and multiplying each value by $c_s^{-m}$,
we can solve for the coefficients $a_{j}$ in polynomial
time and obtain the value of $p(\lambda)=\sum_{0\leq j \leq m}a_{j}\lambda^{j}$
for any $\lambda\in\mathbb{C}$. Therefore, we have
\[ \mathrm{Holant}(=_2|\mathcal{F},h,h')\le_T\mathrm{Holant}(=_2|\mathcal{F},h).\]
\end{proof}

\begin{corollary}\label{polyto-1}
    Let $h$ be a binary signature with full rank. Then we have
    \begin{equation}\label{reducthionfordegen}
        \mathrm{Holant}(=_2|\mathcal{F},h,h^{-1})\le_T\mathrm{Holant}(=_2|\mathcal{F},h),
    \end{equation}where $h^{-1}$ has the signature matrix $M(h^{-1})=M(h)^{-1}$.
\end{corollary}
\begin{proof}
If $M(h)$ has finite projective order $r$, write $M(h)^r=cI$ with $c\ne0$. Then $M(h)^{r-1}=cM(h)^{-1}$, so a chain realizes the inverse up to a known nonzero scalar. For $r=1$, use $h$ itself, which is a scalar multiple of its inverse.

Otherwise, use Jordan normal form. In the diagonalizable case, Lemma~\ref{polyatdia} gives the inverse by taking $\lambda=t^{-1}$. In the nondiagonalizable case, Lemma~\ref{polyatnotdia} with $\lambda=-t$ gives a scalar multiple of the inverse, since $\begin{psmallmatrix}t&1\\0&t\end{psmallmatrix}^{-1}=-t^{-2}\begin{psmallmatrix}-t&1\\0&-t\end{psmallmatrix}$. The known scalar factors are divided out of the queried Holant values.
\end{proof}
\begin{corollary}\label{polytodegen}
    Let $h$ be a nondegenerate binary signature with infinite projective order. Then we have
    \begin{equation}\label{reducthionfordegen-2}
        \mathrm{Holant}(=_2|\mathcal{F},h,h')\le_T\mathrm{Holant}(=_2|\mathcal{F},h),
    \end{equation}where $h'$ is a degenerate binary signature.
\end{corollary}
\begin{lemma}\label{infinitetodegen}
Suppose there is a gadget construction in $\mathrm{Holant}(=_2 |\mathcal{F})$ that can produce a sequence of
polynomially many distinct binary signatures of the form $P\begin{pmatrix}
 1 & 0\\
  0&t
\end{pmatrix}P^{-1}$ with polynomial-size descriptions.
Then we have
    \begin{equation}\label{reducthionfordegen-3}
        \mathrm{Holant}(=_2|\mathcal{F},h')\le_T\mathrm{Holant}(=_2|\mathcal{F}),
    \end{equation}where $h'$ is a degenerate binary signature.
\end{lemma}

\subsection{Tractable signature classes and known dichotomies}

In this subsection, we give known signature classes that
define polynomial-time computable (tractable) counting problems,
and several known dichotomies.

\subsubsection{\texorpdfstring{Product-type signatures $\mathcal{P}$}{Product-type signatures P}}
\begin{definition}\label{product}
A signature is of product type if
it can be expressed as a product of unary signatures, binary
equality signatures ([1,0,1]), and binary disequality signatures ([0,1,0]).
We use $\mathcal{P}$ to denote the set of product-type signatures.
\end{definition}

\subsubsection{\texorpdfstring{Affine signatures $\mathcal{A}$}{Affine signatures A}}
\begin{definition}\label{Affine}
 A signature $f(x_{1},...,x_{n})$ of arity $n$ is affine if
it has the form
\[\lambda \cdot \chi_{AX=0} \cdot \frak{i}^{Q(X)}, \]
where $\lambda\in\mathbb{C}, X=(x_{1},x_{2},...,x_{n},1)$, $A$
is a matrix over $\mathbb{Z}_{2}$, $Q(x_{1},x_{2},...,x_{n})\in \mathbb{Z}_{4}[x_{1},x_{2},...,x_{n}]$ is a
multilinear polynomial with total degree d$(Q)\leq 2$ and
the additional requirement that the coefficients of all
cross terms are even, i.e., $Q$ has the form
\[Q(x_{1},x_{2},...,x_{n})=a_{0}+\sum_{k=1}^{n}a_{k}x_{k}+\sum_{1\leq i<j \leq n}2b_{ij}x_{i}x_{j},\]
and $\chi$ is a 0-1 indicator function such that $\chi _{AX=0}$ is 1 iff $AX=0$. We use $\mathcal{A}$ to denote the
set of all affine signatures.
\end{definition}

The classes $\mathcal{A}$ and $\mathcal{P}$ are known to be the tractable
classes for \#CSP.

\begin{theorem}\label{CSP}
(\cite{CSP}) Let $\mathcal{F}$ be any set of complex-valued signatures in Boolean variables.
Then \#{\rm CSP}$(\mathcal{F})$ is \#P-hard unless $\mathcal{F}\subseteq \mathcal{P}$
or $\mathcal{F}\subseteq \mathcal{A}$, in which case the problem is computable
in polynomial time.
\end{theorem}

\begin{definition}
We say a pair of signature sets $(\mathcal{G}|\mathcal{F})$ is
$\mathscr{C}$-transformable for {\rm Holant}$(\mathcal{G}|\mathcal{F})$
if there exists $T\in \mathrm{SL}(2,\mathbb{C})$ such that $\mathcal{G}T\subseteq \mathscr{C}$ and $T^{-1}\mathcal{F}\subseteq \mathscr{C}$. For $\mathcal{G}=\{(=_2)\}$, we say simply that $\mathcal{F}$ is $\mathscr{C}$-transformable if $(=_2)T\subseteq \mathscr{C}$ and $T^{-1}\mathcal{F}\subseteq \mathscr{C}$.

\end{definition}

Note that if Holant($\mathscr{C}$) is tractable and $(\mathcal{F}|\mathcal{G})$ is $\mathscr{C}$-transformable, then it follows from Theorem~\ref{Holtrans} that Holant($\mathcal{F}|\mathcal{G}$) is tractable by a holographic transformation.
We will define some tractable families of signatures that are expressible under a holographic transformation,
specific to the \#CSP$^{2}$ framework.

Let $T_p=\begin{pmatrix}
  1& 0\\
 0 &p
\end{pmatrix}$, and $\mathcal{A}_{d}^r=\{f|T_\rho^rf\in\mathcal{A},\ \rho\ \text{be a $4d$-th primitive root of unity}\}$, for $r \in [d]$.

\begin{definition}
A signature $f$ is locally affine if for each $\sigma =s_{1}s_{2}...s_{n}\in
\{0,1\}^{n}$ in the support of $f$, $(T_{\alpha^{s_{1}}}\otimes T_{\alpha^{s_{2}}}\otimes ...\otimes T_{\alpha^{s_{n}}})f\in \mathcal{A}$, where $\alpha^2=\frak{i}$.
We use $\mathcal{L}$ to denote the set of locally affine signatures.
\end{definition}

In \cite{CaiLuXia2018}, a dichotomy theorem was proved for \#CSP$^{2}_{c}$ which assumes the presence of
the two pinning signatures in \#CSP$^{2}$. By Lemma 4.3 of \cite{pinning-29} and Lemma 3.1 of \cite{pinning-32},
these pinning signatures can be constructed in \#CSP$^{2}$. We obtain the following dichotomy for \#CSP$^{2}$.

\begin{theorem}\label{CSP2dich}
Let $\mathcal{F}$ be any set of complex-valued signatures in Boolean variables.
Then $\#{\rm CSP}^{2}(\mathcal{F})$ is \#P-hard unless $\mathcal{F}\subseteq \mathcal{P}$, $\mathcal{F}\subseteq \mathcal{A}$, $\mathcal{F}\subseteq \mathcal{A}^{1}_2$ or $\mathcal{F}\subseteq \mathcal{L}$, in which cases the problem
is computable in polynomial time.
\end{theorem}

\begin{theorem}\label{CSPd}\cite{Formquantumentanglementandback}
    Let $\mathcal{F}$ be a set of complex-valued signatures. If $\mathcal{F}\subseteq \mathcal{P}$
or $\mathcal{F} \subseteq \mathcal{A}_{d}^r$
for some $r \in [d]$, then \#$\mathrm{CSP}^d
(\ne_2,\mathcal{F})$ is tractable; otherwise, \#$\mathrm{CSP}^d
(\ne_2,\mathcal{F})$ is \#P-hard.
\end{theorem}

\begin{definition}\label{VanishingDef}
	A set of signatures $\mathcal{F}$ is called vanishing if the value of
	Holant$_{\Omega}(\mathcal{F})$ is zero for every signature grid $\Omega$.
	A signature $f$ is called vanishing if the singleton set $\{f\}$ is vanishing.
\end{definition}

\begin{theorem}\label{Th8.1}
	(\cite{CaiHengWilliams2016})
	Let $\mathcal{F}$ be any set of symmetric, complex-valued signatures in
	Boolean variables. Then {\rm Holant}$(\mathcal{F})$ is \#P-hard unless
	$\mathcal{F}$ satisfies one of the following conditions, in which cases
	the problem is tractable:
	\begin{enumerate}
		
		\item All nondegenerate signatures in $\mathcal{F}$
		are of arity at most 2;
		
		\item $\mathcal{F}$ is $\mathcal{A}$-transformable;
		
		\item $\mathcal{F}$ is $\mathcal{P}$-transformable;

		\item $\mathcal{F}$ is a vanishing signature set.
	\end{enumerate}
\end{theorem}

\subsubsection{\texorpdfstring{Holant\textsuperscript{c}}{Holant c}}

\begin{definition}
We use the following notation.

\begin{itemize}
    \item $\mathcal{T}$ is the set of all unary and binary signatures,
   \item  $\mathcal{E}$ is the set of all signatures that are nonzero only on two inputs $x$ and $\overline{x}$, where $\overline{x}$ denotes the bitwise complement of x,
   \item $\mathcal{M}$ is the set of all signatures that are nonzero only on inputs of Hamming weight at most 1,
\end{itemize}
\end{definition}
\begin{theorem}\label{Hcdich}(\cite{Holantc})
    Let $\mathcal{F}$ be any set of complex-valued functions in Boolean variables. Then  Holant$^c(\mathcal{F})$ is \#P-hard unless:
    \begin{itemize}
        \item $\mathcal{F}\subseteq \left \langle \mathcal{T} \right \rangle $;
        \item $\mathcal{F}\subseteq \left \langle Z\mathcal{M} \right \rangle $ or $\mathcal{F}\subseteq \left \langle ZX\mathcal{M} \right \rangle $;
        \item $\mathcal{F}\cup\{[1,0],[0,1]\}$ is $\mathcal{A}$-transformable;
        \item $\mathcal{F}\cup\{[1,0],[0,1]\}$ is $\mathcal{P}$-transformable;
        \item $\mathcal{F}\subseteq \mathcal{L}$,
    \end{itemize}
     in which cases $\mathrm{Holant}^c(\mathcal{F})$ is computable in polynomial time.
\end{theorem}
\subsubsection{The six-vertex and eight-vertex models}
\begin{theorem}
    \label{six vertex} (\cite{Sixvertex})
    Let $f$ be a 4-ary signature with signature matrix \[M(f)=\begin{pmatrix}
  0& 0 & 0 &b \\
 0 &  c&d  & 0\\
 0 & w &  z& 0\\
  y& 0 &0  &0
\end{pmatrix}.\] Then $\mathrm{Holant}(\ne_2|f)$ is \#P-hard except for the following cases:
    \begin{itemize}
        \item One of $(b,y)$, $(c,z)$, $(d,w)$ is $(0,0)$ and $f\in\mathcal{A}\cup\mathcal{P}$;
        \item There is one zero in each pair $(b,y)$, $(c,z)$, $(d,w)$, i.e. $f\in \left \langle \mathcal{M}  \right \rangle $ or $f\in\left \langle X\mathcal{M}  \right \rangle $,
    \end{itemize}
    in which cases $\mathrm{Holant}(\ne_2|f)$ is computable in polynomial time.
\end{theorem}
\begin{theorem}\label{eightvertex}
(\cite{Eightvertex})
    Let $f$ be a 4-ary signature that has eight-vertex form. Then $\mathrm{Holant}(\ne_2|f)$ is \#P-hard except for the following cases:
    \begin{itemize}
        \item $(\ne_2|f)$ is $ \mathcal{A}\text{-}$transformable;
        \item $(\ne_2|f)$ is $ \mathcal{P}\text{-}$transformable;
        \item $(\ne_2|f)$ is $ \mathcal{L}\text{-}$transformable,
    \end{itemize}
    in which cases $\mathrm{Holant}(\ne_2|f)$ is computable in polynomial time.
\end{theorem}




\subsection{\texorpdfstring{The new tractable class $\mathcal{H}$}{The new tractable class H}}

\begin{definition}\label{Hform}
    A 4-ary signature $f$ has H-form if and only if
\[M(f)=\begin{pmatrix}
 0 &  0&  0& b\\
 0 & c & d &* \\
  0& w &  z& *\\
 y & * &  *&*
\end{pmatrix}\] or \[M(f)=\begin{pmatrix}
 * &  *&  *& b\\
 * & c & d &0 \\
  *& w &  z& 0\\
 y & 0 &  0&0
\end{pmatrix},\] and its reduced form $r(f)$ is a 4-ary signature with the matrix form
\[M(r(f))=\begin{pmatrix}
 0 &  0&  0& b\\
 0 & c & d &0 \\
  0& w &  z& 0\\
 y & 0 &  0&0
\end{pmatrix}.\]
\end{definition}
Note that for any instance
of Holant($\ne_2|f$), where $f$ has H-form, its value is the same as Holant($\ne_2|r(f)$).

\begin{definition}
$f\in\mathcal{H}$ if and only if $f'=(Z^{-1})^{\otimes 4}f$ has H-form and $\mathrm{Holant}$($\ne_2|r(f')$) is tractable.
\end{definition}

\begin{theorem}\label{ff'f''}
Let $f$ be a 4-ary signature such that $f'=(Z^{-1})^{\otimes 4}f$ has H-form. Then $\mathrm{Holant}$($=_2|f$) is \#P-hard unless $f\in\mathcal{H}$, in which case $\mathrm{Holant}$($=_2|f$) is tractable.
\end{theorem}
\begin{proof}
     By a holographic transformation using $Z$, we have \[
\mathrm{Holant}(\ne_2|f')\equiv_T\mathrm{Holant}(=_2|f).
\] In an instance of Holant$(\ne_2|f')$, for any assignment such that the value of the instance is nonzero, the entries of weight less than or greater than 2 of $f'$ never appear since $\ne_2$ on the left-hand side contributes equal numbers of zeros and ones.
Thus, we have \[
    \mathrm{Holant}(\ne_2|r(f'))\equiv_T \mathrm{Holant}(\ne_2|f'),
    \]and the hardness or tractability of $\mathrm{Holant}(\ne_2|r(f'))$ implies that $\mathrm{Holant}(=_2|f)$ is \#P-hard or tractable.
\end{proof}

\section{\texorpdfstring{Reductions to CSP\textsuperscript{2}}{Reductions to CSP squared}}
In this paper, a basic reduction is to CSP$^2$. In this section, we summarize the cases that can be reduced to CSP$^2$.

\begin{lemma}\label{toCSP2}(\cite{PlanarCSP})
    For any set of signatures $\mathcal{F}$,
    \[
\mathrm{CSP}^2(\mathcal{F}) \le_T \mathrm{Holant}(=_2|\mathcal{F},=_4) .
\]
\end{lemma}

\begin{lemma}\label{g=_4}(\cite{PlanarCSP})
    Let $f$ be a 4-ary signature with the signature matrix \[M(f)=\begin{pmatrix}
  a&  0& 0 & b\\
 0 &0  & 0 & 0\\
  0&  0& 0 &0 \\
 y & 0 & 0 &x
\end{pmatrix},\] where $\begin{pmatrix}
 a & b\\
 y &x
\end{pmatrix}$ has full rank. Then, for any signature set $\mathcal{F}$,
\[
\mathrm{Holant}(=_2|\mathcal{F},f,=_4) \equiv_T \mathrm{Holant}(=_2|\mathcal{F},f) .
\]
\end{lemma}
\begin{lemma}\label{twozeropair}
    Let $f$ be a 4-ary signature with the signature matrix \[M(f)=\begin{pmatrix}
 0 &  0& 0 &0 \\
  0& c & d & 0\\
 0 & w &  z& 0\\
  0&  0&  0&0
\end{pmatrix},\] where $\begin{pmatrix}
 c & d\\
 w &z
\end{pmatrix}$ has full rank. Then, for any signature set $\mathcal{F}$,
\[
\mathrm{Holant}(=_2|\mathcal{F},f,=_4) \equiv_T \mathrm{Holant}(=_2|\mathcal{F},f) .
\]
\end{lemma}
\begin{proof}
 Note that $\begin{pmatrix}
  c& d\\
 w &z
\end{pmatrix}$ has full rank, then by the Jordan normal form of $\begin{pmatrix}
  c& d\\
 w &z
\end{pmatrix}$, there exists $P\in \mathrm{SL}(2,\mathbb{C})$ such that $\begin{pmatrix}
  c& d\\
 w &z
\end{pmatrix}=P\begin{pmatrix}
  1& 0\\
 0 &\lambda
\end{pmatrix}P^{-1}$, or $\begin{pmatrix}
  c& d\\
 w &z
\end{pmatrix}=P\begin{pmatrix}
  1& \lambda\\
 0 &1
\end{pmatrix}P^{-1}$ when there is a double root, both up to a scalar. This implies that \[M(f)=\begin{pmatrix}
  0& \textbf{0} & 0\\
 \textbf{0} & P & \textbf{0}\\
 0 & \textbf{0} &0
\end{pmatrix}\begin{pmatrix}
 0 &  0& 0 &0 \\
  0& 1 & 0 & 0\\
 0 & 0 &  \lambda& 0\\
  0&  0&  0&0
\end{pmatrix}\begin{pmatrix}
  0& \textbf{0} & 0\\
 \textbf{0} & P^{-1} & \textbf{0}\\
 0 & \textbf{0} &0
\end{pmatrix}\] or \[M(f)=\begin{pmatrix}
  0& \textbf{0} & 0\\
 \textbf{0} & P & \textbf{0}\\
 0 & \textbf{0} &0
\end{pmatrix}\begin{pmatrix}
 0 &  0& 0 &0 \\
  0& 1 & \lambda & 0\\
 0 & 0 &  1& 0\\
  0&  0&  0&0
\end{pmatrix}\begin{pmatrix}
  0& \textbf{0} & 0\\
 \textbf{0} & P^{-1} & \textbf{0}\\
 0 & \textbf{0} &0
\end{pmatrix}.\] Then we can construct \[\begin{pmatrix}
 0 &  0& 0 &0 \\
  0& 1 & 0 & 0\\
 0 & 0 &  1& 0\\
  0&  0&  0&0
\end{pmatrix}\] by polynomial interpolation for both cases. The polynomial interpolation procedure is the same as in Lemma~\ref{g=_4}. Further, by Lemma~\ref{g=_4} we can construct $(=_4)$.
\end{proof}
\begin{lemma}\label{gen=4}
    Let $f$ be a 4-ary signature with the signature matrix \[M(f)=\begin{pmatrix}
  0&  0& 0 & 0\\
 a &0  & 0 & b\\
  y&  0& 0 &x \\
 0& 0 & 0 &0
\end{pmatrix}\] where $\begin{pmatrix}
 a & b\\
 y &x
\end{pmatrix}$ has full rank. Then, for any signature set $\mathcal{F}$, at least one of the following holds:
\begin{itemize}
    \item
    $
\mathrm{Holant}(=_2|\mathcal{F},f,h) \equiv_T \mathrm{Holant}(=_2|\mathcal{F},f),
$ where $h$ is a nonzero degenerate binary signature.
\item $
\mathrm{Holant}(=_2|\mathcal{F},f,=_4) \equiv_T \mathrm{Holant}(=_2|\mathcal{F},f);
$
\end{itemize}
\end{lemma}
\begin{proof}
Note that $f_{=_2}^{34}=(0,a+b,y+x,0)$. If $f_{=_2}^{34}\equiv0$, then $a=-b$ and $y=-x$, and $\begin{pmatrix}
 a & b\\
 y &x
\end{pmatrix}$ does not have full rank, a contradiction. If $a=-b$ or $y=-x$, then $h=f_{=_2}^{34}$ is a nonzero degenerate binary signature. If $a\ne-b$ and $y\ne -x$, up to a nonzero scalar, let $f_{=_2}^{34}=(0,1,t,0)$, where $t\ne0$.
By a binary modification to input $x_2$ of $f$ using the binary signature $(0,1,t,0)$, we construct a new 4-ary signature $f^*$ with signature matrix \[M(f^*)=\begin{pmatrix}
  at&  0& 0 & bt\\
 0 &0  & 0 & 0\\
  0&  0& 0 &0 \\
 y& 0 & 0 &x
\end{pmatrix}.\] Note that $\begin{pmatrix}
 at & bt\\
 y &x
\end{pmatrix}$ has full rank since $\begin{pmatrix}
 a & b\\
 y &x
\end{pmatrix}$ has full rank. Then by Lemma~\ref{g=_4} we construct $(=_4)$.
\end{proof}


\section{\texorpdfstring{Holant problems with a 4-ary signature $f$ and a degenerate binary signature}{Holant problems with a 4-ary signature f and a degenerate binary signature}}
\begin{lemma}\label{decompose}
    Let $f$ be a 4-ary signature, and $h=u_1\otimes u_2$ be a nonzero degenerate binary signature. Then in $\mathrm{Holant}(=_2|f,h)$, at least one of the following holds:
    \begin{itemize}
        \item \begin{equation}\label{frlemaadecompose}
            \mathrm{Holant}(=_2|f,h,u) \equiv_T \mathrm{Holant}(=_2|f,h),
        \end{equation}
         where $u$ is a nonzero unary signature;
        \item \begin{equation}\label{secrlemaadecompose}
            \mathrm{Holant}(\ne_2|f') \equiv_T\mathrm{Holant}(=_2|f),
        \end{equation}
        where $f'=(Z^{-1})^{\otimes 4}f$ has eight-vertex form.
    \end{itemize}
\end{lemma}
\begin{proof}
    By a holographic transformation using $Z$, we have \[
   \mathrm{Holant}(\ne_2|f',h') \equiv_T\mathrm{Holant}(=_2|f,h),
    \]where $f'=(Z^{-1})^{\otimes 4}f$ and $h'=(Z^{-1})^{\otimes 2}h=u_1'\otimes u_2'$. Assume that $u_1'=[x,y]$ and $u_2'=[s,t]$.
    Clearly if $\mathrm{Holant}(\ne_2|f',h',u') \le_T \mathrm{Holant}(\ne_2|f',h')$ holds, then $\mathrm{Holant}(=_2|f,h,u) \le_T \mathrm{Holant}(=_2|f,h)$ holds.

    For an instance $\Omega'$ of $\mathrm{Holant}(\ne_2|f',h',u'_1)$, assume that $u'_1$ appears $m$ times.
    Note that $(\ne_2),f',h'$ all have even arities, thus $m$ is even.
    If $st\ne0$, by replacing each occurrence of $u_1'$ by the $u_1'$ in $h'$ and connecting the remaining $m$ $u_2'$ using $\ne_2$ pairwise, we construct an instance $\Omega''$ of $\mathrm{Holant}(\ne_2|f',h')$. We also have $(2st)^{\frac{m}{2}}\mathrm{Holant}_{\Omega'}(\ne_2|f',h',u'_1)=\mathrm{Holant}_{\Omega''}(\ne_2|f',h')$.  Then we have $\mathrm{Holant}(=_2|f,h,u_1) \le_T \mathrm{Holant}(=_2|f,h)$. Similarly, if $xy\ne0$, then we have $\mathrm{Holant}(=_2|f,h,u_2) \le_T \mathrm{Holant}(=_2|f,h)$. Further, if $u_1'=u_2'$ up to a nonzero scalar, by replacing two occurrences of $u_1'$ by $h'$, we have $\mathrm{Holant}_{\Omega'}(\ne_2|f',h',u'_1)=\mathrm{Holant}_{\Omega''}(\ne_2|f',h')$, then we have $\mathrm{Holant}(=_2|f,h,u_1) \le_T \mathrm{Holant}(=_2|f,h)$.

    Now we can assume that $xy=st=0$ and $u_1'\ne u_2'$.
    Without loss of generality, we assume that $u_1'=[1,0]$ and $u_2'=[0,1]$, up to a nonzero scalar. Further,
    for any binary signature $g'\in S(f'_{\ne_2})$, by connecting one variable of $g'$ and one variable of $h'$ using $\neq_2$ we have four degenerate signatures,  $[1,0]\otimes[{g'}_{00},{g'}_{01}]$, $[1,0]\otimes[{g'}_{00},{g'}_{10}]$, $[{g'}_{01},{g'}_{11}]\otimes[0,1]$, $[{g'}_{10},{g'}_{11}]\otimes[0,1]$. Note that we have $\mathrm{Holant}(\ne_2|f',h',g',u') \le_T \mathrm{Holant}(\ne_2|f',h',g')$ unless $[{g'}_{00},{g'}_{01}],[{g'}_{00},{g'}_{10}]$ is $[0,1]$ or $[0,0]$, and $[{g'}_{01},{g'}_{11}],[{g'}_{10},{g'}_{11}]$ is $[1,0]$ or $[0,0]$, up to a scalar. This implies that
    ${g'}_{00}={g'}_{11}=0$. Then, we have $S(f'_{\ne_2})\subseteq \{(0,b,c,0)|b,c\in\mathbb{C}\}$, by Lemma~\ref{solveeightvertexform}, $f'$ has eight-vertex form.
\end{proof}

The following lemma has been proved by Meng et al. in \cite{oddarity}.
\begin{lemma}\label{cdne0}(\cite{oddarity})
    Let $\mathcal{F}$ be a set of signatures, then $\mathrm{Holant}(=_2|\mathcal{F},[1,0])$ is \#P-hard except for the following cases:
    \begin{itemize}
        \item $\mathcal{F} \subseteq \left \langle Z\mathcal{M} \right \rangle$ or $ \left \langle ZX\mathcal{M} \right \rangle$;
        \item $\mathcal{F} \subseteq  \left \langle \mathcal{T} \right \rangle$;
        \item $\mathcal{F}\cup [1,0]$ is $ \mathcal{A}\text{-}$transformable;
        \item $\mathcal{F}$ is $ \mathcal{P}\text{-}$transformable;
        \item $\mathcal{F}\cup[1,0]$ is $ \mathcal{L}\text{-}$transformable,
    \end{itemize}
in which cases the problem is computable in polynomial time.
\end{lemma}
Note that if we have a unary signature $[c, d]$ with $c^2+d^2\neq 0$ in Holant$(=_2|f)$, then by a holographic transformation using
$\frac{1}{c^2+d^2}\begin{pmatrix}
 c & -d\\
 d & c
\end{pmatrix}$,  we can transform $[c,d ]$ to $[1, 0]$ and  reduce it to Lemma~\ref{cdne0}(\cite{oddarity}). Thus we have the following corollary.
\begin{corollary}\label{cdne0coro}
    Let $\mathcal{F}$ be a set of signatures and $[c,d]$ be a unary signature with $c^2+d^2\ne0$, then $\mathrm{Holant}(=_2|\mathcal{F},[c,d])$ is \#P-hard except for the following cases:
    \begin{itemize}
        \item $\mathcal{F} \subseteq \left \langle Z\mathcal{M} \right \rangle$ or $ \left \langle ZX\mathcal{M} \right \rangle$;
        \item $\mathcal{F} \subseteq  \left \langle \mathcal{T} \right \rangle$;
        \item $\mathcal{F}\cup [c,d]$ is $ \mathcal{A}\text{-}$transformable;
        \item $\mathcal{F}$ is $ \mathcal{P}\text{-}$transformable;
        \item $\mathcal{F}\cup[c,d]$ is $ \mathcal{L}\text{-}$transformable,
    \end{itemize}
in which cases the problem is computable in polynomial time.
\end{corollary}
\begin{lemma}\label{only10}
    Let $f'$ be a 4-ary signature. Then one of the following holds in $\mathrm{Holant}(\ne_2|f',[1,0]$):
    \begin{itemize}
        \item $\mathrm{Holant}(\ne_2|f',[1,0],[x,y])\equiv_T\mathrm{Holant}(\ne_2|f',[1,0]$), where $y\ne0$;
        \item $f'$ has H-form.
    \end{itemize}
\end{lemma}
\begin{proof}
    Let $h$ be a binary signature that we can construct in $\mathrm{Holant}(\ne_2|f',[1,0]$). If $h_{11}\ne0$, then we have $h_{[0,1]}^1=[h_{10},h_{11}]$ and complete the proof. Thus we can assume that any binary signature we construct has the form $(a,b,c,0)$. Note that ${f'(x_i,x_j,x_k,x_{\ell})}_{\ne_2}^{ij}=(*,*,*,f'_{x_i=0,x_j=1,x_i=1,x_i=1}+f'_{x_i=1,x_j=0,x_i=1,x_i=1})$ for any permutation $(i,j,k,\ell)$ of $(1,2,3,4)$. Then all the weight 3 entries of $f'$ are zero. Further, if $f'_{1111}\ne0$, then we have ${f'}_{[0,1]^{\otimes3}}^{123}=[f'_{1110},f'_{1111}]$ and complete the proof. Thus all the weight 3 and weight 4 entries of $f'$ are zero. By Definition~\ref{Hform}, $f'$ has H-form.
\end{proof}
\begin{lemma}\label{both10and01}
    Let $f'$ be a 4-ary signature. Then one of the following holds in $\mathrm{Holant}(\ne_2|f',[1,0],[0,1]$):
    \begin{itemize}
        \item $f'$ has eight-vertex form;
        \item $f'$ has H-form;
        \item we can construct a generalized equality signature with arity less than 4, i.e. we have
        \[\mathrm{Holant}(\ne_2|f',[1,0],[0,1],[a,b])\equiv_T\mathrm{Holant}(\ne_2|f',[1,0],[0,1]),\] or \[\mathrm{Holant}(\ne_2|f',[1,0],[0,1],[a,0,b])\equiv_T\mathrm{Holant}(\ne_2|f',[1,0],[0,1]),\] or
        \[\mathrm{Holant}(\ne_2|f',[1,0],[0,1],[a,0,0,b])\equiv_T\mathrm{Holant}(\ne_2|f',[1,0],[0,1]),\] where $ab\ne0$.

    \end{itemize}
\end{lemma}
\begin{proof}
    Let $\varepsilon$ and $\mu$ be two bit strings of length 4 such that the Hamming weight of $\varepsilon\oplus \mu$ is $1$. Note that we can construct $[f'_{\varepsilon},f'_{\mu}]$ or $[f'_{\mu},f'_{\varepsilon}]$ by pinning $[0,1]$ and $[1,0]$ to $f'$. If $f'_{\varepsilon}f'_{\mu}\ne0$, then the result follows. Thus we assume that $f'_{\varepsilon}f'_{\mu}=0$. We consider the following cases separately.
    \begin{caselist}
        \caseitem{1}  $f'_{0000}f'_{1111}\ne0$. Then all the weight 1 and weight 3 entries of $f'$ are zero. $f'$ has eight-vertex form.
        \caseitem{2}  $f'_{0000}\ne0$ and $f'_{1111}=0$. Then all the weight 1 entries of $f'$ are zero. Further, if all the weight 3 entries of $f'$ are zero, then $f'$ has H-form. Thus without loss of generality, we assume that $f'_{1110}\ne0$. Then we have $f'_{0110}=f'_{1010}=f'_{1100}=0$, i.e., \[
        M(f')=\begin{pmatrix}
  {f'_{0000}} &  0&  0& {f'_{0011}}\\
  0 &  {f'_{0101}}&  0& {f'_{0111}}\\
  0 &  {f'_{1001}}&  0& {f'_{1011}}\\
  0 &  {f'_{1101}}&  {f'_{1110}}& 0
\end{pmatrix}.
        \] We have ${f'}_{[1,0]}^4=[{f'_{0000}},0,0,{f'_{1110}}]$ and complete the proof.
        \caseitem{3}  $f'_{0000}=0$ and $f'_{1111}\ne0$. Under the holographic transformation using $X$, this case reduces to \caseref{2}.
        \caseitem{4}   $f'_{0000}=f'_{1111}=0$. We assume that at least one weight 1 entry $f_{\alpha}$ is nonzero and at least one weight 3 entry $f_{\beta}$ is nonzero, otherwise, $f'$ has H-form. Note that there are two different cases, i.e., the Hamming weight of $\alpha\oplus\beta$ is 2 or 4.  Then by symmetry, we consider the following cases separately.
        \begin{caselist}
            \subcaseitem{4.1}   $f'_{0001}f'_{1101}\ne0$. We have $h={f'}_{[1,0]\otimes[0,1]}^{34}=(f'_{0001},f'_{0101},f'_{1001},f'_{1101})$. If $f'_{0101}=f'_{1001}=0$, then the result follows. If $f'_{0101}\ne0$ or $f'_{1001}\ne0$, then the result follows from $h_{[1,0]}^{1}=[f'_{0001},f'_{0101}]$ or $h_{[1,0]}^{2}=[f'_{0001},f'_{1001}]$.
            \subcaseitem{4.2}  $f'_{0001}f'_{1110}\ne0$. Then all the weight 2 entries of $f'$ are zero. Further, if $f'_{0010}\ne0$, then we have  ${f'}_{[0,1]\otimes[1,0]}^{34}=(f'_{0010},f'_{0110},f'_{1010},f'_{1110})=[f'_{0010},0,f'_{1110}]$ and complete the proof. If $f'_{1101}\ne0$, then we have ${f'}_{[1,0]\otimes[0,1]}^{34}=(f'_{0001},f'_{0101},f'_{1001},f'_{1101})=[f'_{0001},0,f'_{1101}]$ and complete the proof. Thus we can assume that $f'_{0010}=f'_{1101}=0$. Then we have ${f'}_{\ne2}^{34}=[f'_{0001},0,f'_{1110}]$ and complete the proof.\qedhere
        \end{caselist}
    \end{caselist}
\end{proof}
\begin{lemma}\label{=ktocspk}(\cite{Formquantumentanglementandback})
     For any $d \ge 3$, \#$\mathrm{CSP}^d(\ne_2,\mathcal{F})\le _T\mathrm{Holant}(\ne_2|=_d,\mathcal{F})$
\end{lemma}
We are now ready to prove the main theorem of this section.
\begin{theorem}\label{theoremwithunary}
Let $f$ be a 4-ary signature. Then Holant$(=_2|f,[c,d])$ is \#P-hard except for the following cases:
\begin{itemize}
        \item $f \in \left \langle Z\mathcal{M} \right \rangle$ or $\in \left \langle ZX\mathcal{M} \right \rangle$;
        \item $f \in \left \langle \mathcal{T} \right \rangle$;
        \item $\{f,[c,d]\}$ is $ \mathcal{A}\text{-}$transformable;
        \item $\{f,[c,d]\}$ is $ \mathcal{L}\text{-}$transformable;
        \item $f$ is $ \mathcal{P}\text{-}$transformable;
        \item $f\in\mathcal{H}$,
    \end{itemize}
    in which cases the problem is computable in polynomial time.
\end{theorem}
\begin{proof}
    If $c^2+d^2\ne0$, then by Corollary~\ref{cdne0coro}, the result follows.

    If $c^2+d^2=0$, we only prove the case in which $[c,d]=[1,\frak{i}]$ up to a nonzero scalar; the other case is analogous and is omitted. By a holographic transformation using
    $Z$, we have \[
\mathrm{Holant}(\ne_2|f',[1,0])\equiv_T\mathrm{Holant}(=_2|f,[c,d]),
\]where $f'=(Z^{-1})^{\otimes4}f$.
    By Lemma~\ref{only10}, $f'$ has H-form or we can construct $[x,y]$ with $y\ne0$. If $f'$ has H-form, by Theorem~\ref{ff'f''}, $f\in \mathcal{H}$ or $\mathrm{Holant}(=_2|f)$ is \#P-hard. If we can construct $[x,y]$ with $y\ne0$, then we consider the following cases separately.
    \begin{caselist}
        \caseitem{1}  $x\ne0$. By a holographic transformation using
    $Z^{-1}$, we have \[
\mathrm{Holant}(=_2|f,[1,\frak{i}],[x+y,\frak{i}(x-y)])\equiv_T\mathrm{Holant}(\ne_2|f',[1,0],[x,y]).
\] Note that $(x+y)^2+[\frak{i}(x-y)]^2=4xy\ne0$, then by Corollary~\ref{cdne0coro}, the result follows.
\caseitem{2}  $x=0$. By Lemma~\ref{both10and01}, we consider the following cases separately.
\begin{caselist}
\subcaseitem{2.1}  We have \[\mathrm{Holant}(\ne_2|f',[1,0],[0,1],[a,b])\equiv_T\mathrm{Holant}(\ne_2|f',[1,0],[0,1]),\] where $ab\ne0$. By a holographic transformation using
    $Z^{-1}$ and Corollary~\ref{cdne0coro}, the result follows.
        \subcaseitem{2.2}  We have \[\mathrm{Holant}(\ne_2|f',[1,0],[0,1],a[1,0,\frac{b}{a}])\equiv_T\mathrm{Holant}(\ne_2|f',[1,0],[0,1]),\] where $ab\ne0$. By a holographic transformation using
    $T_{(\frac{b}{a})^{\frac{1}{2}}}=\begin{pmatrix}
  1&0 \\
 0 &(\frac{b}{a})^{\frac{1}{2}}
\end{pmatrix}$, we have \[\mathrm{Holant}(\ne_2|f'',[1,0],[0,1],=_2)\equiv_T\mathrm{Holant}(\ne_2|f',[1,0],[0,1],a[1,0,\frac{b}{a}]),\] where $f''=(T_{(\frac{b}{a})^{\frac{1}{2}}}^{-1})^{\otimes4}f'$. We can also construct $=_2$ on the left-hand side by connecting two copies of $\ne_2$ to inputs $x_1$ and $x_2$ of $=_2$ respectively. Then by Lemma~\ref{cdne0} the result follows.
        \subcaseitem{2.3}  We have \[\mathrm{Holant}(\ne_2|f',[1,0],[0,1],a[1,0,0,\frac{b}{a}])\equiv_T\mathrm{Holant}(\ne_2|f',[1,0],[0,1]),\] where $ab\ne0$. By a holographic transformation using
    $T_{(\frac{b}{a})^{\frac{1}{3}}}=\begin{pmatrix}
  1&0 \\
 0 &(\frac{b}{a})^{\frac{1}{3}}
\end{pmatrix}$, we have \[\mathrm{Holant}(\ne_2|f'',[1,0],[0,1],=_3)\equiv_T\mathrm{Holant}(\ne_2|f',[1,0],[0,1],a[1,0,0,\frac{b}{a}]),\] where $f''=(T_{(\frac{b}{a})^{\frac{1}{3}}}^{-1})^{\otimes4}f'$. By Lemma~\ref{=ktocspk}, we have \[\mathrm{CSP}^3(\ne_2,f'',[1,0],[0,1])\le_T\mathrm{Holant}(\ne_2|f'',[1,0],[0,1],=_3).\] Then by Theorem~\ref{CSPd} the result follows.
        \subcaseitem{2.4}  $f'$ has eight-vertex form. By Theorem~\ref{eightvertex}, the result follows.
        \subcaseitem{2.5}  $f'$ has H-form. By Theorem~\ref{ff'f''}, the result follows.\qedhere
\end{caselist}
    \end{caselist}

\end{proof}

\begin{theorem} \label{dicwithdegen}
    Let $f$ be a 4-ary signature. If we can construct a nonzero degenerate binary signature $h$ in $\mathrm{Holant}(=_2|f)$, then $\mathrm{Holant}(=_2|f)$ is \#P-hard except for the following cases:
\begin{itemize}
        \item $f \in\left \langle Z\mathcal{M} \right \rangle$ or $\left \langle ZX\mathcal{M} \right \rangle$;
        \item $f $ is $ \mathcal{A}\text{-}$transformable;
        \item $f $ is $ \mathcal{L}\text{-}$transformable;
        \item $f $ is $ \mathcal{P}\text{-}$transformable;
        \item $f \in \mathcal{H}$;
        \item $f \in \left \langle \mathcal{T} \right \rangle$,
    \end{itemize}
    in which cases the problem is computable in polynomial time.
\end{theorem}
\begin{proof}
    In Holant$(=_2|f,h)$, by Lemma~\ref{decompose}, one of the following holds:
    \begin{itemize}
        \item we have
    \[
    \mathrm{Holant}(=_2|f,h,u) \le_T \mathrm{Holant}(=_2|f,h),
    \]where $u$ is a nonzero unary signature. By Theorem~\ref{theoremwithunary} the result follows.
\item we have \[\mathrm{Holant}(\ne_2|f') \equiv_T \mathrm{Holant}(=_2|f),\] where $f'=(Z^{-1})^{\otimes 4}f$ has eight-vertex form. By Theorem~\ref{eightvertex} the result follows.
    \end{itemize}
\end{proof}

In the remainder of this paper, by Theorem~\ref{dicwithdegen} and Corollary~\ref{polytodegen}, we assume that any binary signature we can construct is nondegenerate and has finite order.

\section{\texorpdfstring{Signatures $f$ in eight-vertex form}{Signatures f in eight-vertex form}}
In this section, $f$ always denotes a 4-ary signature with the signature matrix \[M(f)=\begin{pmatrix}
a  &  0&0  & b\\
 0 &  c& d & 0\\
 0 & w & z & 0\\
 y &0  &0  &x
\end{pmatrix}.\] We call $(a,x)$ the outer pair of $f$ and $(b,y),(c,z),(d,w)$ the inner pairs of $f$.

First, we introduce some known results in \cite{Eightvertex} and \cite{Sixvertex}.
\begin{lemma}\label{twononzeropairinsixvertexmodel}(\cite{Sixvertex})
    Let $f$ be a 4-ary signature with signature matrix  \[M(f)=\begin{pmatrix}
 0 &  0& 0 &0 \\
 0 &  c&d  &0 \\
  0&  w&z  & 0\\
  0& 0 &0  &0
\end{pmatrix}.\] Then $\mathrm{Holant}(\ne_2|f)$ is \#P-hard unless $f\in\mathcal{A}$ or $f\in\mathcal{P}$, in which case the problem is computable in polynomial time.
\end{lemma}
\begin{lemma}\label{twononzeroineight}\cite{Eightvertex}
     Let $f$ be a 4-ary signature with signature matrix  \[M(f)=\begin{pmatrix}
 a &  0& 0 &0 \\
 0 &  c& 0 &0 \\
  0&  0&z  & 0\\
  0& 0 &0  &a
\end{pmatrix}.\] Then $\mathrm{Holant}(\ne_2|f)$ is \#P-hard unless $\begin{pmatrix}
    1&0\\
    0&\beta
\end{pmatrix}^{\otimes 4}f\in\mathcal{A}$ with $\beta^{16}=1$ or $f\in\mathcal{P}$, in which case the problem is computable in polynomial time.
\end{lemma}
\begin{lemma}\label{M(g)}\cite{Eightvertex}
    Let $g$ be a 4-ary signature with signature matrix  \[M(g)=\begin{pmatrix}
 t &  0& 0 &0 \\
 0 &  1& 0 &0 \\
  0&  0&1  & 0\\
  0& 0 &0  &t
\end{pmatrix},\] where $t\ne0$ and $t^4\ne1$. Then for any signature set $\mathcal{F}$,
\[
\mathrm{Holant}(\ne_2|\mathcal{EQ_2},\mathcal{F},g)\le_T\mathrm{Holant}(\ne_2|\mathcal{F},g).
\]
\end{lemma}
\begin{lemma}\label{sixwithax=0}\cite{Sixvertex}
    Let $f$ be a 4-ary signature with signature matrix  \[M(f)=\begin{pmatrix}
 0 &  0& 0 &b \\
 0 &  c&d  &0 \\
  0&  w&z  & 0\\
  y& 0 &0  &0
\end{pmatrix}\] with $byczdw\ne0$. Then $\mathrm{Holant}(\ne_2|f)$ is \#P-hard.
\end{lemma}
\begin{lemma}\label{foursamepair}(\cite{Eightvertex})
 Let $f$ be a 4-ary signature with signature matrix  \[M(f)=\begin{pmatrix}
 1 &  0& 0 &b \\
 0 &  c&d  &0 \\
  0&  d&c  & 0\\
  b& 0 &0  &1
\end{pmatrix}\] with $bcd\ne0$. Then in $\mathrm{Holant}(\ne_2|\mathcal{F},f)$, we have
    \begin{itemize}
        \item Either $\mathrm{Holant}(\ne_2|\mathcal{F},f)$ is \#P-hard;
        \item or $f\in \mathcal{A}$;
        \item or there exists $T\in \mathrm{SL}(2,\mathbb{C})$, such that $\mathrm{CSP}^2(\mathcal{F'},f')\le_T\mathrm{Holant}(\ne_2|\mathcal{F},f)$, where $f'=(T^{-1})^{\otimes 4}f$, $\mathcal{F'}=(T^{-1})\mathcal{F}$.
    \end{itemize}
\end{lemma}

\begin{corollary}\label{affine support}(\cite{Eightvertex})
    Let $f$ be a 4-ary signature that has eight-vertex form with $f_{0000}f_{1111}\ne0$. If the support of $f$ is not affine, then $\mathrm{Holant}(\ne_2|f)$ is \#P-hard.
\end{corollary}

\subsection{\texorpdfstring{The dichotomy of $\mathrm{Holant}(=_2|f)$, where $f$ has eight-vertex form}{The dichotomy of Holant(=2 | f), where f has eight-vertex form}}
\begin{lemma}\label{mobius}
    Let $f$ be a 4-ary signature with signature matrix  \[M(f)=\begin{pmatrix}
 a & 0 &  0& b\\
  0&  c&  d&0 \\
 0 & w &  z&0 \\
 y &0  & 0 &x
\end{pmatrix},\] where $\begin{pmatrix}
 a & b\\
  y&x
\end{pmatrix}$ has full rank. Suppose that we can construct $g_i=(1,0,0,t_i)$ as binary signatures where $t_i$ are distinct for $1\le i\le 3$. In $\mathrm{Holant}(=_2|\mathcal{F},f)$, if one of the following holds:
\begin{itemize}
    \item $\begin{pmatrix}
 a & b\\
  y&x
\end{pmatrix}$ has the form $\begin{pmatrix}
 e^{\frak{i}\theta} & e^{\frak{i}\theta}\lambda \\
  \overline{\lambda } &1
\end{pmatrix}$, where $|\lambda|\ne1$(up to a nonzero scalar), and $\begin{pmatrix}
 a & b\\
  y&x
\end{pmatrix}$ has infinite projective order;
\item $\begin{pmatrix}
 a & b\\
  y&x
\end{pmatrix}$ does not have the form  $\begin{pmatrix}
 a &0 \\
  0 &x
\end{pmatrix}$, or $\begin{pmatrix}
 0&b\\
  y &0
\end{pmatrix}$, or $\begin{pmatrix}
 e^{\frak{i}\theta} & e^{\frak{i}\theta}\lambda \\
  \overline{\lambda } &1
\end{pmatrix}$(up to a nonzero scalar), where $|\lambda|\ne1$;
\end{itemize}
then we can construct a degenerate binary signature.
\end{lemma}
\begin{proof}
    For $f_{g_i}^{34}=(a+t_ib,0,0,y+t_ix)$, note that at least one of $a+t_ib$ and $y+t_ix$ is nonzero since $\begin{pmatrix}
    a & b\\
    y&x
    \end{pmatrix}$ has full rank.
    If there exist $1\le i \le 3$ such that $a+t_ib=0$ or $y+t_ix=0$, then we construct a degenerate binary signature.
    Thus we assume that for any $1\le i \le 3$, $(a+t_ib)(y+t_ix)\ne0$. Then $f_{g_i}^{34}=(a+t_ib)(1,0,0,\varphi(t_i))$, where $\varphi(t_i)=\frac{y+t_ix}{a+t_ib}$. $\varphi(\mathfrak{z} )=\frac{y+\mathfrak{z}x}{a+\mathfrak{z}b}$ is a $\mathrm{M\ddot{o}bius}$ transformation of the extended complex plane $\widehat{\mathbb{C}}$. Since
a $\mathrm{M\ddot{o}bius}$ transformation is determined by any 3 distinct points, mapping 3 distinct points from $S^1$
to $S^1$
implies that this $\varphi(\mathfrak{z} )$ maps $S^1$ homeomorphically onto $S^1$.

A $\mathrm{M\ddot{o}bius}$ transformation mapping 3 distinct points from $S^1$ to $S^1$ has a special form $\mathcal{M}(\lambda,e^{\frak{i}\theta}): \mathfrak{z}\longmapsto e^{-\frak{i}\theta}\frac{\mathfrak{z}+\overline{\lambda}}{1+\lambda\mathfrak{z}}$, where $|\lambda|\ne1$. This implies that if $\begin{pmatrix}
 a & b\\
  y&x
\end{pmatrix}$ does not have the form $\begin{pmatrix}
 e^{\frak{i}\theta} & e^{\frak{i}\theta}\lambda \\
  \overline{\lambda } &1
\end{pmatrix}$(up to a nonzero scalar), we construct a binary signature of the form $(1,0,0,s)$ with $|s|\ne1$ and then by Corollary~\ref{polytodegen} we can construct a degenerate binary signature. If $\begin{pmatrix}
 a & b\\
  y&x
\end{pmatrix}$ has the form $\begin{pmatrix}
 e^{\frak{i}\theta} & e^{\frak{i}\theta}\lambda \\
  \overline{\lambda } &1
\end{pmatrix}$, and has infinite projective order, then the $\mathrm{M\ddot{o}bius}$ transformation $\varphi(\mathfrak{z} )$ defines an infinite group. If for each $i\in[3]$, there exists some $n_i>0$ such that $\varphi^{n_i}(t_i )=t_i$, then $\varphi^{n_1n_2n_3}(t_i )=t_i$ for all $i\in[3]$. So the
$\mathrm{M\ddot{o}bius}$ transformation is the identity map. This implies that $\varphi(\mathfrak{z} )$ defines a
group of finite order, a contradiction. Therefore, there is an $i\in[3]$ such that $\varphi^{n}(t_i )\ne t_i$
for all $n \in \mathbb{N}$.
This implies that $(1, \varphi^{n}(t_i ))$ are all distinct for $n\in \mathbb{N}$, since $\varphi$ maps $S^1$ 1-1 onto $S^1$.
This implies that we can construct polynomially many distinct binary signatures of the form $(1,0,0,\varphi^{n}(t_i ))$. Then by Lemma~\ref{infinitetodegen}, we can construct a degenerate binary signature.
\end{proof}

\begin{lemma}\label{atleastonenonzeroin(a,x)}
    Let $f$ be a 4-ary signature with signature matrix  \[M(f)=\begin{pmatrix}
  a&  0& 0 & b\\
  0&  c& d & 0\\
  0& w &  z& 0\\
 y &  0&  0&x
\end{pmatrix}.\] If there is at least one nonzero entry in $(a,x)$, and $M(f)$ is not of the form \[\begin{pmatrix}
 \epsilon x & 0 &0  & \epsilon y\\
  0& \epsilon z & \epsilon w & 0\\
 0 & w & z & 0\\
 y & 0 &0  &x
\end{pmatrix},\] where $\epsilon=\pm1$, then we have \[
\mathrm{Holant}(=_2|f,[1,0,t])\le_T\mathrm{Holant}(=_2|f),
\]where $t\ne\pm1,0$, or we have \[
\mathrm{Holant}(=_2|f,h)\le_T\mathrm{Holant}(=_2|f),
\]where $h$ is degenerate binary signature.
\end{lemma}
\begin{proof}
    For any inner pair ($s,t$) and the outer pair $(a,x)$, we say " $S(\begin{bmatrix}
a  & s\\
 t &x
\end{bmatrix})$ by $\epsilon$" if $a=\epsilon x$ and $s=\epsilon t$ for some $\epsilon=\pm1$.
If $S(\begin{bmatrix}
a  & b\\
 y &x
\end{bmatrix})$ by $\epsilon_1$, $S(\begin{bmatrix}
a  & c\\
 z &x
\end{bmatrix})$ by $\epsilon_2$, $S(\begin{bmatrix}
a  & d\\
 w &x
\end{bmatrix})$ by $\epsilon_3$, then by assumption it is not the case in which $\epsilon_1=\epsilon_2=\epsilon_3$. Without loss of generality, we assume that $\epsilon_1\ne\epsilon_2$, then $a=\epsilon_1x=\epsilon_2x$, which implies that $a=x=0$, contradicting the fact that there exists at least one nonzero entry in $(a,x)$. Thus at least one of $S(\begin{bmatrix}
a  & b\\
 y &x
\end{bmatrix})$, $S(\begin{bmatrix}
a  & c\\
 z &x
\end{bmatrix})$, $S(\begin{bmatrix}
a  & d\\
 w &x
\end{bmatrix})$ does not hold for any $\epsilon=\pm1$. Without loss of generality, we assume that $S(\begin{bmatrix}
a  & b\\
 y &x
\end{bmatrix})$ does not hold for any $\epsilon=\pm1$, i.e., there is no $\epsilon=\pm1$ such that $a=\epsilon x$ and $b=\epsilon y$.

For $f_{=_2}^{12}=(a+y,0,0,b+x)$ and $f_{=_2}^{34}=(a+b,0,0,y+x)$, there are four cases:
\begin{caselist}
    \caseitem{1}   $f_{=_2}^{12}$ and $f_{=_2}^{34}$ are both $[0,0,0]$ or $=_2$. In this case, we have \[
    \left\{\begin{matrix}
 a+y=b+x\\
a+b=y+x
\end{matrix}\right. ,
    \] then $a=x$ and $b=y$, i.e. $S(\begin{bmatrix}
a  & b\\
 y &x
\end{bmatrix})$ by $1$, a contradiction.
    \caseitem{2}   $f_{=_2}^{12}$ is $[0,0,0]$ or $=_2$, and $f_{=_2}^{34}$ is $[1,0,-1]$.  In this case, we have \[
    \left\{\begin{matrix}
 a+y=b+x\\
a+b=-y-x
\end{matrix}\right. .
    \] Then we have $a=-y$, $b=-x$ and we can construct $[1,0,-1]$. For $f_{[1,0,-1]}^{12}=(a-y,0,0,b-x)=2(a,0,0,-x)$. If $ax\ne0$, then we have $a=-\epsilon x=\epsilon b=- y$, i.e. $S(\begin{bmatrix}
a  & b\\
 y &x
\end{bmatrix})$ by $\epsilon$, a contradiction. If there exists one zero in $(a,x)$, $f_{[1,0,-1]}^{12}$ is degenerate.
    \caseitem{3}  $f_{=_2}^{34}$ is $[0,0,0]$ or $=_2$, $f_{=_2}^{12}$ is $[1,0,-1]$. This case is the same as \caseref{2} by symmetry.
    \caseitem{4}  $f_{=_2}^{12}$ and $f_{=_2}^{34}$ are both $[1,0,-1]$. In this case, we have \[
    \left\{\begin{matrix}
a+y=-b-x\\
 a+b=-y-x
\end{matrix}\right..
    \] Then we have $a+b+x+y=0$ and we can construct $[1,0,-1]$. For $f_{[1,0,-1]}^{34}=(a-b,0,0,y-x)$, if we have $a-b=y-x$, then $a=-x$ and $b=-y$, a contradiction. If we have $a-b=x-y$, i.e. $a=-y$ and $b=-x$, then for $f_{[1,0,-1]}^{12}=(a-y,0,0,b-x)=2(a,0,0,x)$. If $ax\ne0$, then we have $a=\epsilon x$ and $b=\epsilon y$, a contradiction. If there exists exactly one zero in $(a,x)$, $f_{[1,0,-1]}^{12}$ is degenerate.\qedhere
\end{caselist}

\end{proof}

\begin{corollary}
    If \[M(f)=\begin{pmatrix}
 \epsilon x & 0 &0  & \epsilon y\\
  0& \epsilon z & \epsilon w & 0\\
 0 & w & z & 0\\
 y & 0 &0  &x
\end{pmatrix}\] and  $\epsilon=1$, then $S(f_{=_2})=\{=_2\}$; if $\epsilon=-1$, then $S(f_{=_2})=\{[1,0,-1]\}$.
\end{corollary}

Here $\ne_4$ denotes the signature supported on $0011$ and $1100$, with value $1$ at both inputs; input permutations of this signature are used when needed.
\begin{lemma}\label{=+2and=-2}
    Let $f$ be a 4-ary signature in eight-vertex form. If \[M(f)=\begin{pmatrix}
 \epsilon x & 0 &0  & \epsilon y\\
  0& \epsilon z & \epsilon w & 0\\
 0 & w & z & 0\\
 y & 0 &0  &x
\end{pmatrix},\] where $\epsilon=\pm1$, then we have
    \begin{itemize}
        \item Either $\mathrm{Holant}(=_2|\mathcal{F},f)$ is \#P-hard;
        \item or $f\in \mathcal{A}$;
        \item or $f\in \left \langle \mathcal{T} \right \rangle$;
        \item or there exists $T\in \mathrm{SL}(2,\mathbb{C})$, such that $\mathrm{CSP}^2(\mathcal{F'},f')\le_T\mathrm{Holant}(=_2|\mathcal{F},f)$, where $f'=(T^{-1})^{\otimes 4}f$, $\mathcal{F'}=(T^{-1})\mathcal{F}$.
    \end{itemize}
\end{lemma}
\begin{proof}
We separately prove the cases of $\epsilon=1$ and $\epsilon=-1$.
\begin{caselist}
    \caseitem{1}   $\epsilon=1$. By a holographic transformation using $Z$ we have
    \[
        \mathrm{Holant}(\ne_2|f')\le_T\mathrm{Holant}(\ne_2|\mathcal{F'},f',\ne_2)\equiv_T\mathrm{Holant}(=_2|\mathcal{F},f,=_2),
    \] where $f'=(Z^{-1})^{\otimes 4}f$, $\mathcal{F'}=(Z^{-1})\mathcal{F}$. Then \[M(f')=\begin{pmatrix}
 a &  0& 0 &b \\
  0& c & d & 0\\
 0 & d &  c& 0\\
  b&  0&  0&a
\end{pmatrix},\] where $a=(x-y-z-w)/2$, $b=(x-y+z+w)/2$, $c=(x+y-z+w)/2$, and $d=(x+y+z-w)/2$. There are four cases of $M(f')$.\begin{caselist}
       \subcaseitem{1.1}  There is exactly one nonzero pair in $f'$.

       Clearly $f'=(=_4)$ or $f'=(\ne_4)$, then $f'\in\mathcal{A}$, which implies that $f\in \mathcal{A}$.
       \subcaseitem{1.2}  There are exactly two nonzero pairs in $f'$.

       If $a=0$, by Lemma~\ref{twononzeropairinsixvertexmodel}, Holant($\ne_2|f'$) is \#P-hard or $f'\in\mathcal{A}$ or $\mathcal{P}$.
       Holant($\ne_2|f'$) being \#P-hard implies that Holant($=_2|\mathcal{F},f$) is \#P-hard. $f'\in\mathcal{A}$ implies that $f\in\mathcal{A}$. Moreover, $f'\in \mathcal{P}$ implies that $\{f',f\}\subseteq\left \langle \mathcal{T} \right \rangle$. If $a\ne0$, by Lemma~\ref{twononzeroineight}, Holant($\ne_2|f'$) is \#P-hard or $\begin{pmatrix}
           1&0\\
           0&\beta
       \end{pmatrix}^{\otimes4}f'\in\mathcal{A}$ with $\beta^{16}=1$ or $f'\in\mathcal{P}$.  Holant($\ne_2|f'$) being \#P-hard implies that Holant($=_2|\mathcal{F},f$) is \#P-hard. $f'\in \mathcal{P}$ implies that $\{f',f\}\subseteq\left \langle \mathcal{T} \right \rangle$. If $\beta^8=1$, then it is easy to check that $\begin{pmatrix}
           1&0\\
           0&\beta
       \end{pmatrix}^{\otimes4}f'\in\mathcal{A}$ implies that $f'\in \mathcal{A}$, then $f\in\mathcal{A}$ also. If $\beta$ is a primitive 16th root of unity, by symmetry, assume that \[M(f')=\begin{pmatrix}
 a &  0& 0 &0 \\
  0& 1 & 0 & 0\\
 0 & 0 &  1& 0\\
  0&  0&  0&a
\end{pmatrix}\] up to a nonzero scalar. Then $\begin{pmatrix}
           1&0\\
           0&\beta
       \end{pmatrix}^{\otimes4}f'\in\mathcal{A}$ implies that $a$ is a primitive eighth root of unity, i.e. $a^4\ne1$.
By Lemma~\ref{M(g)}, we have \[
\mathrm{Holant}(\ne_2|=_2,f')\le_T\mathrm{Holant}(\ne_2|f')\le_T\mathrm{Holant}(\ne_2|\mathcal{F'},f',\ne_2)\equiv_T\mathrm{Holant}(=_2|\mathcal{F},f,=_2). 
\] By connecting one copy of $=_2$ to input $x_3$ of $f'$ and another copy of $=_2$ to input $x_4$ of $f'$, both using $\ne_2$, we construct $f^{'*}$ with signature matrix \[M(f^{'*})=\begin{pmatrix}
 0 &  0& 0 &a \\
  0& 0 & 1 & 0\\
 0 & 1 &  0& 0\\
  a&  0&  0&0
\end{pmatrix}.\] As in the case in which $a=0$, by Lemma~\ref{twononzeropairinsixvertexmodel} the result follows.
       \subcaseitem{1.3}  There are exactly three nonzero pairs in $f'$.
       By Lemma~\ref{sixwithax=0} and Corollary~\ref{affine support} the problem is \#P-hard.
       \subcaseitem{1.4}  All pairs in $f'$ are nonzero.
       By Lemma~\ref{foursamepair} the result follows.
   \end{caselist}
   \caseitem{2}  $\epsilon=-1$. By a holographic transformation using $Z$ we have
    \[
        \mathrm{Holant}(\ne_2|\mathcal{F'},f',=_2)\equiv_T\mathrm{Holant}(=_2|\mathcal{F},f,[1,0,-1]),
    \] where $f'=(Z^{-1})^{\otimes 4}f$, $\mathcal{F'}=(Z^{-1})\mathcal{F}$. Then \[M(f')=\begin{pmatrix}
 0 &  a& b &0 \\
  c& 0 & 0 &d\\
 d & 0 &  0& c\\
  0&  b&  a&0
\end{pmatrix}.\] By connecting $=_2$ to input $x_4$ of $f'$ using $\ne_2$, we construct $f^{'*}$ with signature matrix \[M(f^{'*})=\begin{pmatrix}
 a &  0& 0 &b \\
  0& c & d & 0\\
 0 & d &  c& 0\\
  b&  0&  0&a
\end{pmatrix}.\] There are four cases of $M(f^{'*})$.\begin{caselist}
       \subcaseitem{2.1}  There is exactly one nonzero pair in $f^{'*}$.

       Clearly $f^{'*}=(=_4)$ or $f^{'*}=(\ne_4)$, then $f^{'*}\in\mathcal{A}$, which implies that $f\in \mathcal{A}$.
       \subcaseitem{2.2}  There are exactly two nonzero pairs in $f^{'*}$.

       If $a=0$, by Lemma~\ref{twononzeropairinsixvertexmodel}, Holant($\ne_2|f^{'*}$) is \#P-hard or $f^{'*}\in\mathcal{A}$ or $\mathcal{P}$.
       Holant($\ne_2|f^{'*}$) being \#P-hard implies that Holant($=_2|\mathcal{F},f$) is \#P-hard. $f^{'*}\in\mathcal{A}$ implies that $f\in\mathcal{A}$. Moreover, $f^{'*}\in \mathcal{P}$ implies that $\{f^{'*},f',f\}\subseteq\left \langle \mathcal{T} \right \rangle$. If $a\ne0$, without loss of generality, we assume that \[M(f^{'*})=\begin{pmatrix}
 a &  0& 0 &0 \\
  0& c & 0 & 0\\
 0 & 0 &  c& 0\\
  0&  0&  0&a
\end{pmatrix}.\]
       By connecting one copy of $=_2$ to input $x_3$ of $f^{'*}$ and another copy of $=_2$ to input $x_4$ of $f^{'*}$, both using $\ne_2$, we construct $f^{'**}$ with signature matrix \[M(f^{'**})=\begin{pmatrix}
 0 &  0& 0 &a \\
  0& 0 & c & 0\\
 0 & c &  0& 0\\
  a&  0&  0&0
\end{pmatrix}.\] As in the case in which $a=0$, by Lemma~\ref{twononzeropairinsixvertexmodel} the result follows.
       \subcaseitem{2.3}  There are exactly three nonzero pairs in $f^{'*}$.
       By Lemma~\ref{sixwithax=0} and Corollary~\ref{affine support} the problem is \#P-hard.
       \subcaseitem{2.4}  All pairs in $f^{'*}$ are nonzero.
       By Lemma~\ref{foursamepair} the result follows.\qedhere
   \end{caselist}
\end{caselist}

\end{proof}

\begin{lemma}\label{four1}
    Let $f$ be a 4-ary signature with the signature matrix \[M(f)=\begin{pmatrix}
  1&  0&  0& 1\\
 0 &  c&0  & 0\\
 0 & 0 & z &0 \\
  1& 0 &0  &1
\end{pmatrix}\] or \[M(f)=\begin{pmatrix}
  1&  0&  0& -1\\
 0 &  c&0  & 0\\
 0 & 0 & z &0 \\
  -1& 0 &0  &1
\end{pmatrix}\] with $cz\ne0$, then $\mathrm{Holant}(=_2|\mathcal{F},f)$ is \#P-hard.
\end{lemma}
\begin{proof}
    We prove only the case in which \[M(f)=\begin{pmatrix}
  1&  0&  0& 1\\
 0 &  c&0  & 0\\
 0 & 0 & z &0 \\
  1& 0 &0  &1
\end{pmatrix}.\] The proof for the other case is analogous and is omitted here. By connecting variables $x_3,x_4$ of one copy of $f$ with, respectively, variables $x_3,x_4$ of another copy of $f$, we construct a 4-ary signature $f_1$ with the signature matrix \[M(f_1)=M(f)_{x_1,x_2,x_3,x_4}M(f)_{x_4,x_3,x_2,x_1}=\begin{pmatrix}
  1&  0&  0& 1\\
 0 &  \frac{cz}{2}&0  & 0\\
 0 & 0 & \frac{cz}{2} &0 \\
  1& 0 &0  &1
\end{pmatrix}\] up to the scalar $2$. By a holographic transformation using $Z$, we have \[
\mathrm{Holant}(\ne_2|\mathcal{F'},f',f_1')\equiv_T\mathrm{Holant}(=_2|\mathcal{F},f,f_1),
\]where $\mathcal{F'}=(Z^{-1})^{\otimes4}\mathcal{F}$, $f'=(Z^{-1})^{\otimes4}f$ and $f_1'=(Z^{-1})^{\otimes4}f_1$ with the signature matrix \[M(f_1')=\begin{pmatrix}
  -\frac{cz}{2}& 0 &0  & \frac{cz}{2}\\
  0& 2+\frac{cz}{2} & 2-\frac{cz}{2} & 0\\
  0&  2-\frac{cz}{2}&  2+\frac{cz}{2}&0 \\
 \frac{cz}{2} &  0& 0 &-\frac{cz}{2}
\end{pmatrix}.\] Note that $f_1'$ is clearly not in $\mathcal{A}$. By Lemma~\ref{foursamepair}, we conclude that $\mathrm{Holant}(\ne_2|\mathcal{F'},f',f_1')$ is \#P-hard.

\end{proof}

\begin{lemma}\label{onezeropair}
    Let $f$ be a 4-ary signature with signature matrix  \[M(f)=\begin{pmatrix}
 a &  0& 0 &b \\
  0& c & 0 & 0\\
 0 & 0 &  z& 0\\
  y&  0&  0&x
\end{pmatrix}\] with $abcxyz\ne0$. If $ax = by$ or $ax = cz$, then $\mathrm{Holant}(=_2|\mathcal{F},f)$ is \#P-hard. If $ax\ne by$ and $ax\ne cz$, further assume that we have three distinct binary signatures of the form $(1,0,0,t)$, then we can construct a degenerate binary signature unless $\mathrm{Holant}(=_2|\mathcal{F},f)$ is \#P-hard.

\end{lemma}
\begin{proof}
We first deal with the case in which $ax=by$, (by symmetry, the case in which $ax=cz$ is the same), then \[M(f)=\begin{pmatrix}
 a &  0& 0 &b \\
  0& c & 0 & 0\\
 0 & 0 &  z& 0\\
  ka&  0&  0&kb
\end{pmatrix},\] where $k\ne0$. We consider the following cases:
\begin{caselist}
    \caseitem{1}  $a+b\ne0$ and $1+k\ne0$.

    For $f_{=_2}^{34}=(a+b)(1,0,0,k)$, $f_{=_2}^{12}=a(1+k)(1,0,0,\frac{b}{a})$, then we can construct $(f_{=_2}^{34})^{-1}$ and $(f_{=_2}^{12})^{-1}$. By connecting $(f_{=_2}^{34})^{-1}$ to input $x_1$ of $f$ and $(f_{=_2}^{12})^{-1}$ to input $x_3$ of $f$, we obtain a new 4-ary signature with matrix \[M(f^*)=\begin{pmatrix}
 1 &  0& 0 &1 \\
  0& c^* & 0 & 0\\
 0 & 0 &  z^*& 0\\
  1&  0&  0&1
\end{pmatrix}\] up to a scalar $a$, where $c^*=\frac{c}{a}$, $z^*=\frac{z}{kb}$. By Lemma~\ref{four1} the problem is \#P-hard.
\caseitem{2}  $a+b=0$ and $1+k\ne0$.

        For $f_{=_2}^{12}=a(1+k)(1,0,0,\frac{b}{a})$, we construct $[1,0,-1]$. For $f_{[1,0,-1]}^{34}=(a-b)(1,0,0,k)$, and then we can construct $(f_{[1,0,-1]}^{34})^{-1}$ since $a-b\ne0$. By connecting $(f_{[1,0,-1]}^{34})^{-1}$ to input $x_1$ of $f$ and $[1,0,-1]$ to input $x_3$ of $f$, we obtain a new 4-ary signature with signature matrix \[M(f^*)=\begin{pmatrix}
 1 &  0& 0 &1 \\
  0& c^* & 0 & 0\\
 0 & 0 &  z^*& 0\\
  1&  0&  0&1
\end{pmatrix}\] up to $a$. By Lemma~\ref{four1} the problem is \#P-hard.
\caseitem{3}  $a+b\ne0$ and $1+k=0$. By symmetry, this case reduces to \caseref{2}.
\caseitem{4}  $a+b=1+k=0$. In this case, \[M(f)=\begin{pmatrix}
 1 &  0& 0 &-1 \\
  0& c^* & 0 & 0\\
 0 & 0 &  z^*& 0\\
  -1&  0&  0&1
\end{pmatrix}\] up to $a$, where $c^*=\frac{c}{a}$ and $z^*=\frac{z}{a}$. By Lemma~\ref{four1} the problem is \#P-hard.
\end{caselist}
We conclude that if $ax=by$ or $ax=cz$, Holant($=_2|f$) is \#P-hard.

Then we deal with the case in which $ax\ne by$, $ax\ne cz$ and we have three binary signatures of the form $(1,0,0,t)$. If no degenerate binary signature can be constructed, Lemma~\ref{mobius} and symmetry imply that \[M(f)=\begin{pmatrix}
 e^{\frak{i}\theta_1} &  0& 0 &e^{\frak{i}\theta_1}\lambda_1 \\
  0& e^{\frak{i}\theta_1}\lambda_2 & 0 & 0\\
 0 & 0 &  \overline{\lambda_2}& 0\\
  \overline{\lambda_1}&  0&  0&1
\end{pmatrix}\] up to a nonzero scalar $x$, where $|\lambda_1|\ne1$ and $|\lambda_2|\ne1$, and $\begin{pmatrix}
  e^{\frak{i}\theta_1}&e^{\frak{i}\theta_1}\lambda_1 \\
 \overline{\lambda_1} &1
\end{pmatrix}$, $\begin{pmatrix}
  e^{\frak{i}\theta_1}&e^{\frak{i}\theta_1}\lambda_2 \\
 \overline{\lambda_2} &1
\end{pmatrix}$ has finite projective order. Note that $det\begin{pmatrix}
  e^{\frak{i}\theta_1}&e^{\frak{i}\theta_1}\lambda_1 \\
 \overline{\lambda_1} &1
\end{pmatrix}=e^{\frak{i}\theta_1}(1-|\lambda_1|^2)$. Assume that $\begin{pmatrix}
  e^{\frak{i}\theta_1}&e^{\frak{i}\theta_1}\lambda_1 \\
 \overline{\lambda_1} &1
\end{pmatrix}^n=\begin{pmatrix}
  e^{\frac{n}{2}\frak{i}\theta_1}(1-|\lambda_1|^2)^{\frac{n}{2}}&0 \\
 0 &e^{\frac{n}{2}\frak{i}\theta_1}(1-|\lambda_1|^2)^{\frac{n}{2}}
\end{pmatrix}$. Then we construct $f_n$ by a chain of $n$ copies of $f_{x_1,x_2,x_3,x_4}$, and \[ M(f^n)=M_{x_1,x_2,x_3,x_4}(f)^n=\begin{pmatrix}
  e^{\frac{n}{2}\frak{i}\theta_1}(1-|\lambda_1|^2)^{\frac{n}{2}} &  0& 0 &0\\
  0& e^{n\frak{i}\theta_1}\lambda_2^n& 0 & 0\\
 0 & 0 & \overline{\lambda_2}^n& 0\\
  0 &  0&  0& e^{\frac{n}{2}\frak{i}\theta_1}(1-|\lambda_1|^2)^{\frac{n}{2}}
\end{pmatrix}.
\]If $e^{n\frak{i}\theta_1}\lambda_2^n\overline{\lambda_2}^n=e^{n\frak{i}\theta_1}|\lambda_2|^{2n}\ne e^{n\frak{i}\theta_1}(1-|\lambda_1|^2)^n$, then we can construct $=_4$. If $e^{n\frak{i}\theta_1}\lambda_2^n\overline{\lambda_2}^n=e^{n\frak{i}\theta_1}|\lambda_2|^{2n}=e^{n\frak{i}\theta_1}(1-|\lambda_1|^2)^n$, note that $|\lambda_1|,|\lambda_2|$ are real numbers, we have $|\lambda_1|^2+|\lambda_2|^2=1$ or $|\lambda_1|^2-|\lambda_2|^2=1$. By symmetry in $\lambda_1$ and $\lambda_2$, $|\lambda_1|^2-|\lambda_2|^2=1$ implies that $|\lambda_1|^2=1$, a contradiction. If $|\lambda_1|^2+|\lambda_2|^2=1$, by connecting $(e^{-\frak{i}\theta_1},0,0,1)$ to input $x_1$ of $f$, we construct a new 4-ary signature \[M(f^*)=\begin{pmatrix}
 1 &  0& 0 &\lambda_1 \\
  0& \lambda_2 & 0 & 0\\
 0 & 0 &  \overline{\lambda_2}& 0\\
  \overline{\lambda_1}&  0&  0&1
\end{pmatrix}.\] Then \[
M(g)=M_{x_1,x_4,x_3,x_2}(f^*)M_{x_2,x_3,x_4,x_1}(f^*)=\begin{pmatrix}
 1 &  0& 0 &0 \\
  0& 1 & 2\lambda_1\lambda_2 & 0\\
 0 & 2\overline{\lambda_1}\overline{\lambda_2} &  1& 0\\
  0&  0&  0&1
\end{pmatrix}
,\] then by Lemma~\ref{four1} the problem is \#P-hard.

\end{proof}

\begin{lemma}\label{a=x=0}
     Let $f$ be a 4-ary signature with the signature matrix \[M(f)=\begin{pmatrix}
 0 &  0& 0 &b \\
  0& c & d & 0\\
 0 & w &  z& 0\\
  y&  0&  0&0
\end{pmatrix}.\] Then $\mathrm{Holant}(=_2|f)$ is \#P-hard except for the following cases:
\begin{itemize}
        \item $f \in \left \langle Z\mathcal{M} \right \rangle$ or $\in \left \langle ZX\mathcal{M} \right \rangle$;
        \item $f$ is $ \mathcal{A}\text{-}$transformable;
        \item $f$ is $ \mathcal{L}\text{-}$transformable;
        \item $f$ is $ \mathcal{P}\text{-}$transformable;
        \item $f\in\mathcal{H}$;
        \item $f \in \left \langle \mathcal{T} \right \rangle$,
    \end{itemize}
    in which cases the problem is computable in polynomial time.
\end{lemma}
\begin{proof}
     Note that for any pair $(s,t)$, we can construct a binary signature which is $[s,0,t]$. If there is a pair containing exactly one zero, then we can construct a degenerate binary signature, a contradiction. If there is a zero pair, permute inputs so that the remaining entries form the middle $2\times2$ block. If this block has rank at most one, then $f\in\langle\mathcal T\rangle$. Otherwise, Lemma~\ref{twozeropair} supplies $=_4$, and then by Lemma~\ref{toCSP2} together with Theorem~\ref{CSP2dich} the result follows. Then we assume that $byczwd\ne0$. Note that any binary signature must have finite order. Then \[M(f)=\begin{pmatrix}
 0 &  0& 0 &b \\
  0& c & d & 0\\
 0 & e^{\frak{i}\theta_3}d &  e^{\frak{i}\theta_2}c& 0\\
  e^{\frak{i}\theta_1}b&  0&  0&0
\end{pmatrix}.\] If $e^{2\frak{i}\theta_2}\ne 1$, then we can construct at least three distinct binary signatures of the form $(1,0,0,t)$.
For
\[M_{x_1,x_2,x_3,x_4}(f)^2=\begin{pmatrix}
 e^{\frak{i}\theta_1}b^2 &  0& 0 &0 \\
  0& c^2+e^{\frak{i}\theta_3}d^2 & cd(1+e^{\frak{i}\theta_2}) & 0\\
 0 & e^{\frak{i}\theta_3}cd(1+e^{\frak{i}\theta_2}) &  e^{2\frak{i}\theta_2}c^2+e^{\frak{i}\theta_3}d^2& 0\\
  0&  0&  0&e^{\frak{i}\theta_1}b^2
\end{pmatrix},\] $e^{\frak{i}\theta_1}b^2\ne0$, $cd(1+e^{\frak{i}\theta_2})\ne0$ and there is at most one zero in $c^2+e^{\frak{i}\theta_3}d^2$ and $e^{2\frak{i}\theta_2}c^2+e^{\frak{i}\theta_3}d^2$ since $c^2+e^{\frak{i}\theta_3}d^2\ne e^{2\frak{i}\theta_2}c^2+e^{\frak{i}\theta_3}d^2$. If there exists one zero in $c^2+e^{\frak{i}\theta_3}d^2$ and $e^{2\frak{i}\theta_2}c^2+e^{\frak{i}\theta_3}d^2$, by Lemma~\ref{mobius} we can construct a degenerate binary signature. If all are nonzero, the result follows from Lemma~\ref{onezeropair}.

Then by symmetry, we can assume that $e^{2\frak{i}\theta_1}=e^{2\frak{i}\theta_2}=e^{2\frak{i}\theta_3}=1$. If $e^{\frak{i}\theta_1}=e^{\frak{i}\theta_2}=e^{\frak{i}\theta_3}=\pm1$, then $f_{=_2}=\{=_2\}$ or $f_{=_2}=\{[1,0,-1]\}$, the result follows from Lemma~\ref{=+2and=-2}. Then by symmetry there are two different cases, namely $e^{\frak{i}\theta_1}=-e^{\frak{i}\theta_2}=-e^{\frak{i}\theta_3}=1$ and $-e^{\frak{i}\theta_1}=e^{\frak{i}\theta_2}=e^{\frak{i}\theta_3}=1$. If $e^{\frak{i}\theta_1}=-e^{\frak{i}\theta_2}=-e^{\frak{i}\theta_3}=1$,
then
\[M_{x_1,x_2,x_3,x_4}(f)^2=\begin{pmatrix}
 b^2 &  0& 0 &0 \\
  0& c^2-d^2 & 0 & 0\\
 0 & 0 &  c^2-d^2& 0\\
  0&  0&  0&b^2
\end{pmatrix},\]
\[M_{x_1,x_4,x_3,x_2}(f)^2=\begin{pmatrix}
 -d^2 &  0& 0 &0 \\
  0& c^2+b^2 & 0 & 0\\
 0 & 0 &  c^2+b^2& 0\\
  0&  0&  0&-d^2
\end{pmatrix},\]
\[M_{x_1,x_3,x_4,x_2}(f)^2=\begin{pmatrix}
 -c^2 &  0& 0 &0 \\
  0& b^2+d^2 & 0 & 0\\
 0 & 0 &  b^2+d^2& 0\\
  0&  0&  0&-c^2
\end{pmatrix}.\] By Lemma~\ref{g=_4}, failure of all three displayed constructions to yield $=_4$ requires \[\left\{\begin{matrix}
 (c^2-d^2)^2=b^4\\
 (b^2+c^2)^2=d^4\\
(b^2+d^2)^2=c^4
\end{matrix}\right.\] Set $u=b^2$, $v=c^2$, and $w=d^2$. The first two equations imply $2v(u+v-w)=0$. Since $v\ne0$, we have $w=u+v$. The third equation now gives $4uw=0$, contradicting $uw\ne0$. Thus we can construct $=_4$. Lemma~\ref{toCSP2} and Theorem~\ref{CSP2dich} give the stated hardness or tractable alternatives.

If $-e^{\frak{i}\theta_1}=e^{\frak{i}\theta_2}=e^{\frak{i}\theta_3}=1$, note that we construct $[1,0,-1]$ in this case. By connecting $[1,0,-1]$ to input $x_1$ of $f$, the new 4-ary signature \[M(f^*)=\begin{pmatrix}
 0 &  0& 0 &b \\
  0& c & d & 0\\
 0 & -d &  -c& 0\\
  b&  0&  0&0
\end{pmatrix},\] then the proof is the same as above.
\end{proof}
\begin{lemma}\label{specialg}
    Let $g_1,g_2,g_3,g_4$ be 4-ary signatures, with signature matrices \[M(g_1)=\begin{pmatrix}
 1 &  0& 0 &1 \\
  0& 1 & 1 & 0\\
 0 & 1 &  1& 0\\
  1&  0&  0&1
\end{pmatrix},\] \[M(g_2)=\begin{pmatrix}
 1 &  0& 0 &-1 \\
  0& 1 & 1 & 0\\
 0 & 1 &  1& 0\\
  -1&  0&  0&1
\end{pmatrix},\]
\[M(g_3)=\begin{pmatrix}
 1 &  0& 0 &1 \\
  0& -1 & -1 & 0\\
 0 & -1 &  -1& 0\\
  1&  0&  0&1
\end{pmatrix},\] \[M(g_4)=\begin{pmatrix}
 1 &  0& 0 &-1 \\
  0& -1 & -1 & 0\\
 0 & -1 &  -1& 0\\
  -1&  0&  0&1
\end{pmatrix}.\] Let $h=(1,0,0,t)$ be a binary signature with $t\ne\pm1,0$, then the following four reductions hold:
\begin{itemize}
    \item $\mathrm{CSP^2}((H^{-1})\mathcal{F})\le_T\mathrm{Holant}(=_2|\mathcal{F},g_1)$;
    \item $\mathrm{CSP^2}((Z^{-1})\mathcal{F},[1,\frac{1+t}{1-t},1])\le_T\mathrm{Holant}(=_2|\mathcal{F},h,g_2)$;
    \item $\mathrm{CSP^2}((H^{-1})\mathcal{F})\le_T\mathrm{Holant}(=_2|\mathcal{F},g_3)$;
    \item $\mathrm{CSP^2}((Z^{-1})\mathcal{F},[1,\frac{1+t}{1-t},1])\le_T\mathrm{Holant}(=_2|\mathcal{F},h,g_4)$.
\end{itemize}
\end{lemma}
\begin{proof}
Let $H=\frac{-1}{\sqrt{2}}\begin{pmatrix}
 1 & 1\\
 1 &-1
\end{pmatrix}$.
Note that

$(H^{-1})^{\otimes 4}g_1=2(=_4)$; $(H^{-1})^{\otimes 4}g_3=2(\ne_4)$;
$(H^{-1})^{\otimes 2}h=\tfrac12[1+t,1-t,1+t]=\tfrac{1-t}{2}[\frac{1+t}{1-t},1,\frac{1+t}{1-t}]$;

$(Z^{-1})^{\otimes 4}g_2=2(\ne_4)$;
$(Z^{-1})^{\otimes 4}g_4=2(=_4)$;
$(Z^{-1})^{\otimes 2}h=\tfrac12[1-t,1+t,1-t]=\tfrac{1-t}{2}[1,\frac{1+t}{1-t},1]$.

Consider the following four cases.

    \caseheading{1} In $\mathrm{Holant}(=_2|\mathcal{F},g_1)$, by a holographic transformation using $H$, we obtain the stated reduction.

    \caseheading{2} In $\mathrm{Holant}(=_2|\mathcal{F},h,g_2)$, by a holographic transformation using $Z$, we have
    \[
    \mathrm{Holant}(\ne_2|(Z^{-1})\mathcal{F},[1,\frac{1+t}{1-t},1],\ne_4)\equiv_T\mathrm{Holant}(=_2|\mathcal{F},h,g_2).
    \]
    Note that if $t\ne1$, then
    ${(\ne_4)}_{[1,\frac{1+t}{1-t},1]}^{12}=(=_2)$. Then by Lemma~\ref{g=_4} we have
    \[\begin{aligned}
    \mathrm{CSP^2}((Z^{-1})\mathcal{F},[1,\frac{1+t}{1-t},1])&\le_T\mathrm{Holant}(=_2,\ne_2|(Z^{-1})\mathcal{F},[1,\frac{1+t}{1-t},1],\ne_4)\\ &\equiv_T\mathrm{Holant}(\ne_2|(Z^{-1})\mathcal{F},[1,\frac{1+t}{1-t},1],\ne_4).
    \end{aligned}\]

    \caseheading{3} In $\mathrm{Holant}(=_2|\mathcal{F},g_3)$, by a holographic transformation using $H$, we have \[
    \mathrm{Holant}(=_2|(H^{-1})\mathcal{F},\ne_4)\equiv_T\mathrm{Holant}(=_2|\mathcal{F},g_3).
    \] Then by Lemma~\ref{g=_4} we have \[
    \mathrm{CSP^2}((H^{-1})\mathcal{F})\le_T \mathrm{Holant}(=_2|(H^{-1})\mathcal{F},\ne_4)\equiv_T\mathrm{Holant}(=_2|\mathcal{F},g_3).
    \]

    \caseheading{4} In $\mathrm{Holant}(=_2|\mathcal{F},h,g_4)$, by a holographic transformation using $Z$, we have
    \[
    \mathrm{Holant}(\ne_2|(Z^{-1})\mathcal{F},[1,\frac{1+t}{1-t},1],=_4)\equiv_T\mathrm{Holant}(=_2|\mathcal{F},h,g_4).
    \] Note that if $t\ne1$, then
    ${(=_4)}_{[1,\frac{1+t}{1-t},1]}^{12}= (=_2)$. Then by Lemma~\ref{g=_4} we have
    \[\begin{aligned}
    \mathrm{CSP^2}((Z^{-1})\mathcal{F},[1,\frac{1+t}{1-t},1])&\le_T\mathrm{Holant}(=_2,\ne_2|(Z^{-1})\mathcal{F},[1,\frac{1+t}{1-t},1],=_4)\\ &\equiv_T\mathrm{Holant}(\ne_2|(Z^{-1})\mathcal{F},[1,\frac{1+t}{1-t},1],=_4).
    \end{aligned}\]

\end{proof}
\begin{lemma}\label{allnonzero}
     Let $f$ be a 4-ary signature with the signature matrix \[M(f)=\begin{pmatrix}
 a &  0& 0 &b \\
  0& c & d & 0\\
 0 & w &  z& 0\\
  y&  0&  0&x
\end{pmatrix}\] with $abcdxywz\ne0$. Assume that we can construct three distinct binary signatures of the form $(1,0,0,t)$. Then for $\mathrm{Holant}(=_2|\mathcal{F},f)$, one of the following holds:
\begin{itemize}
        \item Either $\mathrm{Holant}(=_2|\mathcal{F},f)$ is \#P-hard;
        \item or we can construct a degenerate binary signature;
        \item or there exists $T\in \mathrm{SL}(2,\mathbb{C})$, such that $\mathrm{CSP}^2(\mathcal{F'},f')\le_T\mathrm{Holant}(=_2|\mathcal{F},f)$, where $f'=(T^{-1})^{\otimes 4}f$, $\mathcal{F'}=(T^{-1})\mathcal{F}$.
    \end{itemize}
\end{lemma}
\begin{proof}
We consider the following cases up to symmetry:
\begin{caselist}
    \caseitem{1}  $by,cz,dw\ne ax$. In this case, by Lemma~\ref{mobius}, we can construct a degenerate binary signature or \[M(f)=\begin{pmatrix}
 e^{\frak{i}\theta_1} &  0& 0 &e^{\frak{i}\theta_1}\lambda_1 \\
  0& e^{\frak{i}\theta_1}\lambda_2 & e^{\frak{i}\theta_1}\lambda_3 & 0\\
 0 & \overline{\lambda_3} &  \overline{\lambda_2}& 0\\
  \overline{\lambda_1}&  0&  0&1
\end{pmatrix}\] up to a nonzero scalar $x$, where $|\lambda_1|,|\lambda_2|,|\lambda_3|\ne 1 ,0$. Note that for $1\le j\le3$, $\begin{pmatrix}
  e^{\frak{i}\theta_1}&e^{\frak{i}\theta_1}\lambda_j \\
  \overline{\lambda_j}&1
\end{pmatrix}$ has finite projective order. Assume that twice the projective order of $\begin{pmatrix}
  e^{\frak{i}\theta_1}&e^{\frak{i}\theta_1}\lambda_j \\
  \overline{\lambda_j}&1
\end{pmatrix}$ is $n_j$. Define $d_j= e^{\frak{i}\theta_1}-e^{\frak{i}\theta_1}\lambda_j\overline{\lambda_j}$. This even exponent makes the scalar exactly $d_j^{n_j/2}$: if $A^r=cI$, then $A^{2r}=(\det A)^rI$.
We have \[
M_{x_1,x_2,x_3,x_4}(f)^{n_1}=\begin{pmatrix}
 d_1^{\frac{n_1}{2} } & \mathbf{0}  & 0\\
  \mathbf{0}& \begin{bmatrix}
 e^{\frak{i}\theta_1}\lambda_2 &e^{\frak{i}\theta_1}\lambda_3 \\
 \overline{\lambda_3} &\overline{\lambda_2}
\end{bmatrix}^{n_1} &\mathbf{0} \\
 0 & \mathbf{0} & d_1^{\frac{n_1}{2} }
\end{pmatrix}.
\] In $\begin{bmatrix}
 e^{\frak{i}\theta_1}\lambda_2 &e^{\frak{i}\theta_1}\lambda_3 \\
 \overline{\lambda_3} &\overline{\lambda_2}
\end{bmatrix}^{n_1}$, if there are four nonzero entries, the result follows from Lemma~\ref{onezeropair}. If there exists a pair with exactly one zero, we construct a degenerate binary signature by Lemma~\ref{mobius}. If all entries are zero, we construct $=_4$. Then we can assume that $\begin{bmatrix}
 e^{\frak{i}\theta_1}\lambda_2 &e^{\frak{i}\theta_1}\lambda_3 \\
 \overline{\lambda_3} &\overline{\lambda_2}
\end{bmatrix}^{n_1}$ has the form $\begin{bmatrix}
 s &0 \\
 0 &t
\end{bmatrix}$ or $\begin{bmatrix}
 0 &s \\
 t &0
\end{bmatrix}$. Note that $st=\pm det\begin{bmatrix}
 e^{\frak{i}\theta_1}\lambda_2 &e^{\frak{i}\theta_1}\lambda_3 \\
 \overline{\lambda_3} &\overline{\lambda_2}
\end{bmatrix}^{n_1}$, if $st\ne d_1^{n_1}$, we construct $=_4$ by Lemma~\ref{g=_4}. If $st=d_1^{n_1}$, we have $e^{n_1\frak{i}\theta_1}(|\lambda_2|^2-|\lambda_3|^2)^{n_1}= e^{n_1\frak{i}\theta_1}(1-|\lambda_1|^2)^{n_1}$ or $e^{n_1\frak{i}\theta_1}(|\lambda_3|^2-|\lambda_2|^2)^{n_1}= e^{n_1\frak{i}\theta_1}(1-|\lambda_1|^2)^{n_1}$. We consider only the case in which $e^{n_1\frak{i}\theta_1}(|\lambda_2|^2-|\lambda_3|^2)^{n_1}= e^{n_1\frak{i}\theta_1}(1-|\lambda_1|^2)^{n_1}$, the other case is analogous.

Write $r_j=|\lambda_j|^2$. Taking absolute values in the preceding determinant relation and repeating the argument for all three choices of the outer block gives
\[
(1-r_1)^2=(r_2-r_3)^2,\quad (1-r_2)^2=(r_1-r_3)^2,\quad (1-r_3)^2=(r_1-r_2)^2.
\]
Subtracting the first two equations gives $2(r_2-r_1)(1-r_3)=0$. Since every $r_j\ne1$, cyclic symmetry gives $r_1=r_2=r_3$. Substitution into any equation yields $r_1=1$, a contradiction.
\caseitem{2}  $by,cz\ne ax$ and $dw=ax$. In this case, by Lemma~\ref{mobius}, we can construct a degenerate binary signature or \[M(f)=\begin{pmatrix}
 e^{\frak{i}\theta_1} &  0& 0 &e^{\frak{i}\theta_1}\lambda_1 \\
  0& e^{\frak{i}\theta_1}\lambda_2 & e^{\frak{i}\theta_1}d & 0\\
 0 & \frac{1}{d} &  \overline{\lambda_2}& 0\\
  \overline{\lambda_1}&  0&  0&1
\end{pmatrix}\] up to a nonzero scalar $x$, where $|\lambda_1|,|\lambda_2|\ne1$. Write $r=|\lambda_1|^2$ and $s=|\lambda_2|^2$. The same determinant argument as in \caseref{1} gives $(1-r)^2=(s-1)^2$, so $r+s=2$ or $r=s$. Both $r$ and $s$ are positive and different from $1$.

For any available parameter $t$, the contractions on inputs $2,4$ and $1,3$ have diagonal ratios $d$ and $e^{-\frak{i}\theta_1}/d$, respectively, whenever they are nonzero. Each contraction vanishes for at most one of the three distinct parameters. Chains and inverses therefore supply the diagonal modifications needed below for every $d\ne0$.
Apply $(e^{-\frak{i}\theta_1},0,0,d)$ to input $1$ and $(1,0,0,d^{-1})$ to input $3$. The resulting signature is
\[M(f^*)=\begin{pmatrix}
1&0&0&\frac{\lambda_1}{d}\\
0&\lambda_2&1&0\\
0&1&\overline{\lambda_2}&0\\
\overline{\lambda_1}d&0&0&1
\end{pmatrix}.\]
After dividing out the known scalar $2$, the chained gadget has matrix
\[M(g)=\tfrac12M_{x_3,x_2,x_1,x_4}(f^*)M_{x_4,x_1,x_2,x_3}(f^*)=\begin{pmatrix}
1&0&0&1\\
0&(r+s)/2&\overline{\lambda_1}\lambda_2d&0\\
0&\lambda_1\overline{\lambda_2}/d&(r+s)/2&0\\
1&0&0&1
\end{pmatrix}.\]
If $r+s=2$, permute inputs so that the middle block becomes the outer block. Its eigenvalues are $1\pm\sqrt{rs}$. Since $0<rs<1$, their ratio is a positive real number different from $1$, so this block is invertible and has infinite projective order. Lemma~\ref{mobius} gives a degenerate binary signature.

If $r=s$, the permutation $(1,3,2,4)$ instead gives the outer block $\begin{psmallmatrix}1&r\\r&1\end{psmallmatrix}$. Its eigenvalues are $1+r$ and $1-r$, and their ratio has modulus different from $1$ because $r>0$ and $r\ne1$. Lemma~\ref{mobius} again applies.
\caseitem{3}  $by\ne ax$ and $cz=dw=ax$. In this case, by Lemma~\ref{mobius}, we can construct a degenerate binary signature or \[M(f)=\begin{pmatrix}
 e^{\frak{i}\theta_1} &  0& 0 &e^{\frak{i}\theta_1}\lambda_1 \\
  0& e^{\frak{i}\theta_1}c & e^{\frak{i}\theta_1}d & 0\\
 0 & \frac{1}{d} &  \frac{1}{c}& 0\\
  \overline{\lambda_1}&  0&  0&1
\end{pmatrix}\] up to a nonzero scalar $x$, where $|\lambda_1|\ne1$. Let $L=\begin{psmallmatrix}e^{\frak{i}\theta_1}c&e^{\frak{i}\theta_1}d\\1/d&1/c\end{psmallmatrix}$. This nonzero matrix has rank one, so $L^n=(\operatorname{tr}L)^{n-1}L$ for $n\ge1$.

The outer block has finite projective order, or Lemma~\ref{mobius} already applies. Choose a positive even multiple $n$ of that order, so its $n$th power is a nonzero scalar matrix. If $\operatorname{tr}L=0$, then $L^n=0$ and the chain yields $=_4$ up to scalar. Otherwise, all four entries of $L^n$ are nonzero, and Lemma~\ref{onezeropair} applies.
\caseitem{4}  $ax=by=cz=dw$. Then \[M(f)=\begin{pmatrix}
 1 &  0& 0 &b \\
  0& c & d & 0\\
 0 & \frac{by}{d} &  \frac{by}{c}& 0\\
  y&  0&  0&by
\end{pmatrix}\] by normalizing $a=1$. We can always construct $(1,0,0,\frac{1}{b})$ and $(1,0,0,\frac{1}{y})$ since we have three distinct
binary signatures of the form $(1, 0, 0, t)$. Then we have $(1,0,0,\frac{1}{by})$. By connecting $(1,0,0,\frac{1}{by})$ to input $x_1$ of $f$, we have a new 4-ary signature \[M(f^*)=\begin{pmatrix}
 1 &  0& 0 &b \\
  0& c & d & 0\\
 0 & \frac{1}{d} &  \frac{1}{c}& 0\\
  \frac{1}{b}&  0&  0&1
\end{pmatrix}.\]
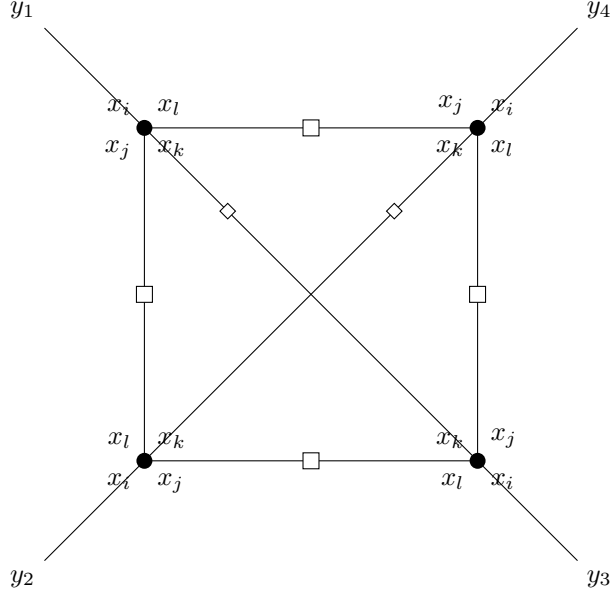
\begin{figure}[tbp]
\centering
\begin{tikzpicture}[
    scale=2.2,
    vertex/.style={circle, fill=black, minimum size=6pt, inner sep=0pt},
    square/.style={rectangle, draw, fill=white, minimum size=6pt, inner sep=0pt},
    diamond/.style={shape=diamond, draw, fill=white, minimum size=6pt, inner sep=0pt},
    every node/.style={font=\normalsize}
]

\coordinate (v1) at (-1, 1);  
\coordinate (v2) at (1, 1);   
\coordinate (v3) at (1, -1);  
\coordinate (v4) at (-1, -1); 

\draw (v1) -- (v2) -- (v3) -- (v4) -- cycle; 
\draw (v1) -- (v3); 
\draw (v2) -- (v4); 
\draw (v1) -- (-1.6, 1.6) node[above left] {$y_1$};
\draw (v4) -- (-1.6, -1.6) node[below left] {$y_2$};
\draw (v3) -- (1.6, -1.6) node[below right] {$y_3$};
\draw (v2) -- (1.6, 1.6) node[above right] {$y_4$};

\node[square] at (0, 1) {};   
\node[square] at (1, 0) {};   
\node[square] at (0, -1) {};  
\node[square] at (-1, 0) {};  
\node[diamond] at (-0.5, 0.5) {}; 
\node[diamond] at (0.5, 0.5) {};  

\node[vertex] at (v1) {};
\node[vertex] at (v2) {};
\node[vertex] at (v3) {};
\node[vertex] at (v4) {};

\node[above left=2pt of v1] {$x_i$};
\node[above right=2pt of v1] {$x_l$};
\node[below left=2pt of v1] {$x_j$};
\node[below right=2pt of v1] {$x_k$};

\node[above left=2pt of v2] {$x_j$};
\node[above right=2pt of v2] {$x_i$};
\node[below left=2pt of v2] {$x_k$};
\node[below right=2pt of v2] {$x_l$};

\node[above left=2pt of v3] {$x_k$};
\node[above right=2pt of v3] {$x_{j}$};
\node[below left=2pt of v3] {$x_l$};
\node[below right=2pt of v3] {$x_i$};

\node[above left=2pt of v4] {$x_l$};
\node[above right=2pt of v4] {$x_k$};
\node[below left=2pt of v4] {$x_i$};
\node[below right=2pt of v4] {$x_j$};

\end{tikzpicture}
\caption{The circles are assigned $f^*$, and the squares and diamonds are assigned $=_2$. The gadget produces
a rotationally symmetric signature.}
\label{fig:tensor_network}
\end{figure}

The gadget in Figure~\ref{fig:tensor_network} is rotationally symmetric. For $(i,j,k,l)=(1,2,3,4)$, use cyclic edge variables $e_1,e_2,e_3,e_4$ and diagonal variables $k_1,k_2$. Its signature is the sum of $\prod_{v=1}^4 f^*(Y_v,e_{v-1},K_v,e_v)$ over these six binary variables, where $e_0=e_4$, $(Y_1,Y_2,Y_3,Y_4)=(y_1,y_4,y_3,y_2)$, and $(K_1,K_2,K_3,K_4)=(k_1,k_2,k_1,k_2)$. Let $A(b,c,d)=c^4+2b^2d^2+4bcd+1$, $B(b,c,d)=bd+2c+\frac{c^2}{bd}+\frac{bd}{c^2}+\frac{2}{c}+\frac{1}{bd}$, $C(b,c,d)=c^2+\frac{2bd}{c}+\frac{2c}{bd}+\frac{1}{c^2}+2$, $D(b,c,d)=\frac{4}{bcd}+\frac{2}{b^2d^2}+\frac{1}{c^4}+1$.
We have:
\begin{itemize}
    \item If the ordered tuple $(i, j, k, l) = (1, 2, 3, 4)$, then the gadget gives a signature $g^{bcd}$ and
    \[
    M_{y_1,y_3,y_2,y_4}(g^{bcd})=\begin{pmatrix}
 A(b,c,d) &  0& 0 & C(b,c,d)\\
 0&   B(b,c,d)&  B(b,c,d) &0 \\
 0 &  B(b,c,d) & B(b,c,d)& 0\\
 C(b,c,d) &  0&  0& D(b,c,d)
\end{pmatrix}
    \]
    \item If $(i, j, k, l) = (1, 3, 2, 4)$, then the gadget gives a signature $g^{cbd}$ and  \[
    M_{y_1,y_3,y_2,y_4}(g^{cbd})=\begin{pmatrix}
 A(c,b,d) &  0& 0 & C(c,b,d)\\
 0&   B(c,b,d)&  B(c,b,d) &0 \\
 0 &  B(c,b,d) & B(c,b,d)& 0\\
 C(c,b,d) &  0&  0& D(c,b,d)
\end{pmatrix}
    \]
    \item If $(i, j, k, l) = (1, 2, 4, 3)$, then the gadget gives a signature $g^{bdc}$ and  \[
    M_{y_1,y_3,y_2,y_4}(g^{bdc})=\begin{pmatrix}
 A(b,d,c) &  0& 0 & C(b,d,c)\\
 0&   B(b,d,c)&  B(b,d,c) &0 \\
 0 &  B(b,d,c) & B(b,d,c)& 0\\
 C(b,d,c) &  0&  0& D(b,d,c)
\end{pmatrix}
    \]
\end{itemize}
Let $M(b,c,d)=\begin{bmatrix}
 A(b,c,d) & C(b,c,d)\\
  C(b,c,d)&D(b,c,d)
\end{bmatrix}$ denote a $2 \times2$ matrix.
If $\det(M(b,c,d))\ne0$ or $\det(M(c,b,d))\ne0$ or $\det(M(b,d,c))\ne0$, then either we can construct $=_4$, or we have Holant$(=_2|g^{bcd})$(or $g^{cbd}$, or $g^{bdc}$) is \#P-hard.
If $\det(M(b,c,d))=\det(M(c,b,d))=\det(M(b,d,c))=0$, we have:
\[
\left\{\begin{matrix}
2\frac{1}{b^2d^2c^4}(bd+c)^2(bd-c)^2(c+1)^2(c-1)^2=0 \\
 2\frac{1}{b^4d^2c^2}(cd+b)^2(cd-b)^2(b+1)^2(b-1)^2=0\\
2\frac{1}{b^2d^4c^2}(bc+d)^2(bc-d)^2(d+1)^2(d-1)^2=0
\end{matrix}\right.
\]
There are four different cases up to symmetry in $b,c,d$:
\begin{caselist}

\subcaseitem{4.1}  $b^2\ne1$, $c^2\ne1$, $d^2\ne1$. Then we have $\left\{\begin{matrix}
 bd=\pm c\\
 cd=\pm b\\
 bc=\pm d
\end{matrix}\right.$. A contradiction since we have $b^2=c^2=d^2=-1$, and $b^2d^2=c^2$.
\subcaseitem{4.2}  $b^2=1$, and $c^2\ne1$, $d^2\ne1$. Then we have $\left\{\begin{matrix}
 bd=\pm c\\
 bc=\pm d
\end{matrix}\right.$.
which implies that $c^2=d^2$.
Write $d=\sigma c$ with $\sigma\in\{1,-1\}$. First suppose that $c^2\ne-1$. A nonzero contraction ${f^*}_{[1,0,t]}^{13}=(1+t/c,0,0,c+t)$ yields $(1,0,0,c)$ up to scalar. If $c$ is not a root of unity, Corollary~\ref{polytodegen} applies. Otherwise, $|c|=1$. Put $\mu=(c+c^{-1})/2$; since $c^2\ne\pm1$, we have $0<|\mu|<1$.
\[
\tfrac12M(f^*)^2=\begin{pmatrix}1&0&0&b\\0&\mu c&\mu d&0\\0&\mu/d&\mu/c&0\\1/b&0&0&1\end{pmatrix}.
\]
Its pair products are $1,\mu^2,\mu^2$, while its outer product is $1$. Thus, after permuting inputs, \caseref{2} above applies. It remains to treat $c^2=d^2=-1$.
Let
\[\begin{aligned}
M(g)&=M_{x_1,x_2,x_3,x_4}(f^*)M_{x_3,x_4,x_1,x_2}(f^*)=\begin{pmatrix}
 1+b^2 & 0 &  0& b+\frac{1}{b}\\
  0& c^2+d^2 & \frac{c}{d}+\frac{d}{c} & 0\\
  0&  \frac{c}{d}+\frac{d}{c}& \frac{1}{c^2}+\frac{1}{d^2} & 0\\
  b+\frac{1}{b}&  0&0  &1+\frac{1}{b^2}
\end{pmatrix}\\ & =2\begin{pmatrix}
 1 & 0 &  0& \frac{1}{b}\\
  0& -1 & \frac{-1}{cd} & 0\\
  0&  \frac{-1}{cd}& -1 & 0\\
  \frac{1}{b}&  0&0  &1
\end{pmatrix}
\end{aligned}\]
Here $g/2$ has middle diagonal entries $-1$ and middle off-diagonal entries $\sigma$. If $\sigma=1$, apply $[1,0,-1]$ to inputs $1$ and $3$; this signature is available by squaring $(1,0,0,c)$. These modifications change the signs of the middle off-diagonal entries and the outer off-diagonal entries. The resulting signature is $g_3$ or $g_4$. Lemma~\ref{specialg} applies, using one of the three available parameters outside $\{1,-1\}$.
\subcaseitem{4.3}  $b^2=c^2=1$ and $d^2\ne 1$. Then we have $bc=\epsilon_3 d$, a contradiction since $b^2c^2=1=d^2$.
\subcaseitem{4.4}  $b^2=c^2=d^2=1$.
Let
\[\begin{aligned}
M(g)&=M_{x_1,x_2,x_3,x_4}(f^*)M_{x_3,x_4,x_1,x_2}(f^*)=\begin{pmatrix}
 1+b^2 & 0 &  0& b+\frac{1}{b}\\
  0& c^2+d^2 & \frac{c}{d}+\frac{d}{c} & 0\\
  0&  \frac{c}{d}+\frac{d}{c}& \frac{1}{c^2}+\frac{1}{d^2} & 0\\
  b+\frac{1}{b}&  0&0  &1+\frac{1}{b^2}
\end{pmatrix}\\ & =2\begin{pmatrix}
 1 & 0 &  0& \frac{1}{b}\\
  0& 1 & \frac{1}{cd} & 0\\
  0&  \frac{1}{cd}& 1 & 0\\
  \frac{1}{b}&  0&0  &1
\end{pmatrix}
\end{aligned}\]
If $cd=1$, then $g/2=g_1$ or $g_2$. If $cd=-1$, one of $c,d$ equals $-1$, so a diagonal loop supplies $[1,0,-1]$; applying it to inputs $1$ and $3$ changes the middle off-diagonal entries to $1$. Thus the resulting signature is again $g_1$ or $g_2$, and Lemma~\ref{specialg} applies.\qedhere
\end{caselist}
\end{caselist}
\end{proof}
\begin{lemma}\label{zeroineachpaireight}
    Let $f$ be a 4-ary signature in eight-vertex form. In $\mathrm{Holant}(=_2,\mathcal{G}|\mathcal{F},f)$, if there is at least one zero in each pair of $f$, then $f\equiv0$ or we can construct a degenerate binary
    signature.
\end{lemma}
\begin{proof}
    By interchanging 0 and 1 and using symmetry, there are three different
cases.
\begin{caselist}
    \caseitem{1}  $f_{0000}=f_{0011}=f_{0101}=f_{0110}=0$. In this case, if $f_{=_2}^{23}$, $f_{=_2}^{24}$ and $f_{=_2}^{34}$ are nondegenerate, we have $f_{1111}=-f_{1100}=-f_{1010}=-f_{1001}$. If $f_{1111}=0$, then $f\equiv0$. If $f_{1111}\ne0$ we construct $[1,0,-1]$ by
    $f_{=_2}^{12}=f_{1111}(1,0,0,-1)$. Then $f_{[1,0,-1]}^{34}$ is degenerate.
    \caseitem{2}  $f_{0000}=f_{0011}=f_{0101}=f_{1001}=0$. In this case, if $f_{=_2}^{14}$, $f_{=_2}^{24}$ and $f_{=_2}^{34}$ are nondegenerate, we have $f_{1111}=-f_{1100}=-f_{1010}=-f_{0110}$. If $f_{1111}=0$, then $f\equiv0$. If $f_{1111}\ne0$ we construct $[1,0,-1]$ by
    $f_{=_2}^{12}=f_{1111}(1,0,0,-1)$. Then $f_{[1,0,-1]}^{34}$ is degenerate.
    \caseitem{3}  $f_{0000}=f_{0011}=f_{1010}=f_{1001}=0$. In this case, if $f_{=_2}^{13}$, $f_{=_2}^{14}$ and $f_{=_2}^{34}$ are nondegenerate, we have $f_{1111}=-f_{1100}=-f_{0101}=-f_{0110}$. If $f_{1111}=0$, then $f\equiv0$. If $f_{1111}\ne0$ we construct $[1,0,-1]$ by
    $f_{=_2}^{12}=f_{1111}(1,0,0,-1)$. Then $f_{[1,0,-1]}^{34}$ is degenerate.\qedhere
\end{caselist}
\end{proof}

\begin{theorem}\label{eightvertex=2}
    Let $f$ be a 4-ary signature with the signature matrix \[M(f)=\begin{pmatrix}
 a &  0& 0 &b \\
  0& c & d & 0\\
 0 & w &  z& 0\\
  y&  0&  0&x
\end{pmatrix},\] then $\mathrm{Holant}(=_2|f)$ is \#P-hard except for the following cases:
    \begin{itemize}
        \item $f \in \left \langle Z\mathcal{M} \right \rangle$ or $\in \left \langle ZX\mathcal{M} \right \rangle$;
        \item $f$ is $ \mathcal{A}\text{-}$transformable;
        \item $f$ is $ \mathcal{L}\text{-}$transformable;
        \item $f$ is $ \mathcal{P}\text{-}$transformable;
        \item $f\in\mathcal{H}$;
        \item $f \in \left \langle \mathcal{T} \right \rangle$,
    \end{itemize}
    in which cases the problem is computable in polynomial time.
\end{theorem}
\begin{proof}
We consider the following three cases:
\begin{caselist}
    \caseitem{1}  $a=x=0$. The result follows from Lemma~\ref{a=x=0}.
\end{caselist}

In the following cases there is at
least one nonzero entry in $(a, x)$. By Lemma~\ref{atleastonenonzeroin(a,x)} we assume that we can construct at least three distinct binary signatures of the form $(1,0,0,t)$. Otherwise, either we can construct a degenerate binary signature, or the result follows from Lemma~\ref{=+2and=-2}.
\begin{caselist}
    \caseitem{2}  $a=0$ or $x=0$.

    If there exists a nonzero pair, by Lemma~\ref{mobius}, we can construct a degenerate binary signature. Then there exists at least one zero in each pair, by Lemma~\ref{zeroineachpaireight}, we construct a degenerate binary signature or $f\equiv0$, which is a contradiction.
    \caseitem{3}  $ax\ne0$.

    If there exists a pair with exactly one zero, by Lemma~\ref{mobius}, we can construct a degenerate binary signature. Then any pair is nonzero or entirely zero.

    If there exist three zero pairs, then $f\in \mathcal{P}$.

    If there exist two zero pairs, then by Lemma~\ref{g=_4}, we construct $=_4$ or $f\in \left \langle \mathcal{T} \right \rangle$

    If there exists one zero pair, the result follows from Lemma~\ref{onezeropair}.

    If any pair is nonzero, the result follows from Lemma~\ref{allnonzero}.\qedhere

\end{caselist}
\end{proof}
\begin{lemma}\label{haveeightvertex}
    Let $f^*$ be a 4-ary signature with the signature matrix \[M(f^*)=\begin{pmatrix}
 a &  0& 0 &b \\
  0& c & d & 0\\
 0 & w &  z& 0\\
  y&  0&  0&x
\end{pmatrix}\] and suppose that $f^*$ is not in $\left \langle \mathcal{T} \right \rangle$; let $f$ be a 4-ary signature and $h=(0,1,-1,0)$ be a binary signature. Then
we have
    \begin{itemize}
        \item Either $\mathrm{Holant}(=_2|f^*,f,h)$ is \#P-hard;
        \item or we can construct a degenerate binary signature in $\mathrm{Holant}(=_2|f^*,f,h)$;
        \item or $f^*\in \mathcal{A}$;

        \item or there exists $T\in \mathrm{SL}(2,\mathbb{C})$ such that
        $
\mathrm{CSP}^2(f^{*'},f',h')\le_T\mathrm{Holant}(=_2|f^*,f,h),
        $
        where $f^{*'}=T^{\otimes4}f^*$, $f'=T^{\otimes4}f$ and $h'=T^{\otimes2}h$.
    \end{itemize}
\end{lemma}
\begin{proof}
We consider the following cases of $f^*$:
    \begin{caselist}
    \caseitem{1}  There is at least one zero in each pair of $f^*$.

    Note that $f^*\not\equiv 0$. By Lemma~\ref{zeroineachpaireight}, we construct a degenerate binary signature.
    \end{caselist}
    In the following cases, there is at least one nonzero pair. If the nonzero pair is not $\{f^*_{0000},f^*_{1111}\}$, without loss of generality, we assume that $f^*_{0011}f^*_{1100}\ne0$, by connecting two copies of $h$ to inputs $x_3$ and $x_4$ of $f^*$, we construct a new 4-ary signature $f^{**}$ with eight-vertex form.
    Moreover, $f^{**}_{0000}f^{**}_{1111}\ne0$. Note that $f^{**}$ has full rank if and only if $f^*$ has full rank, and $f^{**}\in\mathcal{A}$ if and only if $f^*\in\mathcal{A}$. Thus we can always assume that $f^{*}_{0000}f^{*}_{1111}\ne0$.
    \begin{caselist}
    \caseitem{2}  There exists one nonzero pair and one pair with exactly one zero in $f^*$.

    By Lemma~\ref{atleastonenonzeroin(a,x)} we construct $[1,0,t]$ with $t\ne\pm1$. If $t\ne0$, we can construct at least three distinct binary signatures of the form $[1,0,t_i]$, then by Lemma~\ref{mobius} we construct a degenerate binary signature. If $t=0$, then we already construct a degenerate binary signature.

    \caseitem{3}  Any pair of $f^*$ is either all zero or all nonzero.

    By Lemma~\ref{atleastonenonzeroin(a,x)}, we can construct $[1,0,t]$ with $t\ne\pm1$ unless $S(f^*_{=_2})=\{=_2\}$ or $S(f^*_{=_2})=\{[1,0,-1]\}$.
    If $S(f^*_{=_2})=\{=_2\}$ or $S(f^*_{=_2})=\{[1,0,-1]\}$, the result follows from Lemma~\ref{=+2and=-2}. If we can construct $[1,0,t]$, then
    we construct a degenerate binary signature or three distinct binary signatures of the form $(1,0,0,t_i)$. Thus we assume that we can construct three distinct binary signatures of the form $(1,0,0,t_i)$.
    \begin{caselist}
        \subcaseitem{3.1}  There exist at least two zero pairs. By Lemma~\ref{g=_4} we construct $=_4$, since $f^*\notin \left \langle \mathcal{T} \right \rangle$.
        \subcaseitem{3.2}  There exists exactly one zero pair. By Lemma~\ref{onezeropair}, Holant($=_2|f^*,f,h$) is \#P-hard or we can construct a degenerate binary signature.
        \subcaseitem{3.3}  All pairs are nonzero. The result follows from Lemma~\ref{allnonzero}.\qedhere
    \end{caselist}
\end{caselist}
\end{proof}

\section{Proof of the main theorem}
In this section, we give the main theorem, i.e., the computational complexity classification of Holant$(=_2|f)$, where $f$ is a 4-ary complex-valued signature.
We note that for two nondegenerate binary signatures $g$, $h$, by connecting input $x_2$ of $g$ to input $x_1$ of $h$, we construct a binary signature with the signature matrix $M(g)M(h)$, i.e., connecting two nondegenerate binary signatures in gadget construction corresponds to matrix multiplication in $\mathrm{SL}(2,\mathbb{C})$. The signature set constructed by the binary signatures $g_1, g_2, \cdots, g_n$ via chains is equivalent to the group closure of $M(g_1), M(g_2), \cdots, M(g_n)$ in  $\mathrm{SL}(2,\mathbb{C})$. In the following, we denote the closure generated by a binary signature set $S$ as $\left \langle S\right \rangle$.

The key strategy of the proof is the following construction.
\begin{definition}
For $\mathrm{ Holant}(=_2|f$), its binary subgroup sequence is constructed as follows:

 \[H^{(1)}_f=S(f_{=_2}),\]
 \[\mathbf{H}^{(1)}_f=\left \langle H^{(1)}_f\right \rangle,\]
 \[\vdots \]
 \[H^{(m+1)}_f=S(f_{\mathbf{H}^{(m)}_f}),\]
 \[\mathbf{H}^{(m)}_f=\left \langle H^{(1)}_f,H^{(2)}_f,...,H^{(m)}_f\right \rangle,\]
 \[\vdots\]
  where $\mathbf{H}^{(m)}_f=\left \langle H^{(1)}_f,H^{(2)}_f,...,H^{(m)}_f\right \rangle$ is the closure of $H^{(1)}_f,H^{(2)}_f,...,H^{(m)}_f$ for $m\ge 1$.
\end{definition}
 It is obvious that $H^{(1)}_f\subseteq H^{(2)}_f,H^{(2)}_f\subseteq H^{(3)}_f,\cdots,H^{(n)}_f\subseteq H^{(n+1)}_f,\cdots $. Further, $\mathbf{H}^{(m)}_f$ is closed under transposition since both $M(g)$ and $M(g)^T$ are the matrix forms of the binary signature $g$.
 \begin{definition}
     If there exists $m\in\mathbb{N}$, such that $H^{(m)}_f\subseteq \mathbf{H}^{(m-1)}_f$, we say that the sequence is stable at $\mathbf{H}^{(m)}_f$. Otherwise, we say that the sequence is unstable.
 \end{definition}

Note that $\mathbf{H}^{(m)}_f$ is a finitely generated periodic subgroup. By Theorem~\ref{Schur}(Schur's theorem), $\mathbf{H}^{(m)}_f$ is a finite subgroup of $\mathrm{SL}(2,\mathbb{C})$. Further, by Theorem~\ref{classification}, $\mathbf{H}^{(m)}_f$ is conjugate to $C_n,\ BD_{4n},\ BT_{24},\ BO_{48},$ or $BI_{120}$.

\begin{lemma}\label{normal}
    Let $\mathbf{H}^{(m)}_f=PGP^{-1}$, where $P\in \mathrm{SL}(2,\mathbb{C})$, $G\in \{C_n,\ BD_{4n},\ BT_{24},\ BO_{48},\ BI_{120}\}$. Let $S=(P^TP)^{-1}$, then $S$ belongs to the normalizer of $G$, i.e., $S\in \mathcal{N}_{\mathrm{SL}(2,\mathbb{C})}(G)$.
\end{lemma}
\begin{proof}
By Lemma~\ref{norma}, we only need to prove that $G$ is closed under transposition.
$C_n$ is obviously closed under transposition. Note that $BD_{4n},\ BT_{24},\ BO_{48}$ and $BI_{120}$ all contain $\begin{pmatrix}
 0 &1 \\
 -1 &0
\end{pmatrix}$. For any matrix $A=\begin{pmatrix}
 a &b \\
 c &d
\end{pmatrix}\in \mathrm{SL}(2,\mathbb{C})$, $\begin{pmatrix}
 0 &1 \\
 -1 &0
\end{pmatrix}A^{-1}\begin{pmatrix}
 0 &1 \\
 -1 &0
\end{pmatrix}^{-1}=\begin{pmatrix}
 0 &1 \\
 -1 &0
\end{pmatrix}\begin{pmatrix}
 d &-b \\
 -c &a
\end{pmatrix}\begin{pmatrix}
 0 &-1 \\
 1 &0
\end{pmatrix}=\begin{pmatrix}
 a &c \\
 b &d
\end{pmatrix}=A^T$. This implies that any group containing $\begin{pmatrix}
 0 &1 \\
 -1 &0
\end{pmatrix}$ is closed under transposition.
\end{proof}
\subsection{Unstable subgroup sequences}

\begin{lemma}\label{notstable}
    If the sequence is unstable, then for any $m\in\mathbb{N}$, $\mathbf{H}^{(m)}_f$ is one of the following cases (up to conjugacy):
    \begin{itemize}
        \item $\left \{ \begin{pmatrix}
 \lambda & 0\\
 0 & \lambda ^{-1}
\end{pmatrix} \middle|\lambda\in R_n \right \} $;
\item $\left \{ \begin{pmatrix}
 \lambda & 0\\
 0 & \lambda ^{-1}
\end{pmatrix} ,\begin{pmatrix}
 0 & \lambda\\
 -\lambda ^{-1} & 0
\end{pmatrix}\middle|\lambda\in R_n \right \}\ $,
    \end{itemize}
    where $R_n$ is a group generated by polynomially many distinct roots of unity.
\end{lemma}
\begin{proof}
    From $\mathbf{H}^{(m-1)}_{f}$ to $\mathbf{H}^{(m)}_{f}$, we adjoin a number of generators not contained in $\mathbf{H}^{(m-1)}_{f}$. Each such adjunction ensures that the order of $\mathbf{H}^{(m)}_{f}$ is strictly greater than that of $\mathbf{H}^{(m-1)}_{f}$. Furthermore, since every $\mathbf{H}^{(m)}_{f}$ is a finite subgroup of $\mathrm{SL}(2,\mathbb{C})$, and among all finite subgroups of $\mathrm{SL}(2,\mathbb{C})$, only $C_n$ and $BD_{4n}$ can have arbitrarily large orders. Therefore, $\mathbf{H}^{(m)}_{f}$ must be conjugate to $C_n$ or $BD_{4n}$.
\end{proof}

\begin{lemma}\label{infsettodegen}
  For $\mathrm{Holant}(=_2|f)$, if its binary subgroup sequence is unstable, then $\mathrm{Holant}(=_2|f)$ is \#P-hard except for the following cases:
\begin{itemize}
        \item $f \in \left \langle Z\mathcal{M} \right \rangle$ or $\in \left \langle ZX\mathcal{M} \right \rangle$;
        \item $f $ is $ \mathcal{A}\text{-}$transformable;
        \item $f$ is $ \mathcal{P}\text{-}$transformable;
        \item $f $ is $ \mathcal{L}\text{-}$transformable;
        \item $f\in\mathcal{H}$;
        \item $f \in \left \langle \mathcal{T} \right \rangle$,
    \end{itemize}
    in which cases $\mathrm{Holant}(=_2|f)$ is computable in polynomial time.
\end{lemma}
\begin{proof}
    By Lemma~\ref{notstable}, we can construct polynomially many distinct binary signatures. Then by Lemma~\ref{infinitetodegen}, we have \[ \mathrm{Holant}(=_2|f,h)\equiv_T\mathrm{Holant}(=_2|f),\] where $h$ is a degenerate binary signature. Further, by Theorem~\ref{dicwithdegen} the result follows.
\end{proof}

\subsection{\texorpdfstring{Subgroup sequences stable at $\mathbf{H}^{(m)}_{f}$}{Subgroup sequences stable at H\_f\textasciicircum (m)}}
Write $\mathbb C^{\times}=\mathbb C\setminus\{0\}$. For a set $G$, write $\mathbb{C}{G}=\{\lambda g:\lambda\in\mathbb{C},\allowbreak \ g\in G\}$ and $\mathbb{C}^{\times}{G}=\{\lambda g:\lambda\in\mathbb{C}\setminus\{0\},\allowbreak \ g\in G\}$.  Throughout the stable-case analysis, an inclusion of loop images in a displayed family of matrix representatives is understood to permit arbitrary scalar multiples. Actual group equations, normalizers, orders, and conjugacy statements always concern determinant-one matrices.

\begin{lemma}\label{sym}
    Let $f$ be a 4-ary signature. If $S(f_{=_2})=\{[0,0,0]\}$, then $f$ is symmetric.
\end{lemma}
\begin{proof}
Note that \[\begin{aligned}f(x_i,x_j,x_k,x_{\ell})^{ij}_{=_2}&=(f_{x_i=0,x_j=0,x_k=0,x_{\ell}=0}+f_{x_i=1,x_j=1,x_k=0,x_{\ell}=0},\\ & f_{x_i=0,x_j=0,x_k=0,x_{\ell}=1}+f_{x_i=1,x_j=1,x_k=0,x_{\ell}=1},\\ & f_{x_i=0,x_j=0,x_k=1,x_{\ell}=0}+f_{x_i=1,x_j=1,x_k=1,x_{\ell}=0},\\ & f_{x_i=0,x_j=0,x_k=1,x_{\ell}=1}+f_{x_i=1,x_j=1,x_k=1,x_{\ell}=1})\end{aligned}.\] By the permutation of $(i,j,k,\ell)$, we have the following system of linear equations:
    \[\begin{cases}
f_{0000}+f_{0011}=f_{0100} + f_{0111} = f_{1000} + f_{1011} = f_{1100}+f_{1111}=0 \\
f_{0000}+f_{0101}=f_{0010} + f_{0111} = f_{1000} + f_{1101} = f_{1010}+f_{1111}=0 \\
f_{0000}+f_{0110}=f_{0001} + f_{0111} = f_{1000} + f_{1110} = f_{1001}+f_{1111}=0 \\
f_{0000}+f_{1001}=f_{0010} + f_{1011} = f_{0100} + f_{1101} = f_{0110}+f_{1111}=0 \\
f_{0000}+f_{1010}=f_{0001} + f_{1011} = f_{0100} + f_{1110} =f_{0101}+f_{1111}= 0 \\
f_{0000}+f_{1100}=f_{0001} + f_{1101} = f_{0010} + f_{1110} = f_{0011}+f_{1111}=0
\end{cases}.\] Clearly, this implies that $f$ is symmetric.
\end{proof}
By Lemma~\ref{sym}, without loss of generality, we can always assume that  $f_{=_2}^{34}\not\equiv 0$, otherwise $f$ is symmetric and is already covered by a dichotomy.
\begin{lemma}
    Let $f$ be a 4-ary signature, $h$ be a binary signature whose matrix $M(h)$ has full rank and $h\in\mathcal{A}$, and apply a binary modification to an input of $f$ using $h$, to construct a new 4-ary signature $f^*$. Then $f^*\in\mathcal{A}$ or $f^*\in\left \langle \mathcal{T} \right \rangle$ if and only if $f\in\mathcal{A}$ or $f\in\left \langle \mathcal{T} \right \rangle$.
\end{lemma}
\begin{proof}
    $f$ can be obtained by applying a binary modification to the same input of $f^*$ using $h^{-1}$. Note that $h\in\mathcal{A}$ if and only if $h^{-1}\in\mathcal{A}$.
    Moreover, any gadget formed from two affine signatures is affine. Thus $f^*\in\mathcal{A}$ if and only if $f\in\mathcal{A}$.

    $f\in\left \langle \mathcal{T} \right \rangle$ implies that it can be decomposed into a tensor product of signatures of arity less than 3. Clearly, a binary modification using $h$ preserves its tensor product structure since $M(h)$ has full rank.
\end{proof}
\subsubsection{\texorpdfstring{Conjugacy of $\mathbf{H}^{(m)}_{f}$ to $C_n$}{Conjugacy of H\_f\textasciicircum (m) to C\_n}}
\begin{lemma}\label{Cn}
     If $\mathbf{H}^{(m)}_{f}=PC_nP^{-1}$, then $\mathrm{Holant}(=_2|f)$ is \#P-hard except for the following cases:
\begin{itemize}
        \item $f \in \left \langle Z\mathcal{M} \right \rangle$ or $\in \left \langle ZX\mathcal{M} \right \rangle$;
        \item $f $ is $ \mathcal{A}\text{-}$transformable;
        \item $f$ is $ \mathcal{P}\text{-}$transformable;
        \item $f $ is $ \mathcal{L}\text{-}$transformable;
        \item $f\in\mathcal{H}$;
        \item $f \in \left \langle \mathcal{T} \right \rangle$,
    \end{itemize}
    in which cases $\mathrm{Holant}(=_2|f)$ is computable in polynomial time.
\end{lemma}
\begin{proof}
    We divide the proof into two cases: $n=1,2$ in \caseref{1} and $n\ge3$ in \caseref{2}.
    \begin{caselist}
        \caseitem{1}   $\mathbf{H}^{(m)}_{f}=PC_nP^{-1}$, where $n=1,2$.

        Note that $C_1=\{I\}$ and $C_2=\{I,-I\}$. Then $\mathbf{H}^{m}_f=H^{(1)}_f=\{=_2\}$. By a holographic transformation using $Z$, we have
\[
\mathrm{Holant}(=_2|f,=_2)\equiv_T\mathrm{Holant}(\ne_2|f',\ne_2),
\] and the binary signatures constructed by adding loops satisfy $S(f'_{\ne_2})=\{\ne_2\}$. Then by Lemma~\ref{solveeightvertexform},  $f'$ has eight-vertex form. By Theorem~\ref{eightvertex} the result follows.
        \caseitem{2}  $\mathbf{H}^{(m)}_{f}=PC_nP^{-1}$, where $n\ge3$.

        By Lemma~\ref{normalizer}, $\mathcal{N}_{\mathrm{SL}(2,\mathbb{C})}(C_n)=\left \{ \begin{pmatrix}
  a&0 \\
  0&a^{-1}
\end{pmatrix}, \begin{pmatrix}
  0&b \\
  -b^{-1}&0
\end{pmatrix} \Bigg| a,b\in \mathbb{C}\setminus\{0\} \right \}$.
Note that $S=(P^TP)^{-1}=P^{-1}(P^T)^{-1}$ is symmetric. By Lemma~\ref{normal}, $S\in\mathcal{N}_{\mathrm{SL}(2,\mathbb{C})}(C_n)$. Thus $S$ has the form $\begin{pmatrix}
  a&0 \\
  0&a^{-1}
\end{pmatrix}$ or $\begin{pmatrix}
  0&1 \\
  1&0
\end{pmatrix}$ since $S$ is symmetric.
In Holant($=_2|f,PC_nP^{-1}$), by a holographic transformation using $P$, we have \[\text{Holant}(=_2|f,PC_nP^{-1})\equiv_T\text{Holant}(P^TP|f',P^{-1}PC_nP^{-1}(P^{-1})^T)\equiv_T\text{Holant}(S^{-1}|f',C_nS),\]
where $f'=(P^{-1})^{\otimes 4}f$.
\begin{caselist}
    \subcaseitem{2.1}  $S$ has the form $\begin{pmatrix}
  a&0 \\
  0&a^{-1}
\end{pmatrix}$.

By a holographic transformation using $\begin{pmatrix}
  a^{\frac{1}{2}}&0 \\
  0&a^{-\frac{1}{2}}
\end{pmatrix}$, we have
\[\text{Holant}(S^{-1}|f',C_nS)\equiv_T\text{Holant}(=_2|f'',C_n),\] where $f''=\begin{pmatrix}
  a^{-\frac{1}{2}}&0 \\
  0&a^{\frac{1}{2}}
\end{pmatrix}^{\otimes4}f'$. Then in $\text{Holant}(=_2|f'',C_n)$, its binary subgroup sequence is stable at $C_n$. Note that $n\ge3$, we always have a signature of the form $\begin{pmatrix}
  1&0 \\
  0&t
\end{pmatrix}$ with $t\ne1$. Further, $S(f''_{=_2})\subseteq C_n$ and $S(f''_{[1,0,t]})\subseteq C_n$ implies that for any binary signature $h\in S(f''_{=_2})$ or $h\in S(f''_{[1,0,t]})$, $h_{01}=h_{10}=0$, then we have:
\[\begin{cases}
f''_{0100} + f''_{0111} = f''_{1000} + f''_{1011} = 0 \\
f''_{0010} + f''_{0111} = f''_{1000} + f''_{1101} = 0 \\
f''_{0001} + f''_{0111} = f''_{1000} + f''_{1110} = 0 \\
f''_{0010} + f''_{1011} = f''_{0100} + f''_{1101} = 0 \\
f''_{0001} + f''_{1011} = f''_{0100} + f''_{1110} = 0 \\
f''_{0001} + f''_{1101} = f''_{0010} + f''_{1110} = 0
\end{cases}
\ \text{and}\ 
\begin{cases}
f''_{0100} + t f''_{0111} = f''_{1000} + t f''_{1011} = 0 \\
f''_{0010} + t f''_{0111} = f''_{1000} + t f''_{1101} = 0 \\
f''_{0001} + t f''_{0111} = f''_{1000} + t f''_{1110} = 0 \\
f''_{0010} + t f''_{1011} = f''_{0100} + t f''_{1101} = 0 \\
f''_{0001} + t f''_{1011} = f''_{0100} + t f''_{1110} = 0 \\
f''_{0001} + t f''_{1101} = f''_{0010} + t f''_{1110} = 0
\end{cases}.
\] The system of linear equations implies that all the weight 1 and weight 3 entries of $f''$ are zero, i.e., $f''$ has eight-vertex form. By Theorem~\ref{eightvertex=2} the result follows.

\subcaseitem{2.2}  $S$ has the form $\begin{pmatrix}
  0&1 \\
  1&0
\end{pmatrix}$.

In this case, Holant($S^{-1}|f',C_nS$) is exactly Holant($\ne_2|f',C_n(\ne_2)$). Moreover, $S(f'_{\ne_2})\subseteq C_n(\ne_2)\subseteq\{(0,b,c,0)|b,c\in\mathbb{C}\}$, by Lemma~\ref{solveeightvertexform}, $f'$ has eight-vertex form. By Theorem~\ref{eightvertex} the result follows.\qedhere

\end{caselist}
\end{caselist}
\end{proof}

\subsubsection{\texorpdfstring{Conjugacy of $\mathbf{H}^{(m)}_{f}$ to $BD_{4n}$}{Conjugacy of H\_f\textasciicircum (m) to BD\_(4n)}}
\begin{lemma}\label{BD8}
    Let $f$ be a 4-ary signature,
    $\mathcal{B}=\{=_2,[i,0,-i],[0,i,0],(0,1,-1,0)\}$ be a set of binary signatures and $S(f_\mathcal{B})\subseteq \mathbb{C}\mathcal{B}$. Then $\mathrm{Holant}(=_2|f,\mathcal{B})$ is \#P-hard except for the following cases:
    \begin{itemize}
        \item $f \in \left \langle Z\mathcal{M} \right \rangle$ or $\in \left \langle ZX\mathcal{M} \right \rangle$;
        \item $f$ is $ \mathcal{A}\text{-}$transformable;
        \item $f$ is $ \mathcal{L}\text{-}$transformable;
        \item $f$ is $ \mathcal{P}\text{-}$transformable;
        \item $f\in\mathcal{H}$;
        \item $f \in \left \langle \mathcal{T} \right \rangle$,
    \end{itemize}
    in which cases the problem is computable in polynomial time.
\end{lemma}
\begin{proof}
    Up to some nonzero scalars, let $\mathcal{B}=\{=_2,[1,0,-1],\ne_2,(0,1,-1,0)\}$.
    By Lemma~\ref{sym}, we assume that ${f}^{34}_{=_2}\not \equiv 0$. Note that ${f}^{34}_{=_2}\in \mathcal{B}$, ${f}^{34}_{=_2}\in \mathcal{A}$ and has full rank. By Lemma~\ref{polyto-1}, we can construct $({f}^{34}_{=_2})^{-1}$. Moreover, $({f}^{34}_{=_2})^{-1}\in \mathcal{B}$, $({f}^{34}_{=_2})^{-1}\in \mathcal{A}$.
    By connecting input $x_2$ of $({f}^{34}_{=_2})^{-1}$ to input $x_1$ of $f$, we have a new 4-ary signature $f^*$.  Then $f^*\in \left \langle \mathcal{T} \right \rangle$ or $\mathcal{A}$ if and only if $f\in \left \langle \mathcal{T} \right \rangle$ or $\mathcal{A}$. Moreover, $S(f^*_\mathcal{B})\subseteq \mathcal{B}$.
    We have ${f^*}^{34}_{=_2}=(=_2)$, i.e.,
    \begin{equation}\label{f34=21}
        \left\{\begin{matrix}
 f^*_{0000}+f^*_{0011}=f^*_{1100}+f^*_{1111}=1\\
f^*_{0100}+f^*_{0111}=f^*_{1000}+f^*_{1011}=0
\end{matrix}\right. .
    \end{equation}
For ${f^*}^{34}_{[1,0,-1]}$, there are three cases:
\begin{caselist}
    \caseitem{1}  ${f^*}^{34}_{[1,0,-1]}=(=_2)$ up to a nonzero scalar or ${f^*}^{34}_{[1,0,-1]}\equiv0$.
    In this case, we have
    \begin{equation}\label{casea1}
    \left\{\begin{matrix}
 f^*_{0000}-f^*_{0011}=f^*_{1100}-f^*_{1111}\\
f^*_{0100}-f^*_{0111}=f^*_{1000}-f^*_{1011}=0
\end{matrix}\right. ,
    \end{equation}
    By (\ref{f34=21}) and (\ref{casea1}), we have \[
    \left\{\begin{matrix}
 f^*_{0000}=f^*_{1100}\\
 f^*_{0011}=f^*_{1111}\\
f^*_{0100}=f^*_{0111}=f^*_{1000}=f^*_{1011}=0
\end{matrix}\right. ,
    \] i.e.,
   \[ M(f^*)=\begin{pmatrix}
  f^*_{0000} &f^*_{0001}  & f^*_{0010} & f^*_{0011}\\
 0 &  f^*_{0101}& f^*_{0110} & 0\\
  0& f^*_{1001} & f^*_{1010} &0 \\
  f^*_{0000} &f^*_{1101}  & f^*_{1110} &f^*_{0011}
\end{pmatrix}.\] Note that
\[
 \left\{\begin{matrix}
 {f^*}^{12}_{=_2}=(2f^*_{0000},f^*_{0001}+f^*_{1101},f^*_{0010}+f^*_{1110},2f^*_{0011})\\
 {f^*}^{12}_{[1,0,-1]}=(0,f^*_{0001}-f^*_{1101},f^*_{0010}-f^*_{1110},0)
\end{matrix}\right. .
\] Then ${f^*}^{12}_{=_2}$ must have the form $(=_2)$ up to a nonzero scalar since $f^*_{0000}+f^*_{0011}=1$(by (\ref{f34=21})). Further, ${f^*}^{12}_{[1,0,-1]}$ must have the form $(\ne_2)$ or $(0,1,-1,0)$, up to a nonzero scalar. This implies that we have
\[
 \left\{\begin{matrix}
 f^*_{0000}=f^*_{0011}\ne0(by\ {f^*}^{12}_{=_2})\\
 f^*_{0001}+f^*_{1101}=f^*_{0010}+f^*_{1110}=0(by\  {f^*}^{12}_{=_2})\\
 \epsilon_1(f^*_{0001}-f^*_{1101})=f^*_{0010}-f^*_{1110}(by\ {f^*}^{12}_{[1,0,-1]})
\end{matrix}\right. ,
\]where $\epsilon_1=\pm1$.
Then we have \[ \left\{\begin{matrix}
 f^*_{0001}=-f^*_{1101}\\
 f^*_{0010}=-f^*_{1110}\\
\epsilon_1f^*_{0001}=f^*_{0010}\\
\epsilon_1f^*_{1101}=f^*_{1110}
\end{matrix}\right. .\]
Now \[
M(f^*)=\begin{pmatrix}
  f^*_{0000} &  f^*_{0001}&  \epsilon_1f^*_{0001}& f^*_{0000}\\
 0 &  f^*_{0101}& f^*_{0110} & 0\\
  0& f^*_{1001} & f^*_{1010} &0 \\
  f^*_{0000} &  -f^*_{0001}&  -\epsilon_1f^*_{0001}& f^*_{0000}
\end{pmatrix}.\] If $f^*_{0001}=0$
,
$f^*$ has eight-vertex form and then the result follows from Lemma~\ref{haveeightvertex}. If $f^*_{0001}\ne0$, then ${f^*}^{34}_{(0,1,\epsilon_1,0)}=(2f^*_{0001},f^*_{0101}+\epsilon_1f^*_{0110},f^*_{1001}+\epsilon_1f^*_{1010},-2f^*_{0001})$ must have the form $[1,0,-1]$ up to a nonzero scalar. We also have $f^*_{0101}+\epsilon_1f^*_{0110},f^*_{1001}+\epsilon_1f^*_{1010}=0$.
By ${f^*}^{34}_{(0,1,\epsilon_1,0)}=(2f^*_{0001},f^*_{0101}+\epsilon_1f^*_{0110},f^*_{1001}+\epsilon_1f^*_{1010},-2f^*_{0001})$ and ${f^*}^{34}_{(0,1,-\epsilon_1,0)}=(0,f^*_{0101}-\epsilon_1f^*_{0110},f^*_{1001}-\epsilon_1f^*_{1010},0)$
, we have
\begin{equation}
    \left\{\begin{matrix}
f^*_{0101}+\epsilon_1f^*_{0110}=f^*_{1001}+\epsilon_1f^*_{1010}=0\\
\epsilon_2(f^*_{0101}-\epsilon_1f^*_{0110})=f^*_{1001}-\epsilon_1f^*_{1010}
\end{matrix}\right. ,
\end{equation}where $\epsilon_2=\pm1$. Now,
\[
M(f^*)=\begin{pmatrix}
  f^*_{0000} &  f^*_{0001}&  \epsilon_1f^*_{0001}& f^*_{0000}\\
  0 &  f^*_{0101}&  -\epsilon_1f^*_{0101}& 0\\
  0 &  \epsilon_2f^*_{0101}&  -\epsilon_1\epsilon_2f^*_{0101}& 0\\
  f^*_{0000} &  -f^*_{0001}&  -\epsilon_1f^*_{0001}& f^*_{0000}
\end{pmatrix}.\] Note that \[
 \left\{\begin{matrix}
 {f^*}^{23}_{=_2}=(f^*_{0000}-\epsilon_1 f^*_{0101},f^*_{0001},-\epsilon_1f^*_{0001},\epsilon_2f^*_{0101}+f^*_{0000})\\
{f^*}^{23}_{[1,0,-1]}=(f^*_{0000}+\epsilon_1 f^*_{0101},f^*_{0001},-\epsilon_1f^*_{0001},\epsilon_2f^*_{0101}-f^*_{0000})
\end{matrix}\right. .
\] Note that $f^*_{0001}\ne0$, then ${f^*}^{23}_{=_2}$ and ${f^*}^{23}_{[1,0,-1]}$ must have the form $(\ne_2)$ or $(0,1,-1,0)$, i.e., \[f^*_{0000}-\epsilon_1 f^*_{0101}=f^*_{0000}+\epsilon_1 f^*_{0101}=0.\] This implies that $f^*_{0000}=0$, which contradicts the fact that $f^*_{0000}\neq0$.
    \caseitem{2}  ${f^*}^{34}_{[1,0,-1]}= [1,0,-1]$ up to a nonzero scalar. In this case, we have
    \begin{equation}\label{casea}
    \left\{\begin{matrix}
 f^*_{0000}-f^*_{0011}=-(f^*_{1100}-f^*_{1111})\ne0\\
f^*_{0100}-f^*_{0111}=f^*_{1000}-f^*_{1011}=0
\end{matrix}\right. .
    \end{equation}
    By (\ref{f34=21}) and (\ref{casea}), we have
    \[
M(f^*)=\begin{pmatrix}
  f^*_{0000} &  f^*_{0001}&  f^*_{0010}& f^*_{0011}\\
  0 &  *&  *& 0\\
  0 &  *&  *& 0\\
  f^*_{0011} &  f^*_{1101}&  f^*_{1110}& f^*_{0000}
\end{pmatrix},\] where $f^*_{0000}+f^*_{0011}\ne0$ and $f^*_{0000}\ne f^*_{0011}$.
 By ${f^*}^{12}_{=_2}=(f^*_{0000}+f^*_{0011},f^*_{0001}+f^*_{1101},f^*_{0010}+f^*_{1110},f^*_{0011}+f^*_{0000})$ and ${f^*}^{12}_{[1,0,-1]}=(f^*_{0000}-f^*_{0011},f^*_{0001}-f^*_{1101},f^*_{0010}-f^*_{1110},f^*_{0011}-f^*_{0000})$, we have
\begin{equation}
    \left\{\begin{matrix}
 f^*_{0001}+f^*_{1101}=f^*_{0010}+f^*_{1110}=0\\
f^*_{0001}-f^*_{1101}=f^*_{0010}-f^*_{1110}=0
\end{matrix}\right. .
\end{equation} Then $ f^*_{0001}=f^*_{1101}=f^*_{0010}=f^*_{1110}=0$, and
\[
M(f^*)=\begin{pmatrix}
  f^*_{0000} &  0&  0& f^*_{0011}\\
  0 &  *&  *& 0\\
  0 &  *&  *& 0\\
  f^*_{0011} &  0&  0& f^*_{0000}
\end{pmatrix},\] $f^*$ has eight-vertex form. Then the result follows from Lemma~\ref{haveeightvertex}.

    \caseitem{3}  ${f^*}^{34}_{[1,0,-1]}=(\ne_2)$ or ${f^*}^{34}_{[1,0,-1]}=(0,1,-1,0)$, both up to a nonzero scalar. In this case, we have
    \begin{equation}\label{casea11}
    \left\{\begin{matrix}
 f^*_{0000}-f^*_{0011}=f^*_{1100}-f^*_{1111}=0\\
\epsilon_1(f^*_{0100}-f^*_{0111})=f^*_{1000}-f^*_{1011}\ne0
\end{matrix}\right. ,
    \end{equation}where $\epsilon_1=\pm1$. By (\ref{f34=21}) and (\ref{casea11}), we have
\[
M(f^*)=\begin{pmatrix}
  f^*_{0000} &  *&  *& f^*_{0000}\\
  f^*_{0100} &  *&  *& -f^*_{0100}\\
  \epsilon_1f^*_{0100} &  *&  *& -\epsilon_1f^*_{0100}\\
  f^*_{0000} &  *&  *& f^*_{0000}
\end{pmatrix},\]where $f^*_{0100}\ne0$.
     If $f^*_{0000}=0$, then ${f^*}^{34}_{=_2}\equiv0$, a contradiction. Thus $f^*_{0000}\ne0$. Then by ${f^*}^{12}_{=_2}=(2f^*_{0000},f^*_{0001}+f^*_{1101},f^*_{0010}+f^*_{1110},2f^*_{0000})$ and ${f^*}^{12}_{[1,0,-1]}=(0,f^*_{0001}-f^*_{1101},f^*_{0010}-f^*_{1110},0)$, we have
\begin{equation}
    \left\{\begin{matrix}
 f^*_{0001}+f^*_{1101}=f^*_{0010}+f^*_{1110}=0\\
\epsilon_2(f^*_{0001}-f^*_{1101})=f^*_{0010}-f^*_{1110}
\end{matrix}\right. ,
\end{equation}where $\epsilon_2=\pm1$. Now,
\[
M(f^*)=\begin{pmatrix}
  f^*_{0000} &  f^*_{0001}&  \epsilon_2f^*_{0001}& f^*_{0000}\\
  f^*_{0100} &  *&  *& -f^*_{0100}\\
  \epsilon_1f^*_{0100} &  *&  *& -\epsilon_1f^*_{0100}\\
  f^*_{0000} &  -f^*_{0001}&  -\epsilon_2f^*_{0001}& f^*_{0000}
\end{pmatrix}.\]
If $f^*_{0001}=0$, i.e. ${f^*}^{12}_{=_2}=(=_2)$ and ${f^*}^{12}_{[1,0,-1]}\equiv0$. By the permutation of the variables of $f^*$, ${f^*}(x_3,x_4,x_1,x_2)$ falls under \caseref{1}. If $f^*_{0001}\ne0$, by ${f^*}^{34}_{(0,1,\epsilon_2,0)}=(2f^*_{0001},f^*_{0101}+\epsilon_2f^*_{0110},f^*_{1001}+\epsilon_2f^*_{1010},-2f^*_{0001})$ and ${f^*}^{34}_{(0,1,-\epsilon_2,0)}=(0,f^*_{0101}-\epsilon_2f^*_{0110},f^*_{1001}-\epsilon_2f^*_{1010},0)$
, we have
\begin{equation}
    \left\{\begin{matrix}
f^*_{0101}+\epsilon_2f^*_{0110}=f^*_{1001}+\epsilon_2f^*_{1010}=0\\
\epsilon_3(f^*_{0101}-\epsilon_2f^*_{0110})=f^*_{1001}-\epsilon_2f^*_{1010}
\end{matrix}\right. ,
\end{equation}where $\epsilon_3=\pm1$. Now
\[
M(f^*)=\begin{pmatrix}
  f^*_{0000} &  f^*_{0001}&  \epsilon_2f^*_{0001}& f^*_{0000}\\
  f^*_{0100} &  f^*_{0101}&  -\epsilon_2f^*_{0101}& -f^*_{0100}\\
  \epsilon_1f^*_{0100} &  \epsilon_3f^*_{0101}&  -\epsilon_2\epsilon_3f^*_{0101}& -\epsilon_1f^*_{0100}\\
  f^*_{0000} &  -f^*_{0001}&  -\epsilon_2f^*_{0001}& f^*_{0000}
\end{pmatrix}.\]
By ${f^*}^{13}_{=_2}=(f^*_{0000}-\epsilon_2\epsilon_3f^*_{0101},f^*_{0001}-\epsilon_1f^*_{0100},f^*_{0100}-\epsilon_2f^*_{0001},f^*_{0101}+f^*_{0000})$ and ${f^*}^{13}_{[1,0,-1]}=(f^*_{0000}+\epsilon_2\epsilon_3f^*_{0101},f^*_{0001}+\epsilon_1f^*_{0100},f^*_{0100}+\epsilon_2f^*_{0001},f^*_{0101}-f^*_{0000})$. We have
\begin{equation}
    \left\{\begin{matrix}
(f^*_{0000}-\epsilon_2\epsilon_3f^*_{0101})^2=(f^*_{0101}+f^*_{0000})^2\\
(f^*_{0001}-\epsilon_1f^*_{0100})^2=(f^*_{0100}-\epsilon_2f^*_{0001})^2\\
(f^*_{0000}+\epsilon_2\epsilon_3f^*_{0101})^2=(f^*_{0101}-f^*_{0000})^2\\
(f^*_{0001}+\epsilon_1f^*_{0100})^2=(f^*_{0100}+\epsilon_2f^*_{0001})^2
\end{matrix}\right. ,
\end{equation}
which implies that
\[
\left\{\begin{matrix}
 -\epsilon_2\epsilon_3=1\ \text{or}\ f^*_{0101}=0 ,\\
\epsilon_1=\epsilon_2
\end{matrix}\right.
\]
If $f^*_{0101}=0$, by connecting $\ne_2$ to input $x_4$ of $f^*$, and $[1,0,\epsilon_2]$ to input $x_3$ of $f^*$, $[1,0,-1]$ to input $x_1$ of $f^*$, we construct \[M(f^{**})=\begin{pmatrix}
  f^*_{0001} &  f^*_{0000}&  \epsilon_2f^*_{0000}& f^*_{0001}\\
  0 &  f^*_{0100}&  -\epsilon_2f^*_{0100}& 0\\
  0 &  -\epsilon_1f^*_{0100}&  \epsilon_1\epsilon_2f^*_{0100}& 0\\
  f^*_{0001} &  -f^*_{0000}&  -\epsilon_2f^*_{0000}& f^*_{0001}
\end{pmatrix}.\] Then $f^{**}$ falls under \caseref{1} since ${f^{**}}^{34}_{=_2}=(=_2)$ and ${f^{**}}^{34}_{[1,0,-1]}\equiv0$. If $f^*_{0101}\ne0$, then \[
M(f^*)=\begin{pmatrix}
  f^*_{0000} &  f^*_{0001}&  \epsilon_1f^*_{0001}& f^*_{0000}\\
  f^*_{0100} &  f^*_{0101}&  -\epsilon_1f^*_{0101}& -f^*_{0100}\\
  \epsilon_1f^*_{0100} &  -\epsilon_1f^*_{0101}&  f^*_{0101}& -\epsilon_1f^*_{0100}\\
  f^*_{0000} &  -f^*_{0001}&  -\epsilon_1f^*_{0001}& f^*_{0000}
\end{pmatrix}.\]
If $f^*_{0000}\ne f^*_{0101}$, by ${f^*}^{13}_{[1,0,-1]}=(f^*_{0000}-f^*_{0101},f^*_{0001}+\epsilon_1f^*_{0100},f^*_{0100}+\epsilon_1f^*_{0001},f^*_{0101}-f^*_{0000})$, we have $f^*_{0001}+\epsilon_1f^*_{0100}=f^*_{0100}+\epsilon_1f^*_{0001}=0$, i.e. $f^*_{0001}=-\epsilon_1f^*_{0100}$. Further by ${f^*}^{13}_{=_2}=(f^*_{0000}+f^*_{0101},f^*_{0001}-\epsilon_1f^*_{0100},f^*_{0100}-\epsilon_1f^*_{0001},f^*_{0101}+f^*_{0000})$ we have $-f^*_{0000}=f^*_{0101}$.

If $f^*_{0000}=f^*_{0101}$, by ${f^*}^{13}_{=_2}=(f^*_{0000}+f^*_{0101},f^*_{0001}-\epsilon_1f^*_{0100},f^*_{0100}-\epsilon_1f^*_{0001},f^*_{0101}+f^*_{0000})$ we have $f^*_{0001}-\epsilon_1f^*_{0100}=f^*_{0100}-\epsilon_1f^*_{0001}=0$, i.e. $f^*_{0001}=\epsilon_1f^*_{0100}$.

In summary, we consider the following four cases separately:
\[
\left\{\begin{matrix}
 f^*_{0000}=f^*_{0101}\\
f^*_{0100}=f^*_{0001}
\end{matrix}\right.\ \text{or} \left\{\begin{matrix}
 -f^*_{0000}=f^*_{0101}\\
-f^*_{0100}=f^*_{0001}
\end{matrix}\right.\  \text{or} \left\{\begin{matrix}
 f^*_{0000}=f^*_{0101}\\
-f^*_{0100}=f^*_{0001}
\end{matrix}\right.\ \text{or} \left\{\begin{matrix}
 -f^*_{0000}=f^*_{0101}\\
f^*_{0100}=f^*_{0001}
\end{matrix}\right. .
\]
\begin{caselist}
    \subcaseitem{3.1}  $\left\{\begin{matrix}
 f^*_{0000}=f^*_{0101}\\
f^*_{0100}=f^*_{0001}
\end{matrix}\right.$. In this case, \[
M(f^*)=\begin{pmatrix}
  f^*_{0000} &  f^*_{0100}&  f^*_{0100}& f^*_{0000}\\
  f^*_{0100} &  f^*_{0000}&  -f^*_{0000}& -f^*_{0100}\\
  f^*_{0100} &  -f^*_{0000}&  f^*_{0000}& -f^*_{0100}\\
  f^*_{0000} &  -f^*_{0100}&  -f^*_{0100}& f^*_{0000}
\end{pmatrix}.\] By a holographic transformation using $Z$, we have \[
\mathrm{Holant}(\ne_2|f',f^{*'},\mathcal{B'})\equiv_T\mathrm{Holant}(=_2|f,f^*,\mathcal{B}),
\]where \[f'=(Z^{-1})^{\otimes4}f,f^{*'}=(Z^{-1})^{\otimes4}f^*=8\begin{pmatrix}
 -\frak{i}f^*_{0100} & 0 & 0 & 0\\
  0&0  & f^*_{0000} & 0\\
  0& f^*_{0000} & 0 & 0\\
  0&  0& 0 &\frak{i}f^*_{0100}
\end{pmatrix}\] and $\mathcal{B'}=(Z^{-1})^{\otimes2}\mathcal{B}$. Note that $f^{*'}\in \left \langle \mathcal{T} \right \rangle$ or $\mathcal{A}$ if and only if $f\in \left \langle \mathcal{T} \right \rangle$ or $\mathcal{A}$.
Moreover, $\mathcal{B'}=(Z^{-1})^{\otimes2}\mathcal{B}=\mathcal{B}$,  $(=_2)\in\mathcal{B}$. Thus we have \[
 \mathrm{Holant}(\ne_2,=_2|f',f^{*'},\mathcal{B}) \equiv_T \mathrm{Holant}(\ne_2|f',f^{*'},\mathcal{B'}).
\] Then the result follows from Lemma~\ref{haveeightvertex}.
\subcaseitem{3.2}  $\left\{\begin{matrix}
 -f^*_{0000}=f^*_{0101}\\
-f^*_{0100}=f^*_{0001}
\end{matrix}\right.$. In this case, \[
M(f^*)=\begin{pmatrix}
  f^*_{0000} &  -f^*_{0100}&  -f^*_{0100}& f^*_{0000}\\
  f^*_{0100} &  -f^*_{0000}&  f^*_{0000}& -f^*_{0100}\\
  f^*_{0100} &  f^*_{0000}&  -f^*_{0000}& -f^*_{0100}\\
  f^*_{0000} &  f^*_{0100}&  f^*_{0100}& f^*_{0000}
\end{pmatrix}.\] Then ${f^*}_{=_2}^{14}=2(f^*_{0000},-f^*_{0100},f^*_{0100},f^*_{0000})$ is not in $\mathcal{B}$, a contradiction.
\subcaseitem{3.3}   $\left\{\begin{matrix}
 f^*_{0000}=f^*_{0101}\\
-f^*_{0100}=f^*_{0001}
\end{matrix}\right.$. In this case, \[
M(f^*)=\begin{pmatrix}
  f^*_{0000} &  -f^*_{0100}&  f^*_{0100}& f^*_{0000}\\
  f^*_{0100} &  f^*_{0000}&  f^*_{0000}& -f^*_{0100}\\
  -f^*_{0100} &  f^*_{0000}&  f^*_{0000}& f^*_{0100}\\
  f^*_{0000} &  f^*_{0100}&  -f^*_{0100}& f^*_{0000}
\end{pmatrix}.\] Then ${f^*}_{=_2}^{14}=2(f^*_{0000},f^*_{0100},f^*_{0100},f^*_{0000})$ is not in $\mathcal{B}$, a contradiction.
\subcaseitem{3.4}  $\left\{\begin{matrix}
 -f^*_{0000}=f^*_{0101}\\
f^*_{0100}=f^*_{0001}
\end{matrix}\right.$. In this case, \[
M(f^*)=\begin{pmatrix}
  f^*_{0000} &  f^*_{0100}&  -f^*_{0100}& f^*_{0000}\\
  f^*_{0100} &  -f^*_{0000}&  -f^*_{0000}& -f^*_{0100}\\
  -f^*_{0100} &  -f^*_{0000}&  -f^*_{0000}& f^*_{0100}\\
  f^*_{0000} &  -f^*_{0100}&  f^*_{0100}& f^*_{0000}
\end{pmatrix}.\]Then ${f^*}_{[1,0,-1]}^{14}=2(f^*_{0000},-f^*_{0100},f^*_{0100},-f^*_{0000})$ is not in $\mathcal{B}$, a contradiction.\qedhere
\end{caselist}
\end{caselist}

\end{proof}
\begin{lemma}\label{BD16/BD8}
    Let $f$ be a 4-ary signature,
    $\mathcal{B}=\{=_2,[\frak{i},0,-\frak{i}],\frac{1}{\sqrt{2}}(0,1+\frak{i},1-\frak{i},0),-\frac{1}{\sqrt{2}}(0,1-\frak{i},1+\frak{i},0)\}$ be a set of binary signatures and $S(f_{\mathcal{B}})\subseteq \mathbb{C}\mathcal{B}$. Then $\mathrm{Holant}(=_2|f)$ is \#P-hard except for the following cases:
    \begin{itemize}
        \item $f \in \left \langle Z\mathcal{M} \right \rangle$ or $\in \left \langle ZX\mathcal{M} \right \rangle$;
        \item $f$ is $ \mathcal{A}\text{-}$transformable;
        \item $f$ is $ \mathcal{L}\text{-}$transformable;
        \item $f$ is $ \mathcal{P}\text{-}$transformable;
        \item $f\in\mathcal{H}$;
        \item $f \in \left \langle \mathcal{T} \right \rangle$,
    \end{itemize}
    in which cases the problem is computable in polynomial time.
\end{lemma}
\begin{proof}
    Using the same procedure as in Lemma~\ref{BD8}, we construct $f^*$.
    We also have ${f^*}^{34}_{=_2}=(=_2)$, i.e. we have
    \begin{equation}\label{f34=22}
        \left\{\begin{matrix}
 f^*_{0000}+f^*_{0011}=f^*_{1100}+f^*_{1111}=1\\
f^*_{0100}+f^*_{0111}=f^*_{1000}+f^*_{1011}=0
\end{matrix}\right. .
    \end{equation}
For ${f^*}^{34}_{[1,0,-1]}$, there are three cases:
\begin{caselist}
    \caseitem{1}  ${f^*}^{34}_{[1,0,-1]}=(=_2)$ up to a nonzero scalar or ${f^*}^{34}_{[1,0,-1]}\equiv0$.
    In this case, we have
    \begin{equation}\label{casea2}
    \left\{\begin{matrix}
 f^*_{0000}-f^*_{0011}=f^*_{1100}-f^*_{1111}\\
f^*_{0100}-f^*_{0111}=f^*_{1000}-f^*_{1011}=0
\end{matrix}\right. .
    \end{equation}
    By (\ref{f34=22}) and (\ref{casea2}), we have \[
    \left\{\begin{matrix}
 f^*_{0000}=f^*_{1100}\\
 f^*_{1111}=f^*_{0011}\\
f^*_{0100}=f^*_{0111}=f^*_{1000}=f^*_{1011}=0
\end{matrix}\right. .
    \] Now
    \[
M(f^*)=\begin{pmatrix}
  f^*_{0000} &  f^*_{0001}&  f^*_{0010}& f^*_{0011}\\
  0 &  f^*_{0101}&  f^*_{0110}& 0\\
  0 &  f^*_{1001}&  f^*_{1010}& 0\\
  f^*_{0000} &  f^*_{1101}&  f^*_{1110}& f^*_{0011}
\end{pmatrix}.\]
Note that \[
    \left\{\begin{matrix}
 {f^*}^{12}_{=_2}=(2f^*_{0000},f^*_{0001}+f^*_{1101},f^*_{0010}+f^*_{1110},2f^*_{0011})\\
{f^*}^{12}_{[1,0,-1]}=(0,f^*_{0001}-f^*_{1101},f^*_{0010}-f^*_{1110},0)
\end{matrix}\right. .
    \] Then ${f^*}^{12}_{=_2}$ must have the form $(=_2)$ or [$1,0,-1$] since $ f^*_{0000}+f^*_{0011}=1$. Thus we have $f^*_{0000}=f^*_{0011}\ne0$. Further, ${f^*}^{12}_{[1,0,-1]}$ must have the form $(0,1,\pm\frak{i},0)$, this implies that we have
\begin{equation}
    \left\{\begin{matrix}
 f^*_{0001}+f^*_{1101}=f^*_{0010}+f^*_{1110}=0(by\ {f^*}^{12}_{=_2})\\
\varepsilon_1(f^*_{0001}-f^*_{1101})=f^*_{0010}-f^*_{1110}(by\ {f^*}^{12}_{[1,0,-1]})
\end{matrix}\right. ,
\end{equation}where $\varepsilon_1=\pm\frak{i}$.
Now, \[
M(f^*)=\begin{pmatrix}
  f^*_{0000} &  f^*_{0001}&  \varepsilon_1f^*_{0001}& f^*_{0000}\\
   0 &  f^*_{0101}&  f^*_{0110}& 0\\
  0 &  f^*_{1001}&  f^*_{1010}& 0\\
  f^*_{0000} &  -f^*_{0001}&  -\varepsilon_1f^*_{0001}& f^*_{0000}
\end{pmatrix}.\] If $f^*_{0001}=0$
,
$f^*$ has eight-vertex form and then the result follows from Lemma~\ref{haveeightvertex}. If $f^*_{0001}\ne0$,
by ${f^*}^{34}_{(0,1,-\varepsilon_1,0)}=(2f^*_{0001},f^*_{0101}-\varepsilon_1f^*_{0110},f^*_{1001}-\varepsilon_1f^*_{1010},-2f^*_{0001})$ and ${f^*}^{34}_{(0,1,\varepsilon_1,0)}=(0,f^*_{0101}+\varepsilon_1f^*_{0110},f^*_{1001}+\varepsilon_1f^*_{1010},0)$
, we have
\begin{equation}
    \left\{\begin{matrix}
f^*_{0101}-\varepsilon_1f^*_{0110}=f^*_{1001}-\varepsilon_1f^*_{1010}=0\\
\varepsilon_2(f^*_{0101}+\varepsilon_1f^*_{0110})=f^*_{1001}+\varepsilon_1f^*_{1010}
\end{matrix}\right. ,
\end{equation}where $\varepsilon_2=\pm\frak{i}$. Now,
\[
M(f^*)=\begin{pmatrix}
  f^*_{0000} &  f^*_{0001}&  \varepsilon_1f^*_{0001}& f^*_{0000}\\
  0 &  f^*_{0101}&  -\varepsilon_1f^*_{0101}& 0\\
  0 &  \varepsilon_2f^*_{0101}&  -\varepsilon_1\varepsilon_2f^*_{0101}& 0\\
  f^*_{0000} &  -f^*_{0001}&  -\varepsilon_1f^*_{0001}& f^*_{0000}
\end{pmatrix}.\] Then we have\[{f^*}^{23}_{=_2}=(f^*_{0000}-\varepsilon_1 f^*_{0101},f^*_{0001},-\varepsilon_1f^*_{0001},\varepsilon_2f^*_{0101}+f^*_{0000})\] and \[{f^*}^{23}_{=_2}=(f^*_{0000}+\varepsilon_1 f^*_{0101},f^*_{0001},-\varepsilon_1f^*_{0001},\varepsilon_2f^*_{0101}-f^*_{0000}).\] Note that $f^*_{0001}\ne0$, then $f^*_{0000}-\varepsilon_1 f^*_{0101}=f^*_{0000}+\varepsilon_1 f^*_{0101}=0$ implies that $f^*_{0000}=0$, a contradiction.
    \caseitem{2}  ${f^*}^{34}_{[1,0,-1]}= [1,0,-1]$ up to a nonzero scalar. In this case, we have
    \begin{equation}\label{casea19}
    \left\{\begin{matrix}
 f^*_{0000}-f^*_{0011}=-(f^*_{1100}-f^*_{1111})\ne0\\
f^*_{0100}-f^*_{0111}=f^*_{1000}-f^*_{1011}=0
\end{matrix}\right. .
    \end{equation}
    Now,
    \[
M(f^*)=\begin{pmatrix}
  f^*_{0000} &  f^*_{0001}&  f^*_{0010}& f^*_{0011}\\
  0 &  f^*_{0101}&  f^*_{0110}& 0\\
  0 &  f^*_{1001}&  f^*_{1010}& 0\\
  f^*_{0011} &  f^*_{1101}&  f^*_{1110}& f^*_{0000}
\end{pmatrix},\] where $f^*_{0000}+f^*_{0011}\ne0$ and $f^*_{0000}\ne f^*_{0011}$.
 By ${f^*}^{12}_{=_2}=(f^*_{0000}+f^*_{0011},f^*_{0001}+f^*_{1101},f^*_{0010}+f^*_{1110},f^*_{0011}+f^*_{0000})$ and ${f^*}^{12}_{[1,0,-1]}=(f^*_{0000}-f^*_{0011},f^*_{0001}-f^*_{1101},f^*_{0010}-f^*_{1110},f^*_{0011}-f^*_{0000})$, we have
\begin{equation}
    \left\{\begin{matrix}
 f^*_{0001}+f^*_{1101}=f^*_{0010}+f^*_{1110}=0\\
f^*_{0001}-f^*_{1101}=f^*_{0010}-f^*_{1110}=0
\end{matrix}\right. .
\end{equation} Then $ f^*_{0001}=f^*_{1101}=f^*_{0010}=f^*_{1110}=0$, and
\[
M(f^*)=\begin{pmatrix}
  f^*_{0000} &  0&  0& f^*_{0011}\\
  0 &  f^*_{0101}&  f^*_{0110}& 0\\
  0 &  f^*_{1001}&  f^*_{1010}& 0\\
  f^*_{0011} &  0&  0& f^*_{0000}
\end{pmatrix},\] $f^*$ has eight-vertex form and then the result follows from Lemma~\ref{haveeightvertex}.

    \caseitem{3}  ${f^*}^{34}_{[1,0,-1]}=(0,1,\frak{i},0)$ or ${f^*}^{34}_{[1,0,-1]}=(0,1,-\frak{i},0)$, both up to a nonzero scalar. In this case, we have
    \begin{equation}\label{case3}
    \left\{\begin{matrix}
 f^*_{0000}-f^*_{0011}=f^*_{1100}-f^*_{1111}=0\\
\varepsilon_1(f^*_{0100}-f^*_{0111})=f^*_{1000}-f^*_{1011}\ne0
\end{matrix}\right. ,
    \end{equation}where $\varepsilon_1=\pm\frak{i}$. Now,
\[
M(f^*)=\begin{pmatrix}
  f^*_{0000} &  f^*_{0001}&  f^*_{0010}& f^*_{0000}\\
  f^*_{0100} &  f^*_{0101}&  f^*_{0110}& -f^*_{0100}\\
  \varepsilon_1f^*_{0100} &  f^*_{1001}&  f^*_{1010}& -\varepsilon_1f^*_{0100}\\
  f^*_{0000} &  *&  *& f^*_{0000}
\end{pmatrix}\]where $f^*_{0100}\ne0$.
     If $f^*_{0000}=0$, then ${f^*}^{34}_{=_2}\equiv0$, a contradiction. If $f^*_{0000}\ne0$, then by ${f^*}^{12}_{=_2}=(2f^*_{0000},f^*_{0001}+f^*_{1101},f^*_{0010}+f^*_{1110},2f^*_{0000})$ and ${f^*}^{12}_{[1,0,-1]}=(0,f^*_{0001}-f^*_{1101},f^*_{0010}-f^*_{1110},0)$, we have
\begin{equation}
    \left\{\begin{matrix}
 f^*_{0001}+f^*_{1101}=f^*_{0010}+f^*_{1110}=0\\
\varepsilon_2(f^*_{0001}-f^*_{1101})=f^*_{0010}-f^*_{1110}
\end{matrix}\right. ,
\end{equation}where $\varepsilon_2=\pm\frak{i}$. Now,
\[
M(f^*)=\begin{pmatrix}
  f^*_{0000} &  f^*_{0001}&  \varepsilon_2f^*_{0001}& f^*_{0000}\\
  f^*_{0100} &  f^*_{0101}&  f^*_{0110}& -f^*_{0100}\\
  \varepsilon_1f^*_{0100} &  f^*_{1001}&  f^*_{1010}& -\varepsilon_1f^*_{0100}\\
  f^*_{0000} &  -f^*_{0001}&  -\varepsilon_2f^*_{0001}& f^*_{0000}
\end{pmatrix}.\]
If $f^*_{0001}=0$, i.e. ${f^*}^{12}_{=_2}=(=_2)$ and ${f^*}^{12}_{[1,0,-1]}\equiv0$. By symmetry, ${f^*}_{x_3,x_4,x_1,x_2}$ falls under \caseref{1}. If $f^*_{0001}\ne0$, by ${f^*}^{34}_{(0,1,-\varepsilon_2,0)}=(2f^*_{0001},f^*_{0101}-\varepsilon_2f^*_{0110},f^*_{1001}-\varepsilon_2f^*_{1010},-2f^*_{0001})$ and ${f^*}^{34}_{(0,1,\varepsilon_2,0)}=(0,f^*_{0101}+\varepsilon_2f^*_{0110},f^*_{1001}+\varepsilon_2f^*_{1010},0)$
, we have
\begin{equation}
    \left\{\begin{matrix}
f^*_{0101}-\varepsilon_2f^*_{0110}=f^*_{1001}-\varepsilon_2f^*_{1010}=0\\
\varepsilon_3(f^*_{0101}+\varepsilon_2f^*_{0110})=f^*_{1001}+\varepsilon_2f^*_{1010}
\end{matrix}\right. ,
\end{equation}where $\varepsilon_3=\pm\frak{i}$.
\[
M(f^*)=\begin{pmatrix}
  f^*_{0000} &  f^*_{0001}&  \varepsilon_2f^*_{0001}& f^*_{0000}\\
  f^*_{0100} &  f^*_{0101}&  -\varepsilon_2f^*_{0101}& -f^*_{0100}\\
  \varepsilon_1f^*_{0100} &  \varepsilon_3f^*_{0101}&  -\varepsilon_2\varepsilon_3f^*_{0101}& -\varepsilon_1f^*_{0100}\\
  f^*_{0000} &  -f^*_{0001}&  -\varepsilon_2f^*_{0001}& f^*_{0000}
\end{pmatrix}.\]
By ${f^*}^{13}_{=_2}=(f^*_{0000}-\varepsilon_2\varepsilon_3f^*_{0101},f^*_{0001}-\varepsilon_1f^*_{0100},f^*_{0100}-\varepsilon_2f^*_{0001},f^*_{0101}+f^*_{0000})$ and ${f^*}^{13}_{[1,0,-1]}=(f^*_{0000}+\varepsilon_2\varepsilon_3f^*_{0101},f^*_{0001}+\varepsilon_1f^*_{0100},f^*_{0100}+\varepsilon_2f^*_{0001},f^*_{0101}-f^*_{0000})$. We have
\begin{equation}
    \left\{\begin{matrix}
(f^*_{0000}-\varepsilon_2\varepsilon_3f^*_{0101})^2=(f^*_{0101}+f^*_{0000})^2\\
(f^*_{0001}-\varepsilon_1f^*_{0100})^4=(f^*_{0100}-\varepsilon_2f^*_{0001})^4\\
(f^*_{0000}+\varepsilon_2\varepsilon_3f^*_{0101})^2=(f^*_{0101}-f^*_{0000})^2\\
(f^*_{0001}+\varepsilon_1f^*_{0100})^4=(f^*_{0100}+\varepsilon_2f^*_{0001})^4
\end{matrix}\right. ,
\end{equation}
which implies that
\[
\left\{\begin{matrix}
 -\varepsilon_2\varepsilon_3=1\ \text{or}\ f^*_{0101}=0 \\
\varepsilon_1=\varepsilon_2\ \text{or}\ (f^*_{0100})^2=(f^*_{0001})^2
\end{matrix}\right. ,
\]
If $f^*_{0101}=0$, by connecting input $x_2$ of $(0,1,-\varepsilon_2,0)$ to input $x_4$ of $f^*$, $[1,0,-1]$ to input $x_1$ of $f^*$, we construct
\[M(f^{**})=\begin{pmatrix}
  f^*_{0001} &  -\varepsilon_2f^*_{0000}&  f^*_{0000}& f^*_{0001}\\
  0 &  -\varepsilon_2f^*_{0100}&  -f^*_{0100}& 0\\
  0 &  \varepsilon_1\varepsilon_2f^*_{0100}&  \varepsilon_1f^*_{0100}& 0\\
  f^*_{0001} &  \varepsilon_2f^*_{0000}&  -f^*_{0000}& f^*_{0001}
\end{pmatrix}.\] Then $f^{**}$ falls under \caseref{1} since $f^*_{0001}\ne0$.
If ${f^*}_{0100}^2=(f^*_{0001})^2$, then $(f^*_{0001}-\varepsilon_1f^*_{0100}),(f^*_{0100}-\varepsilon_2f^*_{0001}),(f^*_{0001}+\varepsilon_1f^*_{0100}),(f^*_{0100}+\varepsilon_2f^*_{0001})$ are nonzero. Then we have $f^*_{0000}-\varepsilon_2\varepsilon_3f^*_{0101}=f^*_{0101}+f^*_{0000}=f^*_{0000}+\varepsilon_2\varepsilon_3f^*_{0101}=f^*_{0101}-f^*_{0000}$, which implies that $f^*_{0000}=f^*_{0101}=0$, a contradiction. Thus we have $\varepsilon_1=\varepsilon_2=\varepsilon_3$, then
\[
M(f^*)=\begin{pmatrix}
  f^*_{0000} &  f^*_{0001}&  \varepsilon_1f^*_{0001}& f^*_{0000}\\
  f^*_{0100} &  f^*_{0101}&  -\varepsilon_1f^*_{0101}& -f^*_{0100}\\
  \varepsilon_1f^*_{0100} &  \varepsilon_1f^*_{0101}&  f^*_{0101}& -\varepsilon_1f^*_{0100}\\
  f^*_{0000} &  -f^*_{0001}&  -\varepsilon_1f^*_{0001}& f^*_{0000}
\end{pmatrix}.\]
${f^*}^{13}_{=_2}=(f^*_{0000}+f^*_{0101},f^*_{0001}-\varepsilon_1f^*_{0100},f^*_{0100}-\varepsilon_1f^*_{0001},f^*_{0101}+f^*_{0000})$ and ${f^*}^{13}_{[1,0,-1]}=(f^*_{0000}-f^*_{0101},f^*_{0001}+\varepsilon_1f^*_{0100},f^*_{0100}+\varepsilon_1f^*_{0001},f^*_{0101}-f^*_{0000})$

If $f^*_{0000} \ne f^*_{0101} $, by ${f^*}^{13}_{[1,0,-1]}=(f^*_{0000} - f^*_{0101},f^*_{0001} +\varepsilon_1f^*_{0100} , f^*_{0100}+\varepsilon_1f^*_{0001} ,f^*_{0101} - f^*_{0000})$, we have $f^*_{0001}+\varepsilon_1 f^*_{0100}= f^*_{0100}+\varepsilon_1 f^*_{0001}=0$, i.e. $f^*_{0100}=f^*_{0001}=0$, a contradiction.

If $f^*_{0000}=f^*_{0101}$, by ${f^*}^{13}_{=_2}=(f^*_{0000}+f^*_{0101},f^*_{0001}-\varepsilon_1f^*_{0100},f^*_{0100}-\varepsilon_1f^*_{0001},f^*_{0101}+f^*_{0000})$ we have $f^*_{0001}-\varepsilon_1f^*_{0100}=f^*_{0100}-\varepsilon_1f^*_{0001}=0$, i.e. $f^*_{0100}=f^*_{0001}=0$, a contradiction.\qedhere

\end{caselist}

\end{proof}
\begin{lemma}\label{bd4n>3}
    Let $f$ be a 4-ary signature, $\mathcal{GEQ}$ be a set of generalized equality binary signatures, $\mathcal{GDE}$ be a set of generalized disequality binary signatures, $|\mathcal{GEQ}|\ge3,|\mathcal{GDE}|\ge3$ and $(\ne_2)\in\mathcal{GDE}$ or $(0,1,-1,0)\in\mathcal{GDE}$.
    $\mathcal{B}=\mathcal{GEQ}\cup\mathcal{GDE}$, $S(f_{\mathcal{B}})\subseteq \mathbb{C}\mathcal{B}$. Then $\mathrm{Holant}(=_2|f)$ is \#P-hard except for the following cases:
    \begin{itemize}
        \item $f \in \left \langle Z\mathcal{M} \right \rangle$ or $\in \left \langle ZX\mathcal{M} \right \rangle$;
        \item $f$ is $ \mathcal{A}\text{-}$transformable;
        \item $f$ is $ \mathcal{L}\text{-}$transformable;
        \item $f$ is $ \mathcal{P}\text{-}$transformable;
        \item $f\in\mathcal{H}$;
        \item $f \in \left \langle \mathcal{T} \right \rangle$,
    \end{itemize}
    in which cases the problem is computable in polynomial time.
\end{lemma}
\begin{proof}
We may assume that $f\notin\langle\mathcal T\rangle$, since otherwise the stated tractable alternative holds. The invertible input modifications below preserve this exclusion, as noted in the definition of binary modification.

    Note that there exist at least three distinct binary signatures in $\mathcal{GEQ}$ or $\mathcal{GDE}$.
    Let \[\{[1,0,t_1],[1,0,t_2],[1,0,t_3]\}\subseteq \mathcal{GEQ}\] and \[\{(0,1,t_4,0),(0,1,t_5,0),(0,1,t_6,0)\}\subseteq \mathcal{GDE}.\] In \[{f}_{[1,0,t_1]}^{34},{f}_{[1,0,t_2]}^{34},{f}_{[1,0,t_3]}^{34},\] there exist at least two distinct $i,j\in\{1,2,3\}$, such that ${f}_{[1,0,t_i]}^{34}$ and ${f}_{[1,0,t_j]}^{34}$ are both generalized equality  binary signatures or both generalized disequality binary signatures. If they are both generalized equality  binary signatures, we have \[
      \left\{\begin{matrix}
f_{0100}+t_if_{0111}=f_{1000}+t_if_{1011}=0\\
f_{0100}+t_jf_{0111}=f_{1000}+t_jf_{1011}=0
\end{matrix}\right. .
    \] This implies that $f_{0100}=f_{0111}=f_{1000}=f_{1011}=0$. If they are both generalized disequality binary signatures, we have \[
      \left\{\begin{matrix}
f_{0000}+t_if_{0011}=f_{1100}+t_if_{1111}=0\\
f_{0000}+t_jf_{0011}=f_{1100}+t_jf_{1111}=0
\end{matrix}\right. .
    \] This implies that $f_{0000}=f_{0011}=f_{1100}=f_{1111}=0$. Similarly,  there are also two cases in \[f_{(0,1,t_4,0)}^{34},f_{(0,1,t_5,0)}^{34},f_{(0,1,t_6,0)}^{34},\] which imply that $f_{0001}=f_{0010}=f_{1101}=f_{1110}=0$ or $f_{0101}=f_{0110}=f_{1001}=f_{1010}=0$. Thus, there are four cases:
    \begin{caselist}
        \caseitem{1}  \[M(f)=\begin{pmatrix}
  f_{0000}&  f_{0001}&  f_{0010}& f_{0011}\\
  0 &  0& 0& 0\\
 0 &  0& 0&0\\
  {f_{1100}} & {f_{1101}}&  {f_{1110}}& {f_{1111}}
\end{pmatrix}.\] In this case, we further consider \[{f}_{[1,0,t_1]}^{12},{f}_{[1,0,t_2]}^{12},{f}_{[1,0,t_3]}^{12}.\] This implies that \[M(f)=\begin{pmatrix}
  f_{0000}& 0&  0& f_{0011}\\
  0 &  0& 0& 0\\
 0 &  0& 0&0\\
  {f_{1100}} & 0&  0& {f_{1111}}
\end{pmatrix}\] or \[M(f)=\begin{pmatrix}
  0&  f_{0001}&  f_{0010}& 0\\
  0 &  0& 0& 0\\
 0 &  0& 0&0\\
  0 & {f_{1101}}&  {f_{1110}}& 0
\end{pmatrix}.\]
 \caseitem{2}  \[M(f)=\begin{pmatrix}
  f_{0000}&  0&  0& f_{0011}\\
  0 &  {f_{0101}}& {f_{0110}}& 0\\
 0 &  {f_{1001}}& {f_{1010}}&0\\
  {f_{1100}} & 0&  0& {f_{1111}}
\end{pmatrix}\]
 \caseitem{3}  \[M(f)=\begin{pmatrix}
  0&  0&  0& 0\\
  {f_{0100}} &  {f_{0101}}& {f_{0110}}& {f_{0111}}\\
 {f_{1000}} &  {f_{1001}}& {f_{1010}}& {f_{1011}}\\
  0 &0&  0& 0
\end{pmatrix}.\] In this case, we further consider \[f_{(0,1,t_4,0)}^{34},f_{(0,1,t_5,0)}^{34},f_{(0,1,t_6,0)}^{34}.\] This implies that \[M(f)=\begin{pmatrix}
  0&  0&  0& 0\\
  {f_{0100}} & 0& 0& {f_{0111}}\\
 {f_{1000}} &  0& 0& {f_{1011}}\\
  0 &0&  0& 0
\end{pmatrix}\] or \[M(f)=\begin{pmatrix}
  0&  0&  0& 0\\
0&  {f_{0101}}& {f_{0110}}& 0\\
 0 &  {f_{1001}}& {f_{1010}}& 0\\
  0 &0&  0& 0
\end{pmatrix}.\]
 \caseitem{4}  \[M(f)=\begin{pmatrix}
  0&  f_{0001}&  f_{0010}&0\\
  {f_{0100}} &  0&0& {f_{0111}}\\
 {f_{1000}} &  0& 0& {f_{1011}}\\
  0 & {f_{1101}}&  {f_{1110}}& 0
\end{pmatrix}\]
    \end{caselist}
    In summary, \[M(f)=\begin{pmatrix}
  f_{0000}&  0&  0& f_{0011}\\
  0 &  {f_{0101}}& {f_{0110}}& 0\\
 0 &  {f_{1001}}& {f_{1010}}&0\\
  {f_{1100}} & 0&  0& {f_{1111}}
\end{pmatrix}\] or \[M(f)=\begin{pmatrix}
  0&  f_{0001}&  f_{0010}&0\\
  {f_{0100}} &  0&0& {f_{0111}}\\
 {f_{1000}} &  0& 0& {f_{1011}}\\
  0 & {f_{1101}}&  {f_{1110}}& 0
\end{pmatrix}.\]

If \[M(f)=\begin{pmatrix}
  f_{0000}&  0&  0& f_{0011}\\
  0 &  {f_{0101}}& {f_{0110}}& 0\\
 0 &  {f_{1001}}& {f_{1010}}&0\\
  {f_{1100}} & 0&  0& {f_{1111}}
\end{pmatrix},\] the result follows from Theorem~\ref{eightvertex=2}. If \[M(f)=\begin{pmatrix}
  0&  f_{0001}&  f_{0010}&0\\
  {f_{0100}} &  0&0& {f_{0111}}\\
 {f_{1000}} &  0& 0& {f_{1011}}\\
  0 & {f_{1101}}&  {f_{1110}}& 0
\end{pmatrix},\] by applying a binary modification using $(\ne_2)$ or $(0,1,-1,0)$, we construct a new 4-ary signature $f^*$ with eight-vertex form and complete the proof by Lemma~\ref{haveeightvertex}.
\end{proof}
\begin{lemma}\label{dicbd4n}
If $\mathbf{H}^{(m)}_{f}=PBD_{4n}P^{-1}$, then $\mathrm{Holant}(=_2|f)$ is \#P-hard except for the following cases:
    \begin{itemize}
        \item $f \in \left \langle Z\mathcal{M} \right \rangle$ or $\in \left \langle ZX\mathcal{M} \right \rangle$;
        \item $f$ is $ \mathcal{A}\text{-}$transformable;
        \item $f$ is $ \mathcal{L}\text{-}$transformable;
        \item $f$ is $ \mathcal{P}\text{-}$transformable;
        \item $f\in\mathcal{H}$;
        \item $f \in \left \langle \mathcal{T} \right \rangle$,
    \end{itemize}
    in which cases the problem is computable in polynomial time.
\end{lemma}

\begin{proof}
    If $n=1$, then $BD_{4}$ degenerates into a cyclic group(up to conjugacy), and we omit it.
In Holant($=_2|f,PBD_{4n}P^{-1}$), by a holographic transformation using $P$, we have \[\text{Holant}(S^{-1}|f',BD_{4n}S)\equiv_T\text{Holant}(=_2|f,PBD_{4n}P^{-1}),\]
where $f'=(P^{-1})^{\otimes 4}f$, $S^{-1}=P^TP$. Then we divide the proof into two cases: $n = 2$ in \caseref{1} and $n \ge 3$ in \caseref{2}.
\begin{caselist}
        \caseitem{1}  $n=2$.

        In this case, by Lemma~\ref{normalizer}, we have  $\mathcal{N}_{\mathrm{SL}(2,\mathbb{C})}(BD_{8})=BO_{48}$. If $S^{-1}\in BD_8$, then $BD_{8}S=BD_{8}$. Note that $(=_2)\in BD_{8}$, then by Lemma~\ref{poly=_2},
        we have
        \[\text{Holant}(=_2|f',BD_{8})\equiv_T\text{Holant}(S^{-1}|f',BD_{8}S).\] Then $\mathbf{H}^{(m)}_{f'}=BD_{8}$ and the result follows from Lemma~\ref{BD8}.
        If $S^{-1}\in BO_{48}\setminus BD_{8}$, note that $S^{-1}=P^{T}P$ is symmetric, then $S$ up to a scalar, has the form $\begin{pmatrix}
  1&0 \\
  0&\epsilon\frak{i}
\end{pmatrix}$, or $\begin{pmatrix}
  1&\epsilon \\
  \epsilon&-1
\end{pmatrix}$, or $\begin{pmatrix}
  1&\epsilon\frak{i} \\
  \epsilon\frak{i}&1
\end{pmatrix}$, where $\epsilon=\pm1$. Let $D_{\pm\epsilon}=\begin{pmatrix}
  1&0 \\
  0&(\pm\epsilon\frak{i})^{-\frac{1}{2}}
\end{pmatrix}$. We consider the following cases.

\begin{caselist}
    \subcaseitem{1.1}  $S^{-1}$ has the form  $\begin{pmatrix}
  1&0 \\
  0&\epsilon\frak{i}
\end{pmatrix}$.
Note that \[BD_{8}S=\{[1,0,\frak{i}],[1,0,-\frak{i}],(0,\frak{i},1,0),(0,-\frak{i},1,0) \}\] for any $\epsilon$.
By a holographic transformation using $D_{\epsilon}$ in $\text{Holant}(S^{-1}|f',BD_{8}S)$, we have \[\text{Holant}(=_2|f'',D_{\epsilon}^{-1}BD_{8}S(D_{\epsilon}^{-1})^T)\equiv_T\text{Holant}(S^{-1}|f',BD_{8}S).\] Note that $D_{\epsilon}^{-1}BD_{8}S(D_{\epsilon}^{-1})^T=\{=_2,[1,0,-1],(0,1,\frak{i},0),(0,1,-\frak{i},0)\}$, then the result follows from Lemma~\ref{BD16/BD8}.

\subcaseitem{1.2}   $S^{-1}$ has the form  $\begin{pmatrix}
  1&\epsilon \\
  \epsilon&-1
\end{pmatrix}$. Note that \[BD_{8}S=\{[1,1,-1],[-1,1,1],(1,1,-1,1),(1,-1,1,1)\}\] for any $\epsilon$ up to a nonzero scalar.
By a holographic transformation using $ZD_{-\epsilon}$, we have

 \[\text{Holant}(=_2|f'',(D_{-\epsilon}^{-1}Z^{-1})BD_{8}S(D_{-\epsilon}^{-1}Z^{-1})^T)\equiv_T\text{Holant}(S^{-1}|f',BD_{8}S).\] Note that $(D_{-\epsilon}^{-1}Z^{-1})BD_{8}S(D_{-\epsilon}^{-1}Z^{-1})^T=\{=_2,[1,0,-1],(0,1,\frak{i},0),(0,1,-\frak{i},0)\}$.
Then the result follows from Lemma~\ref{BD16/BD8}.
\subcaseitem{1.3}   $S^{-1}$ has the form  $\begin{pmatrix}
  1&\epsilon\frak{i} \\
  \epsilon\frak{i}&1
\end{pmatrix}$. By a holographic transformation using $Z^{-1}D_{-\epsilon}$, we have \[\text{Holant}(=_2|f'',(D^{-1}_{-\epsilon}Z)BD_{8}S(D^{-1}_{-\epsilon}Z)^T)\equiv_T\text{Holant}(S^{-1}|f',BD_{8}S).\] Note that $(D^{-1}_{-\epsilon}Z)BD_{8}S(D^{-1}_{-\epsilon}Z)^T=\{=_2,[1,0,-1],(0,1,\frak{i},0),(0,1,-\frak{i},0)\}$, then the result follows from Lemma~\ref{BD16/BD8}.

\end{caselist}
\caseitem{2}  $n>2$. In this case, by Lemma~\ref{normalizer}, we have $\mathcal{N}_{\mathrm{SL}(2,\mathbb{C})}(BD_{4n})=BD_{8n}.$

If $S^{-1}\in BD_{4n}$,
     then $BD_{4n}S=BD_{4n}$. Note that $(=_2)\in BD_{4n}$, by Lemma~\ref{poly=_2}, we have
    \[
    \text{Holant}(=_2|f',BD_{4n})\equiv_T\text{Holant}(S^{-1}|f',BD_{4n}S).
    \] Then we have $S(f'_{BD_{4n}})\subseteq BD_{4n}$. The result follows from Lemma~\ref{bd4n>3}.

    If $S^{-1}\in BD_{8n}\setminus BD_{4n}$, then $S^{-1}$ has the form $\begin{pmatrix}
  a&0 \\
  0&a^{-1}
\end{pmatrix}$ or $\begin{pmatrix}
  0&1 \\
  1&0
\end{pmatrix}$.
If $S^{-1}$ has the form $\begin{pmatrix}
  a&0 \\
  0&a^{-1}
\end{pmatrix}$. By a holographic transformation using $D=\begin{pmatrix}
  a^{-\frac{1}{2}}&0 \\
  0&a^{\frac{1}{2}}
\end{pmatrix}$, we have
\[\text{Holant}(=_2|f'',D^{-1}BD_{4n}S(D^{-1})^T)\equiv_T\text{Holant}(S^{-1}|f',BD_{4n}S).\]
The result follows from Lemma~\ref{bd4n>3}.

If $S^{-1}$ has the form $\begin{pmatrix}
  0&1 \\
  1&0
\end{pmatrix}$, $\text{Holant}(S^{-1}|f',BD_{4n}S)$ reduce to
$\text{Holant}(\ne_2|f',BD_{4n}(\ne_2))$.
$BD_{4n}(\ne_2)$ is made up of $\begin{pmatrix}
 \varepsilon^k & 0\\
  0& -\varepsilon^{-k}
\end{pmatrix}$ and $\begin{pmatrix}
 0 & \varepsilon^k\\
  \varepsilon^{-k}& 0
\end{pmatrix}$ for $k\in[2n]$. By a holographic transformation using $\begin{pmatrix}
 1 & 0\\
  0& i
\end{pmatrix}$, we have \[
\text{Holant}(\ne_2|f'',(BD_{4n}(\ne_2))')\equiv_T \text{Holant}(\ne_2|f',BD_{4n}(\ne_2)),
\] where $(BD_{4n}(\ne_2))'$ is made up of $\begin{pmatrix}
 \varepsilon^k & 0\\
  0& \varepsilon^{-k}
\end{pmatrix}$ and $\begin{pmatrix}
 0 & \varepsilon^k\\
  \varepsilon^{-k}& 0
\end{pmatrix}$. Note that $=_2\in (BD_{4n}(\ne_2))'$. By Lemma~\ref{bd4n>3} the result follows.\qedhere
    \end{caselist}

\end{proof}
\subsubsection{\texorpdfstring{Conjugacy of $\mathbf{H}^{(m)}_{f}$ to $BT_{24}$}{Conjugacy of H\_f\textasciicircum (m) to BT\_24}}
We list representatives, up to nonzero scalar multiples, of the elements of $BT_{24}$ and divide them into the following disjoint sets.
\[\mathcal{GEQ}= \left\{
\begin{pmatrix}
1  & 0\\
 0 &1
\end{pmatrix},
\begin{pmatrix}
1  & 0\\
 0 &-1
\end{pmatrix}\right \}\]
\[\mathcal{GDE}=\left\{\begin{pmatrix}
0 & 1\\
 1 &0
\end{pmatrix},
\begin{pmatrix}
0 & 1\\
 -1 &0
\end{pmatrix}\right \}\]
{\small \[\text{Z-family}=\left \{\begin{pmatrix}
1 & 1\\
 i &-i
\end{pmatrix},
\begin{pmatrix}
1 & -i\\
 -1 &-i
\end{pmatrix},
\begin{pmatrix}
1 & i\\
 1 &-i
\end{pmatrix},
\begin{pmatrix}
1 & -1\\
 -i &-i
\end{pmatrix},\begin{pmatrix}
-i & -1\\
 -i &1
\end{pmatrix},
\begin{pmatrix}
-i & i\\
 1 &1
\end{pmatrix},
\begin{pmatrix}
-i & -i\\
 -1 &1
\end{pmatrix},
\begin{pmatrix}
-i & 1\\
 i &1
\end{pmatrix}\right \}\]}

We use the following fact.
\begin{remark}
    For any $h\in\mathbb{C}^{\times} BT_{24}$, if $\frac{h_{00}}{h_{11}}=\pm1$, then $h\in \mathbb{C}^{\times}\mathcal{GEQ}$; if $\frac{h_{01}}{h_{10}}=\pm1$, then $h\in \mathbb{C}^{\times}\mathcal{GDE}$; if $\frac{h_{00}}{h_{11}}=\pm\frak{i}$ or $\frac{h_{01}}{h_{10}}=\pm\frak{i}$, then $h\in \mathbb{C}^{\times} Z\mathrm{-family}$.
\end{remark}
In the proof of the following lemma, we will frequently use this fact without explanation.

\begin{lemma}\label{BT24nsolve}
 Let $f$ be a 4-ary signature,  $S(f_{BT_{24}})\subseteq BT_{24}$. Then $\mathrm{Holant}(=_2|f)$ is \#P-hard except for the following cases:
    \begin{itemize}
        \item $f \in \left \langle Z\mathcal{M} \right \rangle$ or $\in \left \langle ZX\mathcal{M} \right \rangle$;
        \item $f$ is $ \mathcal{A}\text{-}$transformable;
        \item $f$ is $ \mathcal{L}\text{-}$transformable;
        \item $f$ is $ \mathcal{P}\text{-}$transformable;
        \item $f\in\mathcal{H}$;
        \item $f \in \left \langle \mathcal{T} \right \rangle$,
    \end{itemize}
    in which cases the problem is computable in polynomial time.
\end{lemma}
\begin{proof}
We may assume that $f\notin\langle\mathcal T\rangle$, since otherwise the stated tractable alternative holds. The invertible input modifications below preserve this exclusion, as noted in the definition of binary modification.

     Using the same procedure as in Lemma~\ref{BD8}, we construct $f^*$.
    We also have ${f^*}^{34}_{=_2}=(=_2)$. Then we have $-f^*_{0100}=f^*_{0111}$ and $-f^*_{1000}=f^*_{1011}$.
    For ${f^*}^{34}_{[1,0,-1]}$, we consider the following cases:
\begin{caselist}
\caseitem{1}  ${f^*}^{34}_{[1,0,-1]}\in Z\mathrm{-family}$ up to a nonzero scalar.

Note that ${f^*}^{34}_{[1,0,-1]}=(f^*_{0000}-f^*_{0011},2f^*_{0100},2f^*_{1000},f^*_{1100}-f^*_{1111})$. Then we have $\delta _1f^*_{0100}=f^*_{1000}$, where $\delta_1=\pm \frak{i}$ and $f^*_{0100}\ne 0$, since ${f^*}^{34}_{[1,0,-1]}\in \mathbb{C}^{\times} Z\mathrm{-family}$.
In this case, we have \[
M(f^*)=\begin{pmatrix}
  f^*_{0000} &  f^*_{0001}&  f^*_{0010}& f^*_{0011}\\
  f^*_{0100} &  *&  *& -f^*_{0100}\\
  \delta_1f^*_{0100} &  *&  *& -\delta_1 f^*_{0100}\\
  f^*_{1100} &  f^*_{1101}&  f^*_{1110}& f^*_{1111}
\end{pmatrix}.\]By ${f^*}^{12}_{\ne_2}=((1+\delta_1)f^*_{0100},f^*_{0101}+f^*_{1001},f^*_{0110}+f^*_{1010},-(1+\delta_1)f^*_{0100})$
and ${f^*}^{12}_{(0,1,-1,0)}=((1-\delta_1)f^*_{0100},f^*_{0101}-f^*_{1001},f^*_{0110}-f^*_{1010},-(1-\delta_1)f^*_{0100})$, we have $\{{f^*}^{12}_{\ne_2},{f^*}^{12}_{(0,1,-1,0)}\}\subseteq\mathcal{GEQ}$(up to a nonzero scalar) since $f^*_{0100}\ne0$. Then $f^*_{0101}=f^*_{1001}=f^*_{0110}=f^*_{1010}=0$.

\[
M(f^*)=\begin{pmatrix}
  f^*_{0000} &  f^*_{0001}&  f^*_{0010}& f^*_{0011}\\
  f^*_{0100} &  0&  0& -f^*_{0100}\\
  \delta_1f^*_{0100} &  0&  0& -\delta_1 f^*_{0100}\\
  f^*_{1100} &  f^*_{1101}&  f^*_{1110}& f^*_{1111}
\end{pmatrix}.\] By ${f^*}^{34}_{\ne_2}=[f^*_{0001}+f^*_{0010},0,f^*_{1101}+f^*_{1110}]$ and ${f^*}^{34}_{(0,1,-1,0)}=[f^*_{0001}-f^*_{0010},0,f^*_{1101}-f^*_{1110}]$, we have \[\left\{\begin{matrix}
 f^*_{0001}+f^*_{0010}=f^*_{1101}+f^*_{1110}\\
 f^*_{0001}-f^*_{0010}=f^*_{1101}-f^*_{1110}
\end{matrix}\right.\] or \[\left\{\begin{matrix}
 f^*_{0001}+f^*_{0010}=f^*_{1101}+f^*_{1110}\\
 f^*_{0001}-f^*_{0010}=f^*_{1110}-f^*_{1101}
\end{matrix}\right.\] or \[\left\{\begin{matrix}
 f^*_{0001}+f^*_{0010}=-f^*_{1101}-f^*_{1110}\\
 f^*_{0001}-f^*_{0010}=f^*_{1101}-f^*_{1110}
\end{matrix}\right.\] or \[\left\{\begin{matrix}
 f^*_{0001}+f^*_{0010}=-f^*_{1101}-f^*_{1110}\\
 f^*_{0001}-f^*_{0010}=f^*_{1110}-f^*_{1101}
\end{matrix}\right.\] which implies that
$\left\{\begin{matrix}
 f^*_{0001}=f^*_{1101}\\
 f^*_{0010}=f^*_{1110}
\end{matrix}\right.$ or $\left\{\begin{matrix}
 f^*_{0001}=f^*_{1110}\\
 f^*_{0010}=f^*_{1101}
\end{matrix}\right.$ or $\left\{\begin{matrix}
 f^*_{0001}=-f^*_{1110}\\
 f^*_{0010}=-f^*_{1101}
\end{matrix}\right.$ or $\left\{\begin{matrix}
 f^*_{0001}=-f^*_{1101}\\
 f^*_{0010}=-f^*_{1110}
\end{matrix}\right.$.
We conclude that
$\left\{\begin{matrix}
 f^*_{1101}=\epsilon_1 f^*_{0001}\\
 f^*_{1110}=\epsilon_1 f^*_{0010}
\end{matrix}\right.$ or $\left\{\begin{matrix}
 f^*_{1110}=\epsilon_1 f^*_{0001}\\
 f^*_{1101}=\epsilon_1 f^*_{0010}
\end{matrix}\right.$ where $\epsilon_1=\pm1$. Now, \[
M(f^*)=\begin{pmatrix}
  f^*_{0000} &  f^*_{0001}&  f^*_{0010}& f^*_{0011}\\
  f^*_{0100} &  0&  0& -f^*_{0100}\\
  \delta_1f^*_{0100} &  0&  0& -\delta_1 f^*_{0100}\\
  f^*_{1100} &  \epsilon_1 f^*_{0001}&  \epsilon_1 f^*_{0010}& f^*_{1111}
\end{pmatrix}\ \mathrm{or}\  \begin{pmatrix}
  f^*_{0000} &  f^*_{0001}&  f^*_{0010}& f^*_{0011}\\
  f^*_{0100} &  0&  0& -f^*_{0100}\\
  \delta_1f^*_{0100} &  0&  0& -\delta_1 f^*_{0100}\\
  f^*_{1100} &  \epsilon_1 f^*_{0010}&  \epsilon_1 f^*_{0001}& f^*_{1111}
\end{pmatrix}\]

We first deal with the case in which $\left\{\begin{matrix}
 f^*_{1101}=\epsilon_1 f^*_{0001}\\
 f^*_{1110}=\epsilon_1 f^*_{0010}
\end{matrix}\right.$.
Note that
\[
\left\{\begin{matrix}
 {f^*}^{34}_{=_2}=(f^*_{0000}+f^*_{0011},0,0,f^*_{1100}+f^*_{1111})=(=_2)\\
 {f^*}^{34}_{[1,0,-1]}=(f^*_{0000}-f^*_{0011},2f^*_{0100},2\delta_1f^*_{0100},f^*_{1100}-f^*_{1111})\\
 {f^*}^{12}_{[1,0,\epsilon_1]}=(f^*_{0000}+\epsilon_1f^*_{1100},2f^*_{0001},2f^*_{0010},f^*_{0011}+\epsilon_1f^*_{1111})\\
 {f^*}^{12}_{[1,0,-\epsilon_1]}=(f^*_{0000}-\epsilon_1f^*_{1100},0,0,f^*_{0011}-\epsilon_1f^*_{1111})\\
{f^*}^{13}_{=_2}=(f^*_{0000},f^*_{0001}-\delta_1f^*_{0100},f^*_{0100}+\epsilon_1f^*_{0010},f^*_{1111})\\
{f^*}^{13}_{\ne_2}=(f^*_{0010}+\delta_1f^*_{0100},f^*_{0011},f^*_{1100},-f^*_{0100}+\epsilon_1f^*_{0001})
\end{matrix}\right.
\]
We have \[
\left\{\begin{matrix}
 f^*_{0000}+f^*_{0011}=f^*_{1100}+f^*_{1111}\ne0\\
 f^*_{0000}-f^*_{0011}=\delta_2(f^*_{1100}-f^*_{1111})\ne0\\
f^*_{0000}+\epsilon_1f^*_{1100}=\varepsilon_1(f^*_{0011}+\epsilon_1f^*_{1111})\\
 f^*_{0000}-\epsilon_1f^*_{1100}=\epsilon_2(f^*_{0011}-\epsilon_1f^*_{1111})\\
f^*_{0000}=\varepsilon_2f^*_{1111}\\
f^*_{0011}=\varepsilon_3f^*_{1100}
\end{matrix}\right.
\]
Then $\{f^*_{0000},f^*_{0011},f^*_{1100},f^*_{1111}\}$ satisfies a homogeneous linear system of
equations whose coefficient matrix is:
\[
\begin{pmatrix}
  1& 1 &  -1& -1\\
  1&  -1&  -\delta_2&\delta_2 \\
  1&- \varepsilon_1 &  \epsilon_1 &-\epsilon_1\varepsilon_1 \\
  1&  -\epsilon_2&  -\epsilon_1 & \epsilon_1\epsilon_2\\
  1&  0& 0 &-\varepsilon_2 \\
  0& 1 & -\varepsilon_3 &0
\end{pmatrix},
\]where each $\epsilon=\pm 1$, $\delta=\pm i$, $\varepsilon^4=1$. To prove the needed consequence, write $(a,b,c,d)=(f^*_{0000},f^*_{0011},f^*_{1100},f^*_{1111})$, $s=a+b=c+d\ne0$, and $t=c-d$. The second equation gives $a-b=\delta t$ with $\delta^2=-1$. Thus $(a,b,c,d)=\frac12(s+\delta t,s-\delta t,s+t,s-t)$. The last two equations give $a^4=d^4$ and $b^4=c^4$, hence $0=a^4+b^4-c^4-d^4=-\frac32s^2t^2$. Therefore $t=0$ and $a=b=c=d$, contradicting the required nonzero difference $a-b$.

Then we deal with the case in which $\left\{\begin{matrix}
 f^*_{1110}=\epsilon_1 f^*_{0001}\\
 f^*_{1101}=\epsilon_1 f^*_{0010}
\end{matrix}\right.$. Note that  \[\left\{\begin{matrix}
 {f^*}^{12}_{=_2}=(f^*_{0000}+f^*_{1100},f^*_{0001}+\epsilon_1f^*_{0010},\epsilon_1(f^*_{0001}+\epsilon_1f^*_{0010}),f^*_{0011}+f^*_{1111})\\
 {f^*}^{12}_{[1,0,-1]}=(f^*_{0000}-f^*_{1100},f^*_{0001}-\epsilon_1f^*_{0010},-\epsilon_1(f^*_{0001}-\epsilon_1f^*_{0010}),f^*_{0011}-f^*_{1111})
\end{matrix}\right. .\] If $f^*_{0001}+\epsilon_1f^*_{0010}$ and $f^*_{0001}-\epsilon_1f^*_{0010}$ are nonzero, then $f^*_{0000}=f^*_{1100}=f^*_{0011}=f^*_{1111}=0$, a contradiction since we have ${f^*}^{34}_{=_2}=(=_2)$. If $f^*_{0001}+\epsilon_1f^*_{0010}=0$ or $f^*_{0001}-\epsilon_1f^*_{0010}=0$, then $f^*_{0001}=\epsilon_2\epsilon_1f^*_{0010}$, where $\epsilon_2=\pm1$. Now,
\[\left\{\begin{matrix}
 f^*_{1110}=\epsilon_1 f^*_{0001}=\epsilon_2f^*_{0010}\\
 f^*_{1101}=\epsilon_1 f^*_{0010}=\epsilon_2f^*_{0001}
\end{matrix}\right.,\]
this reduces to the case $\left\{\begin{matrix}
 f^*_{1101}=\epsilon_1 f^*_{0001}\\
 f^*_{1110}=\epsilon_1 f^*_{0010}
\end{matrix}\right.$.

We conclude that there does not exist a valid $f^*$ such that ${f^*}^{34}_{=_2}=(=_2)$ and ${f^*}^{34}_{[1,0,-1]}\in \mathbb{C}^{\times} Z\mathrm{-family}$. By symmetry, if there exists $1\le i,j\le4$, such that
${f^*}^{ij}_{=_2}\not\equiv 0$ and
$({f^*}^{ij}_{=_2})^{-1}{f^*}^{ij}_{[1,0,-1]}\in \mathbb{C}^{\times} Z\mathrm{-family}$, then such an $f^*$ is also not valid.
\caseitem{2}  ${f^*}^{34}_{[1,0,-1]}= [1,0,-1]$ up to a nonzero scalar.

Note that ${f^*}^{34}_{=_2}=(f^*_{0000}+f^*_{0011},f^*_{0100}+f^*_{0111},f^*_{1000}+f^*_{1011},f^*_{1100}+f^*_{1111})$ and ${f^*}^{34}_{[1,0,-1]}=(f^*_{0000}-f^*_{0011},f^*_{0100}-f^*_{0111},f^*_{1000}-f^*_{1011},f^*_{1100}-f^*_{1111})$. Then
    in this case, we have $f^*_{0100}=f^*_{1000}=f^*_{0111}=f^*_{1011}=0$ , $f^*_{0000}=f^*_{1111}$, $f^*_{0011}=f^*_{1100}$ , $f^*_{0000}+f^*_{0011}\ne0$ and $f^*_{0000}-f^*_{0011}\ne0$.

    Further, ${f^*}^{12}_{=_2}=(f^*_{0000}+f^*_{1100},f^*_{0001}+f^*_{1101},f^*_{0010}+f^*_{1110},f^*_{0011}+f^*_{1111})=(f^*_{0000}+f^*_{0011},f^*_{0001}+f^*_{1101},f^*_{0010}+f^*_{1110},f^*_{0000}+f^*_{0011})$ and ${f^*}^{12}_{[1,0,-1]}=(f^*_{0000}-f^*_{0011},f^*_{0001}-f^*_{1101},f^*_{0010}-f^*_{1110},f^*_{0000}-f^*_{0011})$. Then we have
    $f^*_{0001}=f^*_{1101}=f^*_{0010}=f^*_{1110}=0$, since $f^*_{0000}+f^*_{0011}\ne0$ and $f^*_{0000}-f^*_{0011}\ne0$.
Thus $f^*$ has eight-vertex form, by Lemma~\ref{haveeightvertex}, the result follows.

    \caseitem{3}  ${f^*}^{34}_{[1,0,-1]}\in \mathcal{GDE}$ up to a nonzero scalar.

    Note that $Z\ne_2Z^{-1}=[1,0,-1]$, $Z^{-1}(0,1,-1,0)Z=\frak{i}[1,0,-1]$. By connecting $Z$($Z^{-1}$) and $(Z^{-1})^T\ (Z^T)$ to inputs $x_1$ and $x_2$ of $f^*$,  the new 4-ary signature $f^{**}$ satisfies ${f^{**}}^{34}_{=_2}=(=_2)$ and ${f^{**}}^{34}_{[1,0,-1]}=\mathbb{C}[1,0,-1]$. Then $f^{**}$ falls under \caseref{2}. Thus $f^{**}$ has eight-vertex form, by Lemma~\ref{haveeightvertex}, the result follows.
 \caseitem{4}  ${f^*}^{34}_{[1,0,-1]}=(=_2)$ up to a nonzero scalar or ${f^*}^{34}_{[1,0,-1]}\equiv 0$.

 Note that ${f^*}^{34}_{=_2}=(f^*_{0000}+f^*_{0011},f^*_{0100}+f^*_{0111},f^*_{1000}+f^*_{1011},f^*_{1100}+f^*_{1111})$ and ${f^*}^{34}_{[1,0,-1]}=(f^*_{0000}-f^*_{0011},f^*_{0100}-f^*_{0111},f^*_{1000}-f^*_{1011},f^*_{1100}-f^*_{1111})$. Then
    in this case, we have $f^*_{0100}=f^*_{1000}=f^*_{0111}=f^*_{1011}=0$ , $f^*_{0000}=f^*_{1100}$, $f^*_{0011}=f^*_{1111}$ and $f^*_{0000}+f^*_{0011}\ne0$.
    \[
M(f^*)=\begin{pmatrix}
  f^*_{0000} &  f^*_{0001}&  f^*_{0010}& f^*_{0011}\\
  0 &  f^*_{0101}&  f^*_{0110}& 0\\
  0 &  f^*_{1001}&  f^*_{1010}& 0\\
  f^*_{0000} &  f^*_{1101}&  f^*_{1110}& f^*_{0011}
\end{pmatrix},\] where $f^*_{0000}+f^*_{0011}\ne0$.
Note that we also have $f^*_{0000}f^*_{0011}\ne0$, otherwise ${f^*}^{12}_{=_2}\notin BT_{24}$.

For ${f^*}^{12}_{[1,0,-1]}=(0,f^*_{0001}-f^*_{1101},f^*_{0010}-f^*_{1110},0)$, if there exists exactly one zero in $\{(f^*_{0001}-f^*_{1101}),(f^*_{0010}-f^*_{1110})\}$, then ${f^*}^{12}_{[1,0,-1]}$ is degenerate, a contradiction. Then we consider $f^*_{0001}=f^*_{1101}$ and $f^*_{0010}=f^*_{1110}$ in \subcaseref{4.1}, $(f^*_{0001}-f^*_{1101})(f^*_{0010}-f^*_{1110})\ne0$ in \subcaseref{4.2}.
\begin{caselist}
    \subcaseitem{4.1}    $f^*_{0001}=f^*_{1101}$ and $f^*_{0010}=f^*_{1110}$, then \[M(f^*)=\begin{pmatrix}
  f^*_{0000} &  f^*_{0001}&  f^*_{0010}& f^*_{0011}\\
  0 &  f^*_{0101}&   f^*_{0110}& 0\\
  0 &  f^*_{1001}&  f^*_{1010}& 0\\
  f^*_{0000} &  f^*_{0001}&  f^*_{0010}& f^*_{0011}
\end{pmatrix}.\] If $f^*_{0001}=f^*_{0010}=0$, $f^*$ has eight-vertex form, by Lemma~\ref{haveeightvertex}, the result follows. If there exists exactly one zero in $\{f^*_{0001},f^*_{0010}\}$, then ${f^*}^{12}_{=_2}\notin BT_{24}$. If $f^*_{0001}f^*_{0010}\ne0$, by ${f^*}^{13}_{=_2}=(f^*_{0000}+f^*_{1010},f^*_{0001},f^*_{0010},f^*_{0101}+f^*_{0011})$, ${f^*}^{13}_{[1,0,-1]}=(f^*_{0000}-f^*_{1010},f^*_{0001},-f^*_{0010},f^*_{0101}-f^*_{0011})$, we have $f^*_{0010}=\pm f^*_{0001}$ or $f^*_{0010}=\pm \frak{i}f^*_{0001}$.

If $f^*_{0010}=\pm f^*_{0001}$, then we have $f^*_{0000}=f^*_{0011}=f^*_{0101}=f^*_{1010}=0$ since $\{{f^*}^{13}_{=_2},{f^*}^{13}_{[1,0,-1]}\}\subseteq\mathcal{GDE}$, a contradiction.

If $f^*_{0010}=\pm \frak{i}f^*_{0001}$, by ${f^*}^{34}_{\ne_2}=(f^*_{0001}+f^*_{0010},f^*_{0101}+f^*_{0110},f^*_{1001}+f^*_{1010},f^*_{0001}+f^*_{0010})$, ${f^*}^{34}_{(0,1,-1,0)}=(f^*_{0001}-f^*_{0010},f^*_{0101}-f^*_{0110},f^*_{1001}-f^*_{1010},f^*_{0001}-f^*_{0010})$, we have $f^*_{0101}=f^*_{0110}=f^*_{1001}=f^*_{1010}=0$ since $\{ {f^*}^{34}_{\ne_2},{f^*}^{34}_{(0,1,-1,0)}\}\subseteq\mathcal{GEQ}$, then $f^*,f\in\left \langle \mathcal{T}  \right \rangle $.
\subcaseitem{4.2}  $(f^*_{0001}-f^*_{1101})(f^*_{0010}-f^*_{1110})\ne0$. In this case, if ${f^*}_{=_2}^{12}\in Z\mathrm{-family}$, by connecting $[1,0,-1]$ to input $x_1$ of $f^*$, the new 4-ary signature $f^{**}$ satisfies ${f^{**}}_{=_2}^{12}={f^*}_{[1,0,-1]}^{12}\in\mathcal{GDE}$ and ${f^{**}}_{[1,0,-1]}^{12}={f^*}_{=_2}^{12}\in Z\mathrm{-family}$. This implies that $({f^{**}}_{=_2}^{12})^{-1}{f^{**}}_{[1,0,-1]}^{12}\in Z\mathrm{-family}$, then such an $f^{**}$ is not valid. A contradiction.

Then we have ${f^*}_{=_2}^{12}=(2f^*_{0000},f^*_{0001}+f^*_{1101},f^*_{0010}+f^*_{1110},2f^*_{0011})\notin Z\mathrm{-family}$, which implies that
$f^*_{0001}+f^*_{1101}=f^*_{0010}+f^*_{1110}=0$ and $f^*_{0011}=\epsilon_1f^*_{0000}$ since $f^*_{0000}f^*_{0011}\ne0$. Now, \[
M(f^*)=\begin{pmatrix}
  f^*_{0000} &  f^*_{0001}&  f^*_{0010}& \epsilon_1f^*_{0000}\\
  0 &  f^*_{0101}&  f^*_{0110}& 0\\
  0 &  f^*_{1001}&  f^*_{1010}& 0\\
  f^*_{0000} &  -f^*_{0001}&  -f^*_{0010}& \epsilon_1f^*_{0000}
\end{pmatrix}.\]
 As in \subcaseref{4.1}, we assume that $f^*_{0001}f^*_{0010}\ne0$. By ${f^*}^{13}_{=_2}=(f^*_{0000}+f^*_{1010},f^*_{0001},-f^*_{0010},f^*_{0101}+\epsilon_1f^*_{0000})$, ${f^*}^{13}_{[1,0,-1]}=(f^*_{0000}-f^*_{1010},f^*_{0001},f^*_{0010},f^*_{0101}-\epsilon_1f^*_{0000})$, we have $f^*_{0010}=\pm f^*_{0001}$ or $f^*_{0010}=\pm \frak{i}f^*_{0001}$.

If $f^*_{0010}=\pm f^*_{0001}$, then we have $f^*_{0000}=f^*_{0011}=f^*_{0101}=f^*_{1010}=0$ since $\{{f^*}^{13}_{=_2},{f^*}^{13}_{[1,0,-1]}\}\subseteq\mathcal{GDE}$, a contradiction.

If $f^*_{0010}=\pm\frak{i} f^*_{0001}$, by ${f^*}^{34}_{\ne_2}=(f^*_{0001}+f^*_{0010},f^*_{0101}+f^*_{0110},f^*_{1001}+f^*_{1010},-f^*_{0001}-f^*_{0010})$, ${f^*}^{34}_{(0,1,-1,0)}=(f^*_{0001}-f^*_{0010},f^*_{0101}-f^*_{0110},f^*_{1001}-f^*_{1010},-f^*_{0001}+f^*_{0010})$, we have $f^*_{0101}=f^*_{0110}=f^*_{1001}=f^*_{1010}=0$ since $\{{f^*}^{34}_{\ne_2},{f^*}^{34}_{(0,1,-1,0)}\}\subseteq\mathcal{GEQ}$. Then we have ${f^*}^{12}_{=_2}=f^*_{0000}[1,0,\epsilon_1]$ and ${f^*}^{12}_{[1,0,-1]}=2(0,f^*_{0001},f^*_{0010},0)$.
By connecting ${f^*}^{12}_{=_2}$ to input $x_3$ of $f^*$, the new 4-ary signature $f^{**}$ satisfies ${f^{**}}^{12}_{=_2}=(=_2)$ and ${f^{**}}^{12}_{[1,0,-1]}\in\mathcal{GDE}$, then by symmetry, ${f^{**}}$ falls under \caseref{3}.\qedhere

\end{caselist}
\end{caselist}
\end{proof}

We next consider the following conjugate of $BT_{24}$, which is $D_{\alpha}^{-1}BT_{24}{D_{\alpha}}$, where $D_{\alpha}=\begin{pmatrix}
 1 & 0\\
 0 &\alpha
\end{pmatrix}$ and $\alpha^2=-\frak{i}$. We list  representatives, up to nonzero scalar multiples, of the elements of $D_{\alpha}^{-1}BT_{24}{D_{\alpha}}$ and divide them into the following disjoint sets.
    \[\mathcal{GEQ}= \left\{
\begin{pmatrix}
1  & 0\\
 0 &1
\end{pmatrix},
\begin{pmatrix}
1  & 0\\
 0 &-1
\end{pmatrix}\right \}\]
\[\mathcal{GDE}'=\left\{\begin{pmatrix}
0 & 1\\
 \frak{i}&0
\end{pmatrix},
\begin{pmatrix}
0 & 1\\
 - \frak{i} &0
\end{pmatrix}\right \}\]
\[\text{Z-family}'=\left\{\begin{gathered}
\begin{pmatrix}
1 & \alpha\\
 \alpha^{-1} \frak{i} &-\frak{i}
\end{pmatrix},\begin{pmatrix}
1 & -\alpha\frak{i}\\
 -\alpha^{-1} &-\frak{i}
\end{pmatrix},\begin{pmatrix}
1 & \alpha\frak{i}\\
 \alpha^{-1} &-\frak{i}
\end{pmatrix},\begin{pmatrix}
1 & -\alpha\\
 -\alpha^{-1}\frak{i} &-\frak{i}
\end{pmatrix},\\[3pt]
\begin{pmatrix}
-\frak{i} & -\alpha\\
 -\alpha^{-1}\frak{i} &1
\end{pmatrix},\begin{pmatrix}
-\frak{i} & \alpha\frak{i}\\
 \alpha^{-1} &1
\end{pmatrix},\begin{pmatrix}
-\frak{i} & -\alpha\frak{i}\\
 -\alpha^{-1} &1
\end{pmatrix},\begin{pmatrix}
-\frak{i} & \alpha\\
 \alpha^{-1}\frak{i} &1
\end{pmatrix}
\end{gathered}\right\}\]

We use the following fact.
\begin{remark}
    For any $h\in\mathbb{C}^{\times} D_{\alpha}^{-1}BT_{24}{D_{\alpha}}$, if $\frac{h_{00}}{h_{11}}=\pm1$, then $h\in \mathbb{C}^{\times}\mathcal{GEQ}$; if $\frac{h_{01}}{h_{10}}=\pm \frak{i}$, then $h\in \mathbb{C}^{\times}\mathcal{GDE}'$; if $\frac{h_{00}}{h_{11}}=\pm\frak{i}$ or $\frac{h_{01}}{h_{10}}=\pm1$, then $h\in \mathbb{C}^{\times} Z\mathrm{-family}'$.
\end{remark}
In the proof of the following lemma, we will frequently use this fact without explanation.
\begin{lemma}\label{DaBT24Dansolve}
 Let $f$ be a 4-ary signature,  $S(f_{D_{\alpha}^{-1}BT_{24}{D_{\alpha}}})\subseteq \mathbb{C}D_{\alpha}^{-1}BT_{24}{D_{\alpha}}$. Then $\mathrm{Holant}(=_2|f)$ is \#P-hard except for the following cases:
    \begin{itemize}
        \item $f \in \left \langle Z\mathcal{M} \right \rangle$ or $\in \left \langle ZX\mathcal{M} \right \rangle$;
        \item $f$ is $ \mathcal{A}\text{-}$transformable;
        \item $f$ is $ \mathcal{L}\text{-}$transformable;
        \item $f$ is $ \mathcal{P}\text{-}$transformable;
        \item $f\in\mathcal{H}$;
        \item $f \in \left \langle \mathcal{T} \right \rangle$,
    \end{itemize}
    in which cases the problem is computable in polynomial time.
\end{lemma}
\begin{proof}
We may assume that $f\notin\langle\mathcal T\rangle$, since otherwise the stated tractable alternative holds. The invertible input modifications below preserve this exclusion, as noted in the definition of binary modification.

     Using the same procedure as in Lemma~\ref{BD8}, we construct $f^*$.
    We also have ${f^*}^{34}_{=_2}=(=_2)$. Then we have $-f^*_{0100}=f^*_{0111}$ and $-f^*_{1000}=f^*_{1011}$.
    For ${f^*}^{34}_{[1,0,-1]}$, we consider the following cases:
\begin{caselist}
\caseitem{1}  ${f^*}^{34}_{[1,0,-1]}\in \mathbb{C}^{\times} Z\mathrm{-family}'$.

Note that ${f^*}^{34}_{[1,0,-1]}=(f^*_{0000}-f^*_{0011},2f^*_{0100},2f^*_{1000},f^*_{1100}-f^*_{1111})$. Then we have $\epsilon _1f^*_{0100}=f^*_{1000}$, where $\epsilon_1=\pm 1$ and $f^*_{0100}\ne 0$, since ${f^*}^{34}_{[1,0,-1]}\in Z\mathrm{-family}'$.
In this case, we have \[
M(f^*)=\begin{pmatrix}
  f^*_{0000} &  f^*_{0001}&  f^*_{0010}& f^*_{0011}\\
  f^*_{0100} &  *&  *& -f^*_{0100}\\
  \epsilon_1f^*_{0100} &  *&  *& -\epsilon_1 f^*_{0100}\\
  f^*_{1100} &  f^*_{1101}&  f^*_{1110}& f^*_{1111}
\end{pmatrix}.\]By ${f^*}^{12}_{(0,1,\frak{i},0)}=((1+\epsilon_1\frak{i})f^*_{0100},f^*_{0101}+\frak{i}f^*_{1001},f^*_{0110}+\frak{i}f^*_{1010},-(1+\frak{i}\epsilon_1)f^*_{0100})$
and ${f^*}^{12}_{(0,1,-\frak{i},0)}=((1-\frak{i}\epsilon_1)f^*_{0100},f^*_{0101}-\frak{i}f^*_{1001},f^*_{0110}-\frak{i}f^*_{1010},-(1-\frak{i}\epsilon_1)f^*_{0100})$, we have $\{{f^*}^{12}_{(0,1,\frak{i},0))},{f^*}^{12}_{(0,1,-\frak{i},0))}\}\subseteq\mathcal{GEQ}$ since
$((1\pm\epsilon_1\frak{i})f^*_{0100}\ne0$. Then $f^*_{0101}=f^*_{1001}=f^*_{0110}=f^*_{1010}=0$.

\[
M(f^*)=\begin{pmatrix}
  f^*_{0000} &  f^*_{0001}&  f^*_{0010}& f^*_{0011}\\
  f^*_{0100} &  0&  0& -f^*_{0100}\\
  \epsilon_1f^*_{0100} &  0&  0& -\epsilon_1 f^*_{0100}\\
  f^*_{1100} &  f^*_{1101}&  f^*_{1110}& f^*_{1111}
\end{pmatrix}.\] By ${f^*}^{34}_{(0,1,\frak{i},0)}=[f^*_{0001}+\frak{i}f^*_{0010},0,f^*_{1101}+\frak{i}f^*_{1110}]$ and ${f^*}^{34}_{(0,1,-\frak{i},0)}=[f^*_{0001}-\frak{i}f^*_{0010},0,f^*_{1101}-\frak{i}f^*_{1110}]$, we have \[\left\{\begin{matrix}
 f^*_{0001}+\frak{i}f^*_{0010}=f^*_{1101}+\frak{i}f^*_{1110}\\
 f^*_{0001}-\frak{i}f^*_{0010}=f^*_{1101}-\frak{i}f^*_{1110}
\end{matrix}\right.\] or \[\left\{\begin{matrix}
 f^*_{0001}+\frak{i}f^*_{0010}=f^*_{1101}+\frak{i}f^*_{1110}\\
 f^*_{0001}-\frak{i}f^*_{0010}=\frak{i}f^*_{1110}-f^*_{1101}
\end{matrix}\right.\] or \[\left\{\begin{matrix}
 f^*_{0001}+\frak{i}f^*_{0010}=-f^*_{1101}-\frak{i}f^*_{1110}\\
 f^*_{0001}-\frak{i}f^*_{0010}=f^*_{1101}-\frak{i}f^*_{1110}
\end{matrix}\right.\] or

\[\left\{\begin{matrix}
 f^*_{0001}+\frak{i}f^*_{0010}=-f^*_{1101}-\frak{i}f^*_{1110}\\
 f^*_{0001}-\frak{i}f^*_{0010}=\frak{i}f^*_{1110}-f^*_{1101}
\end{matrix}\right.\] which implies that
$\left\{\begin{matrix}
 f^*_{0001}=f^*_{1101}\\
 f^*_{0010}=f^*_{1110}
\end{matrix}\right.$ or $\left\{\begin{matrix}
 f^*_{0001}=\frak{i}f^*_{1110}\\
 f^*_{0010}=-\frak{i}f^*_{1101}
\end{matrix}\right.$ or $\left\{\begin{matrix}
 f^*_{0001}=-\frak{i}f^*_{1110}\\
 f^*_{0010}=\frak{i}f^*_{1101}
\end{matrix}\right.$ or $\left\{\begin{matrix}
 f^*_{0001}=-f^*_{1101}\\
 f^*_{0010}=-f^*_{1110}
\end{matrix}\right.$.
We conclude that
$\left\{\begin{matrix}
 f^*_{1101}=\epsilon_2 f^*_{0001}\\
 f^*_{1110}=\epsilon_2 f^*_{0010}
\end{matrix}\right.$ or $\left\{\begin{matrix}
 f^*_{1110}=\epsilon_2\frak{i} f^*_{0001}\\
 f^*_{1101}=-\epsilon_2\frak{i} f^*_{0010}
\end{matrix}\right.$ where $\epsilon_2=\pm1$. Now, \[
M(f^*)=\begin{pmatrix}
  f^*_{0000} &  f^*_{0001}&  f^*_{0010}& f^*_{0011}\\
  f^*_{0100} &  0&  0& -f^*_{0100}\\
  \epsilon_1f^*_{0100} &  0&  0& -\epsilon_1 f^*_{0100}\\
  f^*_{1100} &  \epsilon_2 f^*_{0001}&  \epsilon_2 f^*_{0010}& f^*_{1111}
\end{pmatrix}\ \mathrm{or}\  \begin{pmatrix}
  f^*_{0000} &  f^*_{0001}&  f^*_{0010}& f^*_{0011}\\
  f^*_{0100} &  0&  0& -f^*_{0100}\\
  \epsilon_1f^*_{0100} &  0&  0& -\epsilon_1 f^*_{0100}\\
  f^*_{1100} &  -\epsilon_2 \frak{if^*_{0010}}&  \epsilon_2 \frak{if^*_{0001}}& f^*_{1111}
\end{pmatrix}\]

We first deal with the case in which $\left\{\begin{matrix}
 f^*_{1101}=\epsilon_2 f^*_{0001}\\
 f^*_{1110}=\epsilon_2 f^*_{0010}
\end{matrix}\right.$.
Note that
\[
\left\{\begin{matrix}
 {f^*}^{34}_{=_2}=(f^*_{0000}+f^*_{0011},0,0,f^*_{1100}+f^*_{1111})=(=_2)\\
 {f^*}^{34}_{[1,0,-1]}=(f^*_{0000}-f^*_{0011},2f^*_{0100},2\epsilon_1f^*_{0100},f^*_{1100}-f^*_{1111})\\
 {f^*}^{12}_{[1,0,\epsilon_2]}=(f^*_{0000}+\epsilon_2f^*_{1100},2f^*_{0001},2f^*_{0010},f^*_{0011}+\epsilon_2f^*_{1111})\\
 {f^*}^{12}_{[1,0,-\epsilon_2]}=(f^*_{0000}-\epsilon_2f^*_{1100},0,0,f^*_{0011}-\epsilon_2f^*_{1111})\\
{f^*}^{13}_{=_2}=(f^*_{0000},f^*_{0001}-\epsilon_1f^*_{0100},f^*_{0100}+\epsilon_2f^*_{0010},f^*_{1111})\\
{f^*}^{13}_{(0,1,\frak{i},0)}=(f^*_{0010}+\epsilon_1\frak{i}f^*_{0100},f^*_{0011},\frak{i}f^*_{1100},-f^*_{0100}+\epsilon_1\frak{i}f^*_{0001})
\end{matrix}\right.
\]
We have \[
\left\{\begin{matrix}
 f^*_{0000}+f^*_{0011}=f^*_{1100}+f^*_{1111}\ne0\\
 f^*_{0000}-f^*_{0011}=\delta_1(f^*_{1100}-f^*_{1111})\ne0\\
f^*_{0000}+\epsilon_2f^*_{1100}=\varepsilon_1(f^*_{0011}+\epsilon_2f^*_{1111})\\
 f^*_{0000}-\epsilon_2f^*_{1100}=\epsilon_3(f^*_{0011}-\epsilon_2f^*_{1111})\\
f^*_{0000}=\varepsilon_2f^*_{1111}\\
f^*_{0011}=\varepsilon_3f^*_{1100}
\end{matrix}\right.
\]
Then $\{f^*_{0000},f^*_{0011},f^*_{1100},f^*_{1111}\}$ satisfies a homogeneous linear system of
equations whose coefficient matrix is:
\[
\begin{pmatrix}
  1& 1 &  -1& -1\\
  1&  -1&  -\delta_1&\delta_1 \\
  1&- \varepsilon_1 &  \epsilon_2 &-\epsilon_2\varepsilon_1 \\
  1&  -\epsilon_3&  -\epsilon_2 & \epsilon_2\epsilon_3\\
  1&  0& 0 &-\varepsilon_2 \\
  0& 1 & -\varepsilon_3 &0
\end{pmatrix},
\]where each $\epsilon=\pm 1$, $\delta=\pm i$, $\varepsilon^4=1$. To prove the needed consequence, write $(a,b,c,d)=(f^*_{0000},f^*_{0011},f^*_{1100},f^*_{1111})$, $s=a+b=c+d\ne0$, and $t=c-d$. The second equation gives $a-b=\delta t$ with $\delta^2=-1$. Thus $(a,b,c,d)=\frac12(s+\delta t,s-\delta t,s+t,s-t)$. The last two equations give $a^4=d^4$ and $b^4=c^4$, hence $0=a^4+b^4-c^4-d^4=-\frac32s^2t^2$. Therefore $t=0$ and $a=b=c=d$, contradicting the required nonzero difference $a-b$.

Then we deal with the case in which $\left\{\begin{matrix}
 f^*_{1110}=\epsilon_2 \frak{i}f^*_{0001}\\
 f^*_{1101}=-\epsilon_2 \frak{i}f^*_{0010}
\end{matrix}\right.$. Note that  \[\left\{\begin{matrix}
 {f^*}^{12}_{=_2}=(f^*_{0000}+f^*_{1100},f^*_{0001}-\epsilon_2\frak{i}f^*_{0010},\epsilon_2\frak{i}(f^*_{0001}-\epsilon_2\frak{i}f^*_{0010}),f^*_{0011}+f^*_{1111})\\
 {f^*}^{12}_{[1,0,-1]}=(f^*_{0000}-f^*_{1100},f^*_{0001}+\epsilon_2\frak{i}f^*_{0010},-\epsilon_2\frak{i}(f^*_{0001}+\epsilon_2\frak{i}f^*_{0010}),f^*_{0011}-f^*_{1111})
\end{matrix}\right. .\] If $f^*_{0001}+\epsilon_2\frak{i}f^*_{0010}$ and $f^*_{0001}-\epsilon_2\frak{i}f^*_{0010}$ are nonzero, then $f^*_{0000}=f^*_{1100}=f^*_{0011}=f^*_{1111}=0$, a contradiction since we have ${f^*}^{34}_{=_2}=(=_2)$. If $f^*_{0001}+\epsilon_2\frak{i}f^*_{0010}=0$ or $f^*_{0001}-\epsilon_2\frak{i}f^*_{0010}=0$, then $f^*_{0001}=\epsilon_2\epsilon_3\frak{i}f^*_{0010}$, where $\epsilon_3=\pm1$. Now
\[\left\{\begin{matrix}
 f^*_{1110}=\epsilon_2 \frak{i}f^*_{0001}=-\epsilon_3f^*_{0010}\\
 f^*_{1101}=-\epsilon_2 \frak{i}f^*_{0010}=-\epsilon_3f^*_{0001}
\end{matrix}\right.,\]
this reduces to the case $\left\{\begin{matrix}
 f^*_{1101}=\epsilon_1 f^*_{0001}\\
 f^*_{1110}=\epsilon_1 f^*_{0010}
\end{matrix}\right.$.

We conclude that there does not exist a valid $f^*$ such that ${f^*}^{34}_{=_2}=(=_2)$ and ${f^*}^{34}_{[1,0,-1]}\in Z\mathrm{-family}'$. By symmetry, if there exists $1\le i,j\le4$, such that
${f^*}^{ij}_{=_2}\not\equiv 0$ and
$({f^*}^{ij}_{=_2})^{-1}{f^*}^{ij}_{[1,0,-1]}\in Z\mathrm{-family}'$, then such an $f^*$ is also not valid.
\caseitem{2}  ${f^*}^{34}_{[1,0,-1]}= \mathbb{C}^{\times}[1,0,-1]$.

Note that ${f^*}^{34}_{=_2}=(f^*_{0000}+f^*_{0011},f^*_{0100}+f^*_{0111},f^*_{1000}+f^*_{1011},f^*_{1100}+f^*_{1111})$ and ${f^*}^{34}_{[1,0,-1]}=(f^*_{0000}-f^*_{0011},f^*_{0100}-f^*_{0111},f^*_{1000}-f^*_{1011},f^*_{1100}-f^*_{1111})$. Then
    in this case, we have $f^*_{0100}=f^*_{1000}=f^*_{0111}=f^*_{1011}=0$ , $f^*_{0000}=f^*_{1111}$, $f^*_{0011}=f^*_{1100}$ , $f^*_{0000}+f^*_{0011}\ne0$ and $f^*_{0000}-f^*_{0011}\ne0$.

    Further, ${f^*}^{12}_{=_2}=(f^*_{0000}+f^*_{1100},f^*_{0001}+f^*_{1101},f^*_{0010}+f^*_{1110},f^*_{0011}+f^*_{1111})=(f^*_{0000}+f^*_{0011},f^*_{0001}+f^*_{1101},f^*_{0010}+f^*_{1110},f^*_{0000}+f^*_{0011})$ and ${f^*}^{12}_{[1,0,-1]}=(f^*_{0000}-f^*_{0011},f^*_{0001}-f^*_{1101},f^*_{0010}-f^*_{1110},f^*_{0000}-f^*_{0011})$. Then we have
    $f^*_{0001}=f^*_{1101}=f^*_{0010}=f^*_{1110}=0$, since $f^*_{0000}+f^*_{0011}\ne0$ and $f^*_{0000}-f^*_{0011}\ne0$.
Thus $f^*$ has eight-vertex form, by Lemma~\ref{haveeightvertex}, the result follows.

    \caseitem{3}  ${f^*}^{34}_{[1,0,-1]}\in
    \mathbb{C}^{\times}\mathcal{GDE}'$.

    Write $B_\pm=\begin{psmallmatrix}0&1\\\pm\frak{i}&0\end{psmallmatrix}$ and recall that $\alpha^2=-\frak{i}$. Set $L_+=D_\alpha^{-1}ZD_\alpha$ and $L_-=D_\alpha^{-1}Z^{-1}D_\alpha$. Direct multiplication gives
\[
L_+B_+L_+^{-1}=\alpha^{-1}\operatorname{diag}(1,-1),\qquad L_-B_-L_-^{-1}=\frak{i}\alpha^{-1}\operatorname{diag}(1,-1).
\]
Apply $L_\pm$ to input $1$ and $(L_\pm^{-1})^T$ to input $2$, choosing the sign corresponding to the contraction. The resulting signature satisfies ${f^{**}}^{34}_{=_2}=(=_2)$, and ${f^{**}}^{34}_{[1,0,-1]}$ is a nonzero scalar multiple of $[1,0,-1]$. It therefore falls under \caseref{2}, which completes this case.
 \caseitem{4}  ${f^*}^{34}_{[1,0,-1]}=\mathbb{C}(=_2)$.

 Note that ${f^*}^{34}_{=_2}=(f^*_{0000}+f^*_{0011},f^*_{0100}+f^*_{0111},f^*_{1000}+f^*_{1011},f^*_{1100}+f^*_{1111})$ and ${f^*}^{34}_{[1,0,-1]}=(f^*_{0000}-f^*_{0011},f^*_{0100}-f^*_{0111},f^*_{1000}-f^*_{1011},f^*_{1100}-f^*_{1111})$. Then
    in this case, we have $f^*_{0100}=f^*_{1000}=f^*_{0111}=f^*_{1011}=0$ , $f^*_{0000}=f^*_{1100}$, $f^*_{0011}=f^*_{1111}$ and $f^*_{0000}+f^*_{0011}\ne0$.
    \[
M(f^*)=\begin{pmatrix}
  f^*_{0000} &  f^*_{0001}&  f^*_{0010}& f^*_{0011}\\
  0 &  f^*_{0101}&  f^*_{0110}& 0\\
  0 &  f^*_{1001}&  f^*_{1010}& 0\\
  f^*_{0000} &  f^*_{1101}&  f^*_{1110}& f^*_{0011}
\end{pmatrix},\] where $f^*_{0000}+f^*_{0011}\ne0$.
Note that we also have $f^*_{0000}f^*_{0011}\ne0$, otherwise ${f^*}^{12}_{=_2}\notin BT_{24}$.

For ${f^*}^{12}_{[1,0,-1]}=(0,f^*_{0001}-f^*_{1101},f^*_{0010}-f^*_{1110},0)$, if there exists exactly one zero in $\{(f^*_{0001}-f^*_{1101}),(f^*_{0010}-f^*_{1110})\}$, then ${f^*}^{12}_{[1,0,-1]}$ is degenerate, a contradiction. Then we consider $f^*_{0001}=f^*_{1101}$ and $f^*_{0010}=f^*_{1110}$ in \subcaseref{4.1}, $(f^*_{0001}-f^*_{1101})(f^*_{0010}-f^*_{1110})\ne0$ in \subcaseref{4.2}.

\begin{caselist}
    \subcaseitem{4.1}   In this case, \[M(f^*)=\begin{pmatrix}
  f^*_{0000} &  f^*_{0001}&  f^*_{0010}& f^*_{0011}\\
  0 &  f^*_{0101}&   f^*_{0110}& 0\\
  0 &  f^*_{1001}&  f^*_{1010}& 0\\
  f^*_{0000} &  f^*_{0001}&  f^*_{0010}& f^*_{0011}
\end{pmatrix}.\] If $f^*_{0001}=f^*_{0010}=0$, $f^*$ has eight-vertex form, by Lemma~\ref{haveeightvertex}, the result follows. If there exists exactly one zero in $\{f^*_{0001},f^*_{0010}\}$, then ${f^*}^{12}_{=_2}\notin BT_{24}$. If $f^*_{0001}f^*_{0010}\ne0$, by ${f^*}^{13}_{=_2}=(f^*_{0000}+f^*_{1010},f^*_{0001},f^*_{0010},f^*_{0101}+f^*_{0011})$, ${f^*}^{13}_{[1,0,-1]}=(f^*_{0000}-f^*_{1010},f^*_{0001},-f^*_{0010},f^*_{0101}-f^*_{0011})$, we have $f^*_{0010}=\pm f^*_{0001}$ or $f^*_{0010}=\pm \frak{i}f^*_{0001}$.

If $f^*_{0010}=\pm \frak{i}f^*_{0001}$, then we have $f^*_{0000}=f^*_{0011}=f^*_{0101}=f^*_{1010}=0$ since $\{{f^*}^{13}_{=_2},{f^*}^{13}_{[1,0,-1]}\}\subseteq\mathcal{GDE}'$, a contradiction.

If $f^*_{0010}=\pm f^*_{0001}$, by ${f^*}^{34}_{(0,1,\frak{i},0)}=(f^*_{0001}+\frak{i}f^*_{0010},f^*_{0101}+\frak{i}f^*_{0110},f^*_{1001}+\frak{i}f^*_{1010},f^*_{0001}+\frak{i}f^*_{0010})$, ${f^*}^{34}_{(0,1,-\frak{i},0)}=(f^*_{0001}-\frak{i}f^*_{0010},f^*_{0101}-\frak{i}f^*_{0110},f^*_{1001}-\frak{i}f^*_{1010},f^*_{0001}-\frak{i}f^*_{0010})$, we have $f^*_{0101}=f^*_{0110}=f^*_{1001}=f^*_{1010}=0$ since $\{{f^*}^{34}_{(0,1,\frak{i},0)},{f^*}^{34}_{(0,1,-\frak{i},0)}\}\subseteq\mathcal{GEQ}$, then $f^*,f\in\left \langle \mathcal{T}  \right \rangle $.
\subcaseitem{4.2}  In this case, if ${f^*}_{=_2}^{12}\in Z\mathrm{-family}'$, by connecting $[1,0,-1]$ to input $x_1$ of $f^*$, the new 4-ary signature $f^{**}$ satisfies ${f^{**}}_{=_2}^{12}={f^*}_{[1,0,-1]}^{12}\in\mathcal{GDE}'$ and ${f^{**}}_{[1,0,-1]}^{12}={f^*}_{=_2}^{12}\in Z\mathrm{-family}'$. This implies that $({f^{**}}_{=_2}^{12})^{-1}{f^{**}}_{[1,0,-1]}^{12}\in Z\mathrm{-family}'$, then such an $f^{**}$ is not valid. A contradiction.

Then we have ${f^*}_{=_2}^{12}=(2f^*_{0000},f^*_{0001}+f^*_{1101},f^*_{0010}+f^*_{1110},2f^*_{0011})\notin Z\mathrm{-family}'$, which implies that
$f^*_{0001}+f^*_{1101}=f^*_{0010}+f^*_{1110}=0$ and $f^*_{0011}=\epsilon_1f^*_{0000}$ since $f^*_{0000}f^*_{0011}\ne0$. Now, \[
M(f^*)=\begin{pmatrix}
  f^*_{0000} &  f^*_{0001}&  f^*_{0010}& \epsilon_1f^*_{0000}\\
  0 &  f^*_{0101}&  f^*_{0110}& 0\\
  0 &  f^*_{1001}&  f^*_{1010}& 0\\
  f^*_{0000} &  -f^*_{0001}&  -f^*_{0010}& \epsilon_1f^*_{0000}
\end{pmatrix}.\]
 As in \subcaseref{4.1}, we assume that $f^*_{0001}f^*_{0010}\ne0$. By ${f^*}^{13}_{=_2}=(f^*_{0000}+f^*_{1010},f^*_{0001},-f^*_{0010},f^*_{0101}+\epsilon_1f^*_{0000})$, ${f^*}^{13}_{[1,0,-1]}=(f^*_{0000}-f^*_{1010},f^*_{0001},f^*_{0010},f^*_{0101}-\epsilon_1f^*_{0000})$, we have $f^*_{0010}=\pm f^*_{0001}$ or $f^*_{0010}=\pm \frak{i}f^*_{0001}$.

If $f^*_{0010}=\pm \frak{i}f^*_{0001}$, then we have $f^*_{0000}=f^*_{0011}=f^*_{0101}=f^*_{1010}=0$ since $\{{f^*}^{13}_{=_2},{f^*}^{13}_{[1,0,-1]}\}\subseteq\mathcal{GDE}'$, a contradiction.

If $f^*_{0010}=\pm f^*_{0001}$, by ${f^*}^{34}_{(0,1,\frak{i},0)}=(f^*_{0001}+\frak{i}f^*_{0010},f^*_{0101}+\frak{i}f^*_{0110},f^*_{1001}+\frak{i}f^*_{1010},-f^*_{0001}-\frak{i}f^*_{0010})$, ${f^*}^{34}_{(0,1,-\frak{i},0)}=(f^*_{0001}-\frak{i}f^*_{0010},f^*_{0101}-\frak{i}f^*_{0110},f^*_{1001}-\frak{i}f^*_{1010},-f^*_{0001}+\frak{i}f^*_{0010})$, we have $f^*_{0101}=f^*_{0110}=f^*_{1001}=f^*_{1010}=0$ since $\{{f^*}^{34}_{(0,1,\frak{i},0)},{f^*}^{34}_{(0,1,-\frak{i},0)}\}\subseteq\mathcal{GEQ}$. Then we have ${f^*}^{12}_{=_2}=f^*_{0000}[1,0,\epsilon_1]$ and ${f^*}^{12}_{[1,0,-1]}=2(0,f^*_{0001},f^*_{0010},0)$.
By connecting ${f^*}^{12}_{=_2}$ to input $x_3$ of $f^*$, the new 4-ary signature $f^{**}$ satisfies ${f^{**}}^{12}_{=_2}=(=_2)$ and ${f^{**}}^{12}_{[1,0,-1]}\in\mathcal{GDE}'$, then by symmetry, ${f^{**}}$ falls under \caseref{3}.\qedhere

\end{caselist}

\end{caselist}
\end{proof}

\begin{lemma}\label{BT24}
If $\mathbf{H}^{(m)}_{f}=PBT_{24}P^{-1}$, then $\mathrm{Holant}(=_2|f)$ is \#P-hard except for the following cases:
    \begin{itemize}
        \item $f \in \left \langle Z\mathcal{M} \right \rangle$ or $\in \left \langle ZX\mathcal{M} \right \rangle$;
        \item $f$ is $ \mathcal{A}\text{-}$transformable;
        \item $f$ is $ \mathcal{L}\text{-}$transformable;
        \item $f$ is $ \mathcal{P}\text{-}$transformable;
        \item $f\in\mathcal{H}$;
        \item $f \in \left \langle \mathcal{T} \right \rangle$,
    \end{itemize}
    in which cases the problem is computable in polynomial time.
\end{lemma}

\begin{proof}

Note that $\mathcal{N}_{\mathrm{SL}(2,\mathbb{C})}(BT_{24})=BO_{48}$.
In Holant($=_2|f,PBT_{24}P^{-1}$), by a holographic transformation using $P$, we have \[\text{Holant}(=_2|f,PBT_{24}P^{-1})\equiv_T\text{Holant}(P^TP|f',P^{-1}PBT_{24}P^{-1}(P^{-1})^T)=\text{Holant}(S^{-1}|f',BT_{24}S),\]
where $f'=(P^{-1})^{\otimes 4}f$. We consider the following cases:
\begin{caselist}
    \caseitem{1}  $S\in BT_{24}$. Then $=_2\in BT_{24}S=BT_{24}$. The result follows from Lemma~\ref{BT24nsolve}.

    \caseitem{2}  $S\in BO_{48}\setminus BT_{24}$. Up to a nonzero scalar, $S=\begin{pmatrix}
1 & 0\\
0 &\frak{i}
\end{pmatrix},\allowbreak\ \mathrm{or}\ \begin{pmatrix}
1 & 0\\
0 &-\frak{i}
\end{pmatrix},\allowbreak\ \mathrm{or}\ \begin{pmatrix}
1 & \frak{i}\\
\frak{i} &1
\end{pmatrix},\allowbreak\ \mathrm{or}\ 
\begin{pmatrix}
1 & -\frak{i}\\
-\frak{i} &1
\end{pmatrix}$, or $\begin{pmatrix}
1 & 1\\
1 & -1
\end{pmatrix}$, or $\begin{pmatrix}
-1 & 1\\
1 &1
\end{pmatrix}$.
Note that $\begin{pmatrix}
1 & 0\\
0 &-\frak{i}
\end{pmatrix}=\begin{pmatrix}
1 & 0\\
0 &-1
\end{pmatrix}\begin{pmatrix}
1 & 0\\
0 &\frak{i}
\end{pmatrix}$, $\begin{pmatrix}
1 & \frak{i}\\
\frak{i} &1
\end{pmatrix}=\begin{pmatrix}
1 & 1\\
\frak{i} &-\frak{i}
\end{pmatrix}\begin{pmatrix}
1 & 0\\
0 &\frak{i}
\end{pmatrix}$, $\begin{pmatrix}
1 & -\frak{i}\\
-\frak{i} &1
\end{pmatrix}=\begin{pmatrix}
1 & -1\\
-\frak{i} &-\frak{i}
\end{pmatrix}\begin{pmatrix}
1 & 0\\
0 &\frak{i}
\end{pmatrix}$, $\begin{pmatrix}
1 & 1\\
1 &-1
\end{pmatrix}=\begin{pmatrix}
1 & -\frak{i}\\
1 &\frak{i}
\end{pmatrix}\begin{pmatrix}
1 & 0\\
0 &\frak{i}
\end{pmatrix}$, $\begin{pmatrix}
-1 & 1\\
1 &1
\end{pmatrix}=\begin{pmatrix}
-1 & -\frak{i}\\
1 &-\frak{i}
\end{pmatrix}\begin{pmatrix}
1 & 0\\
0 &\frak{i}
\end{pmatrix}$. Then for any $S$,  $BT_{24}S=BT_{24}\begin{pmatrix}
1 & 0\\
0 &\frak{i}
\end{pmatrix}$. By a holographic transformation using $D_\alpha=\begin{pmatrix}
1 & 0\\
0 &\alpha
\end{pmatrix}$, where $\alpha^2=-\frak{i}$, we have \[\text{Holant}(S^{-1}|f',BT_{24}\begin{pmatrix}
1 & 0\\
0 &i
\end{pmatrix})\equiv_T\mathrm{Holant}(D_\alpha S^{-1} D_\alpha|f'',D_\alpha^{-1}BT_{24}D_\alpha)\equiv_T \mathrm{Holant}(=_2|f'',D_\alpha^{-1}BT_{24}D_\alpha)\]
where $f''=(D_\alpha^{-1})^{\otimes 4}f'$. Then the result follows from Lemma~\ref{DaBT24Dansolve}.\qedhere
\end{caselist}

\end{proof}
\subsubsection{\texorpdfstring{Conjugacy of $\mathbf{H}^{(m)}_{f}$ to $BO_{48}$}{Conjugacy of H\_f\textasciicircum (m) to BO\_48}}
First, we list representatives, up to nonzero scalar multiples, of the elements of $BO_{48}$ and divide them into the following disjoint sets.
\[\mathcal{GEQ}= \left\{\begin{pmatrix}
1  & 0\\
 0 &i
\end{pmatrix}
\begin{pmatrix}
1  & 0\\
 0 &-i
\end{pmatrix}
\begin{pmatrix}
1  & 0\\
 0 &1
\end{pmatrix}
\begin{pmatrix}
1  & 0\\
 0 &-1
\end{pmatrix}\right \}\]
\[\mathcal{GDE}=\left\{\begin{pmatrix}
0 & 1\\
 1 &0
\end{pmatrix}
\begin{pmatrix}
0 & 1\\
 -1 &0
\end{pmatrix}
\begin{pmatrix}
0 & 1\\
 i &0
\end{pmatrix}
\begin{pmatrix}
0 & 1\\
 -i &0
\end{pmatrix} \right \}\]

{\small \[\text{Z-family}=\left \{\begin{pmatrix}
1 & 1\\
 i &-i
\end{pmatrix}
\begin{pmatrix}
1 & -i\\
 -1 &-i
\end{pmatrix}
\begin{pmatrix}
1 & i\\
 1 &-i
\end{pmatrix}
\begin{pmatrix}
1 & -1\\
 -i &-i
\end{pmatrix}\right \}\]}
\[\text{Z-family}^{-1}=\left \{\begin{pmatrix}
-i & -1\\
 -i &1
\end{pmatrix}
\begin{pmatrix}
-i & i\\
 1 &1
\end{pmatrix}
\begin{pmatrix}
-i & -i\\
 -1 &1
\end{pmatrix}
\begin{pmatrix}
-i & 1\\
 i &1
\end{pmatrix}\right \}\]
\[\text{H-family}=\left \{\begin{pmatrix}
1 & 1\\
 -1 &1
\end{pmatrix}
\begin{pmatrix}
1 &-1\\
1 &1
\end{pmatrix}
\begin{pmatrix}
1 & i\\
i &1
\end{pmatrix}
\begin{pmatrix}
1 & -i\\
-i &1
\end{pmatrix}
\begin{pmatrix}
1 & -i\\
i &-1
\end{pmatrix}
\begin{pmatrix}
1 & i\\
-i &-1
\end{pmatrix}
\begin{pmatrix}
1 & 1\\
1 &-1
\end{pmatrix}
\begin{pmatrix}
-1 & 1\\
 1 &1
\end{pmatrix}\right \}\]
Note that for any $h\in BO_{48}$, $(\frac{h_{00}}{h_{11}})^4=1$ or $h_{00}=h_{11}=0$, $(\frac{h_{01}}{h_{10}})^4=1$ or $h_{01}=h_{10}=0$. This fact helps us set up homogeneous linear equations for the entries of $f$.
\begin{lemma}\label{solvebo48}
	Let $a,\allowbreak b,\allowbreak c,\allowbreak d\in\mathbb{C}$, $\mu_4=\{\pm 1, \pm \frak{i}\}$. Suppose that, for every $z\in\mu_{4}$, either
	\[
	a+bz=c+dz=0,
	\]
	or both quantities are nonzero and $(a+bz)/(c+dz)\in\mu_{4}$. Then, for some $\sigma\in\mu_{4}$, the vector $(a,\allowbreak b,\allowbreak c,\allowbreak d)$ has one of the forms
	\begin{equation}\label{eq:bo-four-forms}
		(t,u,\sigma t,\sigma u),\qquad
		(t,0,0,\sigma t),\qquad
		(0,t,\sigma t,0).
	\end{equation}
	If in addition $a+b=c+d=1$, these reduce to 
	$$
			(a,b,a,b),\qquad
	(1,0,0,1),\qquad
	(0,1,1,0).
	$$
\end{lemma}
\begin{proof}
	The polynomial
	\[
	R(z)=(a+bz)^4-(c+dz)^4
	\]
	vanishes at all four fourth roots of unity. Hence $R(z)$ is a scalar multiple of $z^4-1$. Comparing coefficients gives
	\begin{equation}\label{eq:bo-coefficients}
		a^3b=c^3d,\quad a^2b^2=c^2d^2,\quad ab^3=cd^3,
		\quad a^4+b^4=c^4+d^4.
	\end{equation}
	If $ab\neq0$, then $cd\neq0$. Dividing the first identity by the second and the second by the third gives $a/b=c/d$. Write $(c,\allowbreak d)=\sigma(a,\allowbreak b)$; the identities then give $\sigma^4=1$.
	
	If $ab=0$, then $cd=0$. The last equality in \eqref{eq:bo-coefficients} shows that any remaining nonzero entries differ by a fourth root of unity. Placing those entries in the four possible pairs of positions gives precisely \eqref{eq:bo-four-forms}. A vector with only one nonzero coordinate cannot satisfy the last equality. Finally, the normalization at $z=1$ forces $\sigma=1$ in the proportional case and $t=\sigma=1$ in each sparse case.
\end{proof}

\begin{lemma}\label{BO48}
    Let $f$ be a 4-ary signature,  $S(f_{BO_{48}})\subseteq \mathbb{C}BO_{48}$. Then $\mathrm{Holant}(=_2|f)$ is \#P-hard except for the following cases:
    \begin{itemize}
        \item $f \in \left \langle Z\mathcal{M} \right \rangle$ or $\in \left \langle ZX\mathcal{M} \right \rangle$;
        \item $f$ is $ \mathcal{A}\text{-}$transformable;
        \item $f$ is $ \mathcal{L}\text{-}$transformable;
        \item $f$ is $ \mathcal{P}\text{-}$transformable;
        \item $f\in\mathcal{H}$;
        \item $f \in \left \langle \mathcal{T} \right \rangle$,
    \end{itemize}
    in which cases the problem is computable in polynomial time.
\end{lemma}
\begin{proof}
We may assume that $f\notin\langle\mathcal T\rangle$, since otherwise the stated tractable alternative holds. The invertible input modifications below preserve this exclusion, as noted in the definition of binary modification.

    Using the same procedure as in Lemma~\ref{BD8}, we construct $f^*$.
    We also have ${f^*}^{34}_{=_2}=(=_2)$.
    Note that if there exists $t\ne1$, such that ${f^*}_{[1,0,t]}^{34}\in\mathbb{C}\mathcal{GEQ}$, then $f^*_{0100}=f^*_{0111}=f^*_{1000}=f^*_{1011}=0$. We divide the proof into two parts: suppose there exists $t\ne1$, such that ${f^*}_{[1,0,t]}^{34}\in\mathbb{C}\mathcal{GEQ}$ in \partref{I} and there does not exist such $t$ in \partref{II}.

    \proofpart{I} In this case, $f^*_{0100}=f^*_{0111}=f^*_{1000}=f^*_{1011}=0$. By Lemma~\ref{solvebo48}, \[\begin{pmatrix}
 f^*_{0000}\\
f^*_{0011} \\
 f^*_{1100}\\
f^*_{1111}
\end{pmatrix}= \begin{pmatrix}
 t_1\\
 b_1\\
 t_1\\
 b_1
\end{pmatrix}\text{ or }\begin{pmatrix}
 1\\
 0\\
 0\\
 1
\end{pmatrix}\text{ or }\begin{pmatrix}
 0\\
 1\\
 1\\
 0
\end{pmatrix}.\] We consider the following two cases: \[\begin{pmatrix}
 f^*_{0000}\\
f^*_{0011} \\
 f^*_{1100}\\
f^*_{1111}
\end{pmatrix}= \begin{pmatrix}
 t_1\\
 b_1\\
 t_1\\
 b_1
\end{pmatrix}\] in \caseref{1}, \[\begin{pmatrix}
 f^*_{0000}\\
f^*_{0011} \\
 f^*_{1100}\\
f^*_{1111}
\end{pmatrix}= \begin{pmatrix}
 1\\
 0\\
 0\\
1
\end{pmatrix}\text{ or }\begin{pmatrix}
 0\\
 1\\
 1\\
 0
\end{pmatrix}\] in \caseref{2}.
\begin{caselist}
    \caseitem{1}  In this case, by Lemma~\ref{solvebo48}, \[
\begin{pmatrix}
 f^*_{0001}\\
f^*_{1101} \\
 f^*_{0010}\\
f^*_{1110}
\end{pmatrix}
=\begin{pmatrix}
 t_2\\
 b_2\\
 \sigma_2t_2\\
 \sigma_2b_2
\end{pmatrix}\text{ or }\begin{pmatrix}
 t_2\\
 0\\
 0\\
 \sigma_2 t_2
\end{pmatrix}\text{ or }\begin{pmatrix}
 0\\
 t_2\\
 \sigma_2t_2\\
 0
\end{pmatrix}or \begin{pmatrix}
0 \\
 0\\
0 \\
0
\end{pmatrix}
.\]

If \[\begin{pmatrix}
 f^*_{0001}\\
f^*_{1101} \\
 f^*_{0010}\\
f^*_{1110}
\end{pmatrix}=\begin{pmatrix}
0 \\
 0\\
0 \\
0
\end{pmatrix},\] then $f^*$ has eight-vertex form and the result follows from Lemma~\ref{haveeightvertex}.

If \[\begin{pmatrix}
 f^*_{0001}\\
f^*_{1101} \\
 f^*_{0010}\\
f^*_{1110}
\end{pmatrix}=\begin{pmatrix}
 0\\
 t_2\\
 \sigma_2t_2\\
 0
\end{pmatrix},\] where $t_2\ne0$, then ${f^*}_{=_2}^{12}=(2t_1,t_2,\sigma_2t_2,2b_1)$. Note that if $t_1=0, b_1\ne0$ or $b_1=0, t_1\ne0$, then ${f^*}_{=_2}^{12}\notin BO_{48}$, a contradiction. Then $t_1b_1\ne0$. Thus ${f^*}_{=_2}^{12}\in BO_{48}\setminus (\mathcal{GEQ}\cup\mathcal{GDE})$. Then we have $(\frac{2t_1}{t_2})^4=1$. Further by ${f^*}_{[1,0,\frak{i}]}^{12}=((1+\frak{i})t_1,\frak{i}t_2,\sigma_2t_2,(1+\frak{i})b_1)$, we have $(\frac{(1+\frak{i})t_1}{\frak{i}t_2})^4=1$. Writing $\frac{2t_1}{t_2}=\alpha_1$, $\frac{(1+\frak{i})t_1}{\frak{i}t_2}=\alpha_2$, we have $t_2=\frac{2t_1}{\alpha_1}=\frac{(1+i)t_1}{i\alpha_2}$, i.e. $\frac{i\alpha_2}{\alpha_1}=\frac{1+i}{2}$. However $(\frac{i\alpha_2}{\alpha_1})^4=1\ne(\frac{1+i}{2})^4=-\tfrac14$, a contradiction.

If \[\begin{pmatrix}
 f^*_{0001}\\
f^*_{1101} \\
 f^*_{0010}\\
f^*_{1110}
\end{pmatrix}=\begin{pmatrix}
 t_2\\
 0\\
 0\\
 \sigma_2t_2
\end{pmatrix},\] similar to the case in which \[\begin{pmatrix}
 f^*_{0001}\\
f^*_{1101} \\
 f^*_{0010}\\
f^*_{1110}
\end{pmatrix}=\begin{pmatrix}
 0\\
 t_2\\
 \sigma_2t_2\\
 0
\end{pmatrix},\] we can obtain a contradiction.

If \[
\begin{pmatrix}
 f^*_{0001}\\
f^*_{1101} \\
 f^*_{0010}\\
f^*_{1110}
\end{pmatrix}
=\begin{pmatrix}
 t_2\\
 b_2\\
 \sigma_2t_2\\
 \sigma_2b_2
\end{pmatrix},\] by Lemma~\ref{solvebo48}, we have \[
\begin{pmatrix}
 f^*_{0101}\\
f^*_{0110} \\
f^*_{1001}\\
f^*_{1010}
\end{pmatrix}
=\begin{pmatrix}
 t_3\\
 b_3\\
 \sigma_3t_3\\
 \sigma_3b_3
\end{pmatrix}\text{ or }\begin{pmatrix}
 t_3\\
 0\\
 0\\
 \sigma_3 t_3
\end{pmatrix}\text{ or }\begin{pmatrix}
 0\\
 t_3\\
 \sigma_3t_3\\
 0
\end{pmatrix}or \begin{pmatrix}
0 \\
 0\\
0 \\
0
\end{pmatrix}
.\] Further, we consider the following cases:
\[\begin{pmatrix}
 f^*_{0101}\\
f^*_{0110} \\
f^*_{1001}\\
f^*_{1010}
\end{pmatrix}
=\begin{pmatrix}
 t_3\\
 b_3\\
 \sigma_3t_3\\
 \sigma_3b_3
\end{pmatrix}\] in \subcaseref{1.1}, \[\begin{pmatrix}
 f^*_{0101}\\
f^*_{0110} \\
f^*_{1001}\\
f^*_{1010}
\end{pmatrix}
=\begin{pmatrix}
 t_3\\
 0\\
 0\\
 \sigma_3 t_3
\end{pmatrix}\text{ or }\begin{pmatrix}
 0\\
 t_3\\
 \sigma_3t_3\\
 0
\end{pmatrix}\] in \subcaseref{1.2}, \[\begin{pmatrix}
 f^*_{0101}\\
f^*_{0110} \\
f^*_{1001}\\
f^*_{1010}
\end{pmatrix}
=\begin{pmatrix}
0 \\
 0\\
0 \\
0
\end{pmatrix}\]
\begin{caselist}
    \subcaseitem{1.1}  In this case, \[
M(f^*)=\begin{pmatrix}
  \overset{0000}{t_1} &  \overset{0001}{t_2}&  \overset{0010}{\sigma_2t_2}& \overset{0011}{b_1}\\
  \overset{0100}{0} &  \overset{0101}{t_3}&  \overset{0110}{b_3}& \overset{0111}{0}\\
  \overset{1000}{0} &  \overset{1001}{\sigma_3t_3}&  \overset{1010}{\sigma_3b_3}& \overset{1011}{0}\\
  \overset{1100}{t_1} &  \overset{1101}{b_2}&  \overset{1110}{\sigma_2b_2}& \overset{1111}{b_1}
\end{pmatrix}
\]
and $t_1t_2t_3b_1b_2b_3\ne 0$. Now we solve for $t_1,t_2,t_3,b_1,b_2,b_3$.
By ${f^*}_{\mathcal{GDE}}^{12}$, we have $b_3=\sigma_4t_3$.
By ${f^*}_{\mathcal{GDE}}^{34}$, we have $b_2=\sigma_5t_2$. By ${f^*}_{\mathcal{GEQ}}^{12}$, we have $b_1=\sigma_6t_1$. Now,
\[
M(f^*)=\begin{pmatrix}
  \overset{0000}{t_1} &  \overset{0001}{t_2}&  \overset{0010}{\sigma_2t_2}& \overset{0011}{\sigma_6t_1}\\
  \overset{0100}{0} &  \overset{0101}{t_3}&  \overset{0110}{\sigma_4t_3}& \overset{0111}{0}\\
  \overset{1000}{0} &  \overset{1001}{\sigma_3t_3}&  \overset{1010}{\sigma_3\sigma_4t_3}& \overset{1011}{0}\\
  \overset{1100}{t_1} &  \overset{1101}{\sigma_5t_2}&  \overset{1110}{\sigma_2\sigma_5t_2}& \overset{1111}{\sigma_6t_1}
\end{pmatrix}
\]
If $\sigma_2^2\ne \sigma_4^2$, then $(1+\sigma_2^3\sigma_4)$, $(1+\sigma_4^3\sigma_2)$ are nonzero and $\sigma_2^2+\sigma_4^2=0$.
By
${f^*}_{(0,1,\sigma_2^3,0)}^{34}=(2t_2,(1+\sigma_2^3\sigma_4)t_3,(1+\sigma_2^3\sigma_4)\sigma_3t_3,2\sigma_5 t_2)$, ${f^*}_{(0,1,\sigma_4^3,0)}^{34}=((1+\sigma_4^3\sigma_2)t_2,2t_3,2\sigma_3t_3,(1+\sigma_4^3\sigma_2)\sigma_5 t_2)$,
we have $\frac{2}{1+\sigma_2^3\sigma_4}\frac{t_2}{t_3}=\alpha_1$ and $\frac{1+\sigma_4^3\sigma_2}{2}\frac{t_2}{t_3}=\alpha_2$.  Then $\frac{t_2}{t_3}=\frac{\alpha_1(1+\sigma_2^3\sigma_4)}{2}=\frac{2\alpha_2}{1+\sigma_4^3\sigma_2}$, $\frac{\alpha_1}{\alpha_2}= \frac{4}{(1+\sigma_2^3\sigma_4)(1+\sigma_4^3\sigma_2)}=\frac{4}{2+\sigma_2\sigma_4(\sigma_2^2+\sigma_4^2)}=2$. However $(\frac{\alpha_1}{\alpha_2})^4=1\ne2^4$, a contradiction. Therefore, we have $\sigma_2^2 = \sigma_4^2$. Now we can choose the appropriate $\beta$ such that $\beta_1^4=1$, $\beta_1\sigma_2\ne-1$ and $\beta_1\sigma_4 \ne-1$. By ${{f^*}_{(0,1,\beta_1,0)}^{34}}=((1+\beta_1\sigma_2)t_2,(1+\beta_1\sigma_4)t_3,(1+\beta_1\sigma_4)\sigma_3t_3,(1+\beta_1\sigma_2)\sigma_5t_2)$, we have $(\frac{(1+\beta_1\sigma_2)t_2}{(1+\beta_1\sigma_4)t_3})^4=1$. Note that $(\frac{1+\beta_1\sigma_2}{1+\beta_1\sigma_4})^4=1$ since $\sigma_2^2 = \sigma_4^2$. Thus we have
$t_3=\sigma_7t_2$, where $\sigma_7^4=1$.

Similarly, if $\sigma_5^2\ne1$, then $(1+\sigma_5)$, $(1+\sigma_5^3)$ are nonzero, and $1+\sigma_5^2=0$.
By ${{f^*}_{=_2}^{12}}=(2t_1,(1+\sigma_5)t_2,(1+\sigma_5)\sigma_2t_2,2\sigma_6t_1)$,
${{f^*}_{[1,0,\sigma_5^3]}^{12}}=((1+\sigma_5^3)t_1,2t_2,2\sigma_2t_2,(1+\sigma_5^3)\sigma_6t_1)$, we have $\frac{2}{1+\sigma_5}\frac{t_1}{t_2}=\alpha_1$ and $\frac{1+\sigma_5^3}{2}\frac{t_1}{t_2}=\alpha_2$. Then we obtain a contradiction. Thus $\sigma_5^2=1$ and we can choose an appropriate $\beta_2$, such that $\beta_2\ne-1$ and $\beta_2\sigma_5\ne-1$. Then by  ${{f^*}_{[1,0,\beta_2]}^{12}}=((1+\beta_2)t_1,(1+\beta_2\sigma_5)t_2,(1+\beta_2\sigma_5)\sigma_2t_2,(1+\beta_2)\sigma_6t_1)$, we have $(\frac{(1+\beta_2)t_1}{(1+\beta_2\sigma_5)t_2})^4=1$ and $(\frac{1+\beta_2}{1+\beta_2\sigma_5})^4=1$.
Thus, we have $t_1=\sigma_8t_2$. Now,
\[
M(f^*)=t_2\begin{pmatrix}
  \overset{0000}{\sigma_8} &  \overset{0001}{1}&  \overset{0010}{\sigma_2}& \overset{0011}{\sigma_6\sigma_8}\\
  \overset{0100}{0} &  \overset{0101}{\sigma_7}&  \overset{0110}{\sigma_4\sigma_7}& \overset{0111}{0}\\
  \overset{1000}{0} &  \overset{1001}{\sigma_3\sigma_7}&  \overset{1010}{\sigma_3\sigma_4\sigma_7}& \overset{1011}{0}\\
  \overset{1100}{\sigma_8} &  \overset{1101}{\sigma_5}&  \overset{1110}{\sigma_2\sigma_5}& \overset{1111}{\sigma_6\sigma_8}
\end{pmatrix}
\]
Obviously we can always choose an appropriate $\beta_3$ such that $\sigma_7+\beta_3\sigma_6\sigma_8\ne0$ and $(\frac{1}{\sigma_7+\beta_3\sigma_6\sigma_8})^4 \ne1$.
Then
${{f^*}_{[1,0,\beta_3]}^{13}}=(\sigma_8+\beta_3\sigma_3\sigma_4\sigma_7,1,\beta_3\sigma_2\sigma_5,\sigma_7+\beta_3\sigma_6\sigma_8)\notin BO_{48}$, since $(\frac{1}{\sigma_7+\beta_3\sigma_6\sigma_8})^4 \ne1$.

\subcaseitem{1.2}  In this case, we can obtain a contradiction using steps similar to those in the case in which \[\begin{pmatrix}
 f^*_{0001}\\
f^*_{1101} \\
 f^*_{0010}\\
f^*_{1110}
\end{pmatrix}=\begin{pmatrix}
 0\\
 t_2\\
 \sigma_2t_2\\
 0
\end{pmatrix}\]
\subcaseitem{1.3}  In this case, \[
M(f^*)=\begin{pmatrix}
  \overset{0000}{t_1} &  \overset{0001}{t_2}&  \overset{0010}{\sigma_2t_2}& \overset{0011}{b_1}\\
  \overset{0100}{0} &  \overset{0101}{0}&  \overset{0110}{0}& \overset{0111}{0}\\
  \overset{1000}{0} &  \overset{1001}{0}&  \overset{1010}{0}& \overset{1011}{0}\\
  \overset{1100}{t_1} &  \overset{1101}{b_2}&  \overset{1110}{\sigma_2b_2}& \overset{1111}{b_1}
\end{pmatrix}.
\]
By ${f^*}_{(0,1,\sigma_2^3,0)}^{34}=2[t_2,0,b_2]$, we have $b_2=\sigma_3t_2$. By
${f^*}_{=_2}^{12}=(2t_1,t_2+b_2,\sigma_2(t_2+b_2),2b_1)$, we have $b_1=\sigma_4t_1$.
If $\sigma_3^2\ne1$, we have $1+\sigma_3$ and $1+\sigma_3^3$ are nonzero and $1+\sigma_3^2=0$.
Now, by ${f^*}_{=_2}^{12}=(2t_1,(1+\sigma_3)t_2,(1+\sigma_3)\sigma_2t_2,2\sigma_4t_1)$,
 ${f^*}_{[1,0,\sigma_3^3]}^{12}=((1+\sigma_3^3)t_1,2t_2,2\sigma_2t_2,(1+\sigma_3^3)\sigma_4t_1)$, we have $(\frac{2t_1}{(1+\sigma_3)t_2})^4=(\frac{(1+\sigma_3^3)t_1}{2t_2})^4=1$, i.e., $2^4\times2^4=(1+\sigma_3^3)^4(1+\sigma_3)^4=[2+\sigma_3(1+\sigma_3^2)]^4=2^4$, then we obtain a contradiction. Then $\sigma_3^2=1$ and we have $t_1=\sigma_5t_2$. Now \[
M(f^*)=t_2\begin{pmatrix}
  \overset{0000}{\sigma_5} &  \overset{0001}{1}&  \overset{0010}{\sigma_2}& \overset{0011}{\sigma_4\sigma_5}\\
  \overset{0100}{0} &  \overset{0101}{0}&  \overset{0110}{0}& \overset{0111}{0}\\
  \overset{1000}{0} &  \overset{1001}{0}&  \overset{1010}{0}& \overset{1011}{0}\\
  \overset{1100}{\sigma_5} &  \overset{1101}{\sigma_3}&  \overset{1110}{\sigma_2\sigma_3}& \overset{1111}{\sigma_4\sigma_5}
\end{pmatrix}.
\] If $\sigma_3=1$, then $f^*, f \in \left \langle \mathcal{T} \right \rangle$. If $\sigma_3=-1$, \[
M(f^*)=t_2\begin{pmatrix}
  \overset{0000}{\sigma_5} &  \overset{0001}{1}&  \overset{0010}{\sigma_2}& \overset{0011}{\sigma_4\sigma_5}\\
  \overset{0100}{0} &  \overset{0101}{0}&  \overset{0110}{0}& \overset{0111}{0}\\
  \overset{1000}{0} &  \overset{1001}{0}&  \overset{1010}{0}& \overset{1011}{0}\\
  \overset{1100}{\sigma_5} &  \overset{1101}{-1}&  \overset{1110}{-\sigma_2}& \overset{1111}{\sigma_4\sigma_5}
\end{pmatrix}.
\]
By connecting input $x_2$ of $[1,0,\sigma_2^3]$ to input $x_3$ of $f^*$, we construct \[M(f^{**})=\begin{pmatrix}
  \overset{0000}{\sigma_5} &  \overset{0001}{1}&  \overset{0010}{1}& \overset{0011}{\sigma_6}\\
  \overset{0100}{0} &  \overset{0101}{0}&  \overset{0110}{0}& \overset{0111}{0}\\
  \overset{1000}{0} &  \overset{1001}{0}&  \overset{1010}{0}& \overset{1011}{0}\\
  \overset{1100}{\sigma_5} &  \overset{1101}{-1}&  \overset{1110}{-1}& \overset{1111}{\sigma_6}
\end{pmatrix},
\]where $\sigma_6=\sigma_2^3\sigma_4\sigma_5$.
Note that ${f^{**}}_{=_2}^{13}=(\sigma_5,1,-1,\sigma_6)$. For a matrix $h\in BO_{48}$, if $\frac{h_{01}}{h_{10}}=-1$, then $h$ is $\begin{pmatrix}
1 & 1\\
 -1 &1
\end{pmatrix}\ \text{or} \
\begin{pmatrix}
1 &-1\\
1 &1
\end{pmatrix}$ or $\begin{pmatrix}
1 & -i\\
i &-1
\end{pmatrix}\ \text{or} \
\begin{pmatrix}
1 & i\\
-i &-1
\end{pmatrix}$. Thus we have $\frac{\sigma_5}{\sigma_6}=\pm1$, i.e. $\sigma_5^2=\sigma_6^2$. Then \[\begin{gathered}f^{**}_{x_1,x_2,x_3,x_4}f^{**}_{x_3,x_4,x_1,x_2}\\ =\begin{pmatrix}
  \overset{0000}{\sigma_5} &  \overset{0001}{1}&  \overset{0010}{1}& \overset{0011}{\sigma_6}\\
  \overset{0100}{0} &  \overset{0101}{0}&  \overset{0110}{0}& \overset{0111}{0}\\
  \overset{1000}{0} &  \overset{1001}{0}&  \overset{1010}{0}& \overset{1011}{0}\\
  \overset{1100}{\sigma_5} &  \overset{1101}{-1}&  \overset{1110}{-1}& \overset{1111}{\sigma_6}
\end{pmatrix}\begin{pmatrix}
  \overset{0000}{\sigma_5} &  \overset{0001}{0}&  \overset{0010}{0}& \overset{0011}{\sigma_5}\\
  \overset{0100}{1} &  \overset{0101}{0}&  \overset{0110}{0}& \overset{0111}{-1}\\
  \overset{1000}{1} &  \overset{1001}{0}&  \overset{1010}{0}& \overset{1011}{-1}\\
  \overset{1100}{\sigma_6} &  \overset{1101}{0}&  \overset{1110}{0}& \overset{1111}{\sigma_6}
\end{pmatrix}\\ =\begin{pmatrix}
  \overset{0000}{\sigma_5^2+\sigma_6^2+2} &  \overset{0001}{0}&  \overset{0010}{0}& \overset{0011}{\sigma_5^2+\sigma_6^2-2}\\
  \overset{0100}{0} &  \overset{0101}{0}&  \overset{0110}{0}& \overset{0111}{0}\\
  \overset{1000}{0} &  \overset{1001}{0}&  \overset{1010}{0}& \overset{1011}{0}\\
  \overset{1100}{\sigma_5^2+\sigma_6^2-2} &  \overset{1101}{0}&  \overset{1110}{0}& \overset{1111}{\sigma_5^2+\sigma_6^2+2}
\end{pmatrix}.\end{gathered}\]
Then $f^{**}_{x_1,x_2,x_3,x_4}f^{**}_{x_3,x_4,x_1,x_2}$ is either $(=_4)$ or $(\ne_4)$. If it is $(\ne_4)$, note that we have $(\ne_2)$, thus we can also construct $(=_4)$. Then by Lemma~\ref{toCSP2} the result follows.
\end{caselist}
\caseitem{2}  If \[\begin{pmatrix}
 f^*_{0000}\\
f^*_{0011} \\
 f^*_{1100}\\
f^*_{1111}
\end{pmatrix}=\begin{pmatrix}
 0\\
 1\\
 1\\
 0
\end{pmatrix},\] by connecting two copies of $\ne_2$ to inputs $x_3,x_4$ of $f^*$, we construct a new 4-ary signature $f^{**}$ and then \[\begin{pmatrix}
 f^{**}_{0000}\\
f^{**}_{0011} \\
 f^{**}_{1100}\\
f^{**}_{1111}
\end{pmatrix}=\begin{pmatrix}
 1\\
 0\\
 0\\
 1
\end{pmatrix}.\] Note that the cases of $f^*$ and $f^{**}$ are equivalent, thus we can assume that \[\begin{pmatrix}
 f^{*}_{0000}\\
f^{*}_{0011} \\
 f^{*}_{1100}\\
f^{*}_{1111}
\end{pmatrix}=\begin{pmatrix}
 1\\
 0\\
 0\\
 1
\end{pmatrix}.\]

Similarly, by Lemma~\ref{solvebo48}, \[
\begin{pmatrix}
 f^*_{0001}\\
f^*_{1101} \\
 f^*_{0010}\\
f^*_{1110}
\end{pmatrix}
=\begin{pmatrix}
 t_2\\
 b_2\\
 \sigma_2t_2\\
 \sigma_2b_2
\end{pmatrix}\text{ or }\begin{pmatrix}
 t_2\\
 0\\
 0\\
 \sigma_2 t_2
\end{pmatrix}\text{ or }\begin{pmatrix}
 0\\
 t_2\\
 \sigma_2t_2\\
 0
\end{pmatrix}or \begin{pmatrix}
0 \\
 0\\
0 \\
0
\end{pmatrix}
.\]

If \[\begin{pmatrix}
 f^*_{0001}\\
f^*_{1101} \\
 f^*_{0010}\\
f^*_{1110}
\end{pmatrix}=\begin{pmatrix}
0 \\
 0\\
0 \\
0
\end{pmatrix},\] then $f^*$ has eight-vertex form and the result follows from Lemma~\ref{haveeightvertex}.

If \[\begin{pmatrix}
 f^*_{0001}\\
f^*_{1101} \\
 f^*_{0010}\\
f^*_{1110}
\end{pmatrix}=\begin{pmatrix}
 t_2\\
 b_2\\
 \sigma_2t_2\\
  \sigma_2b_2
\end{pmatrix},\] where $t_2b_2\ne0$. Choose an appropriate $\beta_1$, such that $\beta_1^4=1$ and $1+\beta_1\sigma_2\ne0$. By ${f^*}_{(0,1,\beta_1,0)}^{34}=((1+\beta_1\sigma_2)t_2,f^*_{0101}+\beta_1 f^*_{0110},f^*_{1001}+\beta_1 f^*_{1010},(1+\beta_1\sigma_2)b_2)$, we have $b_2 =\sigma_3 t_2$.

Choose an appropriate $\beta_2$, such that $\beta_2^4=1$ and $1+\beta_2\sigma_3=(1+\frak{i})$. Then ${f^*}_{[1,0,\beta_2]}^{12}=(1,(1+\beta_2\sigma_3)t_2,(1+\beta_2\sigma_3)\sigma_2t_2,\beta_2)$, we have $(\frac{1}{(1+\beta_2\sigma_3)t_2})^4=(\frac{1}{(1+\frak{i})t_2})^4=1$.
By ${f^*}_{[1,0,\sigma_3^3]}^{12}=(1,2t_2,2\sigma_2t_2,\sigma_3^3 )$, we have $(\frac{1}{2t_2})^4=1$. Thus $(\frac{1}{(1+\frak{i})t_2})^4=1$ and $(\frac{1}{2t_2})^4=1$ implies that $(\frac{2}{1+\frak{i}})^4=1$, a contradiction.

If \[
\begin{pmatrix}
 f^*_{0001}\\
f^*_{1101} \\
 f^*_{0010}\\
f^*_{1110}
\end{pmatrix}
=\begin{pmatrix}
 t_2\\
 0\\
 0\\
 \sigma_2 t_2
\end{pmatrix}\text{ or }\begin{pmatrix}
 0\\
 t_2\\
 \sigma_2t_2\\
 0
\end{pmatrix}\]
Then ${f^*}_{=_2}^{12}=(1,t_2,\sigma_2t_2, 1)$ or ${f^*}_{=_2}^{12}=(1,\sigma_2t_2,t_2, 1)$, and we always have $t_2=\sigma_4$. Further, we consider \[\begin{pmatrix}
 f^*_{0101}\\
f^*_{0110} \\
 f^*_{1001}\\
f^*_{1010}
\end{pmatrix}\] in the following three cases:
\[\begin{pmatrix}
 f^*_{0101}\\
f^*_{0110} \\
 f^*_{1001}\\
f^*_{1010}
\end{pmatrix}=\begin{pmatrix}
 t_3\\
 b_3\\
 \sigma_3t_3\\
  \sigma_3b_3
\end{pmatrix},\] where $t_3b_3\ne0$ in \subcaseref{2.1}; \[\begin{pmatrix}
 f^*_{0101}\\
f^*_{0110} \\
 f^*_{1001}\\
f^*_{1010}
\end{pmatrix}=\begin{pmatrix}
 0\\
0\\
0\\
 0
\end{pmatrix}\] in \subcaseref{2.2}; \[\begin{pmatrix}
 f^*_{0101}\\
f^*_{0110} \\
 f^*_{1001}\\
f^*_{1010}
\end{pmatrix}=\begin{pmatrix}
 t_3\\
 0\\
 0\\
 \sigma_3 t_3
\end{pmatrix}\text{ or }\begin{pmatrix}
 0\\
 t_3\\
 \sigma_3t_3\\
 0
\end{pmatrix}\] in \subcaseref{2.3}.
\begin{caselist}
    \subcaseitem{2.1}  In this case, similar to \[\begin{pmatrix}
 f^*_{0001}\\
f^*_{1101} \\
 f^*_{0010}\\
f^*_{1110}
\end{pmatrix}=\begin{pmatrix}
 t_2\\
 b_2\\
 \sigma_2t_2\\
  \sigma_2b_2
\end{pmatrix},\] we can obtain a contradiction.
\subcaseitem{2.2}  In this case, $f^*$ belongs to $\mathcal{A}$ and hence $f$ belongs to $\mathcal{A}$.
\subcaseitem{2.3}  In this case, we have $t_3=\sigma_5$ by ${f^*}_{\ne_2}^{34}$. Then there are four cases of $f^*$.
If \[M(f^*)=\begin{pmatrix}
  \overset{0000}{1} &  \overset{0001}{\sigma_4}&  \overset{0010}{0}& \overset{0011}{0}\\
  \overset{0100}{0} &  \overset{0101}{\sigma_5}&  \overset{0110}{0}& \overset{0111}{0}\\
  \overset{1000}{0} &  \overset{1001}{0}&  \overset{1010}{\sigma_3 \sigma_5}& \overset{1011}{0}\\
  \overset{1100}{0} &  \overset{1101}{0}&  \overset{1110}{\sigma_2\sigma_4}& \overset{1111}{1}
\end{pmatrix},\] ${f^*}_{[1,0,\beta]}^{13}=(1+\beta\sigma_3\sigma_5,\sigma_4,\beta\sigma_2\sigma_4,\sigma_5+\beta)$, by choosing an appropriate $\beta$ such that $\beta^4=1$ and $1+\beta\sigma_3\sigma_5=1+\frak{i}$, ${f^*}_{[1,0,\beta]}^{13}\notin BO_{48}$. If \[ M(f^*)=\begin{pmatrix}
  \overset{0000}{1} &  \overset{0001}{\sigma_4}&  \overset{0010}{0}& \overset{0011}{0}\\
  \overset{0100}{0} &  \overset{0101}{0}&  \overset{0110}{\sigma_5}& \overset{0111}{0}\\
  \overset{1000}{0} &  \overset{1001}{\sigma_3 \sigma_5}&  \overset{1010}{0}& \overset{1011}{0}\\
  \overset{1100}{0} &  \overset{1101}{0}&  \overset{1110}{\sigma_2\sigma_4}& \overset{1111}{1}
\end{pmatrix},\] similarly, let $1+\beta\sigma_5=1+\frak{i}$, then ${f^*}_{[1,0,\beta]}^{23}= (1+\beta\sigma_5,\sigma_4,\beta\sigma_2\sigma_4,\sigma_3\sigma_5+\beta)\notin BO_{48}$. If \[M(f^*)=\begin{pmatrix}
  \overset{0000}{1} &  \overset{0001}{0}&  \overset{0010}{\sigma_4}& \overset{0011}{0}\\
  \overset{0100}{0} &  \overset{0101}{\sigma_5}&  \overset{0110}{0}& \overset{0111}{0}\\
  \overset{1000}{0} &  \overset{1001}{0}&  \overset{1010}{\sigma_3 \sigma_5}& \overset{1011}{0}\\
  \overset{1100}{0} &  \overset{1101}{\sigma_2\sigma_4}&  \overset{1110}{0}& \overset{1111}{1}
\end{pmatrix},\] let $1+\beta\sigma_5=1+\frak{i}$, then ${f^*}_{[1,0,\beta]}^{24}=(1+\beta\sigma_5,\sigma_4,\beta\sigma_2\sigma_4,\sigma_3\sigma_5+\beta)\notin BO_{48}$. If \[M(f^*)=\begin{pmatrix}
  \overset{0000}{1} &  \overset{0001}{0}&  \overset{0010}{\sigma_4}& \overset{0011}{0}\\
  \overset{0100}{0} &  \overset{0101}{0}&  \overset{0110}{\sigma_5}& \overset{0111}{0}\\
  \overset{1000}{0} &  \overset{1001}{\sigma_3 \sigma_5}&  \overset{1010}{0}& \overset{1011}{0}\\
  \overset{1100}{0} &  \overset{1101}{\sigma_2\sigma_4}&  \overset{1110}{0}& \overset{1111}{1}
\end{pmatrix},\] let $1+\beta\sigma_3\sigma_5=1+\frak{i}$, then ${f^*}_{[1,0,\beta]}^{14}=(1+\beta\sigma_3\sigma_5,\sigma_4,\beta\sigma_2\sigma_4,\sigma_5+\beta)\notin BO_{48}$.
\end{caselist}
\end{caselist}
 \proofpart{II} Note that \[
\begin{pmatrix}
0 & 1\\
 1 &0
\end{pmatrix}=H\begin{pmatrix}
1 & 0\\
 0 &-1
\end{pmatrix} H^{-1},\ 
\begin{pmatrix}
0 & 1\\
 -1 &0
\end{pmatrix}=\frak{i}Z\begin{pmatrix}
1 & 0\\
 0 &-1
\end{pmatrix} Z^{-1},\ 
\begin{pmatrix}
1 & 1\\
 -1 &1
\end{pmatrix}=(1+\frak{i})Z\begin{pmatrix}
1 & 0\\
 0 &-\frak{i}
\end{pmatrix} Z^{-1},\ 
\]\[
\begin{pmatrix}
1 & -1\\
 1 &1
\end{pmatrix}=(1-\frak{i})Z\begin{pmatrix}
1 & 0\\
 0 &\frak{i}
\end{pmatrix} Z^{-1},\ 
\begin{pmatrix}
1 & \frak{i}\\
 \frak{i} &1
\end{pmatrix}=(1+\frak{i})H\begin{pmatrix}
1 & 0\\
 0 &-\frak{i}
\end{pmatrix} H^{-1},\ 
\begin{pmatrix}
1 & -\frak{i}\\
 -\frak{i} &1
\end{pmatrix}=(1-\frak{i})H\begin{pmatrix}
1 & 0\\
 0 &\frak{i}
\end{pmatrix} H^{-1}\ 
\]
Let $\mathcal{S}$ be a set made up of the six matrices above.
If there exists $t\ne1$ such that ${f^*}_{[1,0,t]}^{34}\in \mathcal{S}$, by connecting $H^{-1}$(or $Z^{-1}$) to input $x_1$ of $f^*$, and $H^{T}$(or $Z^{T}$) to input $x_2$ of $f^*$, the new 4-ary signature $f^{**}$ satisfies ${f^{**}}_{=_2}^{34}=(=_2)$ and  ${f^{**}}_{[1,0,t]}^{34}\in \mathcal{GEQ}$. Then $f^{**}$ is in \partref{I}. Further, if there exists $t_1,t_2$, such that ${f^{*}}_{[1,0,t_1]}^{34}={f^{*}}_{[1,0,t_2]}^{34}$, up to a nonzero scalar, then by connecting input $x_1$ of $[1,0,t_1]$ to input $x_3$ of $f^*$, we construct a new 4-ary signature $f^{**}$ that satisfies ${f^{**}}_{=_2}^{34}={f^{*}}_{[1,0,t_1]}^{34}={f^{*}}_{[1,0,t_2]}^{34}={f^{**}}_{[1,0,\frac{t_2}{t_1}]}^{34}$, up to some nonzero scalar. Then by connecting input $x_2$ of $({f^{**}}_{=_2}^{34})^{-1}$ to input $x_1$ of $f^{**}$, we construct a new 4-ary signature $f^{***}$ that satisfies ${f^{***}}_{=_2}^{34}={f^{***}}_{[1,0,\frac{t_2}{t_1}]}^{34}=(=_2)$, up to a nonzero scalar, then $f^{***}$ is in \partref{I}.

We now examine the homogeneous linear system that \[\begin{pmatrix}
 f^*_{0100}\\
f^*_{0111} \\
 f^*_{1000}\\
f^*_{1011}
\end{pmatrix}\] satisfies, whose coefficient matrix is
\[
\begin{pmatrix}
  1&  1& 0 & 0\\
  0&  0& -1 & -1\\
 1 &  -1&  -a_2& a_2\\
 1 &  i& -a_3 & -ia_3\\
  1&  -i& -a_4 &ia_4
\end{pmatrix}.\]
The homogeneous linear system has a nonzero solution if and only if $a_2=a_3=a_4$. This means that the ratio of ${f^{*}}_{[1,0,-1]}^{34}$, ${f^{*}}_{[1,0,\frak{i}]}^{34}$ and ${f^{*}}_{[1,0,-\frak{i}]}^{34}$ at positions 01 and 10 is the same.
We divide $BO_{48}\setminus\{\mathcal{GEQ}\cup\mathcal{S}\}$ into four categories based on the ratio of the entries in positions 01 and 10.
If $\frac{{{f^{*}}_{[1,0,-1]}^{34}}_{01}}{{{f^{*}}_{[1,0,-1]}^{34}}_{10}}=\pm1$, then to ensure the homogeneous linear system has a nonzero solution, we have either ${f^{*}}_{[1,0,\pm\frak{i}]}^{34} \in \mathcal{S}$ or ${f^{*}}_{[1,0,-1]}^{34}={f^{*}}_{[1,0,\pm\frak{i}]}^{34}$, up to a nonzero scalar.

If $\frac{{{f^{*}}_{[1,0,-1]}^{34}}_{01}}{{{f^{*}}_{[1,0,-1]}^{34}}_{10}}=\pm\frak{i}$, then we consider the homogeneous linear system that \[\begin{pmatrix}
 f^*_{0000}\\
f^*_{0011} \\
 f^*_{1100}\\
f^*_{1111}
\end{pmatrix}\] satisfies, whose coefficient matrix is
\[
\begin{pmatrix}
  1&  1& -1 & -1\\
 1 &  -1&  -b_2& b_2\\
 1 &  i& -b_3 & -ib_3\\
  1&  -i& -b_4 &ib_4
\end{pmatrix}.\] To ensure the homogeneous linear system has a nonzero solution, we have either ${f^{*}}_{[1,0,\pm\frak{i}]}^{34} \in \mathcal{S}$ or ${f^{*}}_{[1,0,-1]}^{34}={f^{*}}_{[1,0,\pm\frak{i}]}^{34}$, up to a nonzero scalar or \[\begin{pmatrix}
 f^*_{0000}\\
f^*_{0011} \\
 f^*_{1100}\\
f^*_{1111}
\end{pmatrix}=\begin{pmatrix}
0\\
0 \\
0\\
0
\end{pmatrix}.\] A contradiction since ${f^{*}}_{=_2}^{34}=(=_2)$.
\end{proof}


\begin{lemma}\label{dicbo48}
If $\mathbf{H}^{(m)}_{f}=PBO_{48}P^{-1}$, then $\mathrm{Holant}(=_2|f)$ is \#P-hard except for the following cases:
    \begin{itemize}
        \item $f \in \left \langle Z\mathcal{M} \right \rangle$ or $\in \left \langle ZX\mathcal{M} \right \rangle$;
        \item $f$ is $ \mathcal{A}\text{-}$transformable;
        \item $f$ is $ \mathcal{L}\text{-}$transformable;
        \item $f$ is $ \mathcal{P}\text{-}$transformable;
        \item $f\in\mathcal{H}$;
        \item $f \in \left \langle \mathcal{T} \right \rangle$,
    \end{itemize}
    in which cases the problem is computable in polynomial time.
\end{lemma}

\begin{proof}
     By Lemma~\ref{normalizer}, we have $\mathcal{N}_{\mathrm{SL}(2,\mathbb{C})}(BO_{48})=BO_{48}$.
In Holant($=_2|f,PBO_{48}P^{-1}$), by a holographic transformation using $P$, we have \[\text{Holant}(=_2|f,PBO_{48}P^{-1})\equiv_T\text{Holant}(S^{-1}|f',BO_{48}S),\]
where $f'=(P^{-1})^{\otimes 4}f$. Note that $BO_{48}S=BO_{48}$, then we have \[\text{Holant}(=_2|f,PBO_{48}P^{-1})\equiv_T\text{Holant}(S^{-1}|f',BO_{48}S)\equiv_T\text{Holant}(=_2|f',BO_{48}).\] By Lemma~\ref{BO48}, the result follows.
\end{proof}

\subsubsection{\texorpdfstring{Conjugacy of $\mathbf{H}^{(m)}_{f}$ to $BI_{120}$}{Conjugacy of H\_f\textasciicircum (m) to BI\_120}}



Let $\varepsilon=e^{2\pi\frak{i}/5}$, $\mu_5=\{\varepsilon^j|j\in[5]\}$, $p=\frac{\varepsilon^4-\varepsilon}{\sqrt{5}}$, $q=\frac{ \varepsilon^2-\varepsilon^3}{\sqrt{5}}$. Using the generators
\[
A=\begin{pmatrix}
  \varepsilon^2&  0\\
 0 &  \varepsilon^3
\end{pmatrix}\ ,
B=\begin{pmatrix}
  0&  1\\
 -1 &  0
\end{pmatrix},
C=\begin{pmatrix}
  p&  q\\
 q &  -p
\end{pmatrix}\ ,
\]
we can calculate each element of $BI_{120}$, and divide $BI_{120}$ into four disjoint sets:

$\mathcal{GEQ}$: All diagonal matrices in $BI_{120}$;

$\mathcal{GDE}$: All anti-diagonal matrices in $BI_{120}$;

$\mathcal{Z}^p$:All full-support matrices $h$ in $BI_{120}$, such that $\frac{h_{01}h_{10}}{h_{00}h_{11}}=-(\frac{p}{q})^2$;

$\mathcal{Z}^q$:All full-support matrices $h$ in $BI_{120}$, such that $\frac{h_{01}h_{10}}{h_{00}h_{11}}=-(\frac{q}{p})^2$.

Let $\mathcal{Z}=\mathcal{Z}^p\cup\mathcal{Z}^q$.
 This gives the following necessary ratio conditions for $h\in \mathbb{C}BI_{120}$:

\begin{center}
\begin{tabular}{ccc}
\toprule

Matrix type & $h_{00}/h_{11}$ & $h_{01}/h_{10}$\\\midrule
$\mathcal{GEQ}$ & $\mu_{5}$ & undefined\\
$\mathcal{GDE}$ & undefined & $-\mu_{5}$\\
$\mathcal{Z}$ & $-\mu_{5}$ & $\mu_{5}$\\
\bottomrule
\end{tabular}
\end{center}

The two ratios in the last row are independent; these necessary conditions alone do not characterize group membership.

\begin{lemma}\label{4eq}
Let $x,\allowbreak y,\allowbreak w,\allowbreak z\in\mathbb{C}$. Suppose that, for every $\sigma\in\mu_{5}$, either $x+y\sigma=w+z\sigma=0$, or both expressions are nonzero and their ratio lies in $\mu_{5}$. Then, for some $\varrho\in\mu_{5}$, the vector $(x,\allowbreak y,\allowbreak w,\allowbreak z)$ has one of the forms
\begin{equation}\label{eq:bi-five-forms}
 (t,b,\varrho t,\varrho b),\qquad
 (t,0,0,\varrho t),\qquad
 (0,t,\varrho t,0).
\end{equation}
If in addition $x+y=w+z=1$, these reduce to
\[
 (x,y,x,y)\qquad (x+y=1), \qquad (1,0,0,1),\qquad(0,1,1,0).
\]
\end{lemma}

\begin{proof}
The polynomial $(x+y\sigma)^5-(w+z\sigma)^5$ vanishes on $\mu_{5}$ and is therefore a scalar multiple of $\sigma^5-1$. Comparing coefficients gives
\[
 x^4y=w^4z,\quad x^3y^2=w^3z^2,\quad
 x^2y^3=w^2z^3,\quad xy^4=wz^4,
 \qquad x^5+y^5=w^5+z^5.
\]
If $xy\neq0$, then $wz\neq0$, and division of consecutive identities yields $x/y=w/z$. Thus $(w,\allowbreak z)=\varrho(x,\allowbreak y)$ with $\varrho^5=1$. If $xy=0$, then $wz=0$, and the final equality yields the sparse forms in \eqref{eq:bi-five-forms}. The zero vector is already included. Finally, the normalization at $\sigma=1$ forces $\varrho=1$ in the proportional case and $t=\varrho=1$ in each sparse case.
\end{proof}

\begin{lemma}\label{-4eq}
    Let $x,\allowbreak y,\allowbreak w,\allowbreak z\in\mathbb{C}$. Suppose that, for every $\sigma\in-\mu_{5}$, either $x+y\sigma=w+z\sigma=0$, or both expressions are nonzero and their ratio lies in $\mu_{5}$. Then, for some $\varrho\in\mu_{5}$, the vector $(x,\allowbreak y,\allowbreak w,\allowbreak z)$ has one of the forms
\begin{equation}\label{eq:bi-five-forms-2}
 (t,-b,\varrho t,-\varrho b),\qquad
 (t,0,0,-\varrho t),\qquad
 (0,t,-\varrho t,0).
\end{equation}
\end{lemma}
\begin{proof}
    By Lemma $\ref{4eq}$, the vector ($x,-y,w,-z$) has one of the forms
\[ (t,b,\varrho t,\varrho b),\qquad
 (t,0,0,\varrho t),\qquad
 (0,t,\varrho t,0).\]
 Thus the vector ($x,y,w,z$) has one of the forms
\[ (t,-b,\varrho t,-\varrho b),\qquad
 (t,0,0,-\varrho t),\qquad
 (0,t,-\varrho t,0).\]
\end{proof}

\begin{lemma}\label{BI120}
    Let $f$ be a 4-ary signature,  $S(f_{BI_{120}})\subseteq \mathbb{C}BI_{120}$. Then $\mathrm{Holant}(=_2|f)$ is \#P-hard except for the following cases:
    \begin{itemize}
        \item $f \in \left \langle Z\mathcal{M} \right \rangle$ or $\in \left \langle ZX\mathcal{M} \right \rangle$;
        \item $f$ is $ \mathcal{A}\text{-}$transformable;
        \item $f$ is $ \mathcal{L}\text{-}$transformable;
        \item $f$ is $ \mathcal{P}\text{-}$transformable;
        \item $f\in\mathcal{H}$;
        \item $f \in \left \langle \mathcal{T} \right \rangle$,
    \end{itemize}
    in which cases the problem is computable in polynomial time.
\end{lemma}
\begin{proof}
We may assume that $f\notin\langle\mathcal T\rangle$, since otherwise the stated tractable alternative holds. The invertible input modifications below preserve this exclusion, as noted in the definition of binary modification.

    Using the same procedure as in Lemma~\ref{BD8}, we construct $f^*$.
    We also have ${f^*}^{34}_{=_2}=(=_2)$. We divide the proof into two parts: suppose there exists $t\ne1$, such that ${f^*}_{[1,0,t]}^{34}\in \mathbb{C}\mathcal{GEQ}$  in \partref{I} and there does not exist such $t$ in \partref{II}.

\proofpart{I}
In this case, $f^*_{0100}=f^*_{0111}=f^*_{1000}=f^*_{1011}=0$. Note that when $f^*_{0100}=f^*_{0111}=f^*_{1000}=f^*_{1011}=0$, $S({f^*}_{\mathcal{GEQ}}^{34})\subseteq \mathbb{C} \mathcal{GEQ}$. Moreover, ${f^*}_{[1,0,\varepsilon^i]}^{34}=[f^*_{0000}+\varepsilon^if^*_{0011},0,f^*_{1100}+\varepsilon^if^*_{1111}]$. By Lemma~\ref{4eq},


\[
\begin{pmatrix}
 f^*_{0000}\\
 f^*_{0011}\\
 f^*_{1100}\\
 f^*_{1111}
\end{pmatrix}=\begin{pmatrix}
 t_1\\
 b_1\\
 t_1\\
b_1
\end{pmatrix}\ or\ \begin{pmatrix}
 0\\
 1\\
 1\\
0
\end{pmatrix}\ or\ \begin{pmatrix}
 1\\
 0\\
 0\\
1
\end{pmatrix},
\]where $t_1+b_1=1$.

If \[\begin{pmatrix}
 f^*_{0000}\\
 f^*_{0011}\\
 f^*_{1100}\\
 f^*_{1111}
\end{pmatrix}=\begin{pmatrix}
 0\\
 1\\
 1\\
0
\end{pmatrix}\ or\ \begin{pmatrix}
 1\\
 0\\
 0\\
1
\end{pmatrix},\] then ${f^*}_{(1,0,0,\varepsilon^i)}^{12}=(\varepsilon^i,{f^*}_{0001}+\varepsilon^i{f^*}_{1101},{f^*}_{0010}+\varepsilon^i{f^*}_{1110},1)$ or ${f^*}_{(1,0,0,\varepsilon^i)}^{12}=(1,{f^*}_{0001}+\varepsilon^i{f^*}_{1101},{f^*}_{0010}+\varepsilon^i{f^*}_{1110},\varepsilon^i)$ for $i\in[5]$. Thus we have $\frac{{f^*}_{(1,0,0,\varepsilon^i)}^{12}(00)}{{f^*}_{(1,0,0,\varepsilon^i)}^{12}(11)}=\varepsilon^i$(or $\varepsilon^{-i}$). This implies that $S({f^*}_{\mathcal{GEQ}}^{12})\subseteq \mathbb{C}^{\times}\mathcal{GEQ}$. Thus $f^*_{0001}=f^*_{0010}=f^*_{1101}=f^*_{1110}=0$. $f^*$ has eight-vertex form. If $f^*\in\mathcal{A}$, by applying a binary modification using $[1,0,\varepsilon]$ to $f^*$, we construct a new 4-ary signature $f^{**}$ such that $f^{**}\notin\mathcal{A}$ and $f^{**}$ has eight-vertex form. Thus we always have a 4-ary signature that has eight-vertex form and does not belong to $\mathcal{A}$.  By Lemma~\ref{haveeightvertex}, the result follows.

If \[\begin{pmatrix}
 f^*_{0000}\\
 f^*_{0011}\\
 f^*_{1100}\\
 f^*_{1111}
\end{pmatrix}=\begin{pmatrix}
 t_1\\
 b_1\\
 t_1\\
b_1
\end{pmatrix},\] where $t_1+b_1=1$, then
\[
M({f^*})=\begin{pmatrix}
  \overset{0000}{t_1} &  \overset{0001}{*}&  \overset{0010}{*}& \overset{0011}{b_1}\\
  \overset{0100}{0} &  \overset{0101}{*}&  \overset{0110}{*}& \overset{0111}{0}\\
  \overset{1000}{0} &  \overset{1001}{*}&  \overset{1010}{*}& \overset{1011}{0}\\
  \overset{1100}{t_1} &  \overset{1101}{*}&  \overset{1110}{*}& \overset{1111}{b_1}
\end{pmatrix}.
\] Note that ${f^*}_{(1,0,0,\varepsilon^i)}^{12}=((1+\varepsilon^i)t_1,{f^*}_{0001}+\varepsilon^i{f^*}_{1101},{f^*}_{0010}+\varepsilon^i{f^*}_{1110},(1+\varepsilon^i)b_1)$ for $i\in[5]$. Then we have $t_1b_1\ne0$, otherwise ${f^*}_{(1,0,0,\varepsilon^i)}^{12}\notin BI_{120}$. Moreover, $\frac{{f^*}_{(1,0,0,\varepsilon^i)}^{12}(00)}{{f^*}_{(1,0,0,\varepsilon^i)}^{12}(11)}=\frac{t_1}{b_1}$.
If $\frac{t_1}{b_1}\ne \pm\varepsilon^k$, then ${f^*}_{(1,0,0,\varepsilon^i)}^{12}\notin BI_{120}$, a contradiction.
If $\frac{t_1}{b_1}=\varepsilon^k$, then $\frac{{f^*}_{(1,0,0,\varepsilon^i)}^{12}(00)}{{f^*}_{(1,0,0,\varepsilon^i)}^{12}(11)}=\varepsilon^k$, thus $S({f^*}_{\mathcal{GEQ}}^{12})\subseteq\mathbb{C}\mathcal{GEQ}$, and $f^*_{0001}=f^*_{0010}=f^*_{1101}=f^*_{1110}=0$. $f^*$ has eight-vertex form.
Similarly, we can always construct a 4-ary signature that has eight-vertex form and does not belong to $\mathcal{A}$.  By Lemma~\ref{haveeightvertex}, the result follows.  If $\frac{t_1}{b_1}=-\varepsilon^k=-\varrho_1$, then $\frac{{f^*}_{(1,0,0,\varepsilon^i)}^{12}(00)}{{f^*}_{(1,0,0,\varepsilon^i)}^{12}(11)}=-\varrho_1$. Moreover, $S({f^*}_{\mathcal{GEQ}}^{12})\subseteq\mathbb{C}^{\times}\mathcal{Z}$. Thus by Lemma~\ref{4eq}, \[\begin{pmatrix}
 f^*_{0001}\\
 f^*_{1101}\\
 f^*_{0010}\\
 f^*_{1110}
\end{pmatrix}=\begin{pmatrix}
 \varrho_2t_2\\
 \varrho_2b_2\\
 t_2\\
b_2
\end{pmatrix}\ or\ \begin{pmatrix}
 t_2\\
 0\\
 0\\
\varrho_2t_2
\end{pmatrix}\ or\ \begin{pmatrix}
 0\\
 t_2\\
 \varrho_2t_2\\
0
\end{pmatrix}.\] If \[\begin{pmatrix}
 f^*_{0001}\\
 f^*_{1101}\\
 f^*_{0010}\\
 f^*_{1110}
\end{pmatrix}=\begin{pmatrix}
 t_2\\
 0\\
 0\\
\varrho_2t_2
\end{pmatrix}\ or\ \begin{pmatrix}
 0\\
 t_2\\
 \varrho_2t_2\\
0
\end{pmatrix},\] then ${f^*}_{[1,0,\varepsilon^i]}^{12}=((1+\varepsilon^i)t_1,t_2,\varepsilon^i\varrho_2t_2,(1+\varepsilon^i)b_1)$ or ${f^*}_{[1,0,\varepsilon^i]}^{12}=((1+\varepsilon^i)t_1,\varepsilon^it_2,\varrho_2t_2,(1+\varepsilon^i)b_1)$ and ${f^*}_{[1,0,\varepsilon^i]}^{12}\in \mathbb{C}^{\times}\mathcal{Z}$ for $i\in[5]$ since $t_1t_2b_2\ne0$. If $\{{f^*}_{=_2}^{12},{f^*}_{[1,0,\varepsilon]}^{12}\}\subseteq\mathbb{C}^{\times}\mathcal{Z}^p$ or $\mathbb{C}^{\times}\mathcal{Z}^q$, we have $\frac{\varrho_2t_2^2}{4t_1b_1}=\frac{\varepsilon \varrho_2t_2^2}{(1+\varepsilon)^2t_1b_1}=-(\frac{p}{q})^2$ or $-(\frac{q}{p})^2$, i.e., $\frac{1}{4}=\frac{\varepsilon}{(1+\varepsilon)^2}$, a contradiction. If  ${f^*}_{=_2}^{12}\in\mathbb{C}^{\times}\mathcal{Z}^p$ and ${f^*}_{[1,0,\varepsilon]}^{12}\in\mathbb{C}^{\times}\mathcal{Z}^q$, then we have $\frac{\varrho_2t_2^2}{4t_1b_1}=-(\frac{p}{q})^2$ and $\frac{\varepsilon\varrho_2 t_2^2}{(1+\varepsilon)^2t_1b_1}=-(\frac{q}{p})^2$, i.e., $\frac{\varepsilon}{(1+\varepsilon)^2}=\frac{q^4}{4p^4}$. Note that $\frac{q}{p}=\varepsilon^2+\varepsilon^3+1$, a contradiction. The case in which ${f^*}_{=_2}^{12}\in\mathbb{C}^{\times}\mathcal{Z}^q$ and ${f^*}_{[1,0,\varepsilon]}^{12}\in\mathbb{C}^{\times}\mathcal{Z}^p$ is analogous and is omitted.

Now we consider the case in which \[\begin{pmatrix}
 f^*_{0001}\\
 f^*_{1101}\\
 f^*_{0010}\\
 f^*_{1110}
\end{pmatrix}=\begin{pmatrix}
 \varrho_2t_2\\
 \varrho_2b_2\\
 t_2\\
b_2
\end{pmatrix},\] i.e.,
\[
M(f^*)=\begin{pmatrix}
  \overset{0000}{-\varrho_1 b_1} &  \overset{0001}{\varrho_2t_2}&  \overset{0010}{t_2}& \overset{0011}{b_1}\\
  \overset{0100}{0} &  \overset{0101}{}&  \overset{0110}{}& \overset{0111}{0}\\
  \overset{1000}{0} &  \overset{1001}{}&  \overset{1010}{}& \overset{1011}{0}\\
  \overset{1100}{-\varrho_1 b_1} &  \overset{1101}{\varrho_2b_2}&  \overset{1110}{b_2}& \overset{1111}{b_1}
\end{pmatrix}.
\]
If $t_2=b_2=0$, then $f^*$ has eight-vertex form and similarly, we can always construct a 4-ary signature that has eight-vertex form and does not belong to $\mathcal{A}$. By Lemma~\ref{haveeightvertex}, the result follows. Note that ${f^*}_{(0,1,-\varepsilon^i,0)}^{34}=((\varrho_2-\varepsilon^i)t_2,f^*_{0101}-\varepsilon^if^*_{0110},f^*_{1001}-\varepsilon^if^*_{1010},(\varrho_2-\varepsilon^i)b_2)$. If $t_2=0$ or $b_2=0$, then ${f^*}_{(0,1,-\varepsilon^i,0)}^{34}\notin\ \mathbb{C}BI_{120}$, a contradiction. Thus we assume that $t_2b_2\ne0$.

If $\frac{t_2}{b_2}\notin\pm \mu_5$, then ${f^*}_{(0,1,-\varepsilon^i,0)}^{34}\notin BI_{120}$, a contradiction.

If $\frac{t_2}{b_2}=-\varrho_3\in-\mu_5$. Note that ${f^*}_{[1,0,\varepsilon^i]}^{13}=(f^*_{0000}+\varepsilon^if^*_{1010},\varrho_2t_2,\varepsilon^ib_2,f^*_{0101}+\varepsilon^if^*_{1111})$, then $\frac{{f^*}_{[1,0,\varepsilon^i]}^{13}(01)}{{f^*}_{[1,0,\varepsilon^i]}^{13}(10)}=\frac{\varrho_2t_2}{\varepsilon^ib_2}\in-\mu_5$, which implies that $S({f^*}_{\mathcal{GEQ}}^{13})\subseteq\mathbb{C}^{\times}\mathcal{GDE}$. Then we have \[\begin{pmatrix}
 f^*_{0000}\\
 f^*_{1010}\\
 f^*_{0101}\\
f^*_{1111}
\end{pmatrix}=\begin{pmatrix}
0\\
0\\
0\\
0
\end{pmatrix},\] a contradiction since $f^*_{0000}f^*_{1111}\ne0$.

 If $\frac{t_2}{b_2}=\varrho_3\in\mu_5$, note that there exist at least four distinct $i$ such that $\varrho_2\ne\varepsilon^i$, then in $S({f^*}_{\mathcal{GDE}}^{34})$, there are at least four binary signatures in $\mathbb{C}\mathcal{GEQ}$, thus
\[\begin{pmatrix}
 f^*_{0101}\\
 f^*_{0110}\\
 f^*_{1001}\\
f^*_{1010}
\end{pmatrix}=\begin{pmatrix}
0\\
0\\
0\\
0
\end{pmatrix}.\] Now
\[
M(f^*)=\begin{pmatrix}
  \overset{0000}{-\varrho_1 b_1} &  \overset{0001}{\varrho_2\varrho_3b_2}&  \overset{0010}{\varrho_3b_2}& \overset{0011}{b_1}\\
  \overset{0100}{0} &  \overset{0101}{0}&  \overset{0110}{0}& \overset{0111}{0}\\
  \overset{1000}{0} &  \overset{1001}{0}&  \overset{1010}{0}& \overset{1011}{0}\\
  \overset{1100}{-\varrho_1 b_1} &  \overset{1101}{\varrho_2 b_2}&  \overset{1110}{b_2}& \overset{1111}{b_1}
\end{pmatrix},
\] where $b_1b_2\ne0$. Note that ${f^*}_{=_2}^{12}=(-2\varrho_1b_1,(\varrho_3+1)\varrho_2b_2,(\varrho_3+1)b_2,2b_1)$ and
${f^*}_{=_2}^{23}=(f^*_{0000}+f^*_{0110},f^*_{0001}+f^*_{0111},f^*_{1000}+f^*_{1110},f^*_{1001}+f^*_{1111})=(-\varrho_1 b_1,\varrho_2\varrho_3b_2,b_2,b_1)$.
Then $\{{f^*}_{=_2}^{12},{f^*}_{=_2}^{23}\}\subseteq \mathbb{C}^{\times}\mathcal{Z}$. If $\{{f^*}_{=_2}^{12},{f^*}_{=_2}^{23}\}\subseteq \mathbb{C}^{\times}\mathcal{Z}^p$ or $\mathbb{C}^{\times}\mathcal{Z}^q$, then we have $\frac{(\varrho_3+1)^2\varrho_2b_2^2}{-4\varrho_1b_1^2}=\frac{\varrho_2\varrho_3b_2^2}{-\varrho_1b_1^2}$, i.e., $\varrho_3=1$. Thus $f^*\in \left \langle \mathcal{T} \right \rangle$ and then $f\in\left \langle \mathcal{T} \right \rangle$. If ${f^*}_{=_2}^{12}\in \mathbb{C}^{\times}\mathcal{Z}^p$ and ${f^*}_{=_2}^{23} \in \mathbb{C}^{\times}\mathcal{Z}^q$, then we have $\frac{(\varrho_3+1)^2\varrho_2b_2^2}{-4\varrho_1b_1^2}=-(\frac{p}{q})^2$ and $\frac{\varrho_2\varrho_3b_2^2}{-\varrho_1b_1^2}=-(\frac{q}{p})^2$, i.e., $\frac{(\varrho_3+1)^2}{4\varrho_3}=(\frac{p}{q})^4$. Note that $\frac{p}{q}=\varepsilon^2+\varepsilon^3$ and $\varrho_3^5=1$, then we can obtain a contradiction. The case in which ${f^*}_{=_2}^{12}\in \mathbb{C}^{\times}\mathcal{Z}^q$ and ${f^*}_{=_2}^{23} \in \mathbb{C}^{\times}\mathcal{Z}^p$ is analogous and is omitted.

\proofpart{II}
In this case, we have
\[
\left\{\begin{matrix}
 f^*_{0000}+f^*_{0011}=f^*_{1100}+f^*_{1111}=1\\
f^*_{0100}+f^*_{0111}=f^*_{1000}+f^*_{1011}=0
\end{matrix}\right..
\]
Then $f^*_{0011}=1-f^*_{0000}$, $f^*_{1111}=1-f^*_{1100}$, and
\[\begin{pmatrix}
 f^*_{0100}\\
 f^*_{0111}\\
 f^*_{1000}\\
f^*_{1011}
\end{pmatrix}= \begin{pmatrix}
 t_1\\
 -t_1\\
 b_1\\
-b_1
\end{pmatrix},\] where $t_1b_1\ne0$. Note that ${f^*}_{[1,0,\varepsilon^i]}^{34}=((1-\varepsilon^i)f^*_{0000}+\varepsilon^i,(1-\varepsilon^i)t_1,(1-\varepsilon^i)b_1,(1-\varepsilon^i)f^*_{1100}+\varepsilon^i)$.
If $\frac{t_1}{b_1}\ne\pm \varepsilon^k$, then ${f^*}_{[1,0,\varepsilon^i]}^{34}\notin BI_{120}$, a contradiction. If $\frac{t_1}{b_1} = -\varepsilon^k$, then for $i\in[4]$, ${f^*}_{[1,0,\varepsilon^i]}^{34}\in \mathbb{C}\mathcal{GDE}$, and then $f^*_{0000}=f^*_{0011}=f^*_{1100}=f^*_{1111}=0$, a contradiction. Indeed, both diagonal entries vanish for two distinct values of the parameter, so all four coefficients vanish, contradicting their prescribed sums. Thus we have $\frac{t_1}{b_1} =\varepsilon^k=\varrho_1$,
then  ${f^*}_{[1,0,\varepsilon^i]}^{34}\in\mathbb{C}^{\times}\mathcal{Z}$. This implies that $\frac{{f^*}_{[1,0,\varepsilon^i]}^{34}(00)}{{f^*}_{[1,0,\varepsilon^i]}^{34}(11)}\in-\mu_5$.
If there exists $n,m\in[4]$ and $n\ne m$, such that $\frac{{f^*}_{[1,0,\varepsilon^n]}^{34}(00)}{{f^*}_{[1,0,\varepsilon^n]}^{34}(11)}=\frac{{f^*}_{[1,0,\varepsilon^m]}^{34}(00)}{{f^*}_{[1,0,\varepsilon^m]}^{34}(11)}$, i.e., $\frac{(1-\varepsilon^n)f^*_{0000}+\varepsilon^n}{(1-\varepsilon^n)f^*_{1100}+\varepsilon^n}=\frac{(1-\varepsilon^m)f^*_{0000}+\varepsilon^m}{(1-\varepsilon^m)f^*_{1100}+\varepsilon^m}$, we have $f^*_{0000}=f^*_{1100}$. But then for $i\in[4]$, $\frac{{f^*}_{[1,0,\varepsilon^i]}^{34}(00)}{{f^*}_{[1,0,\varepsilon^i]}^{34}(11)}=1\notin-\mu_5$, a contradiction. Thus for any $n,m\in[4]$ and $n\ne m$, $\frac{{f^*}_{[1,0,\varepsilon^n]}^{34}(00)}{{f^*}_{[1,0,\varepsilon^n]}^{34}(11)}\ne\frac{{f^*}_{[1,0,\varepsilon^m]}^{34}(00)}{{f^*}_{[1,0,\varepsilon^m]}^{34}(11)}$.

By $S({f^*}_{\mathcal{GEQ}}^{34})$,
\[\begin{pmatrix}
 f^*_{0000}\\
 f^*_{0011}\\
 f^*_{1100}\\
f^*_{1111}
\end{pmatrix}\] satisfies a homogeneous linear system of equations whose coefficient matrix is:
\[
\begin{pmatrix}
  1&  1&  -1& -1\\
  1&  \varepsilon&  -a_1& -\varepsilon a_1\\
  1&  \varepsilon^2&  -a_2& -\varepsilon^2 a_2\\
  1&  \varepsilon^3&  -a_3& -\varepsilon^3 a_3\\
  1&  \varepsilon^4&  -a_4&-\varepsilon^4 a_4
\end{pmatrix},
\] where $a_i=\frac{{f^*}_{[1,0,\varepsilon^{i}]}^{34}(00)}{{f^*}_{[1,0,\varepsilon^{i}]}^{34}(11)}\in-\mu_5$ and are all different for $i\in[4]$. Then we can solve this homogeneous linear system of equations by going through all the possibilities of $(a_1,a_2,a_3,a_4)$. There are two cases in which \[\begin{pmatrix}
 f^*_{0000}\\
 f^*_{0011}\\
 f^*_{1100}\\
f^*_{1111}
\end{pmatrix}\] has a nonzero solution:
\begin{caselist}
    \caseitem{1}   $(a_1,a_2,a_3,a_4)=(-\varepsilon^2, -\varepsilon, -\varepsilon^4, -\varepsilon^3)$. In this case, \[\begin{pmatrix}
 f^*_{0000}\\
 f^*_{0011}\\
 f^*_{1100}\\
f^*_{1111}
\end{pmatrix}=f^*_{0000}\begin{pmatrix}
1\\
 \varepsilon^3+\varepsilon^2+2\\
 \varepsilon^3+\varepsilon^2+2\\
1
\end{pmatrix}.\]
\caseitem{2}  $(a_1,a_2,a_3,a_4)=(-\varepsilon^3, -\varepsilon^4, -\varepsilon, -\varepsilon^2)$. In this case, \[\begin{pmatrix}
 f^*_{0000}\\
 f^*_{0011}\\
 f^*_{1100}\\
f^*_{1111}
\end{pmatrix}=f^*_{0000}\begin{pmatrix}
1\\
 -\varepsilon^3-\varepsilon^2+1\\
 -\varepsilon^3-\varepsilon^2+1\\
1
\end{pmatrix}.\]
\end{caselist}
Now, let $u\in\{\varepsilon^3+\varepsilon^2+2 ,-\varepsilon^3-\varepsilon^2+1\}$, $f^*_{0000}=t$, we have  \[M(f^*)=\begin{pmatrix}
  \overset{0000}{t} &  \overset{0001}{*}&  \overset{0010}{*}& \overset{0011}{ut}\\
  \overset{0100}{\varrho_1b_1} &  \overset{0101}{*}&  \overset{0110}{*}& \overset{0111}{-\varrho_1b_1}\\
  \overset{1000}{b_1} &  \overset{1001}{*}&  \overset{1010}{*}& \overset{1011}{-b_1}\\
  \overset{1100}{ut} &  \overset{1101}{*}&  \overset{1110}{*}& \overset{1111}{t}
\end{pmatrix}.\]

Note that ${f^*}_{[1,0,\varepsilon^i]}^{34}=((1+\varepsilon^iu)t,(1-\varepsilon^i)\varrho_1b_1,(1-\varepsilon^i)b_1,(u+\varepsilon^i)t)$, and ${f^*}_{[1,0,\varepsilon^i]}^{34}\in \mathbb{C}^{\times}\mathcal{Z}$ for $i\in[4]$.
For convenience, let $x_i=(1-\varepsilon^i)^2$ and $y_i=(1+\varepsilon^iu)(u+\varepsilon^i)$. For $n,m\in[4]$ with $n\ne m$, the following table lists the values of  $\frac{x_n}{y_n}-\frac{x_m}{y_m}$ for different $n,m$ and $u$.

\begin{center}
\begin{tabular}{ccc}
\toprule

\multirow{2}{*}{($n,m$)} & \multicolumn{2}{c}{Value of $\frac{x_n}{y_n}-\frac{x_m}{y_m}$} \\
\cmidrule(lr){2-3}
 &$u=\varepsilon^3+\varepsilon^2+2$ &  $u=-\varepsilon^3-\varepsilon^2+1$\\\midrule
$(1,2)$ & $-3\varepsilon^3 - 3\varepsilon^2 + 1$ & $-3\varepsilon^3 - 3\varepsilon^2 - 4$\\
$(1,3)$& $-3\varepsilon^3 - 3\varepsilon^2 + 1$ & $-3\varepsilon^3 - 3\varepsilon^2 - 4$\\
$(1,4)$ & 0 & 0\\
$(2,3)$  & 0 & 0\\
$(2,4)$  & $3\varepsilon^3 + 3\varepsilon^2 - 1$ & $3\varepsilon^3 + 3\varepsilon^2 + 4$\\
$(3,4)$  & $3\varepsilon^3 + 3\varepsilon^2 - 1$ & $3\varepsilon^3 + 3\varepsilon^2 + 4$\\
\bottomrule
\end{tabular}
\end{center}
This implies that
\[
\begin{gathered}
\{{f^*}_{[1,0,\varepsilon^1]}^{34},{f^*}_{[1,0,\varepsilon^4]}^{34}\}\subseteq \mathbb{C}^{\times}\mathcal{Z}^p\quad\text{and}\quad \{{f^*}_{[1,0,\varepsilon^2]}^{34},{f^*}_{[1,0,\varepsilon^3]}^{34}\}\subseteq \mathbb{C}^{\times}\mathcal{Z}^q,\\
\text{or}\\
\{{f^*}_{[1,0,\varepsilon^1]}^{34},{f^*}_{[1,0,\varepsilon^4]}^{34}\}\subseteq \mathbb{C}^{\times}\mathcal{Z}^q\quad\text{and}\quad \{{f^*}_{[1,0,\varepsilon^2]}^{34},{f^*}_{[1,0,\varepsilon^3]}^{34}\}\subseteq \mathbb{C}^{\times}\mathcal{Z}^p.
\end{gathered}
\]

If $\{{f^*}_{[1,0,\varepsilon^1]}^{34},{f^*}_{[1,0,\varepsilon^4]}^{34}\}\subseteq \mathbb{C}^{\times}\mathcal{Z}^p$ and $\{{f^*}_{[1,0,\varepsilon^2]}^{34},{f^*}_{[1,0,\varepsilon^3]}^{34}\}\subseteq \mathbb{C}^{\times}\mathcal{Z}^q$, then we have $\frac{x_1\varrho_1b_1^2}{y_1t^2}=-(\frac{p}{q})^2$ and $\frac{x_2\varrho_1b_1^2}{y_2t^2}=-(\frac{q}{p})^2$, i.e., $\frac{x_1y_2}{y_1x_2}=(\frac{p}{q})^4$. Note that $\frac{x_1y_2}{y_1x_2}-(\frac{p}{q})^4=6\varepsilon^3 + 6\varepsilon^2 + 3$ for any $u$, a contradiction. Thus we have $\{{f^*}_{[1,0,\varepsilon^1]}^{34},{f^*}_{[1,0,\varepsilon^4]}^{34}\}\subseteq \mathbb{C}^{\times}\mathcal{Z}^q$ and $\{{f^*}_{[1,0,\varepsilon^2]}^{34},{f^*}_{[1,0,\varepsilon^3]}^{34}\}\subseteq \mathbb{C}^{\times}\mathcal{Z}^p$. We calculate the eigenvector matrix of all matrices in $\mathcal{Z}$ and check whether this eigenvector matrix belongs to $\mathbb{C}^{\times}BI_{120}$. Note that for a matrix $g\in\mathcal{Z}$, if $g\in \mathcal{Z}^q$ and $\frac{g_{00}}{g_{11}}=-\varepsilon^2$ or $-\varepsilon^3$, then $g$ can be written as $g=UDU^{-1}$, where $D\in\mathcal{GEQ}$ and $U\in BI_{120}$. If $g\in \mathcal{Z}^p$ and $\frac{g_{00}}{g_{11}}=-\varepsilon^1$ or $-\varepsilon^4$, then $g$ can be written as $g=UDU^{-1}$, where $D\in\mathcal{GEQ}$ and $U\in BI_{120}$. Recall that $(a_1,a_2,a_3,a_4)=(-\varepsilon^2, -\varepsilon, -\varepsilon^4, -\varepsilon^3)$ or $(a_1,a_2,a_3,a_4)=(-\varepsilon^3, -\varepsilon^4, -\varepsilon, -\varepsilon^2)$. Thus for $i\in[4]$, ${f^*}_{[1,0,\varepsilon^i]}^{34}$ can be written as $UDU^{-1}$, where $D\in\mathcal{GEQ}$ and $U\in BI_{120}$. By applying the binary modifications $U^{-1}$ to input $x_1$ and $U^T$ to input $x_2$, the new 4-ary signature $f^{**}$ satisfies $M({f^{**}}_{=_2}^{34})=U^{-1}IU=I$ and $M({f^{**}}_{[1,0,\varepsilon^i]}^{34})=U^{-1}(UDU^{-1})U=D\in\mathcal{GEQ}$, for the chosen $i\in[4]$. Since $\varepsilon^i\ne1$, $f^{**}$ is in \partref{I}.

In the following table, set $\rho=e^{\pi\frak{i}/3}$ and $\sigma=e^{2\pi\frak{i}/3}$.

\begingroup
\small
\setlength{\tabcolsep}{4pt}
\renewcommand{\arraystretch}{1.18}
\begin{longtable}{@{}clllc@{}}
\toprule
${g_{00}}/{g_{11}} $ & Matrix $g$ & Eigenvector matrix $U$ & $D=U^{-1}gU$ & In $\mathbb{C}^{\times}\BI$?\\
\midrule
\endfirsthead
\toprule
${g_{00}}/{g_{11}}$ & Matrix $g$ & Eigenvector matrix $U$ & $D=U^{-1}gU$ & In $\mathbb{C}^{\times}\BI$?\\
\midrule
\endhead
\midrule
\multicolumn{5}{r}{Continued on the next page}\\
\endfoot
\bottomrule
\endlastfoot
\multicolumn{5}{@{}l}{\textbf{Ratio: }$g_{01}/g_{10}=1$}\\*
\multicolumn{5}{@{}l}{\textbf{Family: }$\mathcal{Z}^q$}\\*[4pt]
$-\varepsilon^0 $ & $\begin{pmatrix}p&q\\q&-p\end{pmatrix}$ & $\begin{pmatrix}q&q\\\frak{i}-p&-\frak{i}-p\end{pmatrix}$ & $\operatorname{diag}(\frak{i},-\frak{i})$ & \no\\*[4pt]
$-\varepsilon^1 $ & $\begin{pmatrix}-p\varepsilon^3&-q\\-q&p\varepsilon^2\end{pmatrix}$ & $\begin{pmatrix}-q&-q\\\rho+p\varepsilon^3&\overline\rho+p\varepsilon^3\end{pmatrix}$ & $\operatorname{diag}(\rho,\overline\rho)$ & \no\\*[4pt]
$-\varepsilon^2 $ & $\begin{pmatrix}p\varepsilon&q\\q&-p\varepsilon^4\end{pmatrix}$ & $\begin{pmatrix}p&q\\q&-p\end{pmatrix}$ & $\operatorname{diag}(-\varepsilon^2,-\varepsilon^3)$ & \yes\\*[4pt]
$-\varepsilon^3 $ & $\begin{pmatrix}-p\varepsilon^4&-q\\-q&p\varepsilon\end{pmatrix}$ & $\begin{pmatrix}p&q\\q&-p\end{pmatrix}$ & $\operatorname{diag}(-\varepsilon^3,-\varepsilon^2)$ & \yes\\*[4pt]
$-\varepsilon^4 $ & $\begin{pmatrix}p\varepsilon^2&q\\q&-p\varepsilon^3\end{pmatrix}$ & $\begin{pmatrix}q&q\\\rho-p\varepsilon^2&\overline\rho-p\varepsilon^2\end{pmatrix}$ & $\operatorname{diag}(\rho,\overline\rho)$ & \no\\*[4pt]
\multicolumn{5}{@{}l}{\textbf{Family: }$\mathcal{Z}^p$}\\*[4pt]
$-\varepsilon^0 $ & $\begin{pmatrix}q&-p\\-p&-q\end{pmatrix}$ & $\begin{pmatrix}-p&-p\\\frak{i}-q&-\frak{i}-q\end{pmatrix}$ & $\operatorname{diag}(\frak{i},-\frak{i})$ & \no\\*[4pt]
$-\varepsilon^1 $ & $\begin{pmatrix}-q\varepsilon^3&p\\p&q\varepsilon^2\end{pmatrix}$ & $\begin{pmatrix}p&q\\q&-p\end{pmatrix}$ & $\operatorname{diag}(-\varepsilon^4,-\varepsilon)$ & \yes\\*[4pt]
$-\varepsilon^2 $ & $\begin{pmatrix}-q\varepsilon&p\\p&q\varepsilon^4\end{pmatrix}$ & $\begin{pmatrix}p&p\\\rho+q\varepsilon&\overline\rho+q\varepsilon\end{pmatrix}$ & $\operatorname{diag}(\rho,\overline\rho)$ & \no\\*[4pt]
$-\varepsilon^3 $ & $\begin{pmatrix}q\varepsilon^4&-p\\-p&-q\varepsilon\end{pmatrix}$ & $\begin{pmatrix}-p&-p\\\rho-q\varepsilon^4&\overline\rho-q\varepsilon^4\end{pmatrix}$ & $\operatorname{diag}(\rho,\overline\rho)$ & \no\\*[4pt]
$-\varepsilon^4 $ & $\begin{pmatrix}q\varepsilon^2&-p\\-p&-q\varepsilon^3\end{pmatrix}$ & $\begin{pmatrix}p&q\\q&-p\end{pmatrix}$ & $\operatorname{diag}(-\varepsilon,-\varepsilon^4)$ & \yes\\[4pt]
\addlinespace[7pt]
\multicolumn{5}{@{}l}{\textbf{Ratio: }$g_{01}/g_{10}=\varepsilon$}\\*
\multicolumn{5}{@{}l}{\textbf{Family: }$\mathcal{Z}^q$}\\*[4pt]
$-\varepsilon^0 $ & $\begin{pmatrix}-p&-q\varepsilon^3\\-q\varepsilon^2&p\end{pmatrix}$ & $\begin{pmatrix}-q\varepsilon^3&-q\varepsilon^3\\\frak{i}+p&-\frak{i}+p\end{pmatrix}$ & $\operatorname{diag}(\frak{i},-\frak{i})$ & \no\\*[4pt]
$-\varepsilon^1 $ & $\begin{pmatrix}p\varepsilon^3&q\varepsilon^3\\q\varepsilon^2&-p\varepsilon^2\end{pmatrix}$ & $\begin{pmatrix}q\varepsilon^3&q\varepsilon^3\\\sigma-p\varepsilon^3&\overline\sigma-p\varepsilon^3\end{pmatrix}$ & $\operatorname{diag}(\sigma,\overline\sigma)$ & \no\\*[4pt]
$-\varepsilon^2 $ & $\begin{pmatrix}-p\varepsilon&-q\varepsilon^3\\-q\varepsilon^2&p\varepsilon^4\end{pmatrix}$ & $\begin{pmatrix}p\varepsilon^4&q\varepsilon^4\\q\varepsilon&-p\varepsilon\end{pmatrix}$ & $\operatorname{diag}(\varepsilon^2,\varepsilon^3)$ & \yes\\*[4pt]
$-\varepsilon^3 $ & $\begin{pmatrix}p\varepsilon^4&q\varepsilon^3\\q\varepsilon^2&-p\varepsilon\end{pmatrix}$ & $\begin{pmatrix}p\varepsilon^4&q\varepsilon^4\\q\varepsilon&-p\varepsilon\end{pmatrix}$ & $\operatorname{diag}(\varepsilon^3,\varepsilon^2)$ & \yes\\*[4pt]
$-\varepsilon^4 $ & $\begin{pmatrix}p\varepsilon^2&q\varepsilon^3\\q\varepsilon^2&-p\varepsilon^3\end{pmatrix}$ & $\begin{pmatrix}q\varepsilon^3&q\varepsilon^3\\\rho-p\varepsilon^2&\overline\rho-p\varepsilon^2\end{pmatrix}$ & $\operatorname{diag}(\rho,\overline\rho)$ & \no\\*[4pt]
\multicolumn{5}{@{}l}{\textbf{Family: }$\mathcal{Z}^p$}\\*[4pt]
$-\varepsilon^0 $ & $\begin{pmatrix}-q&p\varepsilon^3\\p\varepsilon^2&q\end{pmatrix}$ & $\begin{pmatrix}p\varepsilon^3&p\varepsilon^3\\\frak{i}+q&-\frak{i}+q\end{pmatrix}$ & $\operatorname{diag}(\frak{i},-\frak{i})$ & \no\\*[4pt]
$-\varepsilon^1 $ & $\begin{pmatrix}-q\varepsilon^3&p\varepsilon^3\\p\varepsilon^2&q\varepsilon^2\end{pmatrix}$ & $\begin{pmatrix}p\varepsilon^4&q\varepsilon^4\\q\varepsilon&-p\varepsilon\end{pmatrix}$ & $\operatorname{diag}(-\varepsilon^4,-\varepsilon)$ & \yes\\*[4pt]
$-\varepsilon^2 $ & $\begin{pmatrix}q\varepsilon&-p\varepsilon^3\\-p\varepsilon^2&-q\varepsilon^4\end{pmatrix}$ & $\begin{pmatrix}-p\varepsilon^3&-p\varepsilon^3\\\sigma-q\varepsilon&\overline\sigma-q\varepsilon\end{pmatrix}$ & $\operatorname{diag}(\sigma,\overline\sigma)$ & \no\\*[4pt]
$-\varepsilon^3 $ & $\begin{pmatrix}q\varepsilon^4&-p\varepsilon^3\\-p\varepsilon^2&-q\varepsilon\end{pmatrix}$ & $\begin{pmatrix}-p\varepsilon^3&-p\varepsilon^3\\\rho-q\varepsilon^4&\overline\rho-q\varepsilon^4\end{pmatrix}$ & $\operatorname{diag}(\rho,\overline\rho)$ & \no\\*[4pt]
$-\varepsilon^4 $ & $\begin{pmatrix}q\varepsilon^2&-p\varepsilon^3\\-p\varepsilon^2&-q\varepsilon^3\end{pmatrix}$ & $\begin{pmatrix}p\varepsilon^4&q\varepsilon^4\\q\varepsilon&-p\varepsilon\end{pmatrix}$ & $\operatorname{diag}(-\varepsilon,-\varepsilon^4)$ & \yes\\[4pt]
\addlinespace[7pt]
\multicolumn{5}{@{}l}{\textbf{Ratio: }$g_{01}/g_{10}=\varepsilon^2$}\\*
\multicolumn{5}{@{}l}{\textbf{Family: }$\mathcal{Z}^q$}\\*[4pt]
$-\varepsilon^0 $ & $\begin{pmatrix}p&q\varepsilon\\q\varepsilon^4&-p\end{pmatrix}$ & $\begin{pmatrix}q\varepsilon&q\varepsilon\\\frak{i}-p&-\frak{i}-p\end{pmatrix}$ & $\operatorname{diag}(\frak{i},-\frak{i})$ & \no\\*[4pt]
$-\varepsilon^1 $ & $\begin{pmatrix}-p\varepsilon^3&-q\varepsilon\\-q\varepsilon^4&p\varepsilon^2\end{pmatrix}$ & $\begin{pmatrix}-q\varepsilon&-q\varepsilon\\\rho+p\varepsilon^3&\overline\rho+p\varepsilon^3\end{pmatrix}$ & $\operatorname{diag}(\rho,\overline\rho)$ & \no\\*[4pt]
$-\varepsilon^2 $ & $\begin{pmatrix}p\varepsilon&q\varepsilon\\q\varepsilon^4&-p\varepsilon^4\end{pmatrix}$ & $\begin{pmatrix}p\varepsilon^3&q\varepsilon^3\\q\varepsilon^2&-p\varepsilon^2\end{pmatrix}$ & $\operatorname{diag}(-\varepsilon^2,-\varepsilon^3)$ & \yes\\*[4pt]
$-\varepsilon^3 $ & $\begin{pmatrix}p\varepsilon^4&q\varepsilon\\q\varepsilon^4&-p\varepsilon\end{pmatrix}$ & $\begin{pmatrix}p\varepsilon^3&q\varepsilon^3\\q\varepsilon^2&-p\varepsilon^2\end{pmatrix}$ & $\operatorname{diag}(\varepsilon^3,\varepsilon^2)$ & \yes\\*[4pt]
$-\varepsilon^4 $ & $\begin{pmatrix}-p\varepsilon^2&-q\varepsilon\\-q\varepsilon^4&p\varepsilon^3\end{pmatrix}$ & $\begin{pmatrix}-q\varepsilon&-q\varepsilon\\\sigma+p\varepsilon^2&\overline\sigma+p\varepsilon^2\end{pmatrix}$ & $\operatorname{diag}(\sigma,\overline\sigma)$ & \no\\*[4pt]
\multicolumn{5}{@{}l}{\textbf{Family: }$\mathcal{Z}^p$}\\*[4pt]
$-\varepsilon^0 $ & $\begin{pmatrix}-q&p\varepsilon\\p\varepsilon^4&q\end{pmatrix}$ & $\begin{pmatrix}p\varepsilon&p\varepsilon\\\frak{i}+q&-\frak{i}+q\end{pmatrix}$ & $\operatorname{diag}(\frak{i},-\frak{i})$ & \no\\*[4pt]
$-\varepsilon^1 $ & $\begin{pmatrix}q\varepsilon^3&-p\varepsilon\\-p\varepsilon^4&-q\varepsilon^2\end{pmatrix}$ & $\begin{pmatrix}p\varepsilon^3&q\varepsilon^3\\q\varepsilon^2&-p\varepsilon^2\end{pmatrix}$ & $\operatorname{diag}(\varepsilon^4,\varepsilon)$ & \yes\\*[4pt]
$-\varepsilon^2 $ & $\begin{pmatrix}q\varepsilon&-p\varepsilon\\-p\varepsilon^4&-q\varepsilon^4\end{pmatrix}$ & $\begin{pmatrix}-p\varepsilon&-p\varepsilon\\\sigma-q\varepsilon&\overline\sigma-q\varepsilon\end{pmatrix}$ & $\operatorname{diag}(\sigma,\overline\sigma)$ & \no\\*[4pt]
$-\varepsilon^3 $ & $\begin{pmatrix}q\varepsilon^4&-p\varepsilon\\-p\varepsilon^4&-q\varepsilon\end{pmatrix}$ & $\begin{pmatrix}-p\varepsilon&-p\varepsilon\\\rho-q\varepsilon^4&\overline\rho-q\varepsilon^4\end{pmatrix}$ & $\operatorname{diag}(\rho,\overline\rho)$ & \no\\*[4pt]
$-\varepsilon^4 $ & $\begin{pmatrix}-q\varepsilon^2&p\varepsilon\\p\varepsilon^4&q\varepsilon^3\end{pmatrix}$ & $\begin{pmatrix}p\varepsilon^3&q\varepsilon^3\\q\varepsilon^2&-p\varepsilon^2\end{pmatrix}$ & $\operatorname{diag}(\varepsilon,\varepsilon^4)$ & \yes\\[4pt]
\addlinespace[7pt]
\multicolumn{5}{@{}l}{\textbf{Ratio: }$g_{01}/g_{10}=\varepsilon^3$}\\*
\multicolumn{5}{@{}l}{\textbf{Family: }$\mathcal{Z}^q$}\\*[4pt]
$-\varepsilon^0 $ & $\begin{pmatrix}-p&-q\varepsilon^4\\-q\varepsilon&p\end{pmatrix}$ & $\begin{pmatrix}-q\varepsilon^4&-q\varepsilon^4\\\frak{i}+p&-\frak{i}+p\end{pmatrix}$ & $\operatorname{diag}(\frak{i},-\frak{i})$ & \no\\*[4pt]
$-\varepsilon^1 $ & $\begin{pmatrix}p\varepsilon^3&q\varepsilon^4\\q\varepsilon&-p\varepsilon^2\end{pmatrix}$ & $\begin{pmatrix}q\varepsilon^4&q\varepsilon^4\\\sigma-p\varepsilon^3&\overline\sigma-p\varepsilon^3\end{pmatrix}$ & $\operatorname{diag}(\sigma,\overline\sigma)$ & \no\\*[4pt]
$-\varepsilon^2 $ & $\begin{pmatrix}p\varepsilon&q\varepsilon^4\\q\varepsilon&-p\varepsilon^4\end{pmatrix}$ & $\begin{pmatrix}p\varepsilon^2&q\varepsilon^2\\q\varepsilon^3&-p\varepsilon^3\end{pmatrix}$ & $\operatorname{diag}(-\varepsilon^2,-\varepsilon^3)$ & \yes\\*[4pt]
$-\varepsilon^3 $ & $\begin{pmatrix}-p\varepsilon^4&-q\varepsilon^4\\-q\varepsilon&p\varepsilon\end{pmatrix}$ & $\begin{pmatrix}p\varepsilon^2&q\varepsilon^2\\q\varepsilon^3&-p\varepsilon^3\end{pmatrix}$ & $\operatorname{diag}(-\varepsilon^3,-\varepsilon^2)$ & \yes\\*[4pt]
$-\varepsilon^4 $ & $\begin{pmatrix}p\varepsilon^2&q\varepsilon^4\\q\varepsilon&-p\varepsilon^3\end{pmatrix}$ & $\begin{pmatrix}q\varepsilon^4&q\varepsilon^4\\\rho-p\varepsilon^2&\overline\rho-p\varepsilon^2\end{pmatrix}$ & $\operatorname{diag}(\rho,\overline\rho)$ & \no\\*[4pt]
\multicolumn{5}{@{}l}{\textbf{Family: }$\mathcal{Z}^p$}\\*[4pt]
$-\varepsilon^0 $ & $\begin{pmatrix}q&-p\varepsilon^4\\-p\varepsilon&-q\end{pmatrix}$ & $\begin{pmatrix}-p\varepsilon^4&-p\varepsilon^4\\\frak{i}-q&-\frak{i}-q\end{pmatrix}$ & $\operatorname{diag}(\frak{i},-\frak{i})$ & \no\\*[4pt]
$-\varepsilon^1 $ & $\begin{pmatrix}q\varepsilon^3&-p\varepsilon^4\\-p\varepsilon&-q\varepsilon^2\end{pmatrix}$ & $\begin{pmatrix}p\varepsilon^2&q\varepsilon^2\\q\varepsilon^3&-p\varepsilon^3\end{pmatrix}$ & $\operatorname{diag}(\varepsilon^4,\varepsilon)$ & \yes\\*[4pt]
$-\varepsilon^2 $ & $\begin{pmatrix}q\varepsilon&-p\varepsilon^4\\-p\varepsilon&-q\varepsilon^4\end{pmatrix}$ & $\begin{pmatrix}-p\varepsilon^4&-p\varepsilon^4\\\sigma-q\varepsilon&\overline\sigma-q\varepsilon\end{pmatrix}$ & $\operatorname{diag}(\sigma,\overline\sigma)$ & \no\\*[4pt]
$-\varepsilon^3 $ & $\begin{pmatrix}-q\varepsilon^4&p\varepsilon^4\\p\varepsilon&q\varepsilon\end{pmatrix}$ & $\begin{pmatrix}p\varepsilon^4&p\varepsilon^4\\\sigma+q\varepsilon^4&\overline\sigma+q\varepsilon^4\end{pmatrix}$ & $\operatorname{diag}(\sigma,\overline\sigma)$ & \no\\*[4pt]
$-\varepsilon^4 $ & $\begin{pmatrix}-q\varepsilon^2&p\varepsilon^4\\p\varepsilon&q\varepsilon^3\end{pmatrix}$ & $\begin{pmatrix}p\varepsilon^2&q\varepsilon^2\\q\varepsilon^3&-p\varepsilon^3\end{pmatrix}$ & $\operatorname{diag}(\varepsilon,\varepsilon^4)$ & \yes\\[4pt]
\addlinespace[7pt]
\multicolumn{5}{@{}l}{\textbf{Ratio: }$g_{01}/g_{10}=\varepsilon^4$}\\*
\multicolumn{5}{@{}l}{\textbf{Family: }$\mathcal{Z}^q$}\\*[4pt]
$-\varepsilon^0 $ & $\begin{pmatrix}p&q\varepsilon^2\\q\varepsilon^3&-p\end{pmatrix}$ & $\begin{pmatrix}q\varepsilon^2&q\varepsilon^2\\\frak{i}-p&-\frak{i}-p\end{pmatrix}$ & $\operatorname{diag}(\frak{i},-\frak{i})$ & \no\\*[4pt]
$-\varepsilon^1 $ & $\begin{pmatrix}p\varepsilon^3&q\varepsilon^2\\q\varepsilon^3&-p\varepsilon^2\end{pmatrix}$ & $\begin{pmatrix}q\varepsilon^2&q\varepsilon^2\\\sigma-p\varepsilon^3&\overline\sigma-p\varepsilon^3\end{pmatrix}$ & $\operatorname{diag}(\sigma,\overline\sigma)$ & \no\\*[4pt]
$-\varepsilon^2 $ & $\begin{pmatrix}-p\varepsilon&-q\varepsilon^2\\-q\varepsilon^3&p\varepsilon^4\end{pmatrix}$ & $\begin{pmatrix}p\varepsilon&q\varepsilon\\q\varepsilon^4&-p\varepsilon^4\end{pmatrix}$ & $\operatorname{diag}(\varepsilon^2,\varepsilon^3)$ & \yes\\*[4pt]
$-\varepsilon^3 $ & $\begin{pmatrix}p\varepsilon^4&q\varepsilon^2\\q\varepsilon^3&-p\varepsilon\end{pmatrix}$ & $\begin{pmatrix}p\varepsilon&q\varepsilon\\q\varepsilon^4&-p\varepsilon^4\end{pmatrix}$ & $\operatorname{diag}(\varepsilon^3,\varepsilon^2)$ & \yes\\*[4pt]
$-\varepsilon^4 $ & $\begin{pmatrix}-p\varepsilon^2&-q\varepsilon^2\\-q\varepsilon^3&p\varepsilon^3\end{pmatrix}$ & $\begin{pmatrix}-q\varepsilon^2&-q\varepsilon^2\\\sigma+p\varepsilon^2&\overline\sigma+p\varepsilon^2\end{pmatrix}$ & $\operatorname{diag}(\sigma,\overline\sigma)$ & \no\\*[4pt]
\multicolumn{5}{@{}l}{\textbf{Family: }$\mathcal{Z}^p$}\\*[4pt]
$-\varepsilon^0 $ & $\begin{pmatrix}q&-p\varepsilon^2\\-p\varepsilon^3&-q\end{pmatrix}$ & $\begin{pmatrix}-p\varepsilon^2&-p\varepsilon^2\\\frak{i}-q&-\frak{i}-q\end{pmatrix}$ & $\operatorname{diag}(\frak{i},-\frak{i})$ & \no\\*[4pt]
$-\varepsilon^1 $ & $\begin{pmatrix}q\varepsilon^3&-p\varepsilon^2\\-p\varepsilon^3&-q\varepsilon^2\end{pmatrix}$ & $\begin{pmatrix}p\varepsilon&q\varepsilon\\q\varepsilon^4&-p\varepsilon^4\end{pmatrix}$ & $\operatorname{diag}(\varepsilon^4,\varepsilon)$ & \yes\\*[4pt]
$-\varepsilon^2 $ & $\begin{pmatrix}-q\varepsilon&p\varepsilon^2\\p\varepsilon^3&q\varepsilon^4\end{pmatrix}$ & $\begin{pmatrix}p\varepsilon^2&p\varepsilon^2\\\rho+q\varepsilon&\overline\rho+q\varepsilon\end{pmatrix}$ & $\operatorname{diag}(\rho,\overline\rho)$ & \no\\*[4pt]
$-\varepsilon^3 $ & $\begin{pmatrix}-q\varepsilon^4&p\varepsilon^2\\p\varepsilon^3&q\varepsilon\end{pmatrix}$ & $\begin{pmatrix}p\varepsilon^2&p\varepsilon^2\\\sigma+q\varepsilon^4&\overline\sigma+q\varepsilon^4\end{pmatrix}$ & $\operatorname{diag}(\sigma,\overline\sigma)$ & \no\\*[4pt]
$-\varepsilon^4 $ & $\begin{pmatrix}q\varepsilon^2&-p\varepsilon^2\\-p\varepsilon^3&-q\varepsilon^3\end{pmatrix}$ & $\begin{pmatrix}p\varepsilon&q\varepsilon\\q\varepsilon^4&-p\varepsilon^4\end{pmatrix}$ & $\operatorname{diag}(-\varepsilon,-\varepsilon^4)$ & \yes\\[4pt]
\end{longtable}
\endgroup

\end{proof}

\begin{lemma}\label{dicbi120}
If $\mathbf{H}^{(m)}_{f}=PBI_{120}P^{-1}$, then $\mathrm{Holant}(=_2|f)$ is \#P-hard except for the following cases:
    \begin{itemize}
        \item $f \in \left \langle Z\mathcal{M} \right \rangle$ or $\in \left \langle ZX\mathcal{M} \right \rangle$;
        \item $f$ is $ \mathcal{A}\text{-}$transformable;
        \item $f$ is $ \mathcal{L}\text{-}$transformable;
        \item $f$ is $ \mathcal{P}\text{-}$transformable;
        \item $f\in\mathcal{H}$;
        \item $f \in \left \langle \mathcal{T} \right \rangle$,
    \end{itemize}
    in which cases the problem is computable in polynomial time.
\end{lemma}

\begin{proof}
     By Lemma~\ref{normalizer}, we have $\mathcal{N}_{\mathrm{SL}(2,\mathbb{C})}(BI_{120})=BI_{120}$.
In Holant($=_2|f,PBI_{120}P^{-1}$), by a holographic transformation using $P$, we have \[\text{Holant}(=_2|f,PBI_{120}P^{-1})\equiv_T\text{Holant}(S^{-1}|f',BI_{120}S),\]
where $f'=(P^{-1})^{\otimes 4}f$. Note that $BI_{120}S=BI_{120}$, then we have \[\text{Holant}(=_2|f,PBI_{120}P^{-1})\equiv_T\text{Holant}(S^{-1}|f',BI_{120}S)\equiv_T\text{Holant}(=_2|f',BI_{120}).\] By Lemma~\ref{BI120}, the result follows.
\end{proof}
\subsection{Completion of the proof}
\begin{theorem}
    Let $f$ be a 4-ary signature. Then $\mathrm{Holant}(=_2|f)$ is \#P-hard  except for the following cases:
\begin{itemize}
        \item $f \in \left \langle Z\mathcal{M} \right \rangle$ or $\in \left \langle ZX\mathcal{M} \right \rangle$;
        \item $f$ is $ \mathcal{A}\text{-}$transformable;
        \item $f$ is $ \mathcal{L}\text{-}$transformable;
        \item $f$ is $ \mathcal{P}\text{-}$transformable;
        \item $f\in\mathcal{H}$;
        \item $f \in \left \langle \mathcal{T} \right \rangle$,
    \end{itemize}
    in which cases the problem is computable in polynomial time.
\end{theorem}
\begin{proof}
For Holant$(=_2|f)$, we construct its binary subgroup sequence. If the subgroup sequence is unstable, the result follows from Lemma~\ref{infsettodegen}. If the subgroup sequence is stable at $m$, then $\mathbf{H}^{(m)}_f$ is conjugate to $C_n,\ BD_{4n},\ BT_{24},\ BO_{48},$ or $BI_{120}$, and the result follows, respectively, from Lemma~\ref{Cn}, Lemma~\ref{dicbd4n}, Lemma~\ref{BT24}, Lemma~\ref{dicbo48}, Lemma~\ref{dicbi120}.


\end{proof}

\section*{Acknowledgements}
ChatGPT (OpenAI) was used to assist with language editing, LaTeX formatting, and checking and revising parts of the proofs. The authors take responsibility for the content of this manuscript.

\bibliography{ref}

\end{document}